\documentclass[11pt]{article}
\usepackage[T1]{fontenc}
\usepackage{lmodern,amsmath,amssymb,amsthm,mathtools,microtype,booktabs,enumitem}
\usepackage[margin=1in]{geometry}
\usepackage[colorlinks=true,linkcolor=blue,citecolor=blue,urlcolor=blue]{hyperref}
\newtheorem{theorem}{Theorem}
\newtheorem{lemma}{Lemma}
\newtheorem{proposition}{Proposition}
\newtheorem{corollary}{Corollary}
\theoremstyle{definition}

\theoremstyle{remark}

\newcommand{\Pclass}{\mathsf P}
\newcommand{\NP}{\mathsf{NP}}
\newcommand{\PP}{\mathsf{PP}}
\newcommand{\BPP}{\mathsf{BPP}}
\newcommand{\MA}{\mathsf{MA}}
\newcommand{\PH}{\mathsf{PH}}
\newcommand{\CH}{\mathsf{CH}}
\newcommand{\FP}{\mathsf{FP}}
\newcommand{\GapP}{\mathsf{GapP}}
\newcommand{\SharpP}{\mathsf{\#P}}
\newcommand{\PosSLP}{\mathsf{PosSLP}}
\newcommand{\ETR}{\exists\mathbb R}
\newcommand{\Clevel}[1]{\mathsf C_{#1}\mathsf P}
\newcommand{\TC}{\mathsf{TC}^{0}}
\newcommand{\GCR}{\mathsf{GCR}}
\newcommand{\SRS}{\mathsf{SRS}}
\newcommand{\R}{\mathbb R}
\newcommand{\C}{\mathbb C}
\newcommand{\Z}{\mathbb Z}
\newcommand{\Q}{\mathbb Q}
\newcommand{\F}{\mathbb F}
\newcommand{\poly}{\operatorname{poly}}
\DeclareMathOperator{\sgn}{sgn}
\DeclareMathOperator{\Res}{Res}
\DeclareMathOperator{\Norm}{Norm}
\DeclareMathOperator{\Tr}{Tr}
\DeclareMathOperator{\dist}{dist}
\title{$\exists \mathbb{R} \subseteq \textsf{CH}$}
\author{Alex Meiburg}
\date{29 September 2026}
\begin{document}
\maketitle
\small{\textcolor{red}{The content in this paper is written essentially entirely by ChatGPT, and was completed after a series of conversations ending on September 29th 2026. A group of researchers has been working to digest the proof, and while the most essential arguments appear correct, we are endeavoring to give this result the treatment it deserves and a proper exposition and development to benefit of the community.\\
\indent However, on October 6th, OpenAI released a very similar proof that the $\exists \mathbb{R} \subseteq \textsf{CH}$. Their proof is dated as a work done on October 4th, thus, roughly a week after our co-developing of this manuscript with ChatGPT was completed. The headline result is the same: the existential theory of the reals, indeed with any constant level of alternating existential-universal quantifiers, is contained in the Counting Hierarchy. This manuscript places $\exists \mathbb{R}$ at a significantly lower level of the hierarchy, specifically $\BPP^{\Clevel3}\subseteq\Clevel4$, while OpenAI's result places it under the weaker bound $\exists \mathbb{R} \subseteq \Clevel {26}$. This work also includes new results about the difficulty of $\PosSLP$, the Square-Root-Sum problem, SDPs, and others, that are absent from their version.\\
\indent While we work to improve our exposition of this proof, the current version has been uploaded as a service to the community to compare the different proof techniques.\\
\indent Alex Meiburg is listing himself as an author on this version by the standard that every paper needs an author. While he takes responsibility that the proofs appear to be correct, he has not played a nontrivial role in developing them, and believes the human value will be in good exposition and canonicalization of the results.}}

\begin{abstract}
The existential theory of the reals asks whether polynomial constraints
with integer coefficients have a real solution. We give a proof placing
this problem in the counting hierarchy. The first argument is intended
to expose the essential steps: a coercive minimum reduces feasibility to
a finite critical algebra; residue formulas give its multiplication
traces; and a short collection of derivative signs selects a real root.
Exact formal power series and univariate resultants complete the
calculation. This proof permits singular critical fibers and does not
require a multivariate resultant algorithm.

A separate analysis lowers the bound to
\(\ETR\subseteq\BPP^{\Clevel3}\subseteq\Clevel4\).
It uses a common formal homotopy, ordered labels for real roots, and
randomized integer-sign reconstruction. For each fixed \(w\), sentences
with \(w\) alternating real quantifier blocks lie in
\(\Clevel{9w+17}\). We then treat exact semidefinite feasibility,
\(\PosSLP\), square-root sum, geometric real counting, Euler
characteristic, and complex feasibility in separate applications.
The corresponding bounds include \(\BPP^{\Clevel2}\) for general SDP,
\(\BPP^{\PP}\cap\Pclass^{\NP^{\PP}}\) for \(\PosSLP\) and
square-root sum, and \(\FP^{\Clevel4}\) for total geometric real
counting. Longer arithmetic proofs, conditional certificate routes,
and limitations of other approaches are collected in the supplementary
sections.
\end{abstract}

\tableofcontents

\section{Introduction}
The existential theory of the reals, or ETR, is the decision problem
for finite systems of polynomial equations and inequalities over
\(\R\). Its input is a finite binary string; its witnesses are real
tuples. The associated class \(\ETR\) consists of the languages
polynomial-time many-one reducible to this problem. It captures many
geometric existence problems, including the art gallery problem
\cite{AAM18}. The usual general bounds are
\[
 \NP\subseteq\ETR\subseteq\mathsf{PSPACE}
 \qquad\text{\cite{Canny88}}.
\]
An obstacle to a smaller upper bound is that a solution can require
very large coordinates and high algebraic degree. A bound on the size
of a coordinate is not a short description of that coordinate.

The counting hierarchy is defined by
\[
 \Clevel0=\Pclass,\qquad
 \Clevel{h+1}=\PP^{\Clevel h},\qquad
 \CH=\bigcup_{h\ge0}\Clevel h.
\]
It can process some exponentially large arithmetic objects without
writing them down. Allender, B\"urgisser, Kjeldgaard-Pedersen, and
Miltersen used this principle for integer straight-line programs
\cite{ABKPM09}. Andrews, Garg, and Schost recently placed Hilbert's
Nullstellensatz, the corresponding complex feasibility problem, in
\(\CH\) \cite{AGS26}. Over the reals there is another task: identify
the solutions that are real and satisfy the required order conditions.
The sign of a carefully chosen norm will carry this information.

The paper proceeds in four stages. Section~\ref{sec:headline} proves
\(\ETR\subseteq\CH\) without optimizing the level. It gives the
algebra and the real-root selection argument explicitly, starting from
standard integer arithmetic and the known bit-access theorem for
straight-line programs. Section~\ref{sec:small-depth} changes the
implementation and the selector to obtain the fourth-level bound.
Section~\ref{sec:real-hierarchy} treats fixed real quantifier blocks,
including an explicit linear level bound. Section~\ref{sec:applications}
gives the other complexity consequences, with the input model and
promise stated for each problem.

We distinguish a randomized polynomial-time machine with a \(\PP\)
oracle, denoted \(\BPP^{\PP}\), from an outer probability operator
such as \(\mathsf{BP}\cdot\PP\). We also do not identify
\(\PH^{\PP}\) with \(\Pclass^{\PP}\). Oracle order matters in the
arguments below. Function problems, such as exact point counting, are
stated in function classes rather than decision classes.

\section{The basic containment: from critical points to a real-root test}
\label{sec:headline}

\begin{theorem}\label{thm:etr}
The existential theory of the reals belongs to the counting hierarchy:
\(\ETR\subseteq\CH\).
\end{theorem}

There are two distinct issues in the proof. We first replace an arbitrary
real solution set by a finite algebra whose real points answer the
original question. We then compute with that algebra through traces and
univariate polynomials. No step assumes a short rational witness for
the original instance.

\subsection{What it means to compute an exponential object}
Write \(L\) for the original input length. Throughout this section,
degrees, dimensions, and integer bit lengths are bounded by
\(2^{\poly(L)}\). Their indices therefore have polynomially many bits.
An integer array is \emph{CH-accessible} if the sign and a specified bit
of a specified entry can be decided uniformly in one fixed level of
\(\CH\). The level is fixed for the construction; it does not grow
with an index or the number of variables.

We use two standard computational facts. First, a specified bit of the
integer output of a polynomial-size arithmetic straight-line program
is in \(\CH\) \cite{ABKPM09}. We call this the BitSLP theorem.
Second, division and iterated integer multiplication have suitably
uniform \(\TC\) circuits \cite{HAB02}; addition, comparison, and
parallel selection do as well. Here \(\TC\) means polynomial-size,
constant-depth Boolean circuits with unbounded majority gates.

\begin{lemma}\label{basic:padding}
A fixed-depth, uniformly generated threshold computation on
exponentially many exponentially long integers preserves
CH-accessibility. This applies to a fixed number of compositions of
the preceding integer operations and to dense polynomial addition,
multiplication, differentiation, and bounded composition, provided
their degrees and coefficient bit bounds remain singly exponential.
\end{lemma}
\begin{proof}
Every gate and wire has a polynomial-length address. At a majority
gate, a polynomial-time counting machine guesses an input-wire address
and queries the predicate for that wire. Thus one threshold layer
costs at most one further \(\PP\)-oracle level. There are only
constantly many layers.

For polynomial arithmetic, choose an integer base larger than twice
a bound on every coefficient magnitude and evaluate at that base.
Iterated integer multiplication computes a product of polynomials;
parallel balanced-digit extraction recovers its signed coefficients.
The packed integer has polynomial length in the expanded degree and
coefficient bounds. Derivatives use factorial-ratio weights, and
composition uses parallel powers and sums. Compatible uniformity is
essential: merely having an exponential-size circuit would not prove
the assertion.
\end{proof}

We will also use the following direct consequence of BitSLP. A specified
coefficient of a polynomial-size arithmetic circuit has CH-accessible
bits. Split the circuit into a difference of two circuits with
nonnegative coefficients. For each part, Kronecker substitution
\(x_i\mapsto B^{(D+1)^{i-1}}\), with sufficiently large power-of-two
\(B\), encodes coefficients as nonoverlapping blocks. The substituted
integers and exponents have short arithmetic-circuit descriptions.
BitSLP reads the requested block, and Lemma~\ref{basic:padding}
handles the signed subtraction. An SLP with \(s\) gates has degree
at most \(2^s\) and coefficient bit bound \(2^{\poly(s,L)}\), so
the required addresses remain short.

\subsection{Replace feasibility by a minimum in a finite algebra}
Introduce evaluation variables and the usual witnesses for sign
conditions to reduce ETR to integer quadratic equations
\(f_1(x)=\cdots=f_m(x)=0\). Standard effective real-algebraic
sampling bounds \cite{BPR06} supply
\(B=2^{2^{\poly(L)}}\) such that a feasible system has a solution
in \([-B,B]^n\). Only this bound is used; no sample is computed.
The chosen power of two has a polynomial-size repeated-squaring circuit.

Add one variable and a sphere equation of radius \(R=2(n+1)B\).
Relabel the augmented number of variables as \(n\), and put
\[
 F(x)=\sum_i f_i(x)^2+(\|x\|^2-R^2)^2.
\]
The original instance is feasible exactly when \(F\) has a real zero.
Moreover \(F\ge0\), \(\deg F\le4\), and \(F\) is coercive:
it tends to infinity with \(\|x\|\). Thus its minimum is attained.

\begin{lemma}\label{basic:gap}
One can choose a positive integer \(\Delta\), with polynomial-size
arithmetic circuit and singly exponential bit length, such that
\[
 F^{-1}(0)=\varnothing\quad\Longrightarrow\quad
 F(x)\ge 8/\Delta\quad(x\in\R^n).
\]
\end{lemma}
\begin{proof}
Here is a height argument that does not call an elimination algorithm.
Consider \(F_\eta=F+\eta\sum_i x_i^6\). Its critical equations,
after division by \(6\eta\), have leading monomials \(x_i^5\).
Their quotient over \(\Q(\eta)\) has the fixed basis
\(x^a\), \(0\le a_i<5\), of size \(N=5^n\).

Let \(U_\eta\) be multiplication by \(F_\eta\) in this basis.
A starting monomial has degree at most \(4n+6\); each reduction
lowers degree by at least two. Thus, for \(t=2n+3\), the matrix
\(C(\eta)=(6\eta)^tU_\eta\) is integral polynomial in \(\eta\).
If \(2^\tau\) bounds \(1+n+\|F\|_1+\sum_i\|\partial_iF\|_1\),
its entry coefficient norms are at most
\(H_0=(6\cdot2^\tau)^{t+1}\). Consequently
\[
 D(T,\eta)=\det((6\eta)^tTI-C(\eta))
\]
has coefficient norm at most \((2NH_0)^N\). Its coefficient bit
bound \(u\) is singly exponential and has a polynomial-length
binary description.

Let \(m_\eta=\min F_\eta\) and \(m_0=\min F\). Evaluation at a
minimizer of \(F\) shows \(m_\eta\to m_0\). A minimizer of
\(F_\eta\) is critical, so \(D(m_\eta,\eta)=0\). Remove the
largest common power of \(\eta\) from \(D\) and specialize at
zero. This gives a nonzero integer polynomial \(D_0(T)\), of the
same height bound, with \(D_0(m_0)=0\). If \(m_0>0\), remove a
possible factor \(T^a\) and apply the reciprocal Cauchy root bound:
\(m_0\ge(1+2^u)^{-1}\ge2^{-u-1}\). Taking
\(\Delta=2^{u+4}\) proves the claim. The determinant established a
height bound; the algorithm never computes this matrix determinant.
\end{proof}

Set
\[
 M=16\Delta n(R^6+1),\qquad
 G=F+\frac1{6M}\sum_i x_i^6,\qquad h=1-\Delta F.
\]
If \(F\) has a zero, evaluating \(G\) there gives
\(\min G<1/(96\Delta)\). At a minimizer of \(G\), its nonnegative
sextic term gives \(F\le G<1/\Delta\). The NO-case gap excludes
this inequality everywhere. Hence feasibility is equivalent to
\begin{equation}\label{basic:finite-question}
 \exists x\in\R^n:\quad
 x_i^5+M\partial_iF(x)=0\ (1\le i\le n),\qquad h(x)>0.
\end{equation}
These critical equations are integral and have the form
\(x_i^5-g_i(x)\), with \(\deg g_i\le3\). Their relatively prime
leading monomials give a monic Gr\"obner basis, so the quotient
\[
 A=\Q[x_1,\ldots,x_n]/(x_i^5-g_i:1\le i\le n)
\]
has dimension \(N=5^n\) and basis \(x^a\), \(0\le a_i<5\).
It may have repeated points. We retain those multiplicities rather
than assume that a perturbation has made them disappear.

\subsection{Separate the points without finding them}
The same basis is free over \(\Z\), so every coordinate of a complex
point is an algebraic integer of degree at most \(N\). Choose
\(U\ge1+\max_i\|g_i\|_1\). Every complex coordinate has modulus
at most \(U\): if \(V=\max_i|x_i|>1\), the equation for a maximizing
coordinate gives \(V^5\le\|g_i\|_1V^3\).

For two points and a nonzero coordinate difference \(\beta\), the
number \(\beta\) is an algebraic integer of degree at most \(N^2\),
and all its conjugates have modulus at most \(2U\). Its nonzero
integer norm therefore implies
\[
 |\beta|\ge(2U)^{-(N^2-1)}.
\]
Choose a power of two \(T>n(2U)^{N^2}\), and set
\[
 \ell(x)=x_1+Tx_2+\cdots+T^{n-1}x_n.
\]
For two different points, the highest differing coordinate dominates
all the earlier terms. Thus \(\ell\) separates every pair of distinct
complex points. Although \(T\) is very large, \(\log T\) is singly
exponential and its binary address is polynomial, so \(T\) has a
short powering circuit. It is a coefficient, not an iteration count.

Let \(m_a\) denote multiplication by \(a\) on \(A\), and define
\begin{equation}\label{basic:qr}
 q(X)=\det(XI-m_\ell),\qquad
 r(X)=q(X)\Tr\bigl(m_h(XI-m_\ell)^{-1}\bigr).
\end{equation}
These are integer polynomials, of degrees \(N\) and at most \(N-1\).
For example, the second is \(\Tr(m_h\operatorname{adj}(XI-m_\ell))\).
This establishes integrality, but we will compute their coefficients
from traces rather than from the displayed matrices.

\subsection{Residues give traces of multiplication}
The following formula is the computational advantage of the pure
leading powers. Write \(J=\det(\partial f_i/\partial x_j)\) for the
Jacobian of \(f_i=x_i^d-g_i\), where \(\deg g_i<d\).

\begin{lemma}\label{basic:residue}
For \(\deg a\le e\), in the quotient by these equations,
\[
 \Tr(m_a)=
 \left[\prod_i x_i^{d(e+1)-1}\right]
 aJ\prod_i S_e(g_i,x_i^d),\qquad
 S_e(u,v)=\sum_{k=0}^{e}u^kv^{e-k}.
\]
If \(a,g_i\) have short circuits and \(e\) has polynomially many
binary digits, the expression whose coefficient is taken has a
polynomial-size circuit.
\end{lemma}
\begin{proof}
The global residue trace formula extracts the coefficient of
\(x_1^{-1}\cdots x_n^{-1}\) in
\[
 \frac{aJ}{\prod_i f_i}
 =aJ\prod_i x_i^{-d}
       \sum_{k_i\ge0}(g_ix_i^{-d})^{k_i}.
\]
Comparison of total degrees shows that a contributing term has
\(\sum_i k_i\le e\). Clearing the Laurent powers yields the
displayed finite coefficient. For simple fibers the local Jacobian
cancels the residue denominator, giving the sum of \(a\) over points.
Both sides are polynomial in the coefficients in this fixed monomial
basis, so specialization proves the identity for nonreduced fibers too.

The long geometric sum has a short division-free circuit. If
\(Q_m(u,v)=\sum_{k=0}^{m-1}u^kv^{m-1-k}\), then
\[
 Q_{2m}=(u^m+v^m)Q_m,\qquad
 Q_{2m+1}=vQ_{2m}+u^{2m},\qquad Q_1=1.
\]
Track \(u^m,v^m,Q_m\) along the binary digits of \(m\). This uses
\(O(\log m)\) gates, and \(S_e=Q_{e+1}\). The Jacobian determinant
has polynomial dimension and a division-free polynomial-size circuit.
\end{proof}

Apply the lemma to \(a=\ell^j\) and \(a=h\ell^j\), for \(j\le N\).
Binary powering gives short circuits for these polynomials. The
coefficient-bit observation above therefore puts all the integers
\[
 p_j=\Tr(m_\ell^j)\ (1\le j\le N),\qquad
 t_j=\Tr(m_{h\ell^j})\ (0\le j<N)
\]
in one CH-accessible family. No multiplication matrix is constructed.

\subsection{An exact, parallel replacement for Newton identities}
Sequentially recovering characteristic coefficients from traces would
take \(N\) dependent stages. A formal exponential avoids that depth.

\begin{lemma}\label{basic:trace-reconstruct}
From the first \(N\) power traces of an integral \(N\)-dimensional
matrix, its characteristic coefficients are computable in uniform
\(\TC\), measured in the expanded trace bit bounds and \(N\).
Consequently CH-accessible power traces give CH-accessible
characteristic coefficients.
\end{lemma}
\begin{proof}
Put
\[
 A_0=N!,\qquad Q=A_0^{N+1},\qquad
 P(z)=-\sum_{j=1}^N\frac{A_0}{j}p_jz^j,
 \qquad
 S(z)=\sum_{k=0}^N\frac{Q}{A_0^kk!}P(z)^k.
\]
Every weight is an integer. The formal identity
\[
 \det(I-zM)=\exp\left(-\sum_{j\ge1}\Tr(M^j)z^j/j\right)
\]
gives \(S(z)\equiv Q\det(I-zM)\pmod {z^{N+1}}\).
The polynomial \(S\) has degree at most \(N^2\). If a trace has at
most \(H\) bits, its coefficients have
\(O(NH+N^2\log(N+1))\) bits. Parallel powers, factorials, sums,
packing, and exact division by \(Q\) therefore give the claimed
fixed-depth computation. The packing bound is chosen for all of
\(S\), including its large factor \(Q\) and degrees above \(N\).
No analytic exponential is evaluated.
\end{proof}

Write \(q(X)=\sum_{i=0}^Nc_iX^{N-i}\), with \(c_0=1\).
Expanding the resolvent at infinity gives
\begin{equation}\label{basic:r-convolution}
 r(X)=\sum_{k=0}^{N-1}
       \left(\sum_{i=0}^k c_it_{k-i}\right)X^{N-1-k}.
\end{equation}
This is another fixed-depth coefficient operation, so both polynomials
in \eqref{basic:qr} are CH-accessible.

\begin{corollary}\label{basic:univariate-resultant}
The exact resultant of two dense integer univariate polynomials is
computable by uniform \(\TC\) circuits in their dense bit length.
For singly exponential CH-accessible arrays, its bits are in a fixed
level of \(\CH\).
\end{corollary}
\begin{proof}
For monic \(f=x^d-g\) and a second polynomial \(a\), the univariate
case of Lemma~\ref{basic:residue} computes
\(\Tr(m_a^j)\), \(1\le j\le d\), using products of polynomially
many dense polynomials of polynomial degree and bit height. Packed
integer arithmetic implements them in \(\TC\).
Lemma~\ref{basic:trace-reconstruct} gives \(\det(m_a)=\Res(f,a)\).

For \(f\) of degree \(d\ge1\) and leading coefficient \(b\ne0\),
and \(\deg a=e\), use the monic integral polynomial
\(F(y)=b^{d-1}f(y/b)\) and \(A(y)=b^ea(y/b)\). Then
\[
 \Res(F,A)=b^{e(d-1)}\Res(f,a).
\]
Scaling and signed exact division have polynomial bit bounds and are
in \(\TC\). Constants and leading zeros are handled by parallel
degree selection. Common or repeated roots require no exception.
\end{proof}

\subsection{Select a real root using only a short list of signs}
A norm multiplies values over all complex points. Nonreal points pair
with their conjugates and contribute positively. Negative factors at
real points are therefore detectable, but only modulo two. We arrange
that some short choice selects exactly one qualifying real root.

\begin{lemma}\label{basic:real-select}
Let \(b,e,v\) be CH-accessible integer univariate polynomials of
singly exponential degree and bit height, with \(b\ne0\). Suppose
every real root satisfying \(e=0,v>0\) is simple in \(b\).
Existence of such a root is decidable in a fixed level of \(\CH\).
\end{lemma}
\begin{proof}
Normalize the leading coefficient of \(b\) to be positive, and write
\(D=\deg b\). A constant polynomial has no roots, so assume
\(D\ge1\). Thom's lemma says that the vectors
\[
 (\sgn b'(\alpha),\ldots,\sgn b^{(D)}(\alpha))
\]
distinguish its distinct real roots. In a nonempty collection of
qualifying roots, choose a coordinate that varies and retain a smallest
nonempty sign class. This removes at least half the roots. Repeating
at most \(\lceil\log_2D\rceil\) times isolates one qualifying root.
Thus a list of \(O(\log D)\) derivative indices and trits suffices;
its binary length is polynomial in \(L\).

To turn a guessed list into a polynomial, bound all nonzero test values
at roots of \(b\) between \(2^{-K}\) and \(2^K\), with singly
exponential \(K\). Cauchy's bound gives the upper estimate. For the
lower estimate, the resultant of a primitive irreducible factor of
\(b\) with a test polynomial is a nonzero integer; all its other
factors have the upper bound. This argument uses a factor only to prove
the bound, and does not compute a factorization.

There is an explicit rational approximation \(R=A/B\) to the ternary
sign on these values, with \(B>0\) on \(\R\) and \(R(0)=0\).
For clarity, take an even integer \(m\) and define
\[
 P(y)=\prod_{k=-K}^K
       (2^{\max(0,-k)}y+2^{\max(0,k)})^m,
 \qquad R(y)=\frac{P(y)-P(-y)}{P(y)+P(-y)}.
\]
For \(y>0\), a dyadic scale within a factor two of \(y\) gives
\(|P(-y)/P(y)|\le3^{-m}\); all other factors have ratio at most one.
The denominator is positive and \(|R|\le1\) on \(\R\). Taking
\(m=O(\log(1/\varepsilon))\) gives error below \(\varepsilon\).
Its degree and height are polynomial in \(K\) and
\(\log(1/\varepsilon)\), with constant-depth dense arithmetic.

For the guessed derivative tests \((j_i,\sigma_i)\), use
\[
 W(X)=R(e(X))^2+(R(v(X))-1)^2
       +\sum_i(R(b^{(j_i)}(X))-\sigma_i)^2-\tfrac14.
\]
Choose \(\varepsilon<1/(100(s+2))\), where \(s\) is the number
of derivative tests. At a matching root the sum is less than \(1/4\).
At any other real root, some term is at least
\((1-\varepsilon)^2\), or the failed positivity term is at least
one. Clear the positive squared denominators to obtain an integer
polynomial \(w\), negative exactly at the matches and positive at
every other real root of \(b\).

The selector may vanish at nonreal roots. Compute
\[
 F(u)=\Res_X(b(X),w(X)+u)
\]
and take its lowest nonzero coefficient. A nonreal conjugate pair
contributes a positive factor, or an even power of \(u\) if the
selector vanishes there. A nonselected real root contributes positively
with any multiplicity. A selected root is simple and contributes one
negative factor. The coefficient sign is consequently
\((-1)^{\#\text{matches}}\).

Corollary~\ref{basic:univariate-resultant} computes this polynomial's
coefficients in CH; evaluation and interpolation, or one larger packed
evaluation, handles its parameter \(u\). Parallel zero tests locate
its first nonzero coefficient. All degrees and heights are still
singly exponential and all constructions have fixed arithmetic depth.
Accept a guessed list when that coefficient is negative. Soundness
requires only a nonzero odd number of matches. Completeness follows
from the existence of a list isolating one root. The final
polynomial-length guess costs one further level, since
\(\NP^{\Clevel k}\subseteq\Clevel{k+1}\).
\end{proof}

\subsection{Repeated roots and the completion of the proof}
Over \(\C\), decompose \(A\) into its local algebras. At a point
\(z\) of local dimension \(\mu_z\), multiplication by \(\ell\)
is \(\ell(z)I\) plus a nilpotent operator. Every higher pole in its
resolvent has trace zero, even after multiplication by \(h\).
Therefore, with \(\alpha_z=\ell(z)\),
\[
 q(X)=\prod_z(X-\alpha_z)^{\mu_z},\qquad
 r(X)=q(X)\sum_z\frac{\mu_zh(z)}{X-\alpha_z}.
\]
At \(\alpha=\ell(z)\) of multiplicity \(j=\mu_z\), differentiation
gives the useful exact identity
\begin{equation}\label{basic:multiplicity}
 r^{(j-1)}(\alpha)=h(z)q^{(j)}(\alpha).
\end{equation}
In particular \(r^{(j-1)}q^{(j)}\) has the sign of \(h(z)\) there.

For each \(j=1,\ldots,N\), form
\[
 b_j=q^{(j-1)},\qquad
 e_j=\sum_{i=0}^{j-1}(q^{(i)})^2,\qquad
 v_j=r^{(j-1)}q^{(j)}.
\]
At a real root of \(b_j\), the conditions \(e_j=0,v_j>0\)
select exactly the roots of \(q\) of multiplicity \(j\) for which
\(h(z)>0\). Such a root is simple in \(b_j\), because
\(b_j'=q^{(j)}\ne0\). A real projected root corresponds to a real
point: conjugation preserves the fiber, and separation by \(\ell\)
forces \(z=\overline z\) when \(\ell(z)\) is real.

Apply Lemma~\ref{basic:real-select} to these strata. The index \(j\)
and the isolating list can be guessed together with polynomially many
bits. All derivatives and sums above are formed in parallel using
their coefficient formulas, so the number of strata does not become
the depth of the computation. By \eqref{basic:finite-question} and
\eqref{basic:multiplicity}, acceptance is equivalent to the original
real feasibility question. This proves Theorem~\ref{thm:etr}.

The proof has used only a fixed number of CH arithmetic stages above
BitSLP. It has not assigned them a small numerical depth. In particular,
the norm or the packed characteristic polynomial is not asserted to
have a polynomial-size integer SLP. The next section obtains a smaller
bound by a different implementation of the finite-algebra calculation.

\section{Lowering the bound to the fourth level}
\label{sec:small-depth}

\begin{theorem}\label{thm:depth}
\[
 \ETR\subseteq\BPP^{\Clevel3}\subseteq\Clevel4.
\]
\end{theorem}

The preceding proof deliberately hides the depth of integer arithmetic.
Unhiding it is not enough to obtain this bound. We make three changes.
First, a special critical-point construction gives simple points and
distinct critical values. Second, these values receive short, distinct
integer labels, so a random affine hash can isolate a qualifying value.
Third, modular norm evaluation and integer-sign reconstruction are
combined before an unnecessary oracle layer is introduced. We explain
the interfaces here; the complete arithmetic proofs are in
Sections~\ref{sec:etr}, \ref{push4:sec}, and \ref{push6:depth} of the
supplementary material.

\subsection{A critical algebra with distinct values}
Use the bounded version of ETR in which variables lie in \([1/2,2]\)
and constraints have the forms \(x=1\), \(x+y=z\), and \(xy=1\)
\cite{AAM18}. Enforce the interval constraints with
\[
 2x-1-2u^2=0,\qquad 2-x-v^2=0.
\]
The sum \(F\) of the squares of all these equations is a coercive
quartic. Its zeros have all coordinates bounded by two. Write \(m\)
for the number of coordinates. An effective minimum bound supplies a
short binary integer \(E_0\) such that
\[
 \min F=0\quad\hbox{or}\quad\min F\ge 2^{-E_0}.
\]
The parameters can be chosen with \(E_0\ge2m\). Set
\[
 C_0=91(2^m-1),\quad A_0=C_0+1,\quad
 Q_0=A_0 2^{E_0},\quad T=3A_0,\quad w_i=2^{i-1},
\]
and define
\begin{equation}\label{depth:objective}
 J(X)=6Q_0F(X)+\sum_{i=1}^m w_i
       \bigl(X_i^6+3(X_i+1)^2\bigr).
\end{equation}
At a zero of \(F\), the second sum is at most \(C_0<T\).
If \(F\) has no zero, then \(J\ge6A_0>T\) everywhere. Therefore
\[
 \text{the input is feasible}
 \quad\Longleftrightarrow\quad
 \text{some real critical point of }J\text{ has }J<T.
\]

After dividing the \(i\)-th gradient equation by \(6w_i\), the
critical system is
\[
 G_i=X_i^5+X_i+1+(Q_0/w_i)\,\partial_i F=0.
\]
Its integral quotient has basis \(X^\alpha\), \(0\le\alpha_i<5\),
and rank \(D=5^m\). Modulo two each equation becomes
\[
 p(U)=U^5+U+1=(U^2+U+1)(U^3+U^2+1).
\]
This polynomial has five distinct roots over \(\F_{64}\). Hensel
lifting gives \(D\) simple characteristic-zero points. Moreover,
\(U^6+3(U+1)^2\equiv U+1\pmod{2,p(U)}\).
At the first coordinate where two lifted points differ, the weights
\(1,2,4,\ldots\) give a unique lowest two-adic valuation in the
difference of their objective values. Thus the values of \(J\) are
distinct. Finally, \(J-T\equiv X_1+1\) modulo two at the critical
points, and \(p(1)\ne0\). No critical value equals the threshold.
Lemma~\ref{lem:generic} gives the full argument and height bounds.

This small-characteristic construction replaces the multiplicity strata
of the basic proof. The polynomial
\(q(Z)=\Norm(Z-J)\) is now squarefree and its real roots correspond
exactly to the real critical points.

\subsection{Integer-sign reconstruction as an interface}
An exponentially large integer cannot be signed by selecting a single
short modulus. We use its residues simultaneously. The following is
the form of the arithmetic interface needed below.

\begin{lemma}\label{depth:sign-interface}
Let \(Z\in\Z\) have a known singly exponential bit bound. Suppose
that for every positive, polynomial-bit modulus \(a\), an integer
representative of \(Z\bmod a\) is given by a raw
\(\GapP^{\mathcal A}\) expression, with polynomial-length branch
addresses. Then the randomized sign algorithm uses
\(\BPP^{\PP^{\mathcal A}}\). The same assertion holds for the
particular representatives below after the indicated functional
counting absorption.
\end{lemma}

Here ``raw'' means that the outer signed sum remains a counting
expression. Calling an oracle to retrieve and then sum its canonical
residues can cost an extra level. The hypothesis also concerns
\emph{all} short moduli, including prime powers. Residues only modulo
large primes do not suffice for this interface.

For orientation, the underlying identity is
\begin{equation}\label{depth:beta}
 \sum_{i=0}^{N}\frac{(-1)^i\binom Ni}{N+1+i}
 =\int_0^1 t^N(1-t)^N\,dt
 =\frac{(N!)^2}{(2N+1)!}>0.
\end{equation}
Replace the binomial coefficients and \(Z\) by integer representatives
modulo \(N+1+i\). The resulting rational sum is an integer plus
\(\delta=C_NZ\), where \(C_N\) is the positive quantity in
\eqref{depth:beta}. Choosing \(N\) above the bit bound makes
\(|\delta|\) tiny, with a known lower bound when it is nonzero.
Dyadic multiplication of this phase eventually makes its sign visible
before wrapping around the unit circle. A short-seed, random-access
generator and a robust search over scales locate such a phase in
randomized polynomial time. All weighted sums stay inside the
\(\PP^{\mathcal A}\) oracle. Section~\ref{push4:sec} proves the
generator, phase test, and search; in particular it does not assume
that independently generated exponential-length randomness is available.

There is also a deterministic first-event search. For the integer SLP
representatives, its oracle is \(\NP^{\PP}\), giving
\(\Pclass^{\NP^{\PP}}\). We return to this in the applications.

\subsection{Ordered labels instead of full derivative sign vectors}
The basic proof isolates a root using a short part of its Thom vector.
For the smaller bound we first assign an integer label to each real
critical value. Let \(\alpha_1,\ldots,\alpha_D\) be all the roots of
\(q\), including the nonreal roots.

The rational normalizer of Section~\ref{push6:depth} is a short
real-coefficient circuit \(R=n/d\). Its denominator is positive on
\(\R\), it vanishes at zero, and it preserves the sign of every
nonzero real argument. These real-line properties hold for every
seed. With high probability, simultaneously for all nonzero
differences \(\alpha_i-\alpha_j\), its value is very close to one
of \(1,-1\). Its additional ordering property is essential:
along each horizontal line, the limiting half-plane signs on the
good inputs are nondecreasing from left to right
(Lemma~\ref{push6:ordering}). No global preservation of the original
complex half-plane is asserted.

Denote these limiting signs by \(\sigma_{ij}\), with
\(\sigma_{ii}=0\). For a real root put
\[
 S_i=\sum_{j=1}^D\sigma_{ij},\qquad
 r_i=\frac{S_i+D-1}{2}.
\]
Conjugate roots contribute paired signs; hence \(r_i\) is an integer
between zero and \(D-1\). As the real root moves to the right,
every column is nondecreasing, and the two columns belonging to a
pair of real roots force a strict increase. The labels \(r_i\) are
therefore distinct on the real roots.

We need not list the \(D\) summands. A first-order norm computes them
all through the identity
\begin{equation}\label{depth:rank-jet}
 D_0(u)+\epsilon N_0(u)
 =\Norm\bigl(d(u-J)+\epsilon n(u-J)\bigr)
 \pmod{\epsilon^2}.
\end{equation}
Where \(D_0(u)\ne0\), its logarithmic coefficient is
\[
 \frac{N_0(u)}{D_0(u)}=\sum_j R(u-\alpha_j).
\]
Consequently
\(X(u)=(N_0(u)+(D-1)D_0(u))/(2D_0(u))\) lies within a small
constant of \(r_i\) at each real root. Short polynomial circuits
read the bits of this nearby integer: the sign of
\(\sin(\pi(X+1/2)/2^b)\) gives the corresponding bit, and range
reduction, Taylor approximation, and repeated squaring supply the
required polynomial approximation. The distance from the sign
boundary is at least an inverse exponential. The proof in
Section~\ref{push6:depth} includes the approximation error and positive
denominator clearing.

Now hash the \(\ell=\lceil\log_2D\rceil\)-bit labels by a random
affine map \(h:\F_2^\ell\to\F_2^k\). A suitable output length
isolates exactly one qualifying label with constant probability.
Guessing that length randomly costs only an inverse polynomial
factor. The polynomial selector in \eqref{push6:hash-selector}
combines the hash equations with the condition \(J<T\).
Its sign is negative at a matching qualifying real root and positive
at every other real root. Even an unsuccessful normalization seed
cannot make a nonqualifying real root negative. After a sufficiently
large positive scaling and a shift by one, the selector is nonzero
at every complex root. The norm is then negative exactly when an
odd number of qualifying real roots were selected: conjugate pairs
contribute positive factors.

\subsection{Why the arithmetic costs three oracle levels}
Both the inner norm \eqref{depth:rank-jet} and the final selector norm
are evaluated with the same formal homotopy. Over a short finite ring,
start from a split reference system with unit Jacobian and continue
each branch formally. At a fixed evaluation node a polynomial number
of Newton updates suffices, because each update doubles the available
order. Degree bounds permit Fourier coefficient extraction without
aliasing. Products and sums have exponentially many factors but
polynomial-length addresses. There is no sequential walk through all
critical points.

For arbitrary prime powers the calculation takes place in a short
unramified Galois ring. Nonunits in a product are handled by their
valuations; unit products and first-order product jets have counting
representatives (Lemmas~\ref{lem:ring-coefficients} and
\ref{lem:ring-products}, and the jet lemma in
Section~\ref{push6:depth}). This supplies the all-moduli hypothesis
of Lemma~\ref{depth:sign-interface}.

The optimized accounting is as follows. The output of a function
computation at a node has polynomially many bits; the large final
integer is always represented by residues.
\begin{center}
\begin{tabular}{@{}ll@{}}
\toprule
Operation & Complexity at a short modulus\\
\midrule
One formal branch and a local circuit value & \(\Pclass\)\\
Inner norm jet, hence rank and hash factors & \(\FP^{\PP}\)\\
Outer product over branches & \(\FP^{\Clevel2}\)\\
Final Fourier sum, kept as a signed sum & \(\GapP^{\Clevel2}\)\\
Randomized sign reconstruction & \(\BPP^{\Clevel3}\)\\
\bottomrule
\end{tabular}
\end{center}
The last Fourier extraction is merged with the weighted count in
sign reconstruction. It is not first turned into a new canonical
residue oracle. Polynomially many isolation trials and error
amplification preserve \(\BPP^{\Clevel3}\). Since
\(\BPP^{\mathcal A}\subseteq\PP^{\mathcal A}\) relativizes,
the second containment in Theorem~\ref{thm:depth} follows.

The intermediate bounds \(\Clevel{1023}\), \(\Clevel7\),
\(\Clevel6\), and \(\Clevel5\) record different implementations
and isolation methods. They are explained in supplementary
Section~\ref{sec:depth}. They are not successive consequences of
a general collapse of the counting hierarchy.

\section{A linear bound for the real polynomial hierarchy}
\label{sec:real-hierarchy}

Write \(\Sigma_w\R\) for the discrete class defined by polynomial-size
integer polynomial formulas with \(w\) alternating blocks of real
quantifiers, beginning existentially, and write \(\Pi_w\R\) for the
universal counterpart. The number of variables in a block is not
bounded. We use the constant-free, discrete real polynomial hierarchy
terminology of \cite{HMMP25}; binary rational constants cause no change.
In particular \(\Sigma_1\R=\ETR\).

\begin{theorem}\label{hier:linear}
For every fixed positive integer \(w\),
\[
 \Sigma_w\R\ \cup\ \Pi_w\R\ \subseteq\ \Clevel{9w+17}.
\]
Consequently the union of the fixed levels of the discrete real
polynomial hierarchy is contained in \(\CH\).
\end{theorem}

The coefficient of \(w\) is an accounting bound, not an optimality
claim. For \(w=1\), Theorem~\ref{thm:depth} gives the much better
level four. The quantifier \emph{fixed} is part of the theorem:
the degree of the running-time polynomial may depend on \(w\).
This does not place sentences of input-dependent quantifier depth
in one constant level of \(\CH\).

\subsection{Why a Boolean ETR oracle is insufficient}
After choosing a sample for an outer block, the coefficients seen
by an inner block can be algebraic numbers. They need not have short
binary representations. Moreover, a later sign test must refer to
the \emph{same} earlier sample, rather than a separately guessed
root with convenient signs. We therefore need a uniform representation
of sample tuples, together with their validity and sign predicates.

The block-elimination construction of Basu, Pollack, and Roy
\cite[Algorithms 14.1--14.2]{BPR06} supplies the geometric starting
point. We retain its coverage of connected sign components and
replace its sequential arithmetic by parallel resultant calculations.
Section~\ref{push10:real-blocks} gives the construction in detail;
we summarize the invariant and then give the numerical accounting.

Each block supplies exponentially indexed univariate representations
of candidate samples. A representation consists of a defining
polynomial \(q\), a real root \(\alpha\), and rational coordinates
\(g_i(\alpha)/g_0(\alpha)\). Its validity includes
\(g_0(\alpha)\ne0\), tested at that root. Bounds on degrees,
coefficient heights, and family sizes remain singly exponential
for fixed \(w\). The signs of bounded polynomial expressions in
earlier samples are available uniformly; merely knowing the signs
of the individual coefficients would not suffice.

There are three points at which coherence could otherwise fail.
First, we keep degree and gcd metadata for every specialization.
These data prevent roots from splitting, colliding, or acquiring
new zeros of a guard inside one connected parameter cell. Thus
the real root branches and the required signs continue on that
cell. Equality of sign vectors on disconnected parameter sets
does not justify such a conclusion.

Second, roots have short, coherent addresses. Hash the full Thom
vector, implement bucket parity by a resultant, and count the odd
buckets. If their sizes are \(n_t\) and there are \(r\) real roots,
\[
 r-\#\{t:n_t\text{ odd}\}
   =\sum_t\bigl(n_t-(n_t\bmod2)\bigr)\ge0.
\]
The hash family contains an injective member. Hence every hash
maximizing the number of odd buckets is injective. A maximizing
hash and an occupied bucket name one root, and subsequent tests
reuse that address. This gives a sign oracle for an entire
indexed family without guessing its exponential sign profile.

Third, all formal infinitesimals are specialized \emph{once}, for
the whole guarded sample tree. A complete candidate tuple uses
only a number of algebraic root variables depending on \(w\).
Fixed-dimensional elimination bounds all polynomials governing
the truth of the finite representation formula. Choose
\[
 \epsilon_i=\eta^{(M+1)^{i-1}},\qquad \eta=2^{-B},
\]
with \(M,B\) above these effective degree and height bounds.
Leading infinitesimal monomials then dominate simultaneously.
The finite tuple and the formal tuple lie in the same relevant
cell and give the same tree truth. Independent choices for
different branches would not prove this. The constants have short
arithmetic circuits despite their very small values.

Finally, recode a complete candidate tuple into a single finite
algebra. Parent--child consistency becomes equality of coordinate
prefixes of complete tuples. A comparison of two tuples uses two
root addresses, rather than a new nested tower at every comparison.
Existential blocks quantify over valid short addresses; universal
blocks use validity as the antecedent of an implication. Duplicate
representations are harmless. This reduces the original real
sentence to fixed Boolean alternation over coherent CH predicates.

\subsection{The arithmetic primitives and the level ledger}
The general resultant primitive has the following useful refinement.
If the coefficients of its input forms are available at level \(h\),
its output coefficient residues are available in
\(\FP^{\Clevel{h+3}}\). If instead their exact homogeneous values
and the needed first derivatives at short nodes are already
available in \(\FP^{\Clevel r}\), the output uses only
\(\FP^{\Clevel{\max(r,1)+2}}\). The first remaining layer computes
the flattened Poisson product, and the second extracts coefficients.
The evaluator includes points at infinity: it may not divide by
a zero homogenizing coordinate.

For restricted elimination with degree drops, use the padded
resultant
\begin{equation}\label{hier:padded}
 \Res_{D,E}\bigl(f+t,\ g+\lambda t+VX^E\bigr),
 \qquad D\ge\deg f,\quad E>\deg g\text{ even}.
\end{equation}
If the specialized degree of \(f\) is \(d>0\), its first nonzero
\(V\)-slice occurs at \(D-d\). Up to a nonzero leading-coefficient
factor independent of \(t\), that slice is the ordinary deformed
resultant \(\Res(f+t,g+\lambda t)\). Its \(t\)-order is
\(\deg\gcd(f,g)\): at a common root the local order is the
minimum of the two multiplicities, and the formal \(\lambda\)
prevents cancellation. A nonzero constant \(f=a\) gives
\((a+t)^E V^D\); the identically zero case is detected separately.
Thus one coefficient family retains the degree and gcd data,
without a sign-dependent truncation inside the arithmetic algorithm.

A further saving comes from evaluating derivative-coordinate
polynomials in the same Fourier sum that constructs them.
If \(q(T)=\sum_{i=0}^Dq_iT^i\), the homogeneous value of its
\(\mu\)-th derivative at \((B,A)\) is
\[
 \sum_{i=\mu}^D(i)_\mu q_i A^{i-\mu}B^{D-i},
 \qquad (i)_\mu=i(i-1)\cdots(i-\mu+1).
\]
In a Fourier representation of the \(q_i\), this is obtained by
inserting the kernel
\begin{equation}\label{hier:kernel}
 K_{D,\mu}(A,B,z)=
   \sum_{i=\mu}^D(i)_\mu A^{i-\mu}B^{D-i}z^{-i}.
\end{equation}
The interval factorial is an \(\FP^{\PP}\) product, and the
kernel is in \(\FP^{\Clevel2}\). It fits inside the existing
sampling-product summand. There is no new sum around already
retrieved resultant coefficients. Only the Fourier node \(z\)
is inverted; \(A=0\), \(B=0\), and small residue characteristic
are permitted. This saving applies to these represented derivative
families, not to arbitrary coefficient arrays.

Here is the complete conservative ledger. Suppose the incoming
coefficient arrays for one block are at level \(h\).
\begin{center}
\begin{tabular}{@{}lr@{}}
\toprule
Operation & Level\\
\midrule
Values of critical equations and required derivatives & \(h+1\)\\
Sampling resultant and coordinate jets & \(h+3\)\\
Derivative-coordinate values using \eqref{hier:kernel} & \(h+3\)\\
Padded restricted resultant \eqref{hier:padded} & \(h+5\)\\
Two further homogeneous coefficient compositions & \(h+9\)\\
\bottomrule
\end{tabular}
\end{center}
One common finite infinitesimal specialization costs one additional
level overall. Thus the finite tuple and leaf arrays have level
\(k\le9w+1\). Flat recoding costs five more levels. A guarded
root-resultant test costs four; comparing two represented tuples
requires two such tests. Integer-sign reconstruction costs one.
Finally, fixed Boolean alternation and relativized Toda absorption
cost at most two more. Altogether,
\begin{equation}\label{hier:ledger}
 (9w+1)+5+4+4+1+2=9w+17.
\end{equation}
All sums and products retain polynomial-length indices. The structural
proof and its affine-depth invariant are in
Sections~\ref{push10:real-blocks} and \ref{push11:sec}; supplementary
Section~\ref{app:new-interfaces} records the evaluator and
specialization details behind the numerical refinement.

The last absorption uses relativized Toda's theorem in its proper
oracle order. It does not use \(\PH^{\PP}=\Pclass^{\PP}\), and
it does not treat an external counting predicate as an ordinary
polynomial guard.

\section{Other problems and complexity consequences}
\label{sec:applications}

The same arithmetic interfaces have different costs in different
applications. A local sign condition uses one norm; the ordered
labels used for full real feasibility require a norm inside another
norm. This distinction accounts for several of the smaller bounds
below. All inputs are finite binary descriptions. The bounds in this
section are consequences of the draft's proofs, with the same status
as its headline theorem.

\subsection{Semidefinite feasibility}
For explicitly sized symmetric rational matrices, let
\[
 A(x)=A_0+\sum_{i=1}^n x_iA_i.
\]
General SDP feasibility asks whether \(A(x)\succeq0\) for some real
\(x\); strict SDP feasibility asks whether \(A(x)\succ0\).
These are exact decision questions, with no numerical gap promise.

\begin{theorem}\label{app:sdp}
General and strict rational SDP feasibility belong to
\(\BPP^{\Clevel2}\subseteq\Clevel3\).
More generally, the same bound holds for feasibility of a
quantifier-free semialgebraic set promised to be closed and convex.
\end{theorem}
\begin{proof}
An effective real sample bound gives a short-circuit radius \(R\)
such that a nonempty input set \(S\) meets the closed ball of radius
\(R\). The intersection \(K\) is compact and convex. Thus
\[
 \chi(K)=\begin{cases}0&K=\varnothing,\\1&K\ne\varnothing.
 \end{cases}
\]
The Euler-parity algorithm of Theorem~\ref{app:euler} decides which
case holds in \(\BPP^{\Clevel2}\). This is why convexity helps:
it makes a single bit of Euler characteristic detect nonemptiness.
The algorithm need not recognize the closed-convex promise.

For a rational symmetric pencil, the PSD condition has a polynomial-size
semialgebraic description. If
\(\det(tI+A)=t^m+c_1t^{m-1}+\cdots+c_m\), it is equivalent to
\(c_j\ge0\) for every \(j\). Indeed the roots are the negatives
of the real eigenvalues, and a monic polynomial with nonnegative
coefficients has no positive real root. Division-free determinant
circuits compute all \(c_j\) with polynomial size. The feasible
set is closed and convex, so the preceding argument applies.

Strict feasibility reduces to a closed SDP by adding variables
\(\delta,v\) and the constraints
\[
 A(x)-\delta I\succeq0,\qquad
 \begin{pmatrix}\delta&1\\1&v\end{pmatrix}\succeq0.
\]
The second block forces \(\delta>0\), and every positive-definite
\(A(x)\) admits such a choice.
\end{proof}

There is also a useful equivalence independent of the CH arguments:
\begin{proposition}\label{app:strict-equivalence}
General rational SDP feasibility and strict rational SDP feasibility
are polynomial-time many-one equivalent. The reduction from general
to strict feasibility can send every NO input to a strongly infeasible
pencil. The same assertion holds for second-order cone feasibility.
\end{proposition}
The reverse reduction is the two-by-two block above. For the forward
direction, first bound a possible feasible point by an effective
radius \(M\), then bound the positive infeasibility gap on
\([-M-1,M+1]^n\) by a small \(\mu\). Polynomially many two-by-two
strict LMIs express
\(M<T<M+1\) and \(0<\epsilon<\mu\), even though the expanded
binary descriptions of \(M\) and \(1/\mu\) are exponential.
The output is
\[
 |x_i|<T\quad(1\le i\le n),\qquad A(x)+\epsilon I\succ0.
\]
Supplementary Section~\ref{app:sdp-details} proves the gadgets and
the gap assertion, including weakly infeasible inputs. This use of
large constants is related to \cite[Theorem 9]{BLS25}; we make no
priority claim for the reduction. Strict feasibility by itself
does not justify a lower complexity bound than general feasibility.

For a commuting rational pencil, simultaneous spectral decomposition
does give a smaller bound:
\begin{equation}\label{app:commuting}
 \textsc{Commuting-SDP}\in
 \NP^{\PosSLP}\cap\mathsf{coNP}^{\PosSLP}
 \subseteq\MA^{\PP}\cap\mathsf{coMA}^{\PP}
 \subseteq\Clevel2.
\end{equation}
The generated rational matrix algebra has dimension at most the
matrix size. A separating element represents every coefficient
matrix as a rational polynomial in one symmetric matrix. Its
polynomially many distinct eigenvalues are specified by a dense
squarefree polynomial and isolating intervals. Feasibility becomes
a polynomial-size LP over these algebraic numbers. Active-constraint
certificates for feasibility and sparse Farkas certificates for
infeasibility are verified in \(\Pclass^{\PosSLP}\).
The first containment in \eqref{app:commuting} does not depend on our
new arithmetic upper bounds. The second uses the PosSLP theorem below;
fresh amplified oracle randomness gives
\(\NP^{\BPP^{\PP}}\subseteq\MA^{\PP}\), and counting over
Merlin's strings after amplification gives
\(\MA^{\PP}\subseteq\PP^{\PP}\).

The general SDP bound also extends to several compact coefficient
encodings: rational arithmetic circuits, explicitly isolated dense
algebraic numbers and arithmetic circuits over those leaves, and
square roots of nonnegative integer-circuit values. Their graph
constructions reduce them to ordinary rational SDP with polynomial
matrix dimension; see Section~\ref{app:sdp-coefficients}.
This statement does not include arbitrary nested algebraic-root DAGs
or succinct descriptions of exponentially sized matrices.

\subsection{PosSLP and coefficient positivity}
\(\PosSLP\) asks whether an integer straight-line program using
\(0,1,+,-,\times\) has positive output. The output may have
exponentially many bits. A related problem asks for the sign of a
specified coefficient of a polynomial given by a short arithmetic
circuit, with the coefficient index in binary.
These circuit problems are studied in \cite{ABKPM09,KP07}.

\begin{theorem}\label{thm:posslp}
Both problems belong to
\[
 \BPP^{\PP}\ \cap\ \Pclass^{\NP^{\PP}}.
\]
In particular they belong to \(\Clevel2\), by the randomized bound,
and to \(\PH^{\PP}\), by the deterministic bound.
\end{theorem}

The common input to both algorithms is an integer residue
representative at every short modulus. For an SLP value this is
direct evaluation. For a specified coefficient, use Kronecker
encoding and Fourier extraction over a short Galois ring, with
unit-character orthogonality valid modulo arbitrary prime powers.
The resulting coefficient representative is a raw
\(\GapP^{\PH}\) expression; functional Toda--Watanabe absorption
places the complete weighted counts in \(\FP^{\PP}\)
\cite{TW92}. The beta identity \eqref{depth:beta} then converts
the sign into a very small signed displacement on the circle.

The deterministic algorithm asks whether any earlier dyadic scale
has reached a fixed distance from zero. Binary search for the first
such event uses \(\NP^{\PP}\); its direction reveals the sign.
The randomized algorithm locates a usable scale by the short-seed
phase test of Lemma~\ref{depth:sign-interface} and uses only a
\(\PP\) oracle. The complete independent proofs are in
Sections~\ref{sec:posslp} and \ref{push4:sec}. They do not pass
through the real-root selection machinery.

The published general CH containment is due to
\cite{ABKPM09}; the smaller bounds displayed here are the proposed
refinements. The 2024 discussions in \cite{BJ24,BDJ24} still cite
the older upper bound; this comparison is not an exhaustive priority
claim. Their relation to the ETR proof is arithmetic: modular
representatives, weighted counting, and sign reconstruction are shared.
There is no established reduction \(\ETR\subseteq\NP^{\PosSLP}\)
in this paper. Consequently the PosSLP theorem alone does not give
the headline ETR theorem, or put ETR in \(\PH^{\PP}\).

\subsection{Square-root sum}
The square-root-sum problem compares
\(\sum_{i=1}^n\sqrt{a_i}\) with an integer \(k\), for
binary nonnegative integers \(a_i\). Fixing either strict or
nonstrict comparison gives the same upper bounds below.

\begin{corollary}\label{app:srs}
\[
 \SRS\in\BPP^{\PP}\cap\Pclass^{\NP^{\PP}}
 \subseteq\CH,
 \qquad \SRS\in\Clevel2.
\]
\end{corollary}
Use the standard reduction to \(\PosSLP\) from \cite{ABKPM09}
and Theorem~\ref{thm:posslp}. Algebraic separation supplies a
singly exponential precision bound. Newton approximations to the
radicals attain that precision with polynomially many shared
arithmetic-circuit operations; the separation bound distinguishes
zero from either nonzero sign. The final integer comparison is a
PosSLP instance. Thus this consequence improves the inherited
fourth-level CH upper bound through the integer-sign algorithm;
it does not require a separate real-feasibility reduction.

A possible certificate route to \(\MA\cap\mathsf{coMA}\) remains
unresolved. Algebraic sums-of-squares existence is insufficient:
the numerators and the identities to be checked need uniformly
short, coherent descriptions. The four-square constructions,
square-class promises, and the precise missing size hypothesis are
kept in supplementary Sections~\ref{push7:sec}--\ref{push10:sec}.
They are not additional unconditional containments for SRS.

\subsection{Geometric real counting}
The standard name here is \emph{GCR}, the geometric real counting
class of B\"urgisser and Cucker \cite{BC06}. It is a function class:
\(\GCR=\mathsf{BP}(\SharpP_{\R}^{0})\), where \(\mathsf{BP}\)
denotes the Boolean part of a constant-free real computation,
not a bounded-error probability operator. Its basic total problem is
\[
 \operatorname{pc}(S)=
 \begin{cases}|S|&S\text{ is finite},\\
 \infty&S\text{ is infinite},\end{cases}
\]
for a binary-described semialgebraic set. It counts distinct points,
not connected components or algebraic multiplicities. The symbol
\(\infty\) is part of the output convention.

\begin{theorem}\label{app:gcr}
The total threshold problem
\(\{(S,k):\operatorname{pc}(S)\ge k\}\) belongs to \(\Clevel4\).
Consequently
\[
 \GCR\subseteq\FP^{\Clevel4},\qquad
 \FP^{\GCR}\subseteq\FP^{\Clevel4}.
\]
\end{theorem}

Here is how counting avoids a fifth level. For a finite algebraic
set \(Z\), add a sphere coordinate to double its points and use
a small compact sublevel of a sum of squares. After the weighted
sextic perturbation of \eqref{depth:objective}, Morse theory gives
\[
 2|Z|=N_+-N_-,
\]
where \(N_+\) and \(N_-\) count critical points below the threshold
with positive and negative Hessian determinant. The ordered labels
of Section~\ref{sec:small-depth} give distinct addresses in a
set of size \(H\), a power of two exceeding the critical-algebra
rank. Choose a uniform address. Accept at a positive critical point,
reject at a negative one, and otherwise toss a fair coin. Its ideal
acceptance probability is
\[
 \frac12+\frac{N_+-N_-}{2H}
 =\frac12+\frac{|Z|}{H}.
\]
The local address and sign computations use
\(\BPP^{\Clevel3}\), with error amplified far below \(1/H\).
Their randomness is sampled inside the outer
\(\PP^{\Clevel3}\) computation. A midpoint threshold therefore
decides the desired integer comparison at level four. We do not
place a new \(\PP\) machine above a \(\Clevel4\) predicate.

Boolean combinations and inequalities admit a quadratic algebraic
encoding with a known constant fiber size; divide out that size
in the threshold. To recognize infinity, an effective one-block
bound supplies \(\eta>0\) below every nonzero squared distance
between points of any finite input set. An infinite semialgebraic
set contains an arc and therefore arbitrarily close distinct
points. Hence
\[
 S\text{ infinite}\quad\Longleftrightarrow\quad
 \exists x,y\in S:\ 0<\|x-y\|^2<\eta.
\]
This is a polynomial-size ETR instance using a short dyadic circuit.
Amplify its \(\BPP^{\Clevel3}\) algorithm and merge its coins into
the same threshold experiment. Finite counts have polynomially
many output bits after quadratic encoding, so threshold queries
recover the entire value in \(\FP^{\Clevel4}\).
Section~\ref{app:count-details} supplies the geometric and error
details.

The published framework gives
\(\SharpP\subseteq\mathsf{GCC}\subseteq\GCR\subseteq
\mathsf{FPSPACE}\), together with completeness results for point
counting and Euler characteristic \cite{BC06}. Our added ingredient
is the signed Morse count combined with the norm and rank machinery.
The full Euler characteristic and its emptiness flag also lie in
\(\FP^{\Clevel4}\) by those reductions. The lower bound for a
few numerical Euler bits below is a different, more restricted
output task.

\subsection{Euler characteristic and its low bits}
For arbitrary Boolean semialgebraic sets, use the additive,
definable Euler characteristic \(\chi^*\). On locally closed
sets this equals Euler characteristic with compact supports;
on compact sets it equals ordinary Euler characteristic.

\begin{theorem}\label{app:euler}
For every fixed \(c\), the residue
\(\chi^*(S)\bmod2^b\), where \(0\le b\le c\log_2(L+2)\)
and \(L\) is the input length, is computable in
\(\mathsf{FBPP}^{\Clevel2}\).
No finiteness, smoothness, or boundedness promise on \(S\) is needed.
\end{theorem}

The first ingredient is a low-bit trick for locally qualified
critical points. To compute
\(\binom Nk\bmod2\), count increasing \(k\)-tuples, ordered by
their distinct critical values. For polynomially bounded \(k\),
this is a single product-algebra parity problem with ordinary
short-circuit sign guards. Its one norm has a raw
\(\GapP^{\PP}\) residue representative, so its sign is in
\(\BPP^{\Clevel2}\). Lucas's identity
\[
 \text{bit}_j(N)=\binom{N}{2^j}\pmod2
\]
then recovers logarithmically many bits, since the required tuple
lengths remain polynomial. Apply this to the positive and negative
Hessian counts in the Morse formula.

The second ingredient preserves Euler characteristic when converting
an arbitrary Boolean semialgebraic set to compact algebraic sets.
An algebraic lift \(Y\) has one independent sign choice per atom.
A canonical guard chooses the nonnegative witness in each case,
giving a semialgebraic bijection \(Y_{\rm can}\to S\).
Boolean and inverse witnesses make \(Y_{\rm can}\) a union of
connected components of \(Y\), with a uniform sign gap on every
compact truncation. Thus the canonical guard can be retained as a
locally constant weight in the Morse calculation.

Choose a short-circuit radius \(R\) beyond which both ball and
sphere Euler characteristics stabilize, and put
\[
 A=Y\cap\{\|u\|\le R\},\qquad
 B=Y\cap\{\|u\|=R\},\qquad
 D=\{(u,z):u\in Y,\ \|u\|^2+z^2=R^2\}.
\]
The set \(D_{\rm can}\) doubles \(A_{\rm can}\) along
\(B_{\rm can}\). Additivity and the contribution \(-1\) of an
open ray give
\begin{equation}\label{app:double-link}
 \chi^*(S)=
 \frac{\chi(D_{\rm can})-\chi(B_{\rm can})}{2}.
\end{equation}
Both terms in the numerator are Euler characteristics of selected
components of compact algebraic sets. Compute them modulo
\(2^{b+1}\) by the tuple-and-Morse algorithm, subtract, and divide
the even residue by two. Only one extra bit is needed, independently
of the number of original sign atoms. The canonical lift, the
stabilization bound, and the Morse step are proved in
Section~\ref{app:euler-details}.

There is no emptiness flag in Theorem~\ref{app:euler}. A nonempty
circle has Euler characteristic zero, so Euler parity cannot
decide general real feasibility. On the other hand, under a promise
that \(|\chi^*(S)|\) is polynomially bounded, enough low bits give
the exact signed value in the same function class. In particular,
finite sets of polynomially bounded cardinality can be counted in
\(\mathsf{FBPP}^{\Clevel2}\), without recognizing that promise.

\subsection{Complex feasibility and dimension}
Let HN denote Hilbert's Nullstellensatz as a decision problem:
do given integer polynomial equations have a common complex zero?

\begin{theorem}\label{app:hn}
\[
 \mathsf{HN}\in\BPP^{\Clevel2}\subseteq\Clevel3.
\]
The complex dimension of the zero set, with value \(-1\) for the
empty set, is in \(\mathsf{FBPP}^{\Clevel2}\).
\end{theorem}

This refinement has a separate dependency: it uses the published
affine reduction of \cite[Proposition 4.4]{AGS26} and the modular
product construction, but not our ETR or integer-sign theorem.
Quadratic wire encoding makes the input degrees at most two.
The AGS reduction expresses nonemptiness through two trailing
coefficient operations on resultants of explicitly given,
constant-degree homogeneous forms, followed by evaluation at zero.
One formal branch is computable in polynomial time at a short
finite-field node. The flattened Poisson product costs one counting
level, and joint Fourier extraction costs a second.

Coefficient-prefix tests locate each first nonzero coefficient
by binary search. Their geometric-sum weights fit inside the same
Fourier sum. Random-prime identity testing distinguishes zero
polynomials; fresh randomness and uniform degree and height bounds
control the adaptive searches. Since only equality is needed,
there is no residue-to-integer-sign step. Random affine sections
then determine dimension with polynomially many calls to the same
oracle. Section~\ref{app:hn-details} records the precise interface
and error analysis.

\subsection{Arithmetic extensions and stochastic fixed points}
Some further decision bounds follow by reductions, without changing
the main norm argument. Adaptive polynomial-time PosSLP computations
compile to one PosSLP query (Theorem~\ref{push10:adaptive}). The
same holds for the supported nested real-root gates of
Theorem~\ref{push10:root-gates}, including polynomial-degree principal
roots; Theorem~\ref{push12:binary-root-theorem} extends this to
binary-encoded root degrees and specified sparse-polynomial roots
under its root-selection promises. Their sign problems therefore
inherit both bounds in Theorem~\ref{thm:posslp}. The supplied
domain and root-order promises are part of these statements.

For the least-fixed-point threshold problems for branching stochastic
systems treated in Proposition~\ref{push10:games}, the strategy
certificates give
\[
 \NP^{\PosSLP}\cap\mathsf{coNP}^{\PosSLP}
 \subseteq\MA^{\PP}\cap\mathsf{coMA}^{\PP}
 \subseteq\Clevel2.
\]
The underlying PPS reduction to PosSLP is published
\cite{ESY12,ESY13}. The supplementary proofs state the precise
systems, strategies, and comparison conventions. Conditional
short-certificate proposals and restricted numerical reconstruction
lemmas are discussed there as supporting material, rather than
as stronger bounds for unrestricted ETR or square-root sum.

\subsection{Exact rational weighted counting}
There is also a class consequence stated without arithmetic-circuit
inputs. Let \(\mathsf{PP}_{\Q}^{(pt)}\) denote comparison with a
rational threshold of a sum of polynomial-time computable rational
weights over polynomial-length witnesses, as in \cite{CSW19}.
Each individual numerator and denominator has polynomial bit length;
the denominator of their sum need not. Theorem~\ref{thm:rational-counting}
gives
\[
 \mathsf{PP}_{\Q}^{(pt)}\subseteq
 \BPP^{\PP}\cap\Pclass^{\NP^{\PP}},\qquad
 \Pclass^{\mathsf{PP}_{\Q}^{(pt)}}
 =\Pclass^{\mathsf{CoeffPosSLP}}.
\]
Here \(\mathsf{CoeffPosSLP}\) is specified coefficient positivity.
An effective common-denominator bound reduces a nonzero rational
sum to a small displacement with a known lower bound, so both
sign-reconstruction algorithms apply directly. For the relation
to coefficients, arithmetize the weight evaluator and replace the
reciprocal denominator by a sufficiently long geometric sum, which
has a short circuit. The coefficient construction preserves the
threshold sign. The forward reduction is many-one; the reverse
one asserted here is a polynomial-time Turing reduction. The full
proof and the additional factorization promises are in
Section~\ref{sec:rational-counting}.

\clearpage
\appendix
\part*{Supplementary information}
\addcontentsline{toc}{part}{Supplementary information}
The main text gives the containment statements in their current form.
The following sections preserve the detailed arithmetic arguments,
additional geometric interfaces, and certificate investigations.
Intermediate numerical bounds are retained to explain which
implementation choices cost levels. Conditional certificate hypotheses
are explicitly separate from the proved implications that would use
them. The finite checks supplied with the manuscript test identities
and examples; they do not certify the asymptotic containment proofs.

We first give details for the hierarchy and applications introduced
above, then the longer homotopy, reconstruction, and certificate
arguments. The bibliography and every cross-reference are contained
in this single source file.

\section{Resultant interfaces for the numerical real-hierarchy bound}
\label{app:new-interfaces}

This section supplements the structural sampling argument in
Section~\ref{push10:real-blocks}. Its purpose is to justify the
explicit ledger \eqref{hier:ledger}. The geometric invariant is
unchanged: samples cover the required connected sign components,
coordinate denominators are guarded at the selected root, and one
finite infinitesimal tuple preserves the full representation formula.

\subsection{Coefficient and value access}
Consider homogeneous forms of singly exponential degree in
polynomially many variables, represented by indexed coefficients
with polynomial-length addresses. Work modulo a short prime power
in an unramified Galois ring whose residue field contains enough
distinct reference nodes. Extension elements still have polynomial
bit length. The reference system is a product of distinct linear
factors in each variable, with unit Jacobian at every reference
point. If the degrees are exponential, its values, derivatives,
and Jacobian products use \(\FP^{\PP}\). Factorization and field
metadata are found once at the charged preprocessing level; they
are not hidden in a purported polynomial-time leaf evaluator.

Interpolate formally from that reference system to the target.
Polynomial Newton and inverse-Jacobian updates double the order
of agreement with the formal branch at each step. There are only
polynomially many updates to reach an exponential order. At a
fixed short node the updates use polynomially many target value
and derivative queries. Their untruncated polynomial degrees
remain singly exponential. The Poisson formula expresses the
resultant through one product indexed by the branches and the
successive coordinate subsystems. Flatten these indices into one
polynomial-length index. All homotopy and parameter coefficient
extractions then form one joint Fourier sum, as in
\eqref{eq:field-fourier} and \eqref{eq:final-fourier}.

If the target value oracle is at level \(r\), branch work is at
\(\max(r,1)\), the product at \(\max(r,1)+1\), and the final
coefficient extraction at \(\max(r,1)+2\). If only coefficients
at level \(h\) are supplied, their dense homogeneous monomial
sum first gives a value oracle at \(h+1\). This proves the
\(h+3\) coefficient interface and the value-interface refinement
used in the main text. The formal Poisson/Newton identities are
those of \cite[Sections 5.2--5.3]{AGS26}; the finite-ring product
implementation is that of Lemma~\ref{lem:ring-products}.

Homogeneous evaluation is substantive. For a coordinate map
\((g_0,g_1,\ldots,g_s)\), use the formal expression
\[
 \sum_{\alpha}c_\alpha
 g_0^{E-|\alpha|}\prod_i g_i^{\alpha_i},
 \qquad E\ge\deg P.
\]
Every exponent is nonnegative, so this is defined when \(g_0=0\).
For a sign test at a valid affine point, choose \(E\) even and
retain the separate condition \(g_0\ne0\). One must not replace
the homogeneous evaluator by affine division at all Fourier nodes.

\subsection{Restriction without adaptive degree truncation}
To check \eqref{hier:padded}, homogenize in two variables.
At a specialization of degree \(d\), the degree-\(D\) form of
\(f+t\) has the factor \(X_0^{D-d}\). Evaluating the other form
at infinity gives \(V\), so resultant multiplicativity extracts
\(V^{D-d}\). The constant \(V\)-coefficient of what remains is
\[
 a_d^{E-e}\Res_{d,e}(f+t,g+\lambda t),
\]
where \(a_d\ne0\) is the leading coefficient of \(f\), and
\(e\) is the degree of \(g+\lambda t\) over the parameter
field. Its leading factor is independent of \(t\).
The even choice of \(E\) removes the harmless ordering sign
from the infinity factor. Nonzero constants and zero defining
polynomials are handled as stated after \eqref{hier:padded}.

At a common root with multiplicities \(a,b\), deform the two
local equations by \(t\) and \(\lambda t\). Their local
resultant has order \(\min(a,b)\). For \(a=b\), its leading
coefficient is a nonzero polynomial in the independent parameter
\(\lambda\); keeping \(\lambda\) formal prevents cancellation.
Adding over the roots gives the gcd degree. Retaining the original
coefficient arrays and all coefficients in \(V,t,\lambda\)
therefore retains the actual degrees and required gcd degrees.

For completeness, these data have the needed geometric meaning.
Near any parameter point of a connected sign cell, choose disjoint
small neighborhoods of the distinct roots. Constancy of degree
and of the number of distinct roots forbids a split or collision.
Conjugation and local uniqueness keep a real root real. The local
gcd contribution of each branch with a guard is upper semicontinuous;
constancy of the total forces each contribution to be locally
constant. A guard is consequently either zero along a branch or
continuously nonzero with fixed real sign. Continuation gives the
same conclusion on the connected cell. This is the branch property
required by the BPR sampling construction.

\subsection{Merged derivative extraction and flattening}
Suppose a sampling resultant has a Fourier representation
\(q_i=|U|^{-1}\sum_{z\in U}P(z)z^{-i}\), with any other
Fourier variables suppressed. Substitute this expression into the
homogeneous derivative sum and interchange the two finite sums.
The result is exactly
\[
 |U|^{-1}\sum_{z\in U}P(z)K_{D,\mu}(A,B,z).
\]
The sampling-product value \(P(z)\) is already at level \(h+2\),
whereas the kernel is at level two. Thus the value stays at
\(h+3\). Higher derivative order and selection of an infinitesimal
coefficient change the kernel or add a character weight in the
same sum. This proves the saving of one level per real block.

The two extra homogeneous coefficient compositions in the ledger
are each charged two levels, even where an implementation might
save more. One common finite specialization is charged once.
Flat recoding evaluates its regularized tuple system at \(k+1\),
takes its sampling resultant at \(k+3\), and composes coordinates
at \(k+5\). Write \(\ell=k+5\). The full Thom-hash selector
has a homogeneous value at \(\ell+2\); the defining polynomial
and its first derivative are at \(\ell+1\). Put the selector
last in the binary resultant. Its derivative is then unnecessary:
only the first polynomial defines the continued roots. The value
interface gives the guarded parity polynomial at \(\ell+4\).
Two root representations cost eight levels, the inherited sign
algorithm costs one, and fixed Boolean alternation with relativized
Toda costs two. This is \(k+5+8+1+2\), as claimed.

The complete representation formula before specialization uses
only a dimension depending on \(w\), although it has exponentially
many indexed coefficients. Fixed-dimensional elimination therefore
gives singly exponential bounds for all its cumulative projection
polynomials. The leading-monomial choice of the \(\epsilon_i\)
in Section~\ref{push10:real-blocks} applies simultaneously to
validity, coverage, and parent--child coordinate equalities.
Hashing is applied after this choice. No lower real bound is
asserted for a value that remains formally infinitesimal.

\section{Compact Morse counts and total point counting}
\label{app:count-details}

\subsection{A topology-preserving weighted critical system}
Let a compact algebraic set \(Z\) be the zero set of a coercive
quartic sum of squares \(F\). The examples in the main text use
a sphere equation, so coercivity follows from
\(F(x)\ge(R^2-\|x\|^2)^2\). There are finitely many critical
values. A one-block existential description of those values and
the effective degree and height bounds of \cite[Section 4.1]{BC06}
give a positive short dyadic \(\epsilon\) below every positive
critical value.

For sufficiently small \(\epsilon\), the compact sublevel
\(K=\{F\le\epsilon\}\) strongly deformation retracts to \(Z\).
One direct justification uses the gradient inequality near the
compact zero set,
\(\|\nabla F\|\ge cF^\theta\), with \(\theta<1\).
Negative gradient flow enters every lower positive sublevel in
uniformly bounded time and its remaining length from level \(a\)
is at most \(a^{1-\theta}/(c(1-\theta))\).
The endpoint maps therefore converge uniformly, giving a continuous
retraction fixing \(Z\). A time rescaling gives the strong
deformation retraction. These qualitative constants prove the
topological assertion; they are not algorithmic inputs.

Put \(W(x)=\sum_i2^{i-1}(x_i^6+3(x_i+1)^2)\).
On the compact band \(\epsilon/2\le F\le\epsilon\), the
gradient norm has a positive lower bound \(\gamma\). If the
band is empty there is nothing to control. Otherwise its squared
minimum is an endpoint of a one-block semialgebraic image, giving
a short dyadic lower bound by the same effective estimates. Bound
\(W\) and \(\nabla W\) on \(K\) by \(A_W,B_W\), and choose
a power of two \(Q\) such that
\[
 A_W/(6Q)<\epsilon/4,\qquad B_W/(6Q)<\gamma/2.
\]
Also make its two-adic valuation exceed the critical-system
requirement in Section~\ref{sec:small-depth}.

Throughout \(0\le t\le1/(6Q)\), the boundary
\(F+tW=\epsilon\) lies in the positive regular band and has
gradient norm at least \(\gamma/2\). The vector field
\(-W\nabla(F+tW)/\|\nabla(F+tW)\|^2\), with a compact cutoff,
moves these boundaries and their sublevels by an isotopy.
Consequently \(J=6QF+W\), at threshold \(T=6Q\epsilon\),
has the same sublevel Euler characteristic as \(Z\).
The two-adic construction makes \(J\) a proper Morse function
with distinct critical values. Hence
\begin{equation}\label{supp:morse}
 \chi(Z)=\sum_{x\in\R^m:\nabla J(x)=0,\ J(x)<T}
                    \sgn\det\operatorname{Hess}J(x).
\end{equation}
This is the usual Morse handle formula
\cite[Theorem 6.7 and Corollary 6.8]{BC06}.
It remains valid with a locally constant component guard.

The bounds here permit already long integer constants: the effective
coefficient bit-height estimate is singly exponential in dimension
and degree and has polynomial dependence on the incoming bit height.
Alternatively, when deriving an initial syntactic bound, introduce
unique quadratic powering wires before applying the estimate.
Neither argument applies a polynomial-size theorem to a secretly
doubly exponential expanded input.

\subsection{Finite and infinite point counts}
For a finite zero set \(Z\), an isolated-point radius bound gives
\(R\) strictly above the norms of all its points. The zero set of
\[
 \sum_i f_i(x)^2+(R^2-\|x\|^2-z^2)^2
\]
has exactly two points over each member of \(Z\). Thus
\eqref{supp:morse} is \(N_+-N_-=2|Z|\); there is no assertion
of one critical point per zero.

For a Boolean input with \(a\) strict sign atoms, the algebraic
lift in the next section, without its canonical guard, has exactly
\(2^a\) points over each original point. The same experiment
therefore has ideal acceptance probability
\(1/2+2^a|S|/H\). Its threshold for \(|S|\ge k\) is
\[
 t_k=\frac12+\frac{2^ak-1/2}{H}.
\]
The signed norm tests and the common normalizer can each be
amplified to error at most \(1/(256H)\), using polynomially many
repetitions because \(\log H\) is polynomial. Their combined
error is below \(3/(256H)\), much smaller than the
\(1/(2H)\) midpoint gap. All random bits are internal to one
\(\PP^{\Clevel3}\) experiment. Ordinary PP padding implements
comparison with the short dyadic \(t_k\).

For total counting, consider the one-block image
\[
 D_S=\{t>0:\exists x,y\in S, t=\|x-y\|^2\}.
\]
An effective coefficient bit bound \(B\) for its quantifier-free
description is chosen from the original input, before adding any
new precision constant. If \(S\) is finite, every member of
\(D_S\) is a zero of a nonzero description polynomial; otherwise
the truth value would hold on a neighborhood. Removing a zero
power and applying the reciprocal Cauchy bound gives
\(t\ge2^{-B-2}\). Set \(\eta=2^{-B-3}\).
If \(S\) is infinite, a positive-dimensional semialgebraic cell
supplies pairs at arbitrarily small distance. This proves the
infinity test in the main text.

Amplify that test to error at most \(1/(256H)\). Accept immediately
when it reports infinity and otherwise run the finite experiment.
For finite inputs the midpoint gap survives; for infinite inputs
the acceptance probability is at least \(1-1/(256H)\). In the
nontrivial finite threshold range \(t_k\le1-1/(2H)\), so the
infinite case also answers correctly. Thresholds below zero are
immediate, and thresholds beyond the uniform finite bound use
the infinity test alone. The finite experiment remains a syntactically
defined computation even on infinite inputs, so its rare execution
on the wrong branch creates no promise-oracle problem.

\section{Euler residues for arbitrary Boolean semialgebraic sets}
\label{app:euler-details}

\subsection{An algebraic lift with a canonical component selection}
Replace comparisons by Boolean combinations of strict atoms \(p>0\).
Give every arithmetic gate and every atom value, including constants,
its own unique evaluation coordinate. For each atom introduce
\(b,c,t,s,v,w\) and the equations
\begin{align*}
 b(b-1)&=0,& b(pt^2-1)&=0,&(1-b)t&=0,\\
 (1-b)(s^2+p)&=0,&bs&=0,&c(c-1)&=0,\\
 cp&=0,&(1-c)(pv-1)&=0,&cv&=0,\\
 c(w^2-1)&=0,&(1-c)w&=0.
\end{align*}
Encode the Boolean gates uniquely and require the output to be one.
Use unique wires to reduce every displayed equation to degree two.
Let \(Y\) be the resulting algebraic set.

When \(p>0\), the atom has \(b=1,c=0,v=1/p\),
\(t=\pm1/\sqrt p\), and \(s=w=0\). When \(p<0\), it has
\(b=c=0,v=1/p\), \(s=\pm\sqrt{-p}\), and \(t=w=0\).
When \(p=0\), it has \(b=0,c=1,v=0\), \(w=\pm1\),
and \(t=s=0\). There are exactly two witnesses in every case.
The guard \(t,s,w\ge0\) for every atom retains exactly one,
so the projection \(Y_{\rm can}\to S\) is a semialgebraic
bijection and preserves additive Euler characteristic.

The Boolean coordinates are constant on each connected component
of \(Y\). In a fixed truth and zero case, exactly one of
\(t,s,w\) is active and never zero; its sign cannot change on
a component. Thus \(Y_{\rm can}\) is a union of components,
closed and open in \(Y\). On \(\|u\|\le R\), both \(p\)
and its inverse coordinate \(v\), when active, have magnitude at
most \(R\). Each nonzero canonical witness therefore has
magnitude at least \(R^{-1/2}\). For \(\Delta=1/(4R)\),
the shifted strict guard \(u_j>-\Delta/2\) agrees with the
canonical choice in a sufficiently small neighborhood of this
compact portion. It correctly retains an inactive zero witness.

\subsection{One radius stabilizes both the ball and its boundary}
Let \(F=\sum f_i^2\) for the quadratic equations of \(Y\),
with ordinary polynomial-bit coefficients. Consider the set in
the \((e,t)\)-plane given by
\[
 e>0,\quad t\ge0,\quad
 \exists x,\lambda:\ F(x)=e,\ \|x\|^2=t,
                       \ x=\lambda\nabla F(x).
\]
For a regular value \(e\), these are the radial critical values
on the smooth hypersurface \(F=e\). There are finitely many of
them by Sard's theorem and semialgebraicity. There are also only
finitely many exceptional critical values of \(F\). The planar
set thus has dimension at most one.

One-block elimination gives a defining family of bivariate
polynomials with degree and coefficient bit height bounded by
\(H=2^{\poly(L)}\). Every point of this planar set is a zero
of one of its nonzero defining polynomials, since otherwise
membership would persist on a neighborhood. Write such a polynomial
as \(e^aQ(e,t)\) with \(Q(0,t)\not\equiv0\).
Any finite limiting radial critical value as \(e\downarrow0\)
is a root of some \(Q(0,t)\), and hence has magnitude at most
\(1+2^H\). Choose a short power-of-two \(R\) with
\(R^2>2+2^H\).

Fix \(R\le r_1<r_2<\infty\). For all sufficiently small
positive regular \(e\), the hypersurface \(F=e\) has no radial
tangency in this closed annulus; otherwise a convergent sequence
would violate the preceding bound. Proper-submersion isotopy for
the smooth tube \(\{F\le e\}\), including its boundary, makes
both its sphere sections and its ball truncations topologically
constant between \(r_1\) and \(r_2\). At the two endpoints,
compact semialgebraic small-sublevel retraction identifies their
Euler characteristics with those of \(Y\). This applies even
when \(Y\) is singular. The robust component guard above is
locally constant on a sufficiently small tube in a slightly
larger compact ball, so the isotopy and endpoint retractions also
preserve the selected components.

Both Euler functions therefore stabilize beyond the explicitly
chosen \(R\). Qualitative conic structure at some possibly
larger radius gives
\(\chi^*(Y_{\rm can})=\chi(A_{\rm can})-\chi(B_{\rm can})\);
stabilization brings this numerical equality back to \(R\).
The double identity \(\chi(D_{\rm can})=2\chi(A_{\rm can})-
\chi(B_{\rm can})\) proves \eqref{app:double-link}.
Bounding a deformation-retraction radius alone would not justify
the boundary term, which is why both Euler functions were controlled.

\subsection{Tuple parity, Lucas's identity, and signed division}
For a fixed local qualification, let \(N\) count the corresponding
real critical points of the simple system. In its \(k\)-fold
product, add the strict inequalities
\(J(x^{(1)})<\cdots<J(x^{(k)})\). Distinct critical values
give exactly \(\binom Nk\) real tuples. A rational sign filter
for the conjunction, followed by positive denominator clearing
and a small nonvanishing shift, gives a norm whose sign is
\((-1)^{\binom Nk}\). For polynomial \(k\), the tuple system
has polynomial dimension, and its formal branch values are
polynomial-time computable at short nodes. The product costs
\(\FP^{\PP}\), the raw Fourier extraction is
\(\GapP^{\PP}\), and randomized sign reconstruction uses
\(\BPP^{\Clevel2}\). No rank norm occurs in this calculation.

Lucas's identity recovers bit \(j\) by choosing \(k=2^j\).
For \(b\le c\log_2(L+2)\), every required tuple has polynomial
size. Apply these tests to both Hessian signs in
\eqref{supp:morse}, retaining the robust component guard.
This computes the two compact Euler terms in
\eqref{app:double-link} modulo \(2^{b+1}\). Their difference
has an even canonical residue \(r\) in \([0,2^{b+1})\).
Return \(r/2\). This is correct even when the integer Euler
characteristic is negative. Polynomial amplification makes the
whole output correct with probability at least two thirds.

\section{Exact SDP reductions and coefficient encodings}
\label{app:sdp-details}

\subsection{Strict intervals expressed by small rational LMIs}
Fix \(\theta>0\) and \(q\ge1\), set \(s_0=\theta\), and use
\[
 C_i(s)=\begin{pmatrix}s_i&s_{i-1}\\s_{i-1}&1\end{pmatrix}.
\]
The PSD minimum of \(s_q\) is \(M_\theta=\theta^{2^q}\),
attained by repeated squaring. Strict feasible values approach
it from above. Define dual blocks and an affine objective by
\[
 Y_i=\begin{pmatrix}a_i&b_i\\b_i&c_i\end{pmatrix},\qquad
 a_q=1,\quad a_i=-2b_{i+1}\ (i<q),\qquad
 D_\theta=-\sum_i c_i-2\theta b_1.
\]
Substitute the equations for the \(a_i\); there are no equality
constraints left in the strict pencil. Direct expansion gives
\begin{equation}\label{supp:power-dual}
 s_q-D_\theta=\sum_i\Tr(C_i(s)Y_i).
\end{equation}
At \(s_i^*=\theta^{2^i}\), set \(a_q^*=1\) and
\(a_i^*=2a_{i+1}^*s_i^*\). The rank-one matrices
\[
 Y_i^*=a_i^*
 \begin{pmatrix}1\\-s_{i-1}^*\end{pmatrix}
 \begin{pmatrix}1&-s_{i-1}^*\end{pmatrix}
\]
are dual feasible and have zero trace product with the optimal
primal blocks. Thus the dual also attains \(M_\theta\).
The small tuple \(a_i=1\), \(b_1=-1\), \(b_i=-1/2\) for
\(i\ge2\), \(c_i=2\), is strictly feasible. Mixing it with
the optimum gives strict dual values approaching \(M_\theta\)
from below. Every strict dual value is below the optimum by
\eqref{supp:power-dual}.

For \(\theta=2\), the strict LMIs
\[
 C_i\succ0,\quad Y_i\succ0,\quad T-s_q>0,\quad
 D_2+1-T>0
\]
project exactly to \(M_2<T<M_2+1\). For \(\theta=1/2\),
the strict LMIs \(Y_i\succ0\), \(\epsilon>0\),
\(D_{1/2}-\epsilon>0\) project exactly to
\(0<\epsilon<2^{-2^q}\). All coefficients are small rationals
and there are \(O(q)\) two-by-two blocks. These are open interval
gadgets, not exact singleton gadgets incorrectly made strict.

\subsection{The bounded-box gap and strong infeasibility}
For an input of length \(L\), choose a fixed polynomial \(p\)
so that a feasible pencil has some point in
\(\|x\|_\infty<M=2^{2^{p(L)}}\). This follows from the ordinary
existential sample bound applied to \(A(x)=BB^{\mathsf T}\).
On an infeasible input the continuous function
\(\max(0,-\lambda_{\min}(A(x)))\) has a positive minimum
\(\eta\) on \([-M-1,M+1]^n\).

To bound it without computing it, describe
\[
 \{t\ge0:\exists x,B, |x_i|\le M+1, A(x)+tI=BB^{\mathsf T}\}.
\]
This set is precisely \([\eta,\infty)\). Express \(M\) by
ordinary quadratic powering wires in this formula. One-block
endpoint bounds and the reciprocal Cauchy bound give a fixed
polynomial \(r\) with \(\eta>\mu=2^{-2^{r(L)}}\).
Equivalently, the stabilization theorem in
\cite[Theorem 7]{BLS25}, applied to
\(zA(x)+I=BB^{\mathsf T}\) on this box, gives the same bound.
First fix \(p\), then bound that formula and fix \(r\); the
choice of precision is not circular.

Combine the two strict gadgets with \(T\pm x_i>0\) and
\(A(x)+\epsilon I\succ0\). A bounded PSD point yields a
strict solution. Conversely a strict solution has \(x\) in the
box and \(\epsilon<\mu<\eta\), which excludes an infeasible
input. This proves the forward reduction of
Proposition~\ref{app:strict-equivalence}.

In fact a NO output is strongly infeasible. Suppose its blocks
were all bounded below by \(-\delta_j I\), with
\(\delta_j\to0\). Evaluate \eqref{supp:power-dual} at the
fixed optimal primal tuple of each gadget. Since these fixed
blocks are PSD,
\[
 D_\theta(Y)\le M_\theta+
       \delta_j\sum_i\Tr(C_i(s^*)).
\]
The scalar constraints then bound
\(|x_i|\le M+1+O(\delta_j)\) and
\(-\delta_j\le\epsilon\le\mu+O(\delta_j)\).
A convergent subsequence gives a PSD \(A(x)+\epsilon I\) in
the closed box with \(0\le\epsilon\le\mu\), contradicting
\(\eta>\mu\). This uses no size bound on the approximate
dual variables. The resulting positive distance is qualitative;
it is not an inverse-polynomial numerical gap.

Zero-dimensional matrix blocks are handled directly. Second-order
cone constraints have standard arrow-matrix LMI representations;
the new power blocks are two-by-two, hence also second-order cone
constraints. This proves the corresponding reduction for SOCP.

\subsection{Exact graphs for compact coefficient encodings}
\label{app:sdp-coefficients}
Suppose a rational SDP has a finite attained minimum \(C\), and
its ordinary dual also attains \(C\):
\[
 \min c^{\mathsf T}u:\ F_0+\sum_j u_jF_j\succeq0,
 \qquad
 \max -\langle F_0,Y\rangle:\ Y\succeq0,+ \langle F_j,Y\rangle=c_j.
\]
Then the ray graph \(z\ge0,w=Cz\) has the rational SDP lift
\begin{align*}
 z&\ge0,& zF_0+\sum_j u_jF_j&\succeq0,& w&\ge c^{\mathsf T}u,\\
 Y&\succeq0,&\langle F_j,Y\rangle&=zc_j,&
 -\langle F_0,Y\rangle&\ge w.
\end{align*}
For \(z>0\), divide the primal variables and dual matrix by
\(z\) and use weak duality. For \(z=0\), a feasible primal
direction cannot have negative cost, and evaluating a homogeneous
dual at an original feasible point gives dual objective at most
zero. Hence \(w=0\); the zero tuple attains it. Attainment on
both sides supplies the converse for positive \(z\).

A nonnegative circuit using addition, squaring, and nonnegative
rational scaling has such a value SDP: replace addition by an
epigraph inequality and squaring by
\(\left(\begin{smallmatrix}v&a\\a&1\end{smallmatrix}\right)
\succeq0\), and minimize its output. Monotonicity proves the
correct optimum. Strict feasible gate values exist; a direct
attained dual is also obtained by back-propagating the nonnegative
derivatives, with square-gate block proportional to
\(\left(\begin{smallmatrix}1&-a\\-a&a^2\end{smallmatrix}\right)\).

Every signed circuit value \(x\) can be maintained as \(P_x-N_x\)
with two such nonnegative circuits. Addition and subtraction are
immediate, and multiplication uses
\begin{align*}
 P_{xy}&=\bigl((P_x+P_y)^2+(N_x+N_y)^2\bigr)/2,\\
 N_{xy}&=\bigl((P_x+N_y)^2+(N_x+P_y)^2\bigr)/2.
\end{align*}
The difference is exactly \(xy\), with constant overhead and
shared gates. Thus every integer-circuit constant has an exact
ray graph. Split each free variable as \(z^+-z^-\), replace
each fixed-coefficient product by its graph, and obtain an
ordinary rational SDP. For rational circuit coefficients, multiply
the whole pencil by the positive product of the squared nonzero
denominators before this compilation.

For an explicitly isolated real algebraic coefficient \(\beta\),
let \(p\) be its monic squarefree polynomial and \((a,b)\) its
rational isolating interval. In \(V=\Q[T]/(p)\), a variable
linear functional \(\ell\) gives moment matrices
\[
 H_0=(\ell(T^{i+j})),\quad
 H_a=(\ell((T-a)T^{i+j})),\quad
 H_b=(\ell((b-T)T^{i+j})),
\]
with products reduced modulo \(p\), for \(0\le i,j<\deg p\).
The constraints \(H_0,H_a,H_b\succeq0\),
\(\ell(1)=z\), \(\ell(T)=w\) project exactly to
\(z\ge0,w=\beta z\). Indeed
\(V\otimes\R\simeq\R^r\times\C^s\); positivity on squares
kills each complex factor and gives nonnegative weights on real
roots. The two localizers retain only roots in \([a,b]\), namely
\(\beta\). Conversely \(z\) times evaluation at \(\beta\)
supplies the functional. All matrix entries are ordinary rational
linear forms of polynomial size.

For \(\beta=\sqrt a\), \(a\ge0\), a simpler graph uses
\[
 \begin{pmatrix}az&w\\w&z\end{pmatrix}\succeq0,\qquad
 \begin{pmatrix}w&az\\az&aw\end{pmatrix}\succeq0.
\]
When \(a>0\) these impose \(w^2\le az^2\),
\(w^2\ge az^2\), and \(z,w\ge0\). When \(a=0\) the first
block forces \(w=0,z\ge0\). Thus the graph is exact also at
the boundary. If \(a\) is an integer-circuit value, eliminate
that coefficient by the previous construction.

Finally, a circuit over nonnegative algebraic leaves has the same
monotone value SDP and attained dual. Only the original leaves
occur as its fixed nonrational coefficients; intermediate values
are variables. Replace those leaf occurrences by their moment
graphs. Signed leaves are classified by ordinary univariate
isolation before the signed-circuit transformation. This uses a
fixed number of polynomial-size compilation stages and establishes
the coefficient extensions stated in the main text.

\section{Two-level modular evaluation for complex feasibility}
\label{app:hn-details}

After unique-wire quadratization, apply
\cite[Proposition 4.4]{AGS26}. It constructs polynomially many
families of homogeneous forms \(G_{i,j}\) of degree at most two
in the geometric variables, degree at most \(n\) in \(w\),
and degree at most one in \(t,u\), such that with high probability
the original variety is nonempty precisely when, for some \(j\),
\[
  \bigl(\operatorname{TT}_w\operatorname{TT}_t
       \Res(G_{0,j},\ldots,G_{n,j})\bigr)(0)=0.
\]
Here \(\operatorname{TT}\) means the trailing nonzero
\emph{coefficient}; \(t\) is handled first and evaluation at
zero is in \(u\). The forms have polynomially many explicit
terms. Their values and first derivatives are polynomial-time
computable at short field nodes.

Use a homotopy parameter \(s\) distinct from \(t\).
Constant degree makes each reference root and inverse Jacobian
polynomial-time computable, so the Poisson product \(P_j\)
is in \(\FP^{\PP}\). If \(E\) bounds the homotopy degree
of the actual resultant, its target value is
\(\sum_{a=0}^E[s^a]P_j\). Untruncated approximant degrees are
also bounded singly exponentially. Choose a short prime above
all those degrees and work over \(\F_p\) throughout. Its
reference system may depend on the prime; the Poisson identity
always computes the same target integer resultant modulo that
prime. Field metadata are supplied as short explicit parameters
after charged \(\FP^{\NP}\) preprocessing.

Let \(U=\F_p^*\), \(N=p-1\), and define
\[
 W_E(s)=\sum_{a=0}^E s^{-a},\qquad
 S_K(z)=\sum_{b=0}^Kz^b.
\]
Binary geometric-sum recursion evaluates both in polynomial time,
including at \(z=1\). For a prefix of the \(t\)-coefficients
of \(R_j\), the exact Fourier formula is
\[
 \sum_{b=0}^K[t^b]R_j(w_0,u_0)t_0^b
 =N^{-2}\sum_{s,z\in U}
 P_j(s,z,w_0,u_0)W_E(s)S_K(t_0/z).
\]
For a specified \(t\)-index and a prefix in \(w\), add one
Fourier variable with its specified-index character and the same
geometric-sum kernel. Both are still one outer sum over
\(\FP^{\PP}\) summands, hence in \(\FP^{\Clevel2}\).
This explicitly avoids counting again over retrieved coefficients.

Uniform degree and integer coefficient bit bounds apply to every
prefix tested, before any adaptive search. A nonzero coefficient
of bit length at most \(H\) has at most \(H\) distinct prime
divisors. Sample a polynomial-bit prime from an interval with
many more than \(H/\delta\) primes and lower endpoint above
the degree bound divided by \(\delta\). Capped rejection
sampling and polynomial-time primality testing suffice. Except
with probability \(O(\delta)\), reduction preserves a nonzero
coefficient; Schwartz--Zippel over the same field then detects
the nonzero polynomial. Zero polynomials always evaluate to zero.

Binary search for the first nonzero prefix takes polynomially
many tests in \(t\), then in \(w\). Conditional on all previous
answers being correct, each next test has the same worst-case
error bound. Fresh randomness and a union bound therefore control
the adaptive search and final evaluation at \(u=0\). Combined
with the AGS reduction error, amplification yields
\(\BPP^{\Clevel2}\). No sign of a large integer is queried.

For dimension, a generic affine subspace of codimension \(k\)
meets a nonempty complex variety exactly when its dimension is
at least \(k\). Standard degree bounds control the bad affine
choices, so polynomial-bit random coefficients suffice. Test the
polynomially many possible \(k\), using fresh amplified HN
calls, and first test emptiness. This returns the dimension in
\(\mathsf{FBPP}^{\Clevel2}\).

\section{Arithmetic conventions and the first-event test}
We use the input bit length \(n\) throughout. An integer or an index is
\emph{short} if it has polynomially many bits. A numerical bound of size
\(2^{\poly(n)}\) is singly exponential, and itself has a short binary
encoding. In particular, an integer of \(2^{\poly(n)}\) bits is much
larger than a short integer. We always distinguish a value from the
binary description of a bound on its length.

An SLP of size \(s\), with constants of magnitude at most \(2^b\),
has output magnitude less than \(2^H\), where we may take
\begin{equation}\label{eq:slp-height}
H=(b+1)2^{s+1}.
\end{equation}
Indeed, a maximum magnitude \(M\ge2\) grows by at most squaring at
each gate. The same argument applies to the coefficient \(\ell_1\)-norm
of a polynomial SLP. Its total degree is at most \(2^s\), and each
coefficient satisfies \eqref{eq:slp-height}.

The class \(\SharpP^A\) consists of accepting-path counts of
polynomial-time nondeterministic machines with oracle \(A\), and
\(\GapP^A\) consists of differences of two such counts. An
exponentially indexed sum of short nonnegative \(\FP^A\) values
is in \(\SharpP^A\): guess an index, and create as many accepting
branches as the value at that index. Signed sums are in \(\GapP^A\).
These classes are closed under sums over polynomial-length indices and
under products of polynomially many factors, when their machines and
length bounds are uniform. A short count can be recovered in
\(\FP^{\PP^A}\) through threshold queries. This last step is a
function computation; arbitrary postprocessing of the result need not
remain a counting function.

For a monic polynomial \(q\), write
\[
q_{[j]}(z)=[u^j]q(z+u)
\]
for its Hasse derivative. Since \(q^{(j)}=j!q_{[j]}\), the two have
the same signs at real inputs. We use Hasse derivatives so that our
integer construction can later be reduced at small primes without
requiring an inverse of \(j!\).

For an element \(a\) of a finite free algebra \(A\),
\(\Norm_A(a)\) is the determinant of multiplication by \(a\).
Over an algebraically closed field this is the product of its values
at the points of the algebra, with local algebra dimensions as
multiplicities. These identities hold even when the algebra is not
reduced. We use determinants to prove identities and height bounds;
we do not expand an exponentially large determinant as our algorithm.

\subsection{Detecting the first displacement}
The following elementary observation will supply the final sign test.
Its statement is independent of the oracle used to compute a summand.
\begin{lemma}\label{lem:first-event}
Suppose
\[
S=\sum_{i\in I}\frac{A_i}{m_i}=K+\delta,
\qquad K\in\Z,\quad |\delta|<\tfrac18,
\]
where \(m_i>0\), all \(A_i,m_i\) have polynomial bit length, and
\(I\) has at most \(2^{\poly(n)}\) uniformly indexed elements.
Suppose \(J\) is short and
\(\delta\ne0\Rightarrow |\delta|\ge2^{-J}\).
There is a polynomial precision \(b\) such that, from the short integers
\begin{equation}\label{eq:coarse-sums}
T_j=\sum_{i\in I}A_i
 \left\lfloor\frac{2^b(2^j\bmod m_i)}{m_i}\right\rfloor,
\qquad 0\le j\le J,
\end{equation}
one can determine zero and sign by finding the first index satisfying
a fixed dyadic comparison. That index has polynomial bit length.
\end{lemma}
\begin{proof}
Choose a known polynomial \(a\) with \(|A_i|\le2^a\), and take
\(b=a+\lceil\log_2|I|\rceil+12\), replacing \(|I|\) by a
known bound if necessary. Put
\[
v_j=(T_j\bmod2^b)/2^b\in[0,1),\qquad \epsilon=2^{-9}.
\]
The error in approximating the rational sum in
\eqref{eq:coarse-sums} before reducing modulo one is at most
\(|I|2^{a-b}\le2^{-12}\). Replacing each modular power by
\(2^j\) changes that rational sum by an integer. Thus \(v_j\)
approximates \(2^j\delta\bmod1\) to circular distance less than
\(\epsilon\). This remains true for signed \(A_i\).

Call \(j\) an event when \(\dist(v_j,\Z)\ge1/8\). Before
an event, the true centered distance from an integer is less than
\(1/8+\epsilon\). Its double has magnitude less than
\(1/4+2\epsilon<1/2\), so it does not wrap around the circle.
Starting with \(\delta\), induction shows that the centered value
is exactly \(2^j\delta\) up to and including the first event.
At that event its magnitude is at least \(1/8-\epsilon\), and
at most \(1/4+2\epsilon\); at index zero the stronger initial
upper bound suffices. Hence the sign is positive if \(v_j<1/2\)
and negative if \(v_j>1/2\).

If \(\delta=0\), there is no event. Otherwise there must be one
by \(J\): the no-wrap induction would give a centered value
\(2^J\delta\) of magnitude at least one. All multipliers in
\eqref{eq:coarse-sums} are polynomial-time computable from the binary
index \(j\), by modular powering. The whole sum has polynomially
many bits because it has exponentially many polynomial-bit summands.
\end{proof}
For example, if individual event tests are in \(\Pclass^A\), the
predicate that some event occurs by \(t\) is in \(\NP^A\).
Binary search on this monotone prefix predicate finds the first event
with polynomially many queries. Alternatively, guessing the first event
and universally excluding earlier ones puts the sign test in
\(\Sigma_2^A\). We will use the deterministic search for
\(\PosSLP\), and the quantified version when combining it with the
existential choices in ETR.

\section{Detailed shared-homotopy proof for real feasibility}\label{sec:etr}
We now construct an integer whose sign can certify a real solution.
The critical system is fixed deterministically, and its objective
separates all critical points. The integer depends only on a short guessed list
of sign conditions; every negative answer certifies a real solution.

\subsection{A minimum with a controlled gap}
We use the \(\ETR\)-complete problem ETR-INV
\cite{AAM18}. Its variables lie in
\([1/2,2]\), and its equations have the forms
\[
x_i=1,\qquad x_i+x_j=x_k,\qquad x_ix_j=1.
\]
For each variable, introduce \(u_i,v_i\) and add
\begin{equation}\label{eq:slacks}
2x_i-1-2u_i^2=0,\qquad 2-x_i-v_i^2=0.
\end{equation}
These equations enforce exactly the allowed interval. If \(f_j\)
are the resulting integer equations in \(m\) variables, set
\(F=\sum_j f_j^2\). This is a nonnegative quartic. Its zero set
is nonempty exactly when the instance is feasible, and at a zero all
coordinates have magnitude at most two.

The slack equations do more than enforce a bound at the zeros. If
\(F\le M\), their residuals have magnitude at most \(\sqrt M\),
so
\[
(1-\sqrt M)/2\le x_i\le2+\sqrt M,
\qquad u_i^2,v_i^2\le\tfrac32+\tfrac32\sqrt M.
\]
Thus every sublevel set is compact. The minimum of \(F\) is attained;
infeasibility cannot be hidden in a sequence tending to infinity.

\begin{lemma}\label{lem:minimum-gap}
There is a polynomial-time computable, short integer \(E_0\) such that
\[
\min F=0\quad\hbox{or}\quad \min F\ge2^{-E_0}.
\]
The integer \(2^{E_0}\) has a polynomial-size integer SLP.
\end{lemma}
\begin{proof}
Let \(\tau\ge1\) bound the coefficient bit lengths of \(F\).
For \(m\ge3\), one sufficient choice is
\[
E_0=(\tau m^4+2m+3)4^{m+1}.
\]
If the minimum is below one, the preceding inequalities put a minimizer
in \([-3,3]^m\). This cube is contained in the image of the standard
simplex under \(X_i=6my_i-3\). The transformed integer quartic has
coefficient bit lengths at most \(\tau+12m+1\le\tau m^4\).
If the instance is infeasible, it is strictly positive on this simplex.
The explicit simplex minimum bound of Jeronimo and Perrucci
\cite[Theorem 1 and its simplified bound]{JP10} gives
\[
2^{-(\tau m^4+1)4^{m+1}}4^{-(m+1)4^{m+1}}=2^{-E_0}.
\]
The case \(\min F\ge1\) is immediate. Instances without variables
are decided directly; otherwise the slack construction gives \(m\ge3\).
The binary length of \(E_0\) is polynomial, and repeated squaring
constructs \(2^{E_0}\) with polynomially many gates.
\end{proof}

We want a finite set containing a global minimum, with a basis that is
simple enough to use in the modular calculation. Let
\begin{equation}\label{eq:critical-constants}
\begin{split}
D&=5^m,\qquad C_0=91(2^m-1),\qquad A_0=C_0+1,\\
Q_0&=A_0 2^{E_0},\qquad T=3A_0.
\end{split}
\end{equation}
Put \(w_i=2^{i-1}\), and define
\begin{equation}\label{eq:objective}
J_w(X)=6Q_0F(X)+\sum_{i=1}^m w_i\bigl(X_i^6+3(X_i+1)^2\bigr).
\end{equation}
The first term makes a positive minimum of \(F\) too expensive.
The sixth powers give pure fifth powers in the critical equations.
The shifted squares make the critical equations squarefree modulo two;
the weights will distinguish their objective values.

A feasible instance has
\(\min J_w\le C_0<T\), by evaluating at a zero of \(F\).
An infeasible instance has \(J_w\ge6A_0>T\) everywhere.
Since \(J_w\) is coercive, its minimum is a critical point. Its critical
system is
\begin{equation}\label{eq:critical-system}
G_i=X_i^5+X_i+1+(Q_0/w_i)F_{X_i}=0\qquad (1\le i\le m).
\end{equation}
The quotients \(Q_0/w_i\) are integers, since \(E_0\ge2m\).
The other terms have degree at most three. For a monomial order refining
total degree, the leading monomials \(X_i^5\) are monic and pairwise
coprime. Consequently they form a Gr\"obner basis, and
\[
A=\Z[X_1,\ldots,X_m]/(G_1,\ldots,G_m)
\]
is free of rank \(D\), with standard monomials
\(X^e\), \(0\le e_i<5\). Over \(\C\) there are exactly
\(D\) points counted with algebraic multiplicity.

\subsection{The objective separates its critical points}
The large power of two needed for the minimum gap also gives
deterministic separation. No generic parameter or projection is needed.
\begin{lemma}\label{lem:generic}
All \(D\) complex critical points of \(J_w\) are simple, and their
objective values are pairwise distinct. None has objective value \(T\).
\end{lemma}
\begin{proof}
Every derivative of the integer sum of squares \(F\) has even
coefficients. Modulo two the critical equations are copies of
\[
 p(U)=U^5+U+1=(U^2+U+1)(U^3+U^2+1).
\]
The factors are irreducible over \(\F_2\), and the five roots are
distinct in \(\F_{64}\). Squarefreeness also follows from
\(U p'(U)=1\) in \(\F_2[U]/(p)\).
The monic standard basis makes \(A\) finite free over \(\Z\).
Its Jacobian norm is odd, hence nonzero, so every characteristic-zero
point is simple. Equivalently, \(A\otimes\Z_2\) is finite
\emph{\'etale}: in the unramified degree-six ring of integers,
the \(5^m\) residue tuples lift uniquely and exhaust the points.

Compare two distinct lifted points \(x,y\), with \(v_2(2)=1\).
If \(x_j\equiv y_j\pmod2\), the divided difference
\[
 \frac{p(x_j)-p(y_j)}{x_j-y_j}\equiv p'(x_j)\not\equiv0\pmod2
\]
is a unit (interpreted polynomially if the coordinates coincide).
The critical equations and the even coefficients of \(F_{X_j}\)
therefore give
\[
 v_2(x_j-y_j)\ge v_2(Q_0)-(j-1)+1.
\]
Set \(H(U)=U^6+3(U+1)^2\). At a root of \(p\) modulo two,
\(H(U)=U+1\). If \(i\) is the first coordinate with different
residues, the \(i\)-th term of
\[
 J_w(x)-J_w(y)=6Q_0(F(x)-F(y))
       +\sum_j2^{j-1}\bigl(H(x_j)-H(y_j)\bigr)
\]
has valuation exactly \(i-1\). Every earlier summand has valuation
at least \(v_2(Q_0)+1\), every later summand has valuation at least
\(i\), and the \(F\) term has valuation at least \(v_2(Q_0)+1\).
Thus \(v_2(J_w(x)-J_w(y))=i-1\), proving separation.
This is an algebraic nonvanishing statement, so it holds over
\(\C\) as well as over an algebraic closure of \(\Q_2\).

Finally \(A_0=91(2^m-1)+1\) is even, so \(T\) is even.
Modulo two, \(J_w-T=X_1+1\), which is a unit since \(p(1)=1\).
Its norm is odd, and it cannot vanish at any complex point.
\end{proof}

The integral pure-monic presentation ensures that reduction modulo two
detects every component. Reduction of an arbitrary nonmonic system
would not justify this proof. In particular, the critical points have
distinct values and the real global minimum is unique.

Define
\begin{equation}\label{eq:projection-norms}
 q(z)=\Norm_A(z-J_w),\qquad g(z)=z-T.
\end{equation}
The integer polynomial \(q\) is monic of degree \(D\) and squarefree.
Its roots are exactly the objective values of the complex critical
points. A real root \(\alpha\) comes from one such point
\(x_\alpha\); conjugation fixes this unique point, so it is real.
Consequently
\begin{equation}\label{eq:objective-sign}
 g(\alpha)=J_w(x_\alpha)-T.
\end{equation}
A root with \(g<0\) certifies feasibility. Conversely, the global
minimum of every feasible instance supplies such a root. We call
these roots \emph{qualifying}. None has \(g=0\).

\subsection{A short list of signs isolates a root}
At the simple real roots of \(q\), the vectors
\[
\bigl(\sgn q_{[1]}(\alpha),\ldots,\sgn q_{[D]}(\alpha)\bigr)
\]
are distinct. Here is the usual Thom argument in the form we need.
A sign condition on derivatives of orders \(j+1,\ldots,D\)
defines an interval, by downward induction. On that interval
\(q^{(j)}\) is monotone, so imposing its sign leaves an interval.
On an interval with the full derivative sign vector of a simple root,
\(q'\) has fixed nonzero sign. The polynomial cannot have two roots
there. Positive factorial scaling gives the same conclusion for Hasse
derivatives.

\begin{lemma}\label{lem:short-list}
In a nonempty collection of at most \(D\) distinct ternary vectors,
one vector can be isolated by fixing at most
\(\lceil\log_2D\rceil\) coordinates.
\end{lemma}
\begin{proof}
If several vectors remain, choose a coordinate on which they differ
and retain its smallest nonempty sign class. At most half remain.
Repeat until one remains.
\end{proof}
For example, among \((+,+),(+,-),(-,0)\), requiring a negative first
coordinate already isolates the third vector. The lemma does not
promise a short description of every prescribed vector. Existence of
one is exactly what the decision problem needs.

Let \(\kappa=\lceil\log_2D\rceil\). Guess a list of at most
\(\kappa\) pairs \((j_i,\sigma_i)\), with
\(1\le j_i\le D\) and \(\sigma_i\in\{-1,0,1\}\).
The empty list is allowed. Applied to the qualifying roots, the lemma
shows that some such list selects exactly one if any exist. An arbitrary
list may select several, which will cause no soundness problem.

\subsection{Turning the list into a polynomial selector}
We next need a rational function that behaves like sign on every
nonzero algebraic value in the list. Effective bounds are easy to obtain
without isolating any roots.
\begin{lemma}\label{lem:algebraic-gap}
Let \(q\in\Z[z]\) be monic of degree \(D\), with coefficients
of magnitude at most \(2^B\), and let \(h\in\Z[z]\) have degree
at most \(E\) and coefficients of magnitude at most \(2^C\).
Set
\[
\beta=C+\lceil\log_2(E+1)\rceil+E(B+1)+1.
\]
At each complex root \(\alpha\) of \(q\),
\[
|h(\alpha)|\le2^\beta,\qquad
h(\alpha)\ne0\ \Longrightarrow\ |h(\alpha)|\ge2^{-(D-1)\beta}.
\]
\end{lemma}
\begin{proof}
Cauchy's bound gives \(|\alpha|\le2^{B+1}\), proving the upper
estimate. In
\[
\Res_z(q(z),h(z)+v)=\prod_{q(\alpha)=0}(h(\alpha)+v),
\]
the lowest nonzero coefficient of \(v\) is the product of all the
nonzero values, including multiplicities. This is a nonzero integer,
so its magnitude is at least one. Bound all factors except the chosen
one by the upper estimate.
\end{proof}

The polynomials in \eqref{eq:projection-norms} have singly exponential
degree and coefficient bit lengths. To see the height bound, reduce a
polynomial of degree at most six times a standard monomial modulo the
pure-monic equations. Each reduction strictly lowers total degree;
a reduction path has \(O(m)\) steps and polynomial branching. The
multiplication-matrix entries have singly exponential bit lengths.
The determinant estimate then multiplies these lengths by at most
\(D\), adding \(O(D\log D)\) for the number of determinant
terms. This is a bound, not a determinant algorithm. Hasse derivatives
increase coefficient magnitudes by at most \(2^D\).

Lemma~\ref{lem:algebraic-gap} therefore gives a short, singly
exponential \(L\ge1\) such that all nonzero values
\(g(\alpha)\) and \(q_{[j]}(\alpha)\) lie between
\(2^{-L}\) and \(2^L\) in magnitude. The same \(L\) can be
chosen uniformly for all allowed parameters and lists.

\begin{lemma}\label{lem:sign-filter}
For this \(L\) and any inverse-polynomial \(\epsilon>0\), there
are integer polynomial SLPs \(N,E\) of polynomial size such that
\(E(y)>0\) for real \(y\), \(N(0)=0\), and
\[
R(y)=N(y)/E(y),\qquad
|R(y)-\sgn y|\le\epsilon
\quad\hbox{when }2^{-L}\le|y|\le2^L.
\]
The function \(R\) preserves the sign of every nonzero real input.
Its degree and coefficient bit lengths are singly exponential.
\end{lemma}
\begin{proof}
For a positive power of two \(a\), use
\(f_a(y)=2ay/(1+a^2y^2)\). The map is odd and sign preserving.
First \(f_1\) puts the positive range in \([2^{-L},1]\).
If the range is \([2^{-K},1]\), with \(K\ge2\), take
\(a=2^{\lceil K/2\rceil}\). Using
\(f_1(y)=f_1(1/y)\), the new range is contained in
\([2^{-\lceil K/2\rceil},1]\). After \(O(\log L)\) steps
the magnitude is at least \(1/2\). Thereafter
\[
1-f_1(y)=\frac{(1-y)^2}{1+y^2}\le(1-y)^2\qquad(0<y\le1)
\]
squares the error at each step. Another
\(O(\log\log(1/\epsilon))\) steps suffice.

Starting with \((N,E)=(y,1)\), a step is implemented without division:
\[
(N,E)\longmapsto(2aNE,\ E^2+a^2N^2).
\]
The denominator stays positive on the real line. Each step doubles the
degree, and its logarithmic coefficient norm obeys
\(H_{\rm new}\le2H+2\log_2a+3\). These give the asserted bounds.
All powers of two have short exponent descriptions and polynomial-size
powering circuits.
\end{proof}

Choose \(\epsilon=1/(16(\kappa+1))\), and for the guessed list
put
\begin{equation}\label{eq:rational-selector}
W(z)=\sum_{i=1}^k\bigl(R(q_{[j_i]}(z))-\sigma_i\bigr)^2
       +\bigl(R(g(z))+1\bigr)^2-\tfrac14.
\end{equation}
The first terms test the requested derivative signs. The last square
tests that the objective is below threshold. At a real root of \(q\),
matching signs contribute at most \((k+1)\epsilon^2<1/4\).
A mismatch contributes at least \((1-\epsilon)^2>1/4\); if
\(g=0\), its term is exactly one. Thus \(W<0\) exactly at the
selected qualifying roots, and \(W\) never vanishes at a real root.

To clear the denominators, let \(h_0=g,\sigma_0=-1\), and for
\(i\ge1\) let \(h_i=q_{[j_i]}\). Write \(N_i=N(h_i)\),
\(E_i=E(h_i)\), and define
\begin{equation}\label{eq:cleared-selector}
V=4\sum_{i=0}^k(N_i-\sigma_iE_i)^2\prod_{j\ne i}E_j^2
      -\prod_{i=0}^kE_i^2.
\end{equation}
This is \(4\prod_iE_i^2\) times \(W\), a positive multiple at
real arguments. Only polynomially many factors occur. Substitution,
the filter circuits, and these products preserve singly exponential
degree and coefficient bit lengths.

\subsection{A nonzero norm and its sign}
The polynomial \(V\) may vanish at nonreal roots of \(q\). If we
simply took its resultant with \(q\), those zeros would erase the
real information. Apply Lemma~\ref{lem:algebraic-gap} once more,
now to \(V\), and let \(\beta_V\) be its upper-value exponent.
The short choice \(K_*=D\beta_V+2\) ensures
\(2^{K_*}|V(\alpha)|>2\) at every nonzero value. Put
\begin{equation}\label{eq:scalar}
U=2^{K_*}V+1,\qquad
Z=\Norm_{A}(U(J_w))=\Res_z(q,U).
\end{equation}
At a zero of \(V\), \(U\) equals one. Elsewhere its magnitude
cannot be zero. At real roots its sign is unchanged. The integer
\(Z\) is therefore nonzero.

Nonreal conjugate roots contribute positive factors
\(U(\alpha)U(\overline\alpha)=|U(\alpha)|^2\), with the
appropriate multiplicities. Multiple real roots also contribute
positively, since they have \(g=0\). The only negative factors are
selected qualifying simple real roots. Consequently
\begin{equation}\label{eq:parity}
\sgn Z=(-1)^{\#\{\text{selected qualifying roots}\}}.
\end{equation}

This proves the two directions of the reduction separately. If any
permitted guess makes \(Z<0\), there is a qualifying root, hence a
real critical point with \(J_w<T\), and the original instance is
feasible. If the instance is feasible, choose a short list isolating
one qualifying root of the fixed system. Its norm is negative. We have proved
\begin{equation}\label{eq:etr-scalar-equivalence}
\text{ETR-INV instance feasible}\quad\Longleftrightarrow\quad
\exists\,\text{list of polynomial length with } Z<0.
\end{equation}

All degree and height estimates above are uniform. For a polynomial
of degree \(E\), a monic reduction path has at most \(E+4m\)
steps. Singly exponential \(E\) therefore still gives singly
exponential coefficient bit lengths in the final multiplication matrix.
The final norm adds at most another singly exponential factor to this
logarithmic height.
There is therefore a known short integer \(H\) such that
\begin{equation}\label{eq:norm-height}
0<|Z|<2^H,\qquad H\le2^{\poly(n)}.
\end{equation}
The constants used to define \(Z\) have polynomial-size SLPs, but
this does not give such an SLP for the whole norm. We next give the
modular algorithm that its implicit description does support.

\subsection{One homotopy for the whole construction}
Our starting system has the explicit grid of roots
\(a\in\{1,\ldots,5\}^m\). Put
\begin{equation}\label{eq:homotopy}
h_i(X_i)=\prod_{a=1}^5(X_i-a),\qquad
F_i(t,X)=(1-t)h_i(X_i)+tG_i(X).
\end{equation}
The quotient \(A_t\) is free over \(\Z[t]\) with the same
standard monomials. Define \(q_t,g,V_t,U_t\) by the earlier
formulas in this algebra, retaining the fixed objective, list, and
filter constants. Then
\[
Z(t)=\Norm_{A_t}(U_t(J_w))\in\Z[t],\qquad Z(1)=Z.
\]
The intermediate systems need no real interpretation. We only use their
polynomial identities.

There is a uniform singly exponential degree bound for \(Z(t)\).
If a polynomial in the quotient has \(X\)-degree \(E\) and
coefficient \(t\)-degree \(L\), a multiplication-matrix entry
starts at total \(X\)-degree at most \(E+4m\). A monic reduction
lowers this degree and adds at most one power of \(t\). Hence
\begin{equation}\label{eq:homotopy-degree}
\deg_t\Norm_{A_t}(V)\le D(L+E+4m).
\end{equation}
Apply this first to the objective norm, then to the selector circuit,
and finally to the outer norm. Choose a short integer \(K\) bounding
the resulting exact degree, before choosing any approximation or prime.

\begin{lemma}\label{lem:branches}
For each prime \(p>5\) and grid point \(a\), a uniform
polynomial-size SLP over \(\F_p\) computes a polynomial vector
\(B_a(t)\) which agrees modulo \(t^{K+1}\) with the formal branch
of \eqref{eq:homotopy} starting at \(a\). The approximant has no
variable denominators. Its syntactic degree has a uniform singly
exponential bound independent of the values of the field constants.
\end{lemma}
\begin{proof}
Write \(X=a+v\). The initial Jacobian is diagonal, with nonzero
entries \(h_i'(a_i)\), so the equations become a proper system
\(v=P_a(t,v)\), with \(P_a(0,0)=0\) and
\((P_a)_v(0,0)=0\). Its coefficients use only inverses of the
nonzero field constants \(h_i'(a_i)\).

Start with \(v_0=0\), which agrees with the solution modulo \(t\).
At an approximant \(v_r\) with \(v_r(0)=0\), the matrix
\(A_r=(P_a)_v(t,v_r)\) is divisible by \(t\). In a Newton
correction, replace the inverse Jacobian by the polynomial
\[
(I-A_r)^{-1}\equiv I+A_r+\cdots+A_r^K\pmod{t^{K+1}}.
\]
Matrix powering and geometric summation take polynomially many gates,
since the matrices have dimension \(m\) and \(K\) is given in
binary. The polynomial Taylor identity gives doubling of precision.
Thus \(O(\log(K+1))\) iterations suffice.

No truncation of an expanded polynomial is performed. The resulting
SLP computes a possibly higher-degree polynomial with the required
congruence. A coarse recurrence
\(\Delta_{r+1}\le C(K+1)(\Delta_r+1)\) bounds its syntactic
degree by \(2^{O((\log(K+1))^2)}\). These constructions are
uniform in the grid index and modulus.
\end{proof}

Over \(\F_p[[t]]\), evaluation on the exact formal branches
identifies the completed quotient with a product of \(D\) copies
of \(\F_p[[t]]\). The evaluation matrix is a tensor Vandermonde
matrix at \(t=0\), hence its determinant is a unit. Thus norms
are products over the branches.

Use the same approximants \(B_a\) in every inner and outer norm.
Coefficient extraction in an auxiliary variable, polynomial substitution,
and the cleared filter operations all preserve congruence modulo
\(t^{K+1}\). This includes substituting an approximate outer branch
into an approximate inner norm. The resulting polynomial \(P(t)\)
therefore satisfies
\begin{equation}\label{eq:formal-congruence}
P(t)\equiv Z(t)\pmod{t^{K+1}}.
\end{equation}
Thus the same branch family computes both norms. The argument uses
formal polynomial congruences, so it requires no analytic continuation
of individual roots through collisions.

Choose a short prime \(p\) with \(p-1\) greater than \(K\),
the syntactic degree bounds of all approximate polynomials, and all
auxiliary extraction indices. The order of these choices is important:
bound the exact degree, choose the precision, bound the approximate
circuit degrees, and only then choose the field. The degree bounds depend
on the circuit shapes, so this is not circular.

\subsection{Products and coefficients modulo a prime}
Two primitives suffice to evaluate the implicit polynomial \(P(t)\).
An exponentially indexed product of short field elements costs one
counting level above access to its factors. To see this, choose a
primitive root \(g\) of \(\F_p^*\) by NP-assisted factoring
and search. If any factor is zero, the product is zero. Otherwise each
factor \(r_i\) has a unique logarithm \(e_i\in\{0,\ldots,p-2\}\).
A counting machine guesses \(i,e,u\), verifies \(g^e=r_i\),
and accepts for \(0\le u<e\). Its count is \(\sum_i e_i\).
Recovering this short sum and exponentiating gives the product.
The primitive-root preprocessing is available in \(\FP^{\PP}\)
in the first application and is already below the later oracle levels.
No efficient discrete-logarithm algorithm is assumed.

For coefficients, if \(\deg A<p-1\) and \(0\le j<p-1\),
finite-field orthogonality gives
\begin{equation}\label{eq:field-fourier}
[v^j]A(v)=-\sum_{\xi\in\F_p^*}\xi^{-j}A(\xi).
\end{equation}
Only the matching exponent has nonzero sum; its sum is \(p-1=-1\)
in the field. Multiple auxiliary variables are handled by one sum over
tuples. The number of variables used here, and the total binary length
of a tuple index, are polynomial.

At a fixed \(t\), first evaluate an inner norm as a product of the
short branch evaluations. Next extract the selected Hasse derivatives
from the product for \(q_t(z+v)\). The qualification polynomial
\(g=z-T\) is directly available. The sign filter
and selector then use polynomially many ordinary field operations and
calls to these values. A second branch product gives the outer norm
approximant \(P(t)\). There are only a fixed number of nested
counting operations.

Finally, \eqref{eq:formal-congruence} and \(\deg Z(t)\le K\)
give the exact specialization
\begin{equation}\label{eq:final-fourier}
Z\bmod p=\sum_{j=0}^K[t^j]P(t)
=-\sum_{\xi\in\F_p^*}P(\xi)\sum_{j=0}^K\xi^{-j}.
\end{equation}
The inner geometric sum is evaluated by binary doubling, including at
\(\xi=1\). This is one further indexed sum. All filter denominators
were cleared as polynomials, so no modular pole is introduced. A prime
may divide \(Z\), or the final critical system may have repeated
roots: neither affects this computation.

\subsection{Recovering the integer sign}
We can now recover the sign of \(Z\) within a fixed CH level.
Choose a short interval of primes above the preceding degree threshold
whose product \(M\) satisfies
\(M>2^{H+4}\), using the height bound \eqref{eq:norm-height}.
For completeness, fixed effective Chebyshev constants give
\(c_1x\le\vartheta(x)\le c_2x\), where
\(\vartheta(x)=\sum_{p\le x}\log p\), above a fixed threshold.
If the lower endpoint is \(P_0\), take
\(X=O(P_0+H+1)\) with
\(c_1X-c_2P_0>(H+4)\log2\). Then the product of primes in
\((P_0,X]\) works. Both endpoints are short, and
\(H_1=X\lceil\log_2X\rceil+1\) bounds \(\log_2M\).
We index integers in this interval and discard composites by primality
testing; we never need an algorithm for the indexed prime.

Let \(z_i=Z\bmod p_i\) be the canonical residue, and set
\[
A_i=z_i\bigl((M/p_i)^{-1}\bmod p_i\bigr).
\]
The inverse cofactor is computed by the same product primitive on the
other primes modulo \(p_i\). These \(A_i\) are short and
accessible within a fixed CH level. The Chinese remainder theorem says
\[
\sum_i\frac{A_i}{p_i}=\text{integer}+\frac ZM.
\]
The displacement has magnitude below \(1/8\), and, when nonzero,
greater than \(2^{-H_1}\). Lemma~\ref{lem:first-event} applies.
Each coarse sum costs one additional counting level. A short existential
choice of the first event, followed by a universal check of earlier
indices, gives the sign. Adding a fixed number of Boolean quantifiers
above a fixed CH oracle remains within CH.

\begin{proof}[Proof of Theorem~\ref{thm:etr}]
Guess the derivative-sign list in
\eqref{eq:etr-scalar-equivalence}. The preceding construction decides
whether its scalar is negative in a fixed level of CH, uniformly over
all guesses. Merge these guesses with the first-event existential
quantifier. Soundness and completeness are exactly those established
in \eqref{eq:etr-scalar-equivalence}. This proves the containment.
\end{proof}
The proof has retained the real information in the sign of the norm,
while computing that norm through complex formal branches. We next
improve the sign reconstruction for integers with a simpler modular
representation, then return to the level used by this real construction.

\section{All-moduli representatives and deterministic sign reconstruction}\label{sec:posslp}

We now turn to the sign of an integer given by a polynomial-size arithmetic
circuit. Modular evaluation is easy, but the integer can have exponentially
many bits. The difficulty is to recover its sign without constructing those
bits or introducing another counting oracle at each reconstruction step.
We use a rational identity whose value is positive and extremely small.
Multiplying this identity by the unknown integer encodes its sign as the
direction of a small displacement from an integer.

The weights in the identity are binomial coefficients of exponential size.
We only need representatives of these weights modulo short integers. Our
first lemma produces such representatives as counting functions. Their
values need not be canonical residues; this flexibility is what permits
the subsequent weighted sum to remain a single counting computation.

\subsection{Coefficient representatives over finite rings}

\begin{lemma}\label{lem:ring-coefficients}
Let \(f\in\Z[X_1,\ldots,X_r]\) be given by a division-free integer SLP.
Let \(D\) bound every individual degree, and let
\(0\le k_j\le D\). There is a uniform \(\SharpP\) function
\(W(f,D,\mathbf k,m)\) satisfying
\[
 W(f,D,\mathbf k,m)
 \equiv [\mathbf X^{\mathbf k}]f\pmod m
\]
for every positive integer \(m\).
The parameters \(D,\mathbf k,m\) are encoded in binary. In particular,
the representative has polynomially many bits in the complete input
length, even when the coefficient has exponentially many bits.
\end{lemma}

\begin{proof}
The case \(m=1\) is immediate. We first construct a representative modulo
one maximal prime-power factor \(p^a\) of \(m\). Choose a prime
\(d>\max(2,\lceil\log_2(D+2)\rceil)\), with \(d\ne p\).
Scanning two successive Bertrand intervals and testing primality finds
such a \(d\) in polynomial time. Put \(Q=p^d\). Then \(Q-1>D\),
and \(d\), \(\log Q\), and \(a\log p\) are polynomially bounded.

We average over a family of field presentations instead of choosing one.
Let \(\mathcal I\) consist of the monic irreducible polynomials of
degree \(d\) over \(\F_p\) whose \(Y^{d-1}\) coefficient is zero.
Its cardinality is
\[
 I_0=\frac{p^{d-1}-1}{d}.
\]
Indeed, there are \(p^{d-1}\) trace-zero elements in \(\F_Q\).
Since \(d\) is prime, an element of smaller degree lies in \(\F_p\);
its trace is \(d\) times the element, so only zero has trace zero.
The remaining elements occur in orbits of size \(d\). Furthermore,
\(dI_0\equiv-1\pmod p\). Thus \(I_0\) is a unit modulo \(p^a\),
and both \(I_0\) and its inverse are computable in polynomial time.
Membership in \(\mathcal I\) is testable in polynomial time.

For each \(F\in\mathcal I\), lift its coefficients to
\(\{0,\ldots,p-1\}\) and form
\[
 \mathcal A=(\Z/p^a\Z)[Y]/(F(Y)).
\]
This is the unramified Galois ring \(\operatorname{GR}(p^a,d)\).
An element has \(d\) coordinates modulo \(p^a\), in the basis
\(1,Y,\ldots,Y^{d-1}\), so ring arithmetic and binary powering take
polynomial time. An element is a unit precisely when its reduction
modulo \(p\) is nonzero. Indeed, a residue-field inverse lifts to a
product of the form \(1+pv\), and a finite geometric series in \(pv\)
corrects this product to 1.

We need a multiplicative interpolation set whose cardinality is a unit.
For each nonzero \(x\in\F_Q\), lift its coordinate vector to
\(\widetilde x\in\mathcal A\) and put
\[
 t_x=\widetilde x^{\,Q^{a-1}}.
\]
The exponent has only \(O(a\log Q)\) bits. These are the Teichm\"uller
units, as can be verified directly. Every principal unit \(1+pu\) has
order dividing \(p^{a-1}\), since
\((1+p^j u)^p\in1+p^{j+1}\mathcal A\) for \(j\ge1\).
This holds also for \(p=2\). As
\(\widetilde x^{Q-1}\equiv1\pmod p\), we obtain
\(t_x^{Q-1}=1\); moreover \(t_x\bmod p=x\).
Two such roots with the same residue have a ratio whose order divides
both \(p^{a-1}\) and \(Q-1\), and hence are equal. Consequently
\(\mathcal T=\{t_x:x\in\F_Q^*\}\) is a cyclic group of order
\(Q-1\), mapped isomorphically onto \(\F_Q^*\).

For this group, ordinary character orthogonality remains valid over the
ring:
\[
 \sum_{t\in\mathcal T}t^j=
 \begin{cases}
 Q-1,&j\equiv0\pmod{Q-1},\\
 0,&j\not\equiv0\pmod{Q-1}.
 \end{cases}
\]
To prove the second case, let \(g\) generate \(\mathcal T\).
The reduction of \(g^j-1\) is nonzero, so \(g^j-1\) is a unit;
the geometric-sum identity then applies. The invertibility assertion
is essential: a nontrivial root of unity in a general ring would not
by itself justify this calculation.

In one variable we therefore have
\[
 (Q-1)^{-1}\sum_{t\in\mathcal T}t^{Q-1-k}f(t)=[X^k]f
 \quad\text{in }\mathcal A.
\]
There is no aliasing, because for \(0\le k,\nu\le D<Q-1\), the
exponent \(Q-1-k+\nu\) is a multiple of \(Q-1\) exactly when
\(\nu=k\). Applying this independently in each variable gives
\begin{equation}\label{eq:ring-coefficient-fourier}
 (Q-1)^{-r}
 \sum_{t_1,\ldots,t_r\in\mathcal T}
 \left(\prod_{j=1}^r t_j^{Q-1-k_j}\right)f(t_1,\ldots,t_r)
 = [\mathbf X^{\mathbf k}]f.
\end{equation}

For each presentation and each summand, include the scalar factor
\(I_0^{-1}(Q-1)^{-r}\) and take
its canonical constant coordinate in \(\{0,\ldots,p^a-1\}\).
We sum these nonnegative integers without reducing the total modulo
\(p^a\). Coordinate extraction is additive modulo \(p^a\), and the
right side of \eqref{eq:ring-coefficient-fourier} is the same scalar
for every presentation. Summing also over \(\mathcal I\) therefore
multiplies the desired coefficient by \(I_0^{-1}|\mathcal I|=1\)
modulo \(p^a\). The resulting integer \(W_{p^a}\) is a representative
of that coefficient. We consistently use the constant coordinate, so
the argument does not rely on an additional trace normalization.

There are \(I_0(Q-1)^r\) summands, each indexed by a polynomial-length
string. A counting machine guesses \(F\), verifies irreducibility and
the trace coefficient, guesses the field vectors, and computes its
short nonnegative contribution \(h\). It accepts on exactly \(h\)
further branches by guessing an integer smaller than \(h\).

To handle the factorization of \(m\), guess its canonical sorted list
of prime powers. Primality, the exact product, the ordering, and canonical
padding can all be checked in polynomial time. Unique factorization
gives exactly one successful certificate. This removes the need for
a factoring oracle; it does not give a deterministic factoring algorithm.

Finally, for the maximal prime powers \(m_\ell\) dividing \(m\), set
\[
 E_\ell=\frac m{m_\ell}
 \left(\left(\frac m{m_\ell}\right)^{-1}\bmod m_\ell\right),
 \qquad
 W=\sum_\ell E_\ell W_{m_\ell}.
\]
These short CRT weights give the required congruence modulo \(m\).
They are nonnegative and less than \(m\), so this weighted sum is still
in \(\SharpP\). For each prime-power factor,
\(W_{p^a}<I_0(Q-1)^r p^a\). There are at most \(\log_2 m\) factors,
\(\log I_0\le d\log p\), and \(\log Q=d\log p\).
The complete count consequently has polynomial bit length in
\(r,\log(D+2),\log m\), as required.
\end{proof}

The lemma supplies a nonnegative representative even for a negative
coefficient. Only its congruence matters. In particular, the proof makes
no assertion that reducing a counting-function value to its canonical
residue preserves the counting class.

The extension ring serves two separate purposes. Its residue field
provides more interpolation points than the formal degree, while its
characteristic retains the requested prime-power modulus. The
interpolation length is \(Q-1\), which is coprime to the characteristic.
Thus increasing the field degree avoids aliasing without introducing a
noninvertible averaging factor. The eventual coefficient belongs to the
base ring. Averaging over presentations is valid because this scalar
is independent of the presentation; arbitrary intermediate coordinate
values would not have this property.

\subsection{A beta kernel for integer signs}

We isolate the reconstruction step so that its oracle requirement is
explicit. Here \(\GapP^{\PH}\) means that a fixed finite level of PH
suffices for the oracle computation.

The next statement is useful beyond SLPs. Its hypothesis specifies an
integer through compatible congruences, together with a height bound.
The algorithm need not know a compact arithmetic expression for that
integer. What matters is the complexity of a representative for each
congruence. The representative may depend on the modulus in any uniform
way, and it may be signed. We impose no joint size bound on the
integer reconstructed from all the congruences except the stated height
bound.

\begin{lemma}\label{lem:beta-sign}
Suppose an integer family \(Z(x)\) has a known bound
\(|Z(x)|<2^{H(x)}\), where \(H(x)\ge1\) is polynomial-time computable
in binary and has polynomially many bits. Suppose a uniform function
\(R(x,m)\in\GapP^{\PH}\) satisfies
\[
 R(x,m)\equiv Z(x)\pmod m
\]
for every positive binary modulus \(m\). Then \(\sgn Z(x)\) is
computable in \(\Pclass^{\NP^{\PP}}\).
\end{lemma}

\begin{proof}
We construct a positive rational multiplier small enough to place the
unknown integer in the range of Lemma~\ref{lem:first-event}. Set
\[
 N=H(x)+4,\qquad m_i=N+1+i\quad(0\le i\le N).
\]
Binary powering gives an SLP of size \(O(\log N)\) for \((1+X)^N\).
Lemma~\ref{lem:ring-coefficients} supplies
\(W_i\in\SharpP\) with
\(W_i\equiv\binom Ni\pmod{m_i}\). Put
\[
 A_i=(-1)^iW_iR(x,m_i),\qquad S=\sum_{i=0}^N\frac{A_i}{m_i}.
\]
Each \(A_i\) is a \(\GapP^{\PH}\) function: we take a product of
two counting functions and a polynomial-time sign. The exact identity
\begin{equation}\label{eq:beta-kernel}
 C_N:=\sum_{i=0}^N\frac{(-1)^i\binom Ni}{N+1+i}
 =\int_0^1x^N(1-x)^N\,dx
 =\frac{(N!)^2}{(2N+1)!}>0
\end{equation}
shows that \(S\) is an integer plus \(\delta=C_NZ(x)\).
Indeed, each \(W_iR(x,m_i)-\binom Ni Z(x)\) is divisible by \(m_i\).
The denominators need not be prime or pairwise coprime.

This use of representatives is local to each denominator. If either
factor in a numerator is changed by a multiple of its modulus, the
rational sum changes by an integer. The displacement therefore depends
only on the original binomial coefficient and \(Z(x)\), even though
the counting machines may return quite different integer lifts at
different indices. We never need to reconcile these lifts by a global
canonical reconstruction.

The integral gives \(C_N\le4^{-N}\), and hence \(|\delta|<1/8\).
If \(Z(x)\ne0\), integrality gives
\[
 |\delta|\ge C_N\ge(2N+1)^{-(2N+1)}\ge2^{-J},
 \qquad J=(2N+1)\lceil\log_2(2N+1)\rceil.
\]
Both \(J\) and every \(m_i\) have short binary descriptions. We use
the bounds on \(C_N\) without constructing its numerator or denominator.

For clarity, we give the counting implementation of the approximants
in Lemma~\ref{lem:first-event}. Fix a polynomial bound
\(|A_i|\le2^{a(n)}\), uniformly for \(n=|x|\), and choose
\[
 b=a(n)+\lceil\log_2(N+1)\rceil+12.
\]
At a binary scale \(0\le j\le J\), form the integer
\begin{equation}\label{eq:beta-approximation}
 T_j=\sum_{i=0}^N A_i
 \left\lfloor\frac{2^b(2^j\bmod m_i)}{m_i}\right\rfloor.
\end{equation}
Every multiplier is a polynomial-time computable short integer.
In particular, modular exponentiation takes time polynomial in the
binary length of \(j\). Products of two GapP functions, multiplication
by these weights, and exponentially indexed summation keep
\(T_j\) in \(\GapP^{\PH}\). Its entire output still has polynomial
bit length.

After retrieving \(T_j\), we take
\(v_j=(T_j\bmod2^b)/2^b\). Its circular distance from
\(2^j\delta\bmod1\) is at most \(\varepsilon=2^{-9}\), since
\[
 \left|2^{-b}T_j-
 \sum_i A_i\frac{2^j\bmod m_i}{m_i}\right|
 \le (N+1)2^{a(n)-b}\le2^{-12}.
\]
Replacing a modular power by \(2^j\) changes the rational sum by an
integer. The estimate uses absolute values and therefore applies to
signed representatives.

To retrieve these counts we use exact functional counting absorption.
Theorem~4.8 of Green et al.~\cite{GKRST95} states that for every
\(f\in\SharpP^{\mathrm{BPP}^{\oplus\Pclass}}\), there are
\(h\in\SharpP\) and a polynomial \(p\) such that
\begin{equation}\label{eq:functional-absorption}
 f(y)\equiv
 \left\lfloor h(y,0^m)/2^{p(|y|)}\right\rfloor\pmod{2^m}
 \qquad\text{for all }y,m.
\end{equation}
This is an explicit block-extraction form of the functional
Toda--Watanabe reduction~\cite{TW92}.
When deriving \eqref{eq:functional-absorption} from the amplification
construction, one uses modulus \(2^{m+p(|y|)}\) before the final shift.
Choose \(m\) above the polynomial output-bit bound of \(f\).
The block then recovers \(f(y)\) exactly from one \(\SharpP\) value
on a polynomial-size input. Toda's inclusion
\(\PH\subseteq\mathrm{BPP}^{\oplus\Pclass}\)~\cite{Toda91}
gives the \(\SharpP^{\PH}\) special case. Taking a difference handles
GapP. Consequently
\[
 \GapP^{\PH}\subseteq\FP^{\SharpP}=\FP^{\PP}.
\]
This is an exact functional reduction; the language containment
\(\PH\subseteq\Pclass^{\SharpP}\) alone would not justify the step.

All the parameters supplied to this reduction are short. In particular,
\(y\) includes the binary scale \(j\), rather than a unary string of
length \(j\). The unary padding \(0^m\) in
\eqref{eq:functional-absorption} is polynomial because it bounds the
output length of \(T_j\), not the output length of \(Z(x)\).
For a signed count, the two nonnegative parts have polynomial bounds
from their counting machines, so the same reasoning applies to each
part before subtraction.

We can therefore evaluate \(v_j\), the event predicate
\(\operatorname{dist}(v_j,\Z)\ge1/8\), and its side
\(v_j<1/2\), in \(\Pclass^{\PP}\). Lemma~\ref{lem:first-event}
shows that there is no event precisely when \(Z(x)=0\), and otherwise
the first event has the sign of \(\delta\). Its hypotheses hold with
the signed weights \(A_i\), precision \(b\), and search bound \(J\)
just established. This invokes the no-wrap argument of that lemma,
including its treatment of zero and of an event at the initial scale.

To locate that event, consider the monotone prefix predicate
\[
 \operatorname{Prefix}(x,t)
 \iff \exists j\le t\;
       \operatorname{dist}(v_j,\Z)\ge1/8,
 \qquad 0\le t\le J.
\]
The guessed index has polynomial length and its verification is in
\(\Pclass^{\PP}\), so Prefix belongs to \(\NP^{\PP}\).
We first query Prefix at \(J\) to detect zero. Otherwise,
\(O(\log(J+1))\) binary-search queries find the least event; its side
is also accessible to the same oracle. Finally, \(C_N>0\) implies
\(\sgn\delta=\sgn Z(x)\).
\end{proof}

The order of the operations is important. We first form the counting
function in \eqref{eq:beta-approximation}, using the representatives only
through sums and products. We then retrieve its short value and perform
ordinary deterministic arithmetic. Thus canonical reduction occurs
after the counting sum, where it does not add a counting level.

The fixed precision does not discard the information needed for the
sign. It identifies a scale at which that information becomes visible.
The first-event lemma makes this scale search reliable, while the
prefix oracle makes it efficient: the number of possible scales is
exponential, but their binary-search depth is polynomial. These two
facts account for the different roles of PP and the outer NP oracle
in the conclusion.

\subsection{Applying the sign criterion to arithmetic circuits}

\begin{proof}[Proof of the deterministic assertion in Theorem~\ref{thm:posslp}]
Let the input SLP have \(s\) arithmetic gates, and choose \(\beta\ge1\)
so that every input integer constant has magnitude at most \(2^\beta\).
Starting with \(M_0=2^\beta\), each arithmetic gate raises a valid
bound on the maximum magnitude to at most its square: addition and
subtraction use \(2M\le M^2\), and multiplication uses \(M^2\).
Hence its integer output satisfies
\[
 |Z|\le2^{\beta2^s}<2^H,
 \qquad H=(\beta+1)2^{s+1}.
\]
The bound \(H\) has \(O(s+\log(\beta+1))\) bits. Evaluation of every
gate modulo an arbitrary binary integer \(m\) is polynomial time,
and supplies the representative required by Lemma~\ref{lem:beta-sign}.
We obtain \(\PosSLP\in\Pclass^{\NP^{\PP}}\). This also covers the
usual formulation with constants 0 and 1, since a binary integer constant
can be constructed with a number of gates linear in its encoding length.

For coefficient positivity, let \(f\) be an integer polynomial SLP
with a target exponent vector supplied in binary. Its total degree is
at most \(2^s\). Its coefficient \(\ell_1\)-norm obeys the same
recurrence as the magnitude bound above: addition and subtraction are
subadditive, multiplication is submultiplicative, and an input variable
has norm 1. Every coefficient therefore satisfies the same bound
\(|c|<2^H\). We immediately return zero if a target exponent exceeds
\(2^s\). Otherwise Lemma~\ref{lem:ring-coefficients} gives uniform
\(\SharpP\) representatives for
\(c=[\mathbf X^{\mathbf k}]f\) modulo every short integer.
These meet the hypothesis of Lemma~\ref{lem:beta-sign}, proving
coefficient positivity in \(\Pclass^{\NP^{\PP}}\) as well.
Neither application requires writing the potentially exponential-length
integer whose sign is being tested.
\end{proof}

The representation hypothesis also describes the limit of this result
for the ETR construction. To invoke the sign criterion there, one would
need \(\GapP^{\PH}\) representatives of its scalar norm expression,
or a reduction to a coefficient of a polynomial-size SLP. The available
indexed norm computations have not been put in that form. For example,
a modular product algorithm may first retrieve a counting value and
then exponentiate a field element by that value. The resulting
polynomial-time computation with a counting oracle is not one of the
sums or fixed products covered by GapP closure. The coefficient lemma
handles an SLP coefficient directly, but does not eliminate that
separate nonlinear operation.

\section{Short-seed randomized sign reconstruction}\label{push4:sec}

The first-event argument locates a useful scale by asking whether an
earlier event exists. We can replace this existential search by a
randomized magnitude test. The test multiplies the unknown displacement
by a random integer and observes the resulting phase. A multiplier with
exponentially many bits is useful here, so we generate its bits from a
short seed and evaluate the phase without constructing the multiplier.
This gives a second upper bound for the integer problems.

\begin{theorem}\label{push4:posslp}
The problems \(\PosSLP\) and positivity of a specified coefficient of an
integer polynomial SLP belong to \(\BPP^{\PP}\subseteq\Clevel2\).
\end{theorem}

We first state the sign-reconstruction interface. All length bounds in
the next lemma refer to the original input length \(n\); in particular,
the indices of exponential sums have polynomial length.

\begin{lemma}\label{push4:sign}
Suppose
\[
 S(x)=\sum_{i<L}\frac{A(x,i)}{m(x,i)}=z+\delta,
 \qquad z\in\Z,
\]
where \(L,J\) are positive integers with polynomial-length binary
descriptions computable in polynomial time from \(x\), the positive
integers \(m(x,i)\) are polynomial-bit and polynomial-time computable,
and \(A(x,i)\) is a uniform \(\GapP^B\) function satisfying
\(\lvert A(x,i)\rvert\le 2^{a(n)}\) for a known polynomial \(a\).
Assume
\[
 \delta=0\quad\text{or}\quad
 2^{-J}\le\lvert\delta\rvert<2^{-32}.
\]
If the relevant \(\GapP^B\) functions can be retrieved exactly in
\(\FP^D\), then the sign of \(\delta\) is decidable in \(\BPP^D\).
\end{lemma}

The retrieval hypothesis applies to the complete weighted sums formed
below, rather than only to individual numerators. For our integer
application we use \(B=\PH\), \(D=\PP\), and exact functional counting
absorption. For the norm construction we use \(B=\Clevel3\) and
\(D=\Clevel4\).

\subsection{Short seeds for long phase products}

A Fourier shape on \(T\) bits is a function
\[
 f(u_0,\ldots,u_{T-1})=\prod_{j<T}f_j(u_j),
 \qquad f_j:\{0,1\}\longrightarrow\C,
 \quad\lvert f_j\rvert\le1.
\]
We need a generator which fools every such product and allows direct
access to an output coordinate. The following recursive construction
uses standard small-bias ideas \cite{AGHP92,NN93}. We give the proof
and a field implementation to make the required random access explicit.

\begin{lemma}\label{push4:generator}
Given \(T\ge1\) in binary and an integer \(q\ge2\), put
\(\epsilon=2^{-q}\). There is
a uniform generator
\[
 G:\{0,1\}^{d}\longrightarrow\{0,1\}^{T},
 \qquad d=O\!\left(\log^2(T/\epsilon)\right),
\]
whose expectation differs from the uniform expectation by at most
\(\epsilon\) for every Fourier shape. Given its seed and a binary
coordinate index, a selected output bit is computable in time
polynomial in \(\log(T/\epsilon)\).
\end{lemma}

\begin{proof}
We first construct a small-bias step on \(d\) bits. Choose a field
\(\F_{2^e}\) with \(2^e\ge d/\lambda\). For independent uniform
\(a,b\in\F_{2^e}\), put
\[
 V_j(a,b)=\operatorname{Tr}_{\F_{2^e}/\F_2}(ab^j),
 \qquad 0\le j<d.
\]
For a nonzero vector \(c\in\F_2^d\), the polynomial
\(P_c(X)=\sum_jc_jX^j\) is nonzero and has degree at most \(d-1\).
Nondegeneracy of the trace pairing gives
\begin{equation}\label{push4:bias}
 \mathbb E_{a,b}(-1)^{c\cdot V(a,b)}
   =\Pr_b[P_c(b)=0]\le\frac{d-1}{2^e}\le\lambda.
\end{equation}
Thus the walk \(x\mapsto x+V(a,b)\) on \(\F_2^d\) has all its
nonconstant character eigenvalues bounded by \(\lambda\). Fourier
expansion, Parseval, and Cauchy--Schwarz imply, for arbitrary complex
functions \(F,H\) of modulus at most one,
\begin{equation}\label{push4:mixing}
 \left|\mathbb E_{x,a,b}F(x)H(x+V(a,b))
           -\mathbb E_xF(x)\,\mathbb E_xH(x)\right|\le\lambda.
\end{equation}
One subtracts the two means before applying the spectral bound; each
centered function has \(L_2\) norm at most one.

Pad the desired length to \(T'=2^k\), and set
\(\lambda=\epsilon/T'\). Start with \(G_0(x)=x\), using one seed
bit. If \(G_r\) has seed length \(d_r\) and output length \(2^r\),
define
\begin{equation}\label{push4:recursion}
 G_{r+1}(x,a,b)=G_r(x)\,\Vert\,G_r(x+V_r(a,b)),
\end{equation}
where \(V_r\) is the step above on \(d_r\) bits. Split a Fourier
shape into its two halves. Equation~\eqref{push4:mixing} changes the
expectation by at most \(\lambda\) when the two seed arguments are
made independent. If \(\epsilon_r\) bounds the error of \(G_r\),
the two marginal expectations introduce at most \(2\epsilon_r\).
Consequently
\[
 \epsilon_{r+1}\le2\epsilon_r+\lambda,
 \qquad \epsilon_0=0,
 \qquad \epsilon_k\le(2^k-1)\lambda<\epsilon.
\]
Discarding padded coordinates preserves this guarantee, since their
Fourier factors can be chosen identically one.

We can choose each field degree \(e_r\) to be the least power of two
at least \(\lceil\log_2(d_r/\lambda)\rceil\). Then
\(d_{r+1}=d_r+2e_r\), giving
\(d_k=O(\log^2(T/\epsilon))\). A requested output coordinate follows
only one branch at every level of \eqref{push4:recursion}. For a right
branch, compute the short vector \(V_r(a,b)\) and change the seed;
for a left branch, retain the seed. There are only \(k\) such steps.

Here the field presentations can also be constructed deterministically
within the claimed time. Start with \(K_0=\F_2\), \(\gamma_0=1\),
and recursively set
\[
 K_{r+1}=K_r[\theta]/(\theta^2+\theta+\gamma_r),
 \qquad \gamma_{r+1}=\gamma_r\theta.
\]
Inductively \(\gamma_r\) has absolute trace one. The quadratic is
therefore irreducible by the Artin--Schreier criterion, and its
relative trace sends \(\theta\) to one. Transitivity of trace shows
that \(\gamma_{r+1}\) again has absolute trace one. This constructs
the required power-of-two degrees. Arithmetic in the quadratic tower,
and trace by iterated Frobenius squaring, take polynomial time in
the final field degree. The coordinate algorithm is therefore
polynomial-time, without any field-search oracle.
\end{proof}

The functions defining a Fourier shape need not be supplied to the
generator or be efficiently computable. In particular, the same
generator fools \(\prod_{\ell<t}\exp(2\pi i\delta2^\ell u_\ell)\)
even though \(\delta\) itself may require exponentially many bits.

\subsection{Evaluating one phase}

Fix a positive binary integer \(t\), possibly exponential in \(n\).
For a seed \(s\), let \(k_\ell=G(s)_\ell\) and consider the
implicit integer
\[
 K=\sum_{\ell<t}k_\ell2^\ell.
\]
Since \(K\) is an integer, \(KS\) and \(K\delta\) have the same
fractional part. We evaluate that fractional part by reducing only
short modular powers. Put
\[
 r_{i,\ell}=2^\ell\bmod m_i,
 \qquad
 b=a(n)+\lceil\log_2L\rceil+\lceil\log_2(t+1)\rceil+30,
\]
and form
\begin{equation}\label{push4:raw-phase}
 T_{t,s}=\sum_{i<L}\sum_{\ell<t}
 A_i k_\ell\left\lfloor\frac{2^b r_{i,\ell}}{m_i}\right\rfloor.
\end{equation}
The pair \((i,\ell)\) has polynomially many bits. All the displayed
weights are polynomial-time computable short integers, including
\(k_\ell\) by Lemma~\ref{push4:generator}. Thus
\(T_{t,s}\in\GapP^B\), by closure under short integer weights and
exponentially indexed sums. Its entire output has polynomial length.

Termwise rounding and the bound on \(A_i\) give
\begin{equation}\label{push4:phase-error}
 \dist_{\R/\Z}\!\left(T_{t,s}/2^b,K\delta\right)
 \le Lt\,2^{a(n)-b}<2^{-30}.
\end{equation}
This estimate applies to signed numerators by the triangle inequality.
The precision depends on \(\log t\), so it remains polynomial even
when \(K\) has exponentially many bits.

Retrieve \(T_{t,s}\) in \(\FP^D\), take its remainder modulo
\(2^b\), and approximate
\[
 \cos\!\left(2\pi\frac{T_{t,s}\bmod2^b}{2^b}\right)
\]
to error \(1/1024\). By \eqref{push4:phase-error}, this approximates
\(\cos(2\pi K\delta)\) to error less than \(1/512\). We clip the
answer to \([-1,1]\). The remainder and cosine occur after one exact
counting retrieval. Neither operation is inserted into the counting
sum \eqref{push4:raw-phase}.

\subsection{Finding a scale with bounded-error tests}

For a genuinely uniform \(t\)-bit integer \(K\), define
\[
 \mu_t=\mathbb E_{0\le K<2^t}\cos(2\pi K\delta).
\]
The following estimates separate small and large scales:
\begin{equation}\label{push4:regimes}
 \begin{aligned}
 2^t|\delta|\le1/16&\quad\Longrightarrow\quad\mu_t>9/10,\\
 2^t|\delta|\ge4&\quad\Longrightarrow\quad|\mu_t|\le1/8.
 \end{aligned}
\end{equation}
The first follows because every phase has absolute value less than
\(\pi/8\). For the second, sum a geometric progression and use
\(\lvert\sin(\pi\delta)\rvert\ge2|\delta|\):
\[
 \left|\mathbb E_K e^{2\pi iK\delta}\right|
 =\left|\frac{1-e^{2\pi i2^t\delta}}
              {2^t(1-e^{2\pi i\delta})}\right|
 \le\frac{1}{2\cdot2^t|\delta|}.
\]

Generate the multiplier bits using Lemma~\ref{push4:generator} with
error \(1/128\). The resulting expectation differs from \(\mu_t\)
by at most \(1/128\), since the complex exponential is a Fourier
shape. The preceding phase evaluation adds less than \(1/512\).
Independent samples estimate their expectation within \(1/64\) with
failure probability \(\rho\), using \(O(\log(1/\rho))\) seeds by
Hoeffding's inequality. Hence we obtain an estimate \(\widehat\mu_t\)
with
\[
 |\widehat\mu_t-\mu_t|<1/16
\]
except with probability \(\rho\). Every sample is an \(\FP^D\)
computation on polynomial-length inputs.

Call a tested scale \emph{high} when \(\widehat\mu_t>1/2\), and
\emph{low} otherwise. On an accurate test,
\begin{equation}\label{push4:labels}
 \text{high at }t\ \Longrightarrow\ 2^t|\delta|<4,
 \qquad
 \text{low at }t\ \Longrightarrow\ 2^t|\delta|>1/16.
\end{equation}
Labels inside the transition range need not be monotone. We use only
these two implications.

\begin{proof}[Proof of Lemma~\ref{push4:sign}]
First test \(t_{\max}=J+3\). If \(\delta\ne0\), then
\(2^{t_{\max}}|\delta|\ge8\), so an accurate test is low. If
\(\delta=0\), every generated phase is zero modulo integers and the
test is high. A high answer therefore gives the zero output, subject
to the stated error probability.

Otherwise start with a high label at \(t=0\), where \(K=0\), and
the low label at \(t_{\max}\). Bisect this interval, retaining a
high left endpoint and a low right endpoint according to each tested
midpoint. After \(O(\log J)\) tests the endpoints are adjacent:
\(h\) is high and \(h+1\) is low. On accurate tests,
\eqref{push4:labels} gives
\[
 1/32<2^h|\delta|<4.
\]
This conclusion uses only the labels of the two endpoints, so arbitrary
behavior within the transition range causes no problem. Since
\(|\delta|<2^{-32}\), the low endpoint is at least 29 and \(h\ge28\).
Set \(j=h-6\). Then
\begin{equation}\label{push4:final-scale}
 1/2048<|2^j\delta|<1/16.
\end{equation}

One deterministic phase query now determines the sign. Take
\(b_0=a(n)+\lceil\log_2L\rceil+30\) and retrieve
\[
 U_j=\sum_{i<L}A_i
       \left\lfloor\frac{2^{b_0}(2^j\bmod m_i)}{m_i}\right\rfloor
 \quad\text{in }\FP^D.
\]
Its normalized centered remainder is within \(2^{-30}\) of
\(2^j\delta\). Equation~\eqref{push4:final-scale} ensures that its
strict sign equals that of \(\delta\).

There is a known polynomial bound \(Q=O(\log(J+4))\) on the number
of tests. Give each test failure probability at most \(1/(10Q)\) and
use fresh seeds. Conditional on any preceding transcript, the next
scale is fixed and its error bound still holds. A union bound gives
success probability at least \(9/10\) for the entire adaptive search.
All seeds and indices have polynomial length, establishing \(\BPP^D\).
The error can be reduced to \(2^{-\poly(n)}\) by standard repetition.
\end{proof}

\subsection{Integer circuits and coefficient signs}

\begin{proof}[Proof of Theorem~\ref{push4:posslp}]
Let \(Z\) be an integer with \(|Z|<2^H\), where \(H\) has a
polynomial-length binary description, and suppose that \(Z\) has
uniform short-modulus representatives in \(\GapP\). Use the
beta construction from Section~\ref{sec:posslp}, now with
\(N=H+16\). Lemma~\ref{lem:ring-coefficients} supplies
\(\SharpP\) representatives \(W_i\) of \(\binom Ni\) modulo
\(m_i=N+1+i\). If \(R_i\) represents \(Z\) modulo \(m_i\), set
\(A_i=(-1)^iW_iR_i\). A product of two counting functions remains in
\(\GapP\), and the exact identity is
\[
 \sum_{i=0}^N\frac{A_i}{m_i}
   =z+C_NZ,\qquad
 C_N=\int_0^1x^N(1-x)^N\,dx
     =\frac{(N!)^2}{(2N+1)!}>0.
\]
Since \(C_N\le4^{-N}\), its displacement satisfies
\(|C_NZ|<2^{-H-32}<2^{-32}\). When \(Z\ne0\),
\[
 |C_NZ|\ge\frac1{(2N+1)!}\ge2^{-J},
 \qquad J=(2N+1)\lceil\log_2(2N+1)\rceil.
\]
Thus the numerical hypotheses of Lemma~\ref{push4:sign} hold.
The complete sum \eqref{push4:raw-phase} is a plain \(\GapP\)
function, so its exact value is retrievable in \(\FP^{\PP}\).
This supplies the retrieval hypothesis without functional absorption
of a PH oracle. We obtain a
\(\BPP^{\PP}\) sign test for \(Z\).

For \(\PosSLP\), a circuit with \(s\) gates and integer constants
of absolute value below \(2^\beta\) has the uniform bit-height bound
\(H=(\beta+1)2^{s+1}\), with harmless initial padding. Its residue
modulo any short integer is computable in polynomial time. For a
specified SLP coefficient, the syntactic degree is at most \(2^s\).
Replacing subtraction by addition, constants by their absolute values,
and variables by one bounds the coefficient \(\ell_1\) norm by the
same height estimate. An exponent outside the degree bound gives zero
immediately; otherwise Lemma~\ref{lem:ring-coefficients} supplies the
needed \(\SharpP\) representatives. These observations prove
both claims. Finally, \(\BPP^{\PP}\subseteq\PP^{\PP}=\Clevel2\)
is the relativized inclusion \(\BPP\subseteq\PP\).
\end{proof}

This bound and the deterministic \(\Delta_2^{\PP}\) bound proved
earlier can both be retained. We do not identify \(\BPP^{\PP}\)
with the class \(\mathrm{BP}\mathbin{\cdot}\PP\) of randomized
reductions to one PP predicate, nor do we obtain \(\Pclass^{\PP}\).
The sampled phase undergoes nonlinear postprocessing after a counting
retrieval. Averaging those postprocessed answers exactly over all seeds
would introduce another counting computation; sampling is what keeps
this final step within bounded-error polynomial time with a PP oracle.

\section{Intermediate CH bounds and the ordered-rank refinement}
\label{sec:depth}

We retain the geometric reduction of Section~\ref{sec:etr} and improve
the implementation of its final sign test.  For every polynomial-length
parameter choice $\omega$, that reduction constructs a nonzero integer
\[
 Z_\omega=\Norm_{\mathcal A}\bigl(U(\ell)\bigr),
 \qquad |Z_\omega|<2^H,
 \qquad H=2^{\poly(n)}.
\]
Every negative value certifies feasibility, and a feasible instance has
a choice with negative value.  The polynomial-size sign-filter circuits
and short Thom conjunction remain unchanged.  The improvement comes
from evaluating their nested norms modulo short integers, rather than
recovering intermediate coefficient arrays.

An earlier implementation used full Thom-code hashing and the general
resultant circuits of \cite{AGS26}. Its conservative arithmetic-depth
ledger gave $\Clevel{1023}$, recorded in the supplementary treatment.
That ledger is not needed below. Here we account directly for the
specialized computation of Section~\ref{sec:etr}, obtaining the intermediate bounds
$\Clevel7$, $\Clevel6$, and $\Clevel5$, followed by the
$\Clevel4$ ordered-rank refinement.

Throughout this section, all indices and finite-ring elements have
polynomial bit length.  Bounds on degrees and integer bit lengths may
have singly exponential values, but have polynomial-length binary
descriptions.  One fixed polynomial bounds these lengths and the
running times uniformly over every permitted $\omega$ and every supplied
modulus.  We use $\Clevel{0}=\Pclass$ and
$\Clevel{h+1}=\PP^{\Clevel{h}}$.

\subsection{A four-level residue computation}

Write the critical equations as
$G_i(X)=X_i^5+\text{lower-total-degree terms}$, with $m=\poly(n)$
variables and quotient rank $D=5^m$.  For a sufficiently large prime
$p$, start from the grid polynomials
\[
 h_i(X_i)=\prod_{a=1}^5(X_i-a),\qquad
 F_i(t,X)=(1-t)h_i(X_i)+tG_i(X).
\]
If $p>5$, all initial Jacobian entries are units.  Normalizing by that
Jacobian and applying Lemma~\ref{lem:branches} gives polynomial-size
polynomial SLPs $B_b(t)$ approximating the formal branches, indexed by
$b\in\{1,\ldots,5\}^m$.  The coefficients are computed modulo $p$;
we specialize at $t=1$.

We use the same branches for the inner norm defining $q_t$, its
selected derivative evaluations, the cleared selector, and its outer
norm.  Polynomial operations, including coefficient extraction and
substitution, preserve agreement modulo $t^K$.  Consequently there is
no need to extract homotopy coefficients between the inner and outer
norms.

Here the required precision can be fixed before choosing the field.
If an element $V(t,X)$ has $t$-degree $L$ and $X$-degree $E$, its
multiplication matrix in the standard monomial basis satisfies
\[
 \deg_t\Norm(V)\le D(L+E+4m).
\]
Indeed, monic reduction starts in total $X$-degree at most $E+4m$;
each reduction decreases this degree and introduces at most one power
of $t$.  Applying this bound to the nested construction gives a uniform
$K_Z\ge\deg_t Z(t)$ of singly exponential value.  Choose branch
precision $K>K_Z$, and then bound the syntactic degrees of the branch
approximants and the resulting polynomial $P(t)$.  Polynomial-size
Newton SLPs have singly exponential syntactic degrees; the two products
of $D$ factors and the selector preserve such bounds.  Finally choose
$p-1$ larger than $K$, those degree bounds, and every extraction index.
This order avoids a precision--field-size circularity.

For $\deg V<p-1$ and $0\le j<p-1$, finite-field orthogonality gives
\[
 [v^j]V(v)=-\sum_{v\in\F_p}V(v)v^{p-1-j}
\]
The exponent is positive, and only the desired coefficient contributes
a power whose exponent is a positive multiple of $p-1$.
Several auxiliary extraction variables can be indexed together in one
counting sum.

An exponential product of polynomial-time computable nonzero field
elements lies in $\FP^{\SharpP}$.  Choose a primitive root by
$\FP^{\NP}$ preprocessing, and count the sum of the factors' unique
discrete logarithms.  Explicitly, for factors $u_i$ and generator $g$,
count triples $(i,e,j)$ satisfying $g^e=u_i$ and $0\le j<e<p-1$.
Exponentiating $g$ by that count gives the product.  Zero factors are
handled separately.  The preprocessing is performed before the
ordinary counting queries, with the generator passed as a parameter.
Relative to $\Clevel{h}$, this costs one further counting level.

The distinction between a short value and its counting representation
is important here.  An exponential sum of short nonnegative values
computable in $\FP^{\Clevel{h}}$ is a
$\SharpP^{\Clevel{h}}$ function: guess the summand index and an
integer below its value.  The answer still has polynomially many bits,
and polynomially many threshold queries retrieve it in
$\FP^{\Clevel{h+1}}$.  We charge this retrieval whenever its answer
will be evaluated inside a subsequent counting computation.  This
explains the separate levels for products and coefficient extraction.

The resulting ledger concerns entire short field values:
\[
\begin{array}{l|c}
\text{operation}&\text{function class}\\ \hline
\text{one branch factor}&\FP\\
\text{inner norm product}&\FP^{\Clevel{1}}\\
\text{selected derivative coefficient}&\FP^{\Clevel{2}}\\
\text{filter and short selector}&\FP^{\Clevel{2}}\\
\text{outer norm product }P(t)&\FP^{\Clevel{3}}\\
\text{final homotopy extraction and specialization}&\FP^{\Clevel{4}}
\end{array}
\]
The Hasse derivatives in our construction are extracted directly from
$q(z+v)$. The final row combines all coefficients through
$K_Z$ in one Fourier sum with a geometric-series weight.  Binary
doubling evaluates that weight in polynomial time.  Polynomially many
adaptive calls to an existing function oracle do not add a level.

To recover the sign, choose an interval of polynomial-bit primes whose
product $M$ exceeds $8|Z_\omega|$.  Elementary prime-product bounds give
such an interval above the required degree cutoff, together with
$M<2^{H_0}$ for a polynomial-bit $H_0$.  Index all integers in the
interval and give composites zero weight.  The CRT cofactors require
only an $\FP^{\Clevel{1}}$ product, so the scalar residues give weights
$A_i\in\FP^{\Clevel{4}}$ satisfying
\[
 \sum_i\frac{A_i}{p_i}=\text{integer}+\frac{Z_\omega}{M}.
\]
The fixed-precision sums in Lemma~\ref{lem:first-event} are then
$\SharpP^{\Clevel{4}}$ functions.  Their polynomial-bit values are
retrievable in $\FP^{\Clevel{5}}$.  Concretely, at scale $j$ the
lemma forms a polynomial-precision approximation to the fractional
part of $2^jZ_\omega/M$.  The scale index can be exponentially large
because $2^j\bmod p_i$ is computed from its binary description.
Before the first approximation whose circular distance from an integer
is at least $1/8$, the true signed distance has magnitude below
$1/8+2^{-9}<1/4$.  Doubling cannot cross the opposite half of the
circle.  The first departure therefore reveals the original sign;
the height bound guarantees that it occurs within the allowed scales.
Guessing the first event and
universally excluding earlier events gives a
$\Sigma_2^{\Clevel{5}}$ predicate.  All geometric and selector guesses
merge into its first existential block.  Relativized Toda
\cite{Toda91} therefore gives
\[
 \Sigma_2^{\Clevel{5}}\subseteq\PH^{\Clevel{5}}
 \subseteq\Pclass^{\Clevel{6}}\subseteq\Clevel{7}.
\]

\subsection{Products over arbitrary prime powers}

To save another level, we will retain the final Fourier sum as a
counting representative.  The beta identity of Section~\ref{sec:posslp}
accepts such representatives, but uses arbitrary composite denominators.
We must therefore remove the large-characteristic restriction.

The earlier choice of integer Hasse derivatives now matters. Since
$q^{(j)}=j!q_{[j]}$ over the integers and $j!>0$, their real sign
patterns agree, while the modular algorithm never needs to invert a
factorial at a small prime. The scalar and its algebraic-value bounds
were defined with these Hasse derivatives from the start.

Now let $M$ be an arbitrary positive polynomial-bit modulus.  Factor it
into maximal prime powers using $\FP^{\NP}$.
For each $p^a$, use the Galois ring
\[
 \mathcal R=GR(p^a,e)
   =(\Z/p^a\Z)[Y]/(f(Y)),\qquad Q=p^e,
\]
with $f\bmod p$ monic irreducible.  Choose
\[
 Q-1>\max\{5,K,\text{all approximate polynomial degree bounds},
                    \text{all extraction indices}\}.
\]
The least suitable $e$ is polynomially bounded.  Also
$a\le\log_2 M$ and $e\log(p^a)=\poly(n+\log M)$, so ring elements
remain short.  Choose the ring presentation and primitive residue-field
element deterministically by $\FP^{\NP}$ search, once for the entire
computation.

Teichm\"uller Fourier extraction uses exactly one short index per
nonzero residue-field element $x$: take its canonical coordinate lift
and raise it to $Q^{a-1}$.
The index has $O(e\log p)$ bits and the powering exponent has
$O(1+(a-1)e\log p)$ bits.  Although $\mathcal R$ has zero divisors, the
orthogonality proof of Lemma~\ref{lem:ring-coefficients} is valid:
an unwanted power of a Teichm\"uller generator differs from one by
a unit, since that difference has nonzero reduction in the residue
field.  Also
$Q-1$ itself is a unit.  These are the only inverses needed for
coefficient extraction.

\begin{lemma}\label{lem:ring-products}
An exponentially indexed product of short elements of $GR(p^a,e)$,
whose factors are uniformly computable in $\FP^{\Clevel{h}}$, is in
$\FP^{\Clevel{h+1}}$ for $h\ge1$.  Polynomial-time factors give
$\FP^{\SharpP}$.
\end{lemma}

\begin{proof}
Read each factor's $p$-valuation $v_i$ from its coordinates, assigning
$v_i=a$ to zero.  One count gives $v=\sum_i v_i$.  If $v\ge a$ the
product is zero.  Otherwise write $u_i=p^{v_i}w_i$ with a unit lift
$w_i$, obtained by coordinatewise division; for zero factors we may set
$w_i=1$ when preparing the remaining counts.

The Teichm\"uller group has order $Q-1$.  Every principal unit has
exponent dividing $p^{a-1}$, including at $p=2$, by the binomial theorem.
Thus
\[
 t_i=w_i^{Q^{a-1}},\qquad z_i=w_it_i^{-1}\in1+p\mathcal R
\]
splits each unit into its Teichm\"uller and principal parts.  A primitive
residue-field element lifts by the same power to a Teichm\"uller
generator.  As above, a unique-logarithm count computes the product of
the $t_i$.

For the principal parts, fix
$g_{j,l}=1+p^jY^l$ for $1\le j<a$ and $0\le l<e$.
Every $z\in1+p\mathcal R$ has a unique greedy representation
\[
 z=\prod_{j=1}^{a-1}\prod_{l=0}^{e-1}g_{j,l}^{d_{j,l}},
 \qquad 0\le d_{j,l}<p.
\]
At stage $j$, read the coordinates of $(z_j-1)/p^j$ modulo $p$ and
divide by the corresponding generator powers.  The remainder is one
modulo $p^{j+1}$, proving the representation inductively.  This is
polynomial-time short-ring arithmetic; unit inverses may be obtained
by exponent $(Q-1)p^{a-1}-1$.

For each of the polynomially many pairs $(j,l)$, count its digit sum
across factors.  Commutativity gives the principal product as
$\prod_{j,l}g_{j,l}^{\sum_i d_{i,j,l}}$, without resolving digit carries.
The valuation, logarithm, and digit counts are parallel
$\SharpP^{\Clevel{h}}$ queries with polynomial-length answers.
Their retrieval and final powering cost $\FP^{\Clevel{h+1}}$.
For polynomial-time factors, the short $\FP^{\NP}$ preprocessing is
performed first using the available counting oracle, and supplied as
ordinary input to each subsequent $\SharpP$ query.
\end{proof}

Choose five residue-distinct Teichm\"uller elements $a_1,\ldots,a_5$
and replace the starting grid by $h_i=\prod_j(X_i-a_j)$.
All root differences and initial Jacobian entries are units.  The
polynomial Newton construction works over $\mathcal R$: inverse
Jacobians are replaced by finite geometric series modulo $t^K$.
The quotient by the homotopy is free over $\mathcal R[t]$ in the same
standard monomials.  Over $\mathcal R[[t]]$, evaluation at its formal
branches is an isomorphism, since its matrix at $t=0$ is a tensor
Vandermonde matrix with unit determinant.  Thus the branch-product norm
identity remains valid despite zero divisors.

The degree bounds depend on the grid elements only as constants.
Consequently they are fixed before selecting the extension ring.
The same nested computation gives
\[
 P_{\mathcal R}(t)\equiv Z_{\mathcal R}(t)\pmod{t^K},
 \qquad \deg Z_{\mathcal R}\le K_Z<K,
 \qquad Z_{\mathcal R}(1)=Z_\omega\pmod{p^a}.
\]
By Lemma~\ref{lem:ring-products} and Teichm\"uller coefficient
extraction, the first five rows of the previous ledger remain valid:
in particular $P_{\mathcal R}(t)\in\FP^{\Clevel{3}}$ for short $t$.

\subsection{Retaining the last Fourier sum}

Let $\mathcal T$ be the Teichm\"uller unit group.  Orthogonality gives
\[
 Z_{\mathcal R}(1)=
 (Q-1)^{-1}\sum_{\xi\in\mathcal T}
 P_{\mathcal R}(\xi)\sum_{k=0}^{K_Z}\xi^{-k}.
\]
Here both the degree and extraction-index restrictions are covered by
our choice of $Q$.  The geometric weight is evaluated by binary
doubling, without division by $\xi-1$.  Take the canonical constant
coordinate of each whole summand and sum these nonnegative integers
\emph{without reducing the total}.  This gives
$W_{p^a}\in\SharpP^{\Clevel{3}}$ with
$W_{p^a}\equiv Z_\omega\pmod{p^a}$.  The total ring element is scalar
because the starting grid disappears at $t=1$; intermediate coefficients
need not be scalar.  We use a coordinate projection, not a field trace.

This is where the implementation differs from the seven-level route.
Computing the canonical residue of the last sum would give an
$\FP^{\Clevel{4}}$ function.  Instead, we leave its accepting-path
count intact.  Its numerical value may exceed $p^a$, but every later
use needs only its congruence modulo $p^a$.  This permits the final
reconstruction to combine counting paths directly, without invoking
an additional oracle that retrieves that count.

Ordinary short CRT weights combine the polynomially many prime-power
representatives into
\[
 W(M)\in\GapP^{\Clevel{3}},\qquad
 W(M)\equiv Z_\omega\pmod M.
\]
Its output has polynomial bit length: it is an exponential sum of short
coordinates with short weights.  Factorization and consistent ring
selection are available inside the same oracle bound.

\begin{proof}[The first-event bound in Theorem~\ref{thm:depth}]
Set $N=H+4$ and $M_i=N+1+i$.  Lemma~\ref{lem:ring-coefficients},
applied to the short SLP for $(1+X)^N$, supplies
$B_i\in\GapP^{\NP}$ with $B_i\equiv\binom Ni\pmod{M_i}$.
The beta identity gives
\[
 \sum_{i=0}^N\frac{(-1)^iB_iW(M_i)}{M_i}
 =\text{integer}+C_NZ_\omega,
 \qquad C_N=\frac{(N!)^2}{(2N+1)!}>0.
\]
We have $|C_NZ_\omega|<1/8$, and a nonzero displacement has magnitude
at least $2^{-H_1}$ for
$H_1=(2N+1)\lceil\log_2(2N+1)\rceil$.  All bounds have short binary
descriptions.

The fixed-precision sums of Lemma~\ref{lem:first-event} stay in
$\GapP^{\Clevel{3}}$: multiplying two counting functions and taking
an exponentially indexed sum concatenate polynomial-length counting
paths.  We never retrieve $W(M_i)$ canonically inside another count.
The approximants are therefore in $\FP^{\Clevel{4}}$.  Guessing the
first event, its negative side, and $\omega$, followed by universally
excluding earlier events, yields $\Sigma_2^{\Clevel{4}}$.  Finally,
\[
 \Sigma_2^{\Clevel{4}}\subseteq\PH^{\Clevel{4}}
 \subseteq\Pclass^{\PP^{\Clevel{4}}}
 =\Pclass^{\Clevel{5}}\subseteq\Clevel{6},
\]
by relativized Toda~\cite{Toda91} and deterministic simulation inside
$\PP$.  This proves the decision containment.  The nested norm
representatives remain in $\GapP^{\Clevel{3}}$, rather than
$\GapP^{\PH}$; the argument does not place $\ETR$ in $\PH^{\PP}$.
\end{proof}

\subsection{A fifth-level bound from randomized reconstruction}
\label{push4:depth}

The preceding $\Clevel{6}$ proof remains valid independently of the
randomized reconstruction lemma.  We now improve its final step while
retaining the same geometric certificate and modular computation.
Randomized sign recovery removes the existential search over scales;
the remaining existential choice is the original ETR certificate.

\begin{theorem}\label{push4:etr}
We have
\[
 \ETR\subseteq\MA^{\Clevel{4}}\subseteq\Clevel{5}.
\]
\end{theorem}

\begin{proof}
Let $x$ have length $n$, and pad the certificate $\omega$ from
Section~\ref{sec:etr} to a fixed polynomial length $m(n)$.  Invalid
encodings are rejected.  For every valid certificate, the construction
gives a nonzero integer $Z_{x,\omega}$ with
$|Z_{x,\omega}|<2^{H(n)}$, and
\[
 \text{$x$ is feasible}\quad\Longleftrightarrow\quad
 \exists\omega\; Z_{x,\omega}<0,
\]
where $H(n)=2^{P(n)}$ for a fixed polynomial $P$.  This height bound is
uniform over all certificates. Soundness includes inaccurate analytic
choices: a negative real selector factor always supplies a qualifying
real critical point, whereas nonreal factors have positive conjugate
products.  Generic choices are needed only for completeness.

The modular ledger is unchanged.  The inner norm has values in
$\FP^{\Clevel{1}}$, the selected Hasse derivatives and cleared selector
in $\FP^{\Clevel{2}}$, and the outer norm at a short ring point in
$\FP^{\Clevel{3}}$.  Leaving its last Fourier sum unreduced gives
\[
 W(x,\omega,M)\in\GapP^{\Clevel{3}},\qquad
 W(x,\omega,M)\equiv Z_{x,\omega}\pmod M
\]
for every positive polynomial-bit modulus $M$.  The output-bit bound is
uniform in $n+|\omega|+\log M$.  All prime powers, including small
characteristics, are covered by the unit-grid construction above.

Set $N=H+16$ and $M_i=N+1+i$, for $0\le i\le N$.
Lemma~\ref{lem:ring-coefficients} supplies polynomial-bit
$B_i\in\GapP^{\NP}$ with $B_i\equiv\binom Ni\pmod{M_i}$.
The beta identity gives
\[
 S=\sum_{i=0}^N\frac{(-1)^iB_iW(x,\omega,M_i)}{M_i}
   =z+\delta,\qquad
 \delta=C_NZ_{x,\omega},\qquad
 C_N=\frac{(N!)^2}{(2N+1)!}>0.
\]
Since $C_N\le4^{-N}$, we have $|\delta|<2^{-32}$.
Since $Z_{x,\omega}$ is a nonzero integer,
\[
 |\delta|\ge C_N\ge2^{-J},\qquad
 J=(2N+1)\lceil\log_2(2N+1)\rceil.
\]
Both $N$ and $J$ have polynomial-length binary descriptions.  The larger
choice of $N$ supplies the small initial displacement required by
Lemma~\ref{push4:sign}.

The signed numerators are in $\GapP^{\Clevel{3}}$, because
$\NP\subseteq\Clevel{3}$ and only a product of two counting functions
is involved.  Every phase query in Lemma~\ref{push4:sign} is an
exponentially indexed linear combination of these numerators with
polynomial-time computable short weights.  Even its double index,
consisting of a denominator index and a multiplier-bit index, has
polynomial length.  The whole query therefore remains in
$\GapP^{\Clevel{3}}$ and is exactly retrievable in
$\FP^{\PP^{\Clevel{3}}}=\FP^{\Clevel{4}}$.
This verifies the lemma's retrieval hypothesis for the complete weighted
sums, not merely for their individual summands.  We obtain
\[
 \{(x,\omega):Z_{x,\omega}<0\}\in\BPP^{\Clevel{4}}.
\]

Merlin now sends $\omega$, and Arthur runs this sign algorithm with
fresh random bits.  A feasible instance has one fixed witness accepted
with probability at least $2/3$.  On an infeasible instance, every
witness is accepted with probability at most $1/3$.  The pointwise
guarantee of Lemma~\ref{push4:sign} is sufficient: the witness is fixed
before Arthur's randomness, and no seed must work simultaneously for
all witnesses.  This proves membership in $\MA^{\Clevel{4}}$.

We spell out the last containment to avoid an extra hierarchy level.
Fix a complete oracle $A$ for $\Clevel{4}$, amplify Arthur's verifier
to error $\varepsilon\le2^{-m(n)-4}$, and pad every computation to use
exactly $r(n)\ge1$ random bits.  The function
\[
 F(x)=\#\{(\omega,\rho):
       |\omega|=m(n),\ |\rho|=r(n),
       V^A(x,\omega;\rho)=1\}
\]
belongs to $\SharpP^A$.  On a yes instance,
$F(x)\ge(1-\varepsilon)2^{r(n)}>2^{r(n)-1}$.
On a no instance,
\[
 F(x)\le2^{m(n)}\varepsilon2^{r(n)}
       \le2^{r(n)-4}<2^{r(n)-1}.
\]
Thus the sign of the $\GapP^A$ function
$F(x)-2^{r(n)-1}$ decides the language.  Consequently
$\MA^{\Clevel{4}}\subseteq\PP^{\Clevel{4}}=\Clevel{5}$.
The amplification is essential: averaging a constant-error verifier
over all witnesses would not give this gap.
\end{proof}

The verifier used in this refinement has bounded error.  The final
$\Clevel{5}$ containment uses an unbounded-error PP computation with a
$\Clevel{4}$ oracle to count its amplified accepting computations.
The inner norm, derivative, and outer norm layers have not been
flattened, and the conclusion remains a decision containment.

\subsection{Ordered ranks instead of high derivatives}
\label{push6:depth}

We can remove the derivative-extraction layer by giving every real root
a short integer label. The label need not be its usual rank among the
real parts of all complex roots. It is enough that the labels increase
strictly along the real roots. This distinction permits a short rational
construction even when a projected difference is extremely close to the
imaginary axis.

\begin{theorem}\label{push6:etr}
\[
 \ETR\subseteq\BPP^{\Clevel3}\subseteq\Clevel4.
\]
\end{theorem}

We first prove the analytic ordering property, then implement the labels
with first-order jets. Finally, random hashing replaces the guessed
isolating list.

\paragraph{A complex normalizer.}
Let a finite set of \(M\) nonzero complex numbers satisfy
\(2^{-L}\le |z|\le2^L\), where \(\log L\) and \(\log M\) are polynomially
bounded. Put
\[
 T_0=\left\lceil10\bigl(\log_2(L+2)+\log_2(M+2)+10\bigr)\right\rceil,
 \qquad C=2^B\ge10^6M(T_0+2).
\]
Here and below the least sufficient power of two may be used.
For positive \(s\), define
\[
 \phi_s(z)=\frac{2sz}{C(1+s^2z^2)}.
\]
At each stage choose a fresh \(t\) uniformly from a \(C^2\)-point
grid in \([1,2]\), and use \(s=At\). Call the stage good at \(z\) if
\begin{equation}\label{push6:good}
 \frac{|1-s^2|z|^2|}{1+s^2|z|^2}\ge C^{-1}.
\end{equation}
For a fixed input, the excluded interval of \(t\)'s has length at most
\(16/C\); discretization adds at most \(2/C^2\) to its probability.
On a good stage,
\[
 C^{-1}\min\{s|z|,(s|z|)^{-1}\}\le|\phi_s(z)|\le1.
\]
Start with \(A=1\), giving the lower exponent \(K=L+B+1\). Thereafter
take \(A=2^{\lceil K/2\rceil}\) and replace
\(K\) by \(\lceil K/2\rceil+B+1\). After polynomially many stages,
at most \(T_0\), all magnitudes lie in \([2^{-K},1]\) with
\(K\le2B+4\). Conditional on all previous stages being good, the next
bad-event estimate applies to every current value. A union bound therefore
gives simultaneous success with probability greater than \(0.999\).

Choose a fresh \(a\) on the corresponding grid in \([2,3]\), and apply
\[
 H_a(z)=\frac{z}{(z+a)^2+1}.
\]
There are no poles on the unit disk, and \(|H_a(z)|\le1\) there.
Writing \(x=\Re z\) and \(r=|z|^2\), we have
\[
 \Re H_a(z)=
 \frac{r}{|(z+a)^2+1|^2}
 \left(\frac{x}{r}a^2+2a+x+\frac{x}{r}\right).
\]
The polynomial in parentheses has coefficient 2 at \(a\). If
\(|x/r|\le1/6\), its derivative has magnitude at least one on \([2,3]\).
Otherwise its two complex roots show that the set where its magnitude
is below \(C^{-2}\) has length at most \(4\sqrt6/C\).
Including the grid correction and taking another union bound, with
probability greater than \(0.99\) every transformed value satisfies
\[
 |\Re H_a(z)|\ge\gamma:=2^{-2K-2B-9}.
\]
The logarithm of \(\gamma^{-1}\) is polynomially bounded.

Finally iterate \(f(z)=2z/(1+z^2)\). For
\(\sigma=\sgn\Re z\), the Cayley coordinate
\((\sigma-z)/(\sigma+z)\) is squared by \(f\). From \(|z|\le1\) and
\(|\Re z|\ge\gamma\), its initial squared modulus is at most
\(1-\gamma\). Thus polynomially many iterations give a rational
function \(R=n/d\) satisfying
\[
 |R(z_i)-\sigma_i|\le\frac1{1000M},\qquad \sigma_i\in\{-1,1\},
\]
at every specified nonzero input. The function has real coefficients,
\(R(0)=0\), and positive denominator on the real line. Every seed
preserves the sign of every nonzero real input. All scale choices have
polynomially many bits; the possibly large powers \(2^{\lceil K/2\rceil}\)
have short powering circuits. The resulting integer-pair SLP has
polynomial size and singly exponential degree and coefficient bit height.

The angular step can change a complex input's half-plane. The following
property is what the rank construction actually needs.

\begin{lemma}\label{push6:ordering}
Let \(g=\phi_{s_n}\circ\cdots\circ\phi_{s_1}\), with all \(s_j>0\)
and \(C\ge1000(n+1)\), and let \(2\le a\le3\).
For fixed \(y\), the signs of
\(\Re H_a(g(x+iy))\) are nondecreasing in \(x\) on the set of inputs
that satisfy \eqref{push6:good} at every stage.
\end{lemma}
\begin{proof}
Every \(\phi_s\) maps the open right half-plane to itself. Write
\(\eta(w)=\Re w/|w|\) there. Direct calculation gives
\begin{equation}\label{push6:eta}
 \eta(\phi_s(w))
 =\eta(w)\frac{1+s^2|w|^2}{|1+s^2w^2|}\ge\eta(w).
\end{equation}
If both input and output have \(\eta\le1/C\), put
\(\theta_0=\arcsin(1/C)\). The numerator \(1-s^2w^2\) in
\[
 \phi_s'(w)=\frac{2s(1-s^2w^2)}{C(1+s^2w^2)^2}
\]
has argument within \(2\theta_0\) of zero. Since
\(\arg(1+s^2w^2)=\arg w-\arg\phi_s(w)\), its denominator has argument
within \(4\theta_0\) of a multiple of \(2\pi\).
Thus \(|\arg\phi_s'(w)|\le6\theta_0\le12/C\), taking the
representative near zero.

Put \(m=n-1\), and for \(t>0\) write
\(w(t)=(\phi_{s_m}\circ\cdots\circ\phi_{s_1})(t+iy)=u(t)+iv(t)\).
For \(m=0\), use \(w(t)=t+iy\). Consider
\[
 I=\{t>0:\eta(w(t))<1/C\}.
\]
By \eqref{push6:eta}, every preceding angle is also small on this set.
Consequently
\[
 |\arg w'(t)|\le12m/C,\qquad
 u'(t)>0,\qquad |v'(t)|\le(24m/C)u'(t).
\]
These inequalities hold also at its boundary \(\eta=1/C\), and give
\[
 \frac{d}{dt}\eta(w(t))=
 \frac{v^2u'-uvv'}{|w|^3}>0.
\]
Here \(u/|v|\le2/C\). It follows that \(I\) is empty or an initial
interval: a differentiable function increasing wherever its value is at
most \(1/C\) cannot enter that sublevel set from above. There are no
poles or zeros for \(t>0\), since all intermediate values are in the
open right half-plane.

For \(t_1<t_2\) in \(I\), integration gives
\[
 |v(t_2)-v(t_1)|
 \le\frac{24m}{C}(u(t_2)-u(t_1))
 \le\frac{48m}{C^2}|v(t_2)|.
\]
The sign of \(v\) is constant on \(I\). Combining this estimate with
\(|w|=|v|/\sqrt{1-\eta^2}\) yields
\begin{equation}\label{push6:modulus}
 \left|\log\frac{|w(t_2)|}{|w(t_1)|}\right|
 \le\frac{100(m+1)}{C^2}.
\end{equation}
Let \(s=s_n\). The last-stage detuning is
\[
 \rho(t)=\frac{1-s^2|w(t)|^2}{1+s^2|w(t)|^2}
        =-\tanh(\log(s|w(t)|)).
\]
Since \(\tanh\) is 1-Lipschitz, if one point of \(I\) has a good
last stage, \eqref{push6:modulus} implies
\[
 |\rho(t)|\ge\frac9{10C}\quad(t\in I),\qquad
 |V(w(t))|\le\frac{10}{9},\quad V=\phi_s.
\]
Thus goodness at one endpoint controls the whole relevant interval;
we have not assumed goodness between the two tested points.

Set \(A=a^2+1\ge5\) and \(T(w)=V(w)+A/V(w)\). Its derivative factors as
\begin{equation}\label{push6:derivative}
 T'(w)=\frac{AC}{2s}(s^2-w^{-2})
                 \left(1-\frac{V(w)^2}{A}\right).
\end{equation}
On \(I\), the first nonconstant factor has argument at most \(4/C\)
in magnitude. The last factor lies in the disk of radius
\(20/81<1/4\) centered at 1. Together with the bound on \(\arg w'\),
this gives \(|\arg(T'(w)w')|<0.27<\pi/2\), so
\(\Re T(w(t))\) increases strictly on \(I\).

Oddness and real coefficients give \(g(-t+iy)=-\overline{g(t+iy)}\).
Using \(H_a(z)=1/(z+2a+A/z)\), the output sign on the left half-plane
is the sign of \(2a-\Re T(w(t))\). But
\[
 \Re T(w(t))\ge
 \frac{AC}{2}\eta(w(t))
 \left(s|w(t)|+\frac1{s|w(t)|}\right)
 \ge AC\eta(w(t)).
\]
A nonnegative output therefore requires
\(\eta(w(t))\le2a/(AC)\le4/(5C)\), hence \(t\in I\).
If this occurs at the more negative of two tested real parts, its good
last stage makes \(\Re T\) strictly increasing throughout \(I\).
The output at the less negative real part is then strictly positive.

For positive real parts the output is always in the right half-plane.
At \(x=0\), a nonzero good input remains nonzero and purely imaginary
through the radial stages, and \(H_a\) gives positive real part.
The zero input stays zero. On the real line every map preserves sign.
These cases complete the ordering statement.
\end{proof}

\paragraph{Distinct labels at all real roots.}
Use the fixed simple critical system and its separating objective, and let
\(\alpha_1,\ldots,\alpha_D\) be its distinct projected roots.
Their nonzero pair differences have singly exponential algebraic upper
and lower bounds. For example,
\(\prod_{i,j}(Y-(\alpha_i-\alpha_j))\) is monic over \(\Z\);
its first nonzero coefficient and a root bound give the lower estimate.
Apply the same normalizer to every nonzero pair difference, taking
\(M=D^2\) as a bound on their number. All random parameters are shared
across all pairs.

On a successful seed let \(\sigma_{ij}\in\{-1,1\}\) be the approximated
signs, and set \(\sigma_{ii}=0\). For real \(\alpha_i\), define
\[
 S_i=\sum_j\sigma_{ij},\qquad r_i=\frac{S_i+D-1}{2}.
\]
The \(D-1\) nonself terms show that \(r_i\) is an integer in
\(\{0,\ldots,D-1\}\). For each fixed \(j\), Lemma~\ref{push6:ordering}
makes \(\sigma_{ij}\) nondecreasing as the real outer root increases.
For two real roots, the two self/cross terms alone increase their sum
by 2. Hence the integers \(r_i\) increase strictly along all real roots.
Conjugate inner roots cause no exception: their signs agree.
No separation between a real root and the real part of a nonreal root
is required.

The rational sum
\[
 S(u)=\sum_jR(u-\alpha_j)
\]
approximates \(S_i\) within \(1/100\) at a real root \(u=\alpha_i\).
It is obtainable from one first-order jet norm:
\begin{equation}\label{push6:rank-jet}
 D_0(u)+\epsilon N_0(u)=
 \Norm_{A}\bigl(d(u-J_w)+\epsilon n(u-J_w)\bigr)
             \pmod{\epsilon^2},
 \qquad S(u)=N_0(u)/D_0(u).
\end{equation}
Thus \(X(u)=(N_0(u)+(D-1)D_0(u))/(2D_0(u))\) is within
\(1/100\) of the integer label at every real root.

\begin{lemma}\label{push6:jet}
An exponentially indexed product in
\(\mathcal S=GR(p^a,e)[\epsilon]/(\epsilon^2)\), with polynomial-time
short factors and polynomial-value \(a,e\), is in \(\FP^{\SharpP}\).
The relativized version costs one counting level above its factor oracle.
\end{lemma}
\begin{proof}
A factor \(u+\epsilon v\) is a unit exactly when \(u\bmod p\ne0\).
The ideal \((p,\epsilon)\) has nilpotence exponent \(a+1\).
Count the nonunit factors; if there are at least \(a+1\), return zero.
Otherwise recover their at most \(a\) indices with prefix counting
queries and multiply them explicitly. For the remaining computation
replace each nonunit by 1, using only its polynomial-time unit test.
The unit product is
\[
 \prod_i(u_i+\epsilon v_i)
   =\left(\prod_i u_i\right)
      \left(1+\epsilon\sum_i v_i/u_i\right).
\]
Lemma~\ref{lem:ring-products} computes the scalar product in one
counting level, and the short-ring sum uses that same level.
All inverses in this expression are of units and are polynomial-time
computable. Multiplication by the recovered exceptional factors adds
no oracle level. The prefix queries use the original factor evaluator,
so the argument relativizes as stated.
\end{proof}

\paragraph{Reading and hashing the labels.}
Put \(\ell=\lceil\log_2D\rceil\). At a real number \(X\) within
\(1/100\) of an integer \(r\in[0,D-1]\), the sign of
\[
 \sin\left(\frac{\pi(X+1/2)}{2^b}\right),\qquad 0\le b<\ell,
\]
is \((-1)^{\operatorname{bit}_b(r)}\), with magnitude at least
\(1/(2D)\). These signs have polynomial-size polynomial SLP
approximants. To see this, approximate \(\pi\) to error
\(1/(1000D^2)\), divide the angle by a power of two at least \(64D\),
and use a Taylor polynomial of degree \(O(\log D)\) for the exponential
of the resulting imaginary argument. Repeated squaring restores the
angle. Taking sufficiently large absolute constants bounds the final
error by \(1/(100D)\). The imaginary part is a real polynomial
\(P_b(X)\) with the desired sign and magnitude at least \(1/(4D)\).
Its SLP size is polynomial in \(\log D\), and its degree and coefficient
bit height are singly exponential.

Substitute the rational \(X(u)\) into \(P_b\) using integer pairs
\((A(u),B_b(u))\), and put \(H_b=AB_b\).
When \(B_b\ne0\), this has the sign of \(P_b(X)\);
when \(B_b=0\), it is simply zero. Apply
Lemma~\ref{lem:sign-filter} to these polynomial values and to \(g\),
using uniform algebraic bounds and accuracy
\[
 \tau=\frac1{100(\ell+2)^2}.
\]
Write \(E_b(u)\) for the resulting bit-sign approximants and \(E_g(u)\)
for the qualification approximant. They are always defined at real
arguments, even for bad normalizer seeds.

Let \(h:\{0,1\}^{\ell}\to\{0,1\}^k\) be a random affine hash
\(h(r)=Mr+\beta\) over \(\F_2\). Its \(j\)-th parity sign is approximated
by \(\prod_{b:M_{jb}=1}E_b(u)\). Use the selector
\begin{equation}\label{push6:hash-selector}
 W_h(u)=
 \sum_{j=1}^k
 \left(\prod_{b:M_{jb}=1}E_b(u)-(-1)^{\beta_j}\right)^2
 +(E_g(u)+1)^2-\frac14.
\end{equation}
Take \(k\) uniformly in \(\{1,\ldots,H\}\), where \(H\) is the least
power of two at least \(\ell+1\), so \(H\le2(\ell+1)\).
Each product contains at most \(\ell\) factors. On good parameters
its error is at most \(2\ell\tau\le1/(50(\ell+2))\).
Thus matching parities and \(g<0\) make \eqref{push6:hash-selector}
negative, while a parity mismatch makes it positive.
For every seed, regardless of rank accuracy, \(g\ge0\) gives
\((E_g+1)^2\ge1\), and therefore \(W_h>0\).

Clear the positive squared real denominators in
\eqref{push6:hash-selector}, obtaining an integer polynomial \(V_h\).
The same algebraic gap argument as in \eqref{eq:scalar} supplies a
uniform short power \(\Lambda\) such that \(U_h=\Lambda V_h+1\)
is nonzero at every complex root and preserves every nonzero real sign.
Consequently \(Z_h=\Norm_{A}(U_h(J_w))\) is a nonzero integer.
For every seed, \(Z_h<0\) implies feasibility. On good parameters,
its sign is the parity of the qualifying roots whose labels hash to zero.

\paragraph{Success probability and the oracle ledger.}
Lemma~\ref{lem:generic} fixes a simple critical system whose objective
separates its points, without randomness. Conditional on a successful normalizer seed,
let \(\mathcal Q\subseteq\{0,1\}^{\ell}\) be the labels of the
qualifying real roots. If the instance is feasible, \(s=|\mathcal Q|\ge1\).
For \(k=\lceil\log_2s\rceil+1\le\ell+1\), pairwise independence gives
\[
 \Pr[|h^{-1}(0)\cap\mathcal Q|=1]
 \ge \frac{s}{2^k}-\frac{s(s-1)}{2^{2k}}\ge\frac3{16}.
\]
This follows by summing, over each possible isolated point, the
probability that it hashes to zero and using a union bound for the
other points. Since \(k\) is sampled from \(H\le2(\ell+1)\) choices,
one trial makes \(Z_h<0\) with probability at least
\(1/(20(\ell+1))\), including the normalization loss.
On an infeasible instance it never does.

For each fixed seed, all new inner products
\eqref{push6:rank-jet} use first-order jets. Qualification uses only
\(g=z-T\). At a fixed short homotopy node and
outer branch, their factors are polynomial-time computable; hence all
rank, bit, hash, and qualification values are in \(\FP^{\Clevel1}\).
The outer norm is in \(\FP^{\Clevel2}\).
Retaining the final Fourier sum unreduced, exactly as above, gives
\[
 W(\text{seed},M)\in\GapP^{\Clevel2},
 \qquad W(\text{seed},M)\equiv Z_h\pmod M.
\]
The whole beta-phase sum stays in \(\GapP^{\Clevel2}\), and is
retrievable in \(\FP^{\Clevel3}\). Lemma~\ref{push4:sign} therefore
tests \(Z_h<0\) in \(\BPP^{\Clevel3}\).

No analytic parameter is chosen separately at different moduli or
homotopy nodes: one seed fixes one rational SLP, one cleared polynomial,
and one integer norm. Every substitution and norm multiplies singly
exponential degree and height bounds by at most another singly
exponential factor. The number of such layers is fixed. Thus the same
shared-homotopy truncation and actual-degree Fourier bounds apply
uniformly to the entire construction.

Repeat \(200(\ell+1)\) independent trials, amplifying each sign test
to error at most \(1/[100\cdot200(\ell+1)]\), and accept if any test
reports a negative norm. On a feasible instance the probability of
never isolating a root is at most \(e^{-10}\). On either kind of
instance, the probability that any sign test errs is at most \(1/100\).
This proves \(\ETR\subseteq\BPP^{\Clevel3}\).
Finally relativized \(\BPP\subseteq\PP\) gives
\(\BPP^{\Clevel3}\subseteq\PP^{\Clevel3}=\Clevel4\).
The sign subroutine has two-sided error, so this argument does not
claim an \(\mathsf{RP}^{\Clevel3}\) containment.

\section{Exact rational weighted counting}\label{sec:rational-counting}

The sign-reconstruction lemmas also apply to a class defined independently
of arithmetic circuits. Following de Campos, Stamoulis, and Weyland
\cite[Definition~9]{CSW19}, let \(\mathsf{PP}_{\mathbb Q}^{(pt)}\)
denote decisions \(F\ge t\) for sums of polynomial-time computable rational
weights over polynomial-length witnesses. Here each weight is given exactly
by a short numerator and denominator. The reduced denominator of the total
sum can have exponentially many bits, so the polynomial-output hypothesis
of their Theorem~10 does not apply in general.

\begin{theorem}\label{thm:rational-counting}
Let \(\mathsf{CoeffPosSLP}\) denote positivity of a specified coefficient
of a division-free integer SLP, with its multiindex in binary. Then
\[
 \mathsf{PP}_{\mathbb Q}^{(pt)}
 \subseteq\BPP^{\PP}\cap\Pclass^{\NP^{\PP}},
 \qquad
 \Pclass^{\mathsf{CoeffPosSLP}}
   =\Pclass^{\mathsf{PP}_{\mathbb Q}^{(pt)}}.
\]
In particular, there is a polynomial-time many-one reduction from
\(\mathsf{PP}_{\mathbb Q}^{(pt)}\) to \(\mathsf{CoeffPosSLP}\).
The reverse reduction proved here is a polynomial-time Turing reduction.
\end{theorem}

\begin{proof}
After incorporating the threshold as one additional term, write the sum as
\[
 F(x)=\sum_{w\in\{0,1\}^p}\frac{u(x,w)}{v(x,w)},
 \qquad |u(x,w)|\le2^q,\quad 1\le v(x,w)\le2^q,
\]
where \(p,q\) are polynomially bounded and \(q\ge1\).
Set \(J=q2^p\). A common denominator has size at most \(2^J\),
so \(F\ne0\) implies \(|F|\ge2^{-J}\).
Divide every summand by \(C=2^{p+q+40}\).
Then \(\delta=F/C\) has absolute value below \(2^{-32}\), and a
nonzero value has a lower bound \(2^{-J'}\) with \(J'\) of polynomial
binary length. Its weights are still computable exactly in polynomial
time. Lemma~\ref{push4:sign} applies with plain \(\GapP\) phase sums,
giving \(\BPP^{\PP}\). The first-event search of
Section~\ref{sec:posslp} applies to the same sum and gives
\(\Pclass^{\NP^{\PP}}\).

For the reduction in the opposite direction, arithmetize the Boolean
circuits computing the weights. This gives integer SLPs \(U(X),V(X)\)
agreeing with \(u,v\) on \(\{0,1\}^p\). Put \(A=2^q\), and choose
a power of two \(K\) with
\[
 A(J+p+q+8)\le K<2A(J+p+q+8).
\]
The polynomial
\[
 P_K(z,A)=\sum_{k=0}^{K-1}z^kA^{K-1-k}
\]
has an \(O(\log K)\)-size SLP, using
\(P_{2r}=P_r(z,A)(z^r+A^r)\) and repeated squaring.
For \(Q(X)=U(X)P_K(A-V(X),A)\), the identity
\(vP_K(A-v,A)=A^K-(A-v)^K\) gives
\[
 \left|F-A^{-K}\sum_w Q(w)\right|
 \le 2^{p+q}e^{-K/A}\le2^{-J-8}.
\]
Write the approximating sum as \(\widetilde F\).

Choose a positive binary bound \(D\) on every individual degree of
\(Q\). The geometric polynomial
\[
 G_D(X)=\sum_{r=0}^{D-1}X^r+2X^D
\]
has an \(O(\log D)\)-size SLP and satisfies
\([X^D]h(X)G_D(X)=h(0)+h(1)\) for every \(\deg h\le D\).
Consequently the coefficient at \((D,\ldots,D)\) of
\[
 2^{J+3}Q(X)\prod_{j=1}^pG_D(X_j)
       +A^K\prod_{j=1}^pX_j^D
\]
is \(A^K(2^{J+3}\widetilde F+1)\). It is positive exactly when
\(F\ge0\): the scaled error is at most \(1/32\), whereas every negative
\(F\) has \(2^{J+3}F\le-8\). Reversing the offset gives the strict
test \(F>0\), with exact zero rejected.
All displayed exponents have polynomial binary length, and all powers
are produced by repeated squaring. This proves the many-one reduction.
It already holds for univariate coefficient positivity: the displayed
polynomial has individual degrees at most \(2D\), so substitute
\(X_j=X^{(2D+1)^{j-1}}\) and use the corresponding binary target index.
No two exponent vectors can collide.

Finally, Lemma~\ref{lem:ring-coefficients} and the beta construction
express a coefficient sign as the direction of a rational weighted sum
\(S=z+\delta\). Expanding the ordinary counting witnesses makes each
individual weight polynomial-time computable. Approximate \(S\) to
error below \(1/8\) by one complete weighted \(\GapP\) sum; since
\(|\delta|<2^{-32}\), its nearest integer is exactly \(z\).
Retrieve this count using \(\PP\subseteq\mathsf{PP}_{\mathbb Q}^{(pt)}\),
then ask whether \(-S\ge-z\) and negate the answer. This tests \(S>z\)
and proves the reverse Turing reduction and
the claimed equality of oracle classes.
\end{proof}

This equivalence concerns coefficient extraction. A reduction of arbitrary
rational weighted counting to scalar \(\PosSLP\) is not supplied by the
argument: the selected coefficient performs the Boolean summation.

\subsection{Two factor promises for scalar circuits}

Two restricted families admit substantially lower bounds. These are
promises on the output integer, not properties established for all SLPs
or for the norm in the ETR reduction.

\begin{proposition}\label{prop:factor-promises}
Let \(Z\) be the output of an integer SLP.
If every maximal prime-power divisor of \(|Z|\), when nonzero, is at
most \(2^B\), where \(B\) is supplied in unary, its sign is decidable
in \(\oplus\Pclass\).
If instead a positive binary integer \(M\) is supplied with the promise
that every prime factor of nonzero \(|Z|\) divides \(M\), its sign is
decidable in deterministic polynomial time.
\end{proposition}

\begin{proof}
For the first promise, evaluation modulo \(2^{B+1}\) detects zero
and otherwise determines \(v=v_2(|Z|)\). Evaluation modulo \(2^{B+2}\)
then gives \(r=(Z/2^v)\bmod4\), which is either 1 or 3.
The \(\SharpP\) count
\[
 h=\#\{(p,k):p\text{ prime},\ k\ge1,\ p\equiv3\pmod4,
                   \ p^k\le2^B,\ p^k\mid Z\}
\]
counts the multiplicities of all prime factors congruent to 3 modulo 4.
The positive odd part of \(|Z|\) is consequently \((-1)^h\bmod4\).
Thus \(Z>0\) exactly when \(r=(-1)^h\bmod4\). A parity machine
tests this equality by adding a known zero or one to \(h\).
Every prime-power divisibility test uses a polynomial-bit modulus.

For the second promise, put \(q=8M-1\) and compute the Jacobi symbol
\(\left(\frac Zq\right)\) using the short residue \(Z\bmod q\).
The integer \(q\) need not be prime. For every odd prime \(p\mid M\),
quadratic reciprocity and \(q\equiv-1\pmod p\) give
\[
 \left(\frac pq\right)
   =\left(\frac qp\right)(-1)^{(p-1)/2}=1.
\]
Also \(\left(\frac2q\right)=1\) and
\(\left(\frac{-1}q\right)=-1\). By multiplicativity,
\(\left(\frac Zq\right)=\operatorname{sgn}Z\) on nonzero promised
inputs. Such inputs are coprime to \(q\), so the zero case is detected
as well. Modular SLP evaluation and the short Jacobi computation are
both in polynomial time.
\end{proof}

The second promise permits arbitrarily large prime-power exponents.
For example, if every prime factor is at most a supplied unary bound
\(B\), use the polynomial-bit integer \(M=B!\).

\section{Ideal membership and conditional square certificates}\label{push7:sec}
The preceding construction uses two nested products over formal roots.
We next consider approaches that do not begin by optimizing those products.
They give two additional restricted results and a precise certificate
problem for a stronger real-feasibility bound. They do not improve the
general containments in Theorems~\ref{thm:etr} and~\ref{thm:posslp}.

\subsection{Exact membership in a finite monic algebra}
Consider integer polynomial circuits describing
\[
 G_i=X_i^{d_i}+r_i,\qquad \deg r_i<d_i,\qquad
 A=\mathbb Q[X_1,\ldots,X_m]/(G_1,\ldots,G_m).
\]
The integers \(d_i\ge1\) are given in binary. The degree inequalities
are promises on the input circuits. In particular, this includes the
degree-five critical systems used above. The leading monomials are
pairwise coprime, so the \(G_i\) are a monic Gr\"obner basis for any
total-degree order. Set
\[
 \rho_i=d_i-1,\qquad D=\prod_i d_i.
\]
The standard monomials \(X^\alpha\), \(0\le\alpha_i<d_i\), form an
integral basis. Normal-form reduction preserves integer coefficients.

\begin{theorem}\label{push7:membership}
Given the preceding system and an integer polynomial circuit \(F\),
both ordinary ideal membership
\[
 F\in(G_1,\ldots,G_m)
\]
and complex radical membership
\[
 F\in\sqrt{(G_1,\ldots,G_m)}
\]
are in \(\mathsf{coRP}^{\PP}\). No reducedness promise is required.
\end{theorem}

We first give a coefficient-zero test that is independent of the
small-characteristic coefficient construction.

\begin{lemma}\label{push7:coefficient-zero}
Zero testing a specified coefficient of an integer polynomial circuit,
with its exponent vector in binary, is in \(\mathsf{coRP}^{\PP}\).
\end{lemma}
\begin{proof}
Let \(c=[X^\alpha]F\). Circuit syntax gives bounds
\(\deg F\le\delta\) and \(|c|<2^H\), with \(\log\delta\) and
\(\log H\) polynomial in the input length. If some \(\alpha_i>\delta\),
the answer is immediate. Otherwise put
\[
 s=4\lceil\log_2(H+\delta+3)\rceil+16,\qquad N=2^s.
\]
Choose a uniform prime \(p\le N\) with \(p>\delta+1\), by sampling
integers and testing primality. Stop after \(16s\) attempts and answer
zero if no eligible prime was found.

Here are elementary bounds for this sampling step. For even \(N\ge4\),
\[
 2^{N/2}\le \binom{N}{N/2}\le N^{\pi(N)}.
\]
For the upper bound, the exponent of any prime \(p\) in the central
binomial coefficient is at most \(\lfloor\log_pN\rfloor\).
Thus \(\pi(N)\ge N/(2s)\). Our choice of \(N\) leaves at least
\(N/(4s)\) eligible primes, and sampling fails with probability at
most \(e^{-4}<1/32\). Conditional on success, the selected prime is
uniform over the eligible primes. If \(c\ne0\), it has at most \(H\)
distinct prime divisors, so the chance that \(p\mid c\) is less than
\(1/16\).

For such a prime, finite-field orthogonality gives
\[
 c\equiv(-1)^m
 \sum_{a\in(\mathbb F_p^\times)^m}
 F(a)\prod_i a_i^{-\alpha_i}\pmod p.                    \tag{*}
\]
There is no aliasing, since every individual degree and every
\(\alpha_i\) is less than \(p-1\).
Represent each summand by its integer residue in \([0,p-1]\).
It is computable in polynomial time. Guessing \(a\) and an additional
integer below that residue expresses the full integer sum as a
\(\SharpP\) count with polynomially many output bits. A \(\PP\) oracle
retrieves this count, and hence \(c\bmod p\).
Answer nonzero exactly when this residue is nonzero.
The zero case is always accepted; the other case has error below
\(1/8\).
\end{proof}

\begin{proof}[Proof of Theorem~\ref{push7:membership}]
Let \(\lambda(F)\) be the coefficient of \(X^\rho\) in the normal
form of \(F\). The pairing
\[
 (a,b)\longmapsto\lambda(ab)
\]
is perfect over \(\mathbb Z\). To see this directly, index its rows
by \(X^\alpha\) and its columns by \(X^{\rho-\gamma}\). If
\(|\alpha|<|\gamma|\), the entry is zero because reductions decrease
total degree. If \(|\alpha|=|\gamma|\), the entry is
\(\delta_{\alpha\gamma}\): an out-of-basis monomial reduces below
degree \(|\rho|\), and the only standard monomial of that degree is
\(X^\rho\). Ordering by total degree makes the matrix block triangular
with identity diagonal blocks.

Consequently the integer
\[
 U(F)=\sum_{0\le\alpha_i<d_i}\lambda(FX^\alpha)^2
\]
is nonnegative and is zero exactly when \(F=0\) in \(A\).
The sum does not need an exponentially long circuit. Define
\[
 K(X,Y)=\prod_i\sum_{a=0}^{d_i-1}(X_iY_i)^a.
\]
Then
\[
 U(F)=(\lambda\otimes\lambda)
       \bigl(F(X)F(Y)K(X,Y)\bigr).                    \tag{**}
\]

The functional itself has a short coefficient representation.
Expand each reciprocal as a Laurent series at infinity:
\[
 \frac1{G_i}=X_i^{-d_i}
       \sum_{j\ge0}(-r_iX_i^{-d_i})^j.
\]
Then \(\lambda(F)\) is the coefficient of
\(X_1^{-1}\cdots X_m^{-1}\) in \(F/\prod_iG_i\).
Indeed, this functional vanishes on every \(G_iH\), since after
cancellation no negative power of \(X_i\) remains, and on standard
monomials it is one exactly at \(X^\rho\).
Total-degree decrease makes the relevant coefficient finite.
If \(E\ge\deg F\), it equals
\[
 \left[\prod_i X_i^{d_i(E+1)-1}\right]
 F(X)\prod_i
 \left(\sum_{j=0}^E(-r_i(X))^jX_i^{d_i(E-j)}\right).
                                                               \tag{***}
\]
Every displayed geometric sum has a circuit of size polynomial in
\(\log E\) and the input size, by binary doubling.
Apply this representation in the two independent sets of variables
in \((**)\). It gives one specified coefficient of a short circuit
equal to \(U(F)\). Its degree and coefficient bit length have singly
exponential bounds. Lemma~\ref{push7:coefficient-zero} proves ordinary
membership.

Finally, the nilradical of a \(D\)-dimensional rational algebra has
nilpotence index at most \(D\). Thus
\[
 F\in\sqrt{(G_1,\ldots,G_m)}
 \quad\Longleftrightarrow\quad
 F^D\in(G_1,\ldots,G_m).
\]
Repeated squaring constructs \(F^D\) with \(O(\log D)\) additional
gates. Apply the ordinary membership algorithm.
\end{proof}

This is specifically a complex radical test. In
\(\mathbb Q[X]/(X^2+1)\), there are no real points, but \(1\) is not
in the ordinary radical. Real emptiness requires further information.

\subsection{Real emptiness and four-square certificates}
There is a particularly small extension of the critical-point
reduction. Use the fixed weighted objective \(J=J_w\), and
adjoin
\[
 Y^8+J(X)-T=0.                                         \tag{1}
\]
Its tail has total degree at most six. Together with the degree-five
\(G_i\), this is again a system of the form in
Theorem~\ref{push7:membership}, of dimension \(8\cdot5^m\).
A real solution of (1) exists over a real critical point exactly when
\(J(X)\le T\). Feasible instances have a critical minimum strictly below
\(T\); infeasible instances have \(J(X)>T\) at every real point.
Thus the original ETR instance is feasible exactly when this finite
rational algebra \(B\) has a real point. Neither genericity nor
reducedness is needed. The large scalar \(Q\) retains its short-SLP
representation.

\begin{lemma}\label{push7:four-squares}
A nonzero finite-dimensional commutative rational algebra \(B\) has
no real points if and only if \(-1\) is a sum of four squares in \(B\).
\end{lemma}
\begin{proof}
The reduced quotient is a product of number fields. If there are no
real points, none of these fields has a real embedding.
The form \(\langle1,1,1,1,1\rangle\) is isotropic over every completion:
the archimedean completions are complex, and quadratic forms of
dimension at least five over nonarchimedean local number fields are
isotropic. Hasse--Minkowski therefore gives a nontrivial zero over
each number field \cite{Lam05}. Dividing by one nonzero coordinate
writes \(-1\) as four squares. Assemble these identities in the
product of the fields.

Lift the four elements arbitrarily to \(B\), obtaining
\(\sum_i u_i^2=-1+n\) for a nilpotent \(n\).
Set \(a=1-n\), \(c_0=1\), and
\[
 c_{j+1}=c_j(3-ac_j^2)/2.
\]
For \(e_j=1-ac_j^2\), direct expansion gives
\[
 e_{j+1}=e_j^2(3+e_j)/4.
\]
Since \(n^{\dim B}=0\), after \(O(\log\dim B)\) iterations
\(ac_j^2=1\). The four elements \(c_ju_i\) then have squares summing
to \(-1\). Conversely, such an identity cannot hold after evaluation
at a real point.
\end{proof}

Clear a common rational denominator. A certificate of real emptiness
therefore has the form
\[
 c^2+U_1^2+U_2^2+U_3^2+U_4^2\in(G_1,\ldots,G_m,Y^8+J-T),
 \qquad c\in\mathbb Z\setminus\{0\}.                    \tag{2}
\]
If the \(U_i\) and \(c\) are supplied as short integer circuits,
Theorem~\ref{push7:membership} checks the identity in
\(\mathsf{coRP}^{\PP}\). Nonzeroness of \(c\) has a short additional
witness: an integer modulus \(q>1\) with \(c\bmod q\ne0\).
The circuit height bound and the prime-counting argument above ensure
that a polynomial-bit such modulus exists.

\begin{proposition}[Conditional certificate bound]\label{push7:conditional}
Suppose every real-empty algebra obtained from (1) admits certificates
(2) whose total integer-SLP size is bounded by a fixed polynomial in
the original input length. Then
\[
 \ETR\subseteq\mathsf{coMA}^{\PP}
 \subseteq\Clevel2\cap\PH^{\PP}.
\]
\end{proposition}
\begin{proof}
Merlin supplies the circuits and the nonzero-modulus witness.
The preceding verifier has perfect completeness and bounded soundness
error. This puts infeasibility in \(\MA^{\PP}\).
The usual inclusions \(\MA\subseteq\PP\) and
\(\MA\subseteq\Sigma_2^{\Pclass}\) relativize; take complements.
\end{proof}

The circuit-size hypothesis is open here. Four elements can still
require exponentially many coordinates in the standard basis.
The short Newton iteration in Lemma~\ref{push7:four-squares} only
lifts representatives that have already been found.

There is a direct variant in the original critical algebra \(A\).
Keep the fixed threshold and put \(h=J-T\).
On an infeasible instance, \(h\) is positive at every real point.
In each residue number field, a nonzero such element is four squares:
apply Hasse--Minkowski to
\(\langle1,1,1,1,-h\rangle\).
If the last coordinate of an isotropic vector is nonzero, division
gives the representation; otherwise \(\langle1,1,1,1\rangle\) is
isotropic and contains a hyperbolic plane, so it represents every
field element. On a component where \(h=0\), use four zeros.
Consequently there are integer polynomial representatives satisfying
\[
 c^2h-\sum_{i=1}^4P_i^2\in\sqrt{(G_1,\ldots,G_m)},\qquad c\ne0.
\]
Theorem~\ref{push7:membership} verifies a short supplied identity.
At a feasible critical minimum its left side evaluates to a strictly
negative number, proving soundness. Thus a polynomial circuit-size
bound for these direct certificates gives the same conditional
\(\mathsf{coMA}^{\PP}\) containment, with no extra variable, unit
condition, or nilpotent lifting. Short representatives in the reduced
algebra are still unproved.

Even fixed algebra dimension does not settle this compression problem.
For an integer SLP value \(a\), consider \(B_a=\mathbb Q[x]/(x^2+a)\).
It is real-empty exactly when \(a>0\).
If all these algebras admitted polynomial-size certificates (2), then
\(\PosSLP\) would be in \(\MA\cap\mathsf{coMA}\).
Indeed, reduce each certificate circuit modulo \(x^2+a\), representing
its value by a pair of scalar SLPs. The two required scalar identities
are testable by the ordinary \(\mathsf{coRP}\) zero test for integer
SLPs. A nonzero modulus witnesses \(c\ne0\).
This gives \(\PosSLP\in\MA\); apply the same argument to \(-a\) and
include the \(\mathsf{coRP}\) zero test to obtain its complement in
\(\MA\). This implication is not a lower bound against short
certificates. It explains why fixed-degree number-field existence
theorems do not supply the missing complexity bound.

A related sufficient hypothesis was studied explicitly by Saranurak
and Jindal~\cite{SJ12}: whether every positive integer \(a\) has an
addition-multiplication chain of length
\(\tau_+(a)\le(1+\tau(a))^{O(1)}\), where \(\tau\) also permits
subtraction. Their Claim~2 guesses such a chain and checks equality
by the integer-SLP zero test, obtaining a conditional
\(\Sigma_2^{\Pclass}\) bound. The same verifier is an \(\MA\) protocol;
handling \(-a\) and zero gives \(\MA\cap\mathsf{coMA}\).
No polynomial domination is proved here. More general arithmetic
circuits for semialgebraic certificates are studied in the cone proof
system of Alekseev et al.~\cite{AGHT19}; their efficient verification
likewise does not assert polynomial-size certificates for all inputs.

\subsection{A square-class promise for integer signs}
There is also a useful strengthening of the earlier factor-support
promise. Large primes can be allowed when their exponents are even.

\begin{theorem}\label{push7:square-class}
Let \(Z\) be an integer SLP value and let \(M>0\) be a supplied
polynomial-bit integer. Promise that either \(Z=0\), or
\[
 v_p(|Z|)\text{ odd}\quad\Longrightarrow\quad p\mid M
 \quad\text{for every prime }p.
\]
Then deciding \(Z>0\) is in \(\mathsf{RP}\). If \(Z\ne0\) is also
promised, its sign is computable in expected polynomial time with
zero error.
\end{theorem}

The nonzero values have the form \(Z=\varepsilon d u^2\), with
\(\varepsilon\in\{-1,1\}\) and \(d\) a squarefree divisor of
\(\operatorname{rad}(M)\). Neither \(d\) nor \(u\) is supplied.
For \(q=8Mt-1\), quadratic reciprocity gives
\[
 \left(\frac Zq\right)=\varepsilon
 \qquad\text{whenever }\gcd(Z,q)=1.                    \tag{3}
\]
Indeed, \(q\equiv7\pmod8\), so \((2/q)=1\) and
\((-1/q)=-1\). For every odd \(p\mid d\), one has
\(q\equiv-1\pmod p\), and the two reciprocity signs cancel to give
\((p/q)=1\). The square factor contributes one. If the gcd is not
one, the Jacobi symbol is zero. Every nonzero answer is therefore
correct.

For completeness, a short sampling interval suffices to find a
nonzero answer. Given \(0<|Z|<2^H\), put
\[
\begin{split}
 h&=\lceil\log_2(H+2)\rceil,\quad b=2h+64,\quad z=2^b,\\
 m_0&=40(1+\lceil\log_2b\rceil)+1,\qquad
 T=z^{2m_0+10}.
\end{split}
\]
For any positive even integer \(a\), a uniform \(t\in\{1,\ldots,T\}\)
satisfies
\[
 \Pr[\gcd(at-1,Z)=1]\ge\frac1{128b^3}.                 \tag{4}
\]
Here \(\log T=O(\log H\log\log(H+2))\), so \(t\) is short.

To prove (4), let \(\mathcal P\) contain the prime divisors \(p\le z\)
of \(Z\) with \(p\nmid a\). They are all at least three.
The elementary central-binomial estimate on dyadic intervals gives
\[
 \lambda=\sum_{p\in\mathcal P}\frac1p\le2(1+\log b),
 \qquad
 V=\prod_{p\in\mathcal P}(1-1/p)\ge e^{-3\lambda/2}
 \ge\frac1{32b^3},
\]
where these logarithms are natural.
Each prime excludes one residue class of \(t\).
Chinese remaindering and the odd Bonferroni truncation at \(m_0\)
give a surviving fraction at least
\[
 V-\frac{\lambda^{m_0+1}}{(m_0+1)!}-\frac{2z^{m_0}}T
 \ge V/2.
\]
Each displayed error is at most \(1/(128b^3)\le V/4\);
use \(m_0+1\ge20\lambda\).
There are at most \(H\) larger bad primes, whose excluded fractions
sum to at most \(H/z+H/T\le1/(128b^3)\le V/4\).
This proves (4).

Apply (4) with \(a=8M\). Modular circuit evaluation and the Euclidean
Jacobi algorithm take polynomial time.
Accept positivity if any of \(256b^3\) independent trials yields
\(+1\). A nonpositive value is never accepted, while a positive value
is missed with probability at most \(e^{-2}\).
Under the nonzero promise, repeat until the symbol is nonzero; the
expected number of trials is at most \(128b^3\).
This proves Theorem~\ref{push7:square-class}.

The representation of \(M\) is essential. For every nonzero integer
\(Z\), the same reciprocity identity gives
\[
 \left(\frac{Z}{8Z^2-1}\right)=\sgn Z.
\]
Here the denominator has a short SLP, but can have exponentially many
bits. Evaluating a Jacobi symbol with such a succinct denominator
already contains the original sign problem.

\section{Deterministic parameters and square-class promises}\label{push8:sec}
The preceding general containments have not decreased further. This
section records two ways to simplify their arithmetic, a stronger
promise result for integer circuits, and a broader certificate target.
The distinction between a short \emph{description} and a small dense
degree bound remains essential.

\subsection{Direct arithmetic at the prime two}
The proof of Lemma~\ref{lem:generic} does more than remove generic
choices. At any ordinary polynomial precision \(a\), all the critical
points can be evaluated directly in the fixed-rank Galois ring
\(\mathrm{GR}(2^a,6)\). An index in \(\{1,\ldots,5\}^m\) specifies a
residue tuple in \(\F_{64}^m\); Newton lifting with its invertible
Jacobian computes the unique lift in polynomial time. There is no
homotopy variable in this calculation. The exponent \(a\) here has
polynomial \emph{value}, rather than merely polynomial binary length.

The separating objective also gives an exact local discriminant calculation.
For each point \(x\), there are \(4\cdot5^{m-i}\) points \(y\) whose first
different residue coordinate is \(i\). Consequently
\begin{equation}\label{push8:discriminant}
 v_2(q'(J_w(x)))=
 4\sum_{i=1}^m(i-1)5^{m-i}
 =\frac{5^m-4m-1}{4}=:w_m,
 \qquad v_2(\operatorname{disc}q)=5^m w_m.
\end{equation}
These are local separation statements; they do not order the real
embeddings. For example, \(m=1\), \(F=X^2\), \(Q=8\) gives
\(G=X^5+17X+1\). Its reduction modulo three is irreducible, while its
reduction modulo two has factor degrees two and three. The Galois
group is transitive of prime degree, hence primitive, and contains a
transposition by cubing the latter Frobenius permutation. It is
therefore \(S_5\), despite the degree-six local splitting.
For \(m>1\), one must lift the original coordinate equations,
not simple roots of \(q\) modulo two: in fact
\[
 q(z)\equiv(z^5+z^4+1)^{5^{m-1}}\pmod2.
\]
The coordinate algebra is \'etale at two, but the suborder generated
by \(J_w\) has index with valuation \(5^m w_m/2\).

There is also a useful closure property for raw counts.
\begin{lemma}\label{push8:small-char}
Fix a constant \(r\). Let \(f_1,\ldots,f_r\) be \(\GapP\) functions,
and let \(C\) be a uniformly given polynomial-size integer SLP in
\(r\) variables. If \(p\) is prime and \(p,a\) have polynomial value,
then \(C(f_1,\ldots,f_r)\) has a raw \(\GapP\) representative modulo
\(p^a\).
\end{lemma}
\begin{proof}
Write the integer polynomial computed by \(C\) in the falling-factorial
basis:
\[
 C(X)=\sum_k b_k\prod_{j=1}^r(X_j)_{k_j},
 \qquad
 b_k=\frac{\Delta^k C(0)}{\prod_j k_j!}\in\Z.
\]
For every integer \(x\), including negative \(x\), \((x)_k\) is
divisible by \(p^a\) when \(k\ge pa\). Thus only \(0\le k_j<pa\)
matter. There are polynomially many such tuples, since \(r\) is
constant. Evaluate the finite-difference formula modulo
\(p^a\prod_jk_j!\); its arguments and modulus have polynomially many
bits. Dividing the resulting canonical residue by the known factorial
product recovers \(b_k\bmod p^a\) in polynomial time. Each surviving
falling factorial is a product of polynomially many raw counts.
The standard sum and product closure of \(\GapP\) completes the proof.
\end{proof}

Similarly, if \(V_i\) have raw \(\GapP\) representatives and the
indices \(i\) have polynomial length, then
\[
 \prod_i(1+pV_i)
 \equiv\sum_{\substack{S\text{ a set of indices}\\|S|<a}}
       p^{|S|}\prod_{i\in S}V_i\pmod{p^a}
\]
is again one raw \(\GapP\) representative. Increasing tuples encode
the subsets uniquely; each tuple has polynomial length.
The same argument holds coordinatewise in a fixed-rank ring.

These observations apply to a modified version of the rational
normalizer. Choose odd denominator constants and even numerator
scales, so every denominator is coefficientwise one modulo two and
every numerator is even. In the radial stage one can take odd \(C\),
an even integral initial radial scale, and odd-over-odd grid
multipliers. In the angular stage use even-over-odd parameters.
The same real approximation margins hold after increasing the
precision bounds. A denominator's algebraic norm is odd, so no
complex critical point is a pole. Its inverse modulo \(2^a\) is the
truncated series \(\sum_{j<a}(1-D)^j\).

The division by two in the rank shift does not divide a counting
value: the normalizer numerator has an explicit factor two, so we
evaluate its short circuit for \(R/2\), and then take
\(\Tr(R/2)+(5^m-1)/2\). This rank input has six raw coordinates.
The bit polynomials also require an adjustment, since their original
scalar denominators can be even. Replace each nonintegral scalar
\(U/V\), \(V>0\), by
\[
 \frac{U2^K}{V2^K+1}.
\]
The denominator is odd and the error is at most \(|U|2^{-K}\).
A gatewise bound on the bit circuit and its derivatives on the
promised rank interval permits a singly exponential value of \(K\)
with polynomial binary length. Thus these replacements use short
circuits and consume arbitrarily little of the approximation margin.
Choose \(K\) before the modulus, so all residues concern one fixed
rational expression. Retaining the always-defined real filter
preserves the odd-denominator property.

All selector operations are now integral at two. Replace their
variable denominator inverses by the finite series above, and apply
Lemma~\ref{push8:small-char} to their six rank coordinates before
count retrieval. Write \(E\) for the sum of the squared label/hash
errors and the qualifying-value error, and set
\[
 W=4E-1,\qquad 1+2W=8E-1.
\]
The outer factor is negative precisely when \(E<1/8\). The good-seed
matching error can be made smaller than \(1/16\), while the
nonmatching error exceeds \(1/2\). On a NO instance the
qualifying-value error alone exceeds one at every real critical
point, for every seed, so the factor is positive there. Each factor
belongs to \(1+2A_{(2)}\), where
\(A_{(2)}=\Z_{(2)}[X]/(G_1,\ldots,G_m)\); its norm is nonzero and its nonreal
conjugate pairs contribute positively. The displayed principal-unit
expansion therefore gives a raw \(\GapP\) representative of the
rational norm modulo each short \(2^a\), with no canonical rank
reduction or division of a raw count.

This does not yet give a sign algorithm. The rational numbers
\[
 \frac{2^H}{1+2^H},\qquad \frac{-2^H}{1+2^H}
\]
have opposite signs and approach \(1\) and \(-1\) over the reals, but
both vanish modulo \(2^a\) for every \(a\le H\). They have short
circuits when \(H\) has short binary length. Another obstruction is
\[
 \Norm_{\Q[X]/(X^D\mp1)}(1+\Lambda X)=1\pm\Lambda^D
 \quad(D\text{ odd}).
\]
The signs differ for \(\Lambda>1\), yet both residues are one
modulo \(p^a\) whenever \(a\le Dv_p(\Lambda)\).
Thus polynomial-precision local arithmetic alone does not supply
the large-modulus reconstruction used by the general proof.
The obstruction already occurs in this construction: replacing
\(F\) by \(F+1\), with common bounds and constants, leaves its
critical equations unchanged and shifts \(J_w\) by \(6Q_0\).
These objectives have the same residues at every precision
\(a\le v_2(6Q_0)\). Taking \(F=\sum_iX_i^2\) gives one feasible
and one infeasible nonnegative coercive instance. Moreover, when
\(a\) is below \(E_0-m\), both the critical equations and the
objective reduce to their \(F\)-independent perturbation terms.

\subsection{A small squarefree kernel need not be supplied}
\begin{theorem}\label{push8:kernel}
Let an integer SLP compute \(A\), with \(|A|<2^H\), and let \(B\ge1\)
be supplied in unary. Promise that either \(A=0\), or
\[
 A=ku^2,\qquad u>0,\quad k\text{ signed squarefree},\quad |k|<2^B.
\]
Then square testing is in \(\mathsf{coRP}\), and positivity is in
\(\mathsf{MA}\cap\mathsf{coMA}\subseteq\PP\), on this promise.
The signed kernel itself can be recovered in
\(\mathsf{FZPP}^{\NP}\) and in \(\mathsf{FP}^{\PP}\).
No factorization of \(k\), and no short circuit for \(u\), is supplied.
\end{theorem}
\begin{proof}
We give the quantitative square test first. Put
\[
 h=\lceil\log_2(H+2)\rceil,\quad
 b=16(h+B+12),\quad z=2^b,\quad
 s=40(1+\lceil\log_2b\rceil)+1,\quad T=z^{2s+10}.
\]
For a nonsquare nonzero \(A\), the Jacobi character
\(\chi(q)=(k/q)\) is nonprincipal on the units modulo \(C=8|k|\).
For negative \(k\), its value at \(C-1\) is \(-1\). For positive
squarefree \(k>1\), quadratic reciprocity and CRT give a nonresidue
class; the case \(k=2\) uses \(q=5\pmod8\).
Exactly \(\varphi(C)/2\) unit classes therefore have value \(-1\).

Primes in the square part can erase samples by making the Jacobi
symbol zero. Let \(\mathcal P\) consist of the primes \(p\le z\)
dividing \(A\) but not \(C\), and put
\[
 \rho=\frac{\varphi(C)}{2C},\qquad
 V=\prod_{p\in\mathcal P}(1-1/p),\qquad
 \lambda=\sum_{p\in\mathcal P}1/p.
\]
The elementary dyadic binomial estimate gives
\(\lambda\le2(1+\log b)\), and
\[
 \rho V
 =\frac14\prod_{\substack{p\mid A,\ p\le z\\p\ {\rm odd}}}(1-1/p)
 \ge \frac1{128b^3}.
\]
Here all primes dividing \(C\) are below \(z\). For the product
bound, use \(\log(1-1/p)\ge-1/p-1/p^2\) for odd primes and
\(\sum_{p\ge3}p^{-2}<1/2\).

For a product \(d\) of primes in \(\mathcal P\), CRT counts negative
character classes divisible by \(d\) in \([1,T]\) with density
\(\rho/d\) and additive error at most \(C\).
Odd-order Bonferroni, truncated at order \(s\), gives the lower
bound \(\rho V\), less three errors:
\[
 \frac{\rho\lambda^{s+1}}{(s+1)!},\qquad
 \frac{2Cz^s}{T},\qquad
 \frac H z+\frac H T.
\]
The first is the truncation tail, the second bounds the rounding
errors of the at most \(2z^s\) subsets, and the third removes
divisibility by primes greater than \(z\). Each is at most
\(1/(512b^3)\) for the displayed parameters. Consequently
\begin{equation}\label{push8:jacobi-density}
 \Pr_{1\le q\le T}\bigl[q\text{ odd},\ (A/q)=-1\bigr]
 \ge\frac1{512b^3}.
\end{equation}
All sampled \(q\) have polynomially many bits. Evaluate \(A\bmod q\)
and its Jacobi symbol in polynomial time. Repeating \(1024b^3\)
times rejects every nonsquare with probability greater than \(2/3\);
a square, including zero, is never rejected.

For positivity Merlin guesses \(1\le d<2^B\) and a modulus
\(2\le r\le2(H+1)\) at which \(A\) is nonzero. Arthur checks that remainder and applies
the square test to \(Ad\). Every guessed \(d\), including a false
one, satisfies
\[
 |\operatorname{kernel}(Ad)|<2^{2B},\qquad |Ad|<2^{H+B}.
\]
If \(A>0\), choosing \(d=k\) gives a square and perfect acceptance.
If \(A<0\), every \(Ad\) is a nonsquare. Zero is excluded by the
modulus witness. The complement has the analogous protocol for
\(-A\), together with a zero tag checked by the usual coRP
integer-circuit identity test. This proves the MA claims.

For recovery, use the parameters of the square test with
\(B'=2B\) and \(H'=H+B\), and denote the resulting values by
\(b',T'\). A common sample of \(1024(b')^3(B+2)\) moduli eliminates
all false pairs \((\sigma,d)\), \(\sigma\in\{-1,1\}\), \(1\le d<2^B\),
with probability at least \(3/4\), by amplification and a union
bound over fewer than \(2^{B+1}\) pairs. True pairs always survive. NP prefix queries find a surviving
pair with smallest \(d\). A further NP query asks whether any
odd \(1<q\le T'\) has \((\sigma Ad/q)=-1\); by
\eqref{push8:jacobi-density} this detects every false pair exactly.
Restart if it does. Otherwise minimality implies \(d=|k|\), since
all true pairs have \(d=|k|v^2\). The expected number of rounds is
at most \(4/3\).

Zero can be decided by one preliminary NP query for
\(2\le r\le2(H+1)\) with \(A\bmod r\ne0\). Such a modulus exists
for nonzero \(A\), since
\(\operatorname{lcm}(1,\ldots,2(H+1))\ge2^{H+1}\).
Finally, each bit of the unique signed kernel has an MA certificate:
for nonzero \(A\), guess \((\sigma,d)\), certify the squarefreeness
of \(d\) by its short sorted prime factorization, supply a nonzero
modulus for \(A\), check the claimed bit, and use the same square test.
Use a sign bit and \(B\) magnitude bits, with zero encoded by all zeros.
A zero tag is permitted only when claiming bit zero and is checked
by the ordinary coRP zero test.
The opposite bit has the same certificate format. Thus each bit
lies in PP and the kernel is recoverable in \(\mathsf{FP}^{\PP}\).
\end{proof}

The promise has not been proved for unrestricted circuit outputs or
for the norms in the ETR construction. A fixed list of Jacobi moduli
can be hidden by a large square factor: \(A=\pm(\prod_jq_j)^2\) gives
zero residues at all the listed \(q_j\). This is a limitation of such
a preselected list, not a uniform lower bound against deterministic
algorithms that inspect the circuit.

\subsection{Allowing weighted and coefficient-defined certificates}
The four-square target in Section~\ref{push7:sec} can be relaxed.
Permit a polynomial-time generator, on a polynomial-length index
\(i\), to specify coefficient-defined integer polynomials \(P_i\)
and scalar weights \(w_i\). A certificate asserts
\[
 c^2(J-T)-\sum_i w_iP_i^2\in\sqrt I,\qquad c\ne0,\quad w_i\ge0.
\]
All circuits, coefficient indices, and the generator must be fixed
before verifier randomness; dependence on the sampled prime would
not certify a rational identity.

For the radical identity, expand
\(\sum_\alpha\Tr(X^\alpha F)^2\) with
\(F=c^2(J-T)-\sum_iw_iP_i^2\).
There are only two square indices in each summand. Four independent
copies of the coefficient variables represent the two squared
polynomials; the tensor residue formula adds the trace factors.
The resulting coefficient sums have one raw modular counting
representation, so their zero test remains in
\(\mathsf{coRP}^{\PP}\).

The weights require a separate sign check. Amplify the
\(\BPP^{\PP}\) coefficient-sign algorithm enough that a union
bound covers every index, using one shared random tape. For a fixed
tape, existence of a weight declared negative is an
\(\NP^{\PP}\) predicate. Consequently a supplied certificate can
be verified in \(\BPP^{\NP^{\PP}}\).
If every NO instance had such a certificate of polynomial
description size, this would give
\[
 \ETR\subseteq\mathsf{coMA}^{\NP^{\PP}}\subseteq\Clevel3.
\]
Alternatively, the deterministic
\(\mathsf P^{\NP^{\PP}}\) coefficient-sign algorithm places
universal weight validity in \(\Pi_2^{\PP}\). Quantifying the
certificate yields the further conditional bound
\(\ETR\subseteq\Pi_3^{\PP}\subseteq\PH^{\PP}\).

The description-size hypothesis is still unproved. The effective
rational finite-set Positivstellensatz of Baldi, Krick and Mourrain
\cite{BKM25} supplies dense degree and height bounds, but its Gram
matrix construction, rounding, and factorization have dimension
polynomial in the full quotient dimension. They do not provide the
short coefficient generator required here.

\section{Fixed scalar denominators for sums of squares}\label{push9:sec}

The certificate approach leaves two different representation questions:
how to describe a common scalar denominator, and how to describe the
polynomials being squared. We can settle the first question without
solving the second. In every reduced finite integral algebra, an element
positive at all real points has a sum-of-squares representation with a
prescribed power of two as denominator. Its exponent is polynomial in
the dense algebra parameters, hence has polynomial binary length in our
critical-point construction. The number of squares and their dense
coefficient heights have the same type of bound. An odd-denominator
refinement uses the special reduction modulo two of the critical algebra.

These are existence and height statements. They do not provide short
circuits for the numerators. We make that distinction explicit before
applying the coefficient-zero verifier.

\subsection{A quantitative dyadic representation}

Let $\mathcal A$ be a commutative finite free $\Z$-algebra of rank $D$,
with integral basis $e_0=1,e_1,\ldots,e_{D-1}$. Suppose
$A=\mathcal A\otimes\Q$ is reduced. Write
\[
 e_i e_j=\sum_k m_{ij}^k e_k,\qquad
 h=\sum_k h_k e_k,\qquad
 |m_{ij}^k|\le B,\quad |h_k|\le C,
\]
where $B,C\ge1$ are integers. Define
\begin{equation}\label{push9:parameters}
 \begin{split}
 W&=\lceil\log_2B\rceil+\lceil\log_2C\rceil
       +2\lceil\log_2(D+1)\rceil+4,\\
 K&=10(D+1)W,\qquad L=40(D+1)W.
 \end{split}
\end{equation}
Larger integer upper bounds for the logarithms give equally valid
parameters. We use generous constants so that all rounding steps share
one explicit denominator.

\begin{theorem}\label{push9:dyadic}
If $h$ is strictly positive at every real point of $A$, there are
$s\le2D^2+3D$ elements $P_1,\ldots,P_s\in\mathcal A$ satisfying
\begin{equation}\label{push9:dyadic-identity}
 2^{2L}h=\sum_{j=1}^s P_j^2.
\end{equation}
Every coefficient of every $P_j$ in the integral basis has absolute
value at most $2^{100(D+1)W}$. Values of $h$ at nonreal points may be
zero; no unit hypothesis is required.
\end{theorem}

We prove the theorem by finding a positive definite real Gram matrix,
rounding it inside the exact affine space of Gram representations, and
decomposing the small remaining error. A discriminant bound controls
the precision needed in the first step.

Define the Gram map on symmetric matrices by
\[
 \mathcal L(M)=\sum_{i,j}M_{ij}e_i e_j.
\]
Off-diagonal entries occur twice. A positive semidefinite real matrix
with $\mathcal L(M)=h$ represents $h$ as a sum of squares over
$A\otimes\R$.

\begin{lemma}\label{push9:interior}
Put $U=DB$ and $M_0=DCU$. There is a real symmetric matrix $G$ with
$\mathcal L(G)=h$ and
\[
 \lambda_{\min}(G)\ge\mu:=\frac{M_0^{-(D-1)}}{(DU)^2},
 \qquad \|G\|_2\le N,
\]
where $\mu^{-1}\le2^{(D+2)W}$ and $N\le2^{(2D+4)W}$.
\end{lemma}
\begin{proof}
Write $A\otimes\R\simeq\R^r\times\C^q$, with $r+2q=D$.
Let $T$ be the real evaluation matrix: its rows evaluate the basis at
the real points and take real and imaginary parts at one point from
each nonreal pair. The matrix of multiplication by any $e_j$ has
entries bounded by $B$, so its eigenvalues have absolute value at
most $U=DB$. Thus every entry of $T$ has absolute value at most $U$,
and $\|T\|_2\le DU$.

Let $E$ be the complex evaluation matrix with all $D$ points included.
The integer $(\det E)^2$ is the nonzero discriminant of the integral
basis. Replacing each pair of real and imaginary rows by the two
conjugate evaluation rows multiplies the determinant in absolute value
by $2$. Consequently
\[
 |\det T|=2^{-q}|\det E|\ge2^{-q}.
\]
For $D>1$, cofactor bounds give
\[
 \|T^{-1}\|_2\le H_*:=D2^q((D-1)U)^{D-1}.
\]
For $D=1$ take $H_*=1$, since the single basis element is $1$.

Every conjugate of $h$ has absolute value at most $M_0$. The product
of its nonzero conjugates, with multiplicity, is up to sign the lowest
nonzero coefficient of the integral characteristic polynomial of
multiplication by $h$. This coefficient has absolute value at least
$1$. Every nonzero conjugate therefore has absolute value at least
$M_0^{-(D-1)}$. This argument also allows zero conjugates; all real
values are nonzero by the positivity assumption.

In real evaluation coordinates, form a block diagonal matrix $Q_0$.
At a real point its block is the positive scalar $h(x)$. At a complex
point where $h(x)=u+iv$, use
\[
 \begin{pmatrix}
 M_0+1+u/2&v/2\\ v/2&M_0+1-u/2
 \end{pmatrix}.
\]
Its eigenvalues are $M_0+1\pm|h(x)|/2$, both at least $1$.
The difference of its diagonal entries is $u$, and twice its
off-diagonal entry is $v$. It therefore represents $h(x)$ under the
complex evaluation $(a,b)\mapsto a+ib$. We obtain
\[
 G=T^{-1}Q_0T^{-T},\qquad \mathcal L(G)=h,
\]
and
\[
 \lambda_{\min}(G)\ge\frac{M_0^{-(D-1)}}{(DU)^2},
 \qquad \|G\|_2\le N:=2(M_0+1)H_*^2.
\]
Equality under all evaluations is equality in the reduced algebra.
Finally, the definitions give
$\log_2M_0,\log_2(DU)\le W-4$ and, with ample slack,
$\log_2H_*\le DW$. These inequalities imply the stated bounds.
\end{proof}

This construction does not compute the real points or their evaluation
matrix succinctly. Its role is to give a controlled interior point of
the Gram spectrahedron, so that the denominator needed for rational
rounding can be fixed in advance.

\begin{proof}[Proof of Theorem~\ref{push9:dyadic}]
The Gram map has a particularly simple half-integral right inverse.
For $v=\sum_i v_i e_i$, let $\mathcal R_2(v)$ have entries
\[
 (\mathcal R_2(v))_{00}=v_0,\qquad
 (\mathcal R_2(v))_{0i}=(\mathcal R_2(v))_{i0}=v_i/2
 \quad(i>0),
\]
and zero elsewhere. Then $\mathcal L\mathcal R_2(v)=v$.
If $|\cdot|_{\max}$ denotes the maximum absolute entry, the projection
$\Pi_2=\mathrm{Id}-\mathcal R_2\mathcal L$ satisfies
\[
 |\Pi_2(M)|_{\max}\le P|M|_{\max},\qquad
 P=1+D^2B\le2^{6W}.
\]

Round the upper-triangular entries of the matrix $G$ from
Lemma~\ref{push9:interior} to multiples of $2^{-K}$, and reflect them
to form a symmetric matrix $G_0$. Correct the affine error exactly:
\[
 G_1=G_0-\mathcal R_2(\mathcal L(G_0)-h).
\]
Its entries belong to $2^{-(K+1)}\Z$, and
\begin{equation}\label{push9:round-error}
 \mathcal L(G_1)=h,\qquad
 \|G_1-G\|_2\le DP2^{-K}/2\le\mu/4.
\end{equation}
Indeed, the error is $\Pi_2(G_0-G)$, and
\eqref{push9:parameters} gives
$2^K\ge2DP\mu^{-1}$. Thus $G_1\succeq3\mu I/4$ and
$\|G_1\|_2\le N+1$.

Set $\epsilon=2^{-K}\le\mu/8$. Choose a real square root matrix
$V$ with $V^TV=G_1-\epsilon I$, and round its entries to multiples
of $2^{-L}$, obtaining $V_0$. Its entrywise rounding error is at most
$2^{-L}/2$, while the entries of $V$ have absolute value at most
$N+1$. The exact residual
\[
 S=G_1-V_0^TV_0
\]
therefore satisfies
\begin{equation}\label{push9:dd-error}
 |S-\epsilon I|_{\max}
 \le D(N+2)2^{-L}\le\frac{\epsilon}{4D}.
\end{equation}
The last inequality follows from
$2^{L-K}\ge4D^2(N+2)$, again using
\eqref{push9:parameters}. In particular, $S$ is strictly diagonally
dominant with positive diagonal.

Let $u_i$ denote the ordinary coordinate columns. With
$d_i=S_{ii}-\sum_{j\ne i}|S_{ij}|>0$, we have the exact decomposition
\begin{equation}\label{push9:dd-decomposition}
 S=\sum_i d_i u_iu_i^T+
   \sum_{i<j}|S_{ij}|
   (u_i+\sgn(S_{ij})u_j)(u_i+\sgn(S_{ij})u_j)^T.
\end{equation}
Zero off-diagonal weights can be omitted. All weights are nonnegative
and have denominator dividing $2^{2L}$, since $K+1\le2L$.
Lagrange's four-square theorem, applied to each integer numerator,
writes each weight as four squares of rationals with denominator
$2^L$. Together with the $D$ rows of $V_0$, this decomposes $G_1$
into at most
\[
 D+4\left(D+\frac{D(D-1)}2\right)=2D^2+3D
\]
rational vector squares. Their denominators divide $2^L$.
Apply $\mathcal L$ and clear that denominator to obtain
\eqref{push9:dyadic-identity}.

The integral numerator vectors from $V_0$ have entries bounded by
$2^L(N+2)$. The vectors coming from
\eqref{push9:dd-decomposition} have at most two nonzero entries and
obey the same bound: every weight is at most $2\epsilon$, and each
integer in its four-square representation is at most the square root
of its numerator. The inequalities above imply
\[
 2^L(N+2)\le2^{100(D+1)W},
\]
which proves the coefficient bound.
\end{proof}

\subsection{An odd denominator when the algebra is étale at two}

In the critical algebra, the Gram map has an integral right inverse.
This lets us replace every dyadic grid by a triadic one and keep the
denominator odd.

\begin{proposition}\label{push9:triadic}
If, in addition, $\mathcal A/2\mathcal A$ is étale over $\F_2$,
Theorem~\ref{push9:dyadic} holds with $3^{2L}$ in place of $2^{2L}$,
with the same square count and coefficient bound.
\end{proposition}
\begin{proof}
Frobenius is a bijective $\F_2$-linear map on
$\mathcal A/2\mathcal A$. For each $e_k$, choose its Frobenius
preimage $a_k$ with all integral basis coefficients in $\{0,1\}$.
The coefficients of $r_k=e_k-a_k^2$ are even. Define
$\operatorname{Star}(r_k)$ to have entry $(r_k)_0$ at $(0,0)$,
entries $(r_k)_i/2$ at $(0,i)$ and $(i,0)$ for $i>0$, and zero
elsewhere. If $\mathbf a_k$ is the coefficient column of $a_k$, then
\[
 B_k=\mathbf a_k\mathbf a_k^T+\operatorname{Star}(r_k),
 \qquad \mathcal L(B_k)=e_k.
\]
All entries of $B_k$ are integers of absolute value at most $3D^2B$.
Thus
\[
 \mathcal R_3\Bigl(\sum_kv_ke_k\Bigr)=\sum_kv_kB_k
\]
is an integral right inverse, and
$\Pi_3=\mathrm{Id}-\mathcal R_3\mathcal L$ has maximum-entry norm
at most $1+3D^5B^2\le2^{6W}$.

Round $G$ on the grid $3^{-K}$ and correct using $\mathcal R_3$.
Now $G_1$ has entries in $3^{-K}\Z$. Take
$\epsilon=3^{-K}$ and round the square root on the grid $3^{-L}$.
The bounds \eqref{push9:round-error} and \eqref{push9:dd-error}
continue to hold, since the required inequalities have the form
$3^K\ge2DP\mu^{-1}$ and $3^{L-K}\ge4D^2(N+2)$.
The residual weights have denominator dividing $3^{2L}$.
Four-square representations of their integer numerators give the
common denominator $3^L$. Finally,
$3^L(N+2)<2^{100(D+1)W}$, using $\log_2 3<8/5$.
\end{proof}

The oddness is useful for compatibility with the local arithmetic of
Section~\ref{push8:sec}, but is unnecessary for the certificate
verifier: the dyadic scalar is already manifestly nonzero. Neither
proof bounds the circuit size of its rounded matrix or of the chosen
four-square decompositions.

\subsection{What is short in the critical-point application}

The exponent $L$ is genuinely short in binary for the algebras used
here. For fixed $d$, suppose
$G_i=X_i^d+r_i$ with $\deg r_i<d$. Let $s_0$ bound the number of
monomials in each $r_i$, and let their coefficient bit heights be at
most $H$. A product of two standard basis monomials has total degree
at most $2m(d-1)$. Every reduction step decreases total degree, so
the structure constants admit the crude bound
\[
 B\le\bigl(\max(2,s_0)2^H\bigr)^{2m(d-1)+1}.
\]
If an observable has degree $e$, at most $t_0$ monomials, and
coefficient bit heights at most $H_h$, its normal-form coefficients
are bounded by
\[
 C\le\max(2,t_0)2^{H_h}
       \bigl(\max(2,s_0)2^H\bigr)^{e+1}.
\]
These estimates just count reduction branches and multiply their
coefficient bounds. We use the logarithms of their right sides in
\eqref{push9:parameters}; there is no need to write the integers
$B$ and $C$ explicitly.

For our weighted critical system, $d=5$, $D=5^m$, and $h=J-T$ has
degree six. All logarithmic bounds above have singly exponential
values and polynomial binary length. Consequently $L$ is computable
in binary in polynomial time, and $2^L$ and $3^L$ have polynomial-size
integer SLPs. On every infeasible instance,
Theorem~\ref{push9:dyadic} gives at most exponentially many integral
numerators, of standard degree at most $4m$, with singly exponential
coefficient bit heights. The scalar is fixed by the input bounds,
rather than supplied as an unknown part of the witness.

What is missing is a coherent short description of those numerators.
A dense algorithm polynomial in $D$ does not give polynomial-time
access in $\log D$ to its output coefficients. In particular, the
explicit linear section of the Gram map does not make evaluation of
the real CRT matrix, matrix square roots, rounding, or integer
four-square decompositions a succinct computation.

\subsection{Verifying a compressed family of squares}

We can allow a more flexible description than an ordinary SLP for
each numerator. An integer polynomial $P_i(X)$ is
\emph{coefficient-defined} if a short integer polynomial circuit
$F_i(X,Y)$ and a binary exponent vector $\alpha_i$ specify
\[
 P_i(X)=[Y^{\alpha_i}]F_i(X,Y).
\]
A family is coherent if one polynomial-size description generates
these data in polynomial time from the polynomial-length index $i$.
The runtime and program sizes have fixed polynomial bounds in the
original input and certificate lengths. The syntactic degrees and
coefficient bit heights have singly exponential upper bounds with
polynomial-length binary descriptions. The generated integer circuits are fixed before verifier
randomness; they cannot depend on the sampled prime.

\begin{proposition}\label{push9:verifier}
In a finite pure-monic algebra of Theorem~\ref{push7:membership}, a
supplied coherent coefficient-defined family permits a
$\mathsf{coRP}^{\PP}$ test of
\begin{equation}\label{push9:certificate}
 c^2h=\sum_{i\in\{0,1\}^r}P_i^2\quad\text{in }A,
\end{equation}
where $r$ is polynomial and $c$ is a prescribed short integer SLP,
such as $2^L$ or $3^L$.
\end{proposition}
\begin{proof}
Let $F$ be the difference of the two sides. The perfect residue
pairing from Theorem~\ref{push7:membership} gives
\[
 F=0\quad\Longleftrightarrow\quad
 U(F):=\sum_\alpha\lambda(e_\alpha F)^2=0.
\]
The expression on the right is a nonnegative integer. Expand the
square before retrieving any modular values. Only two independent
square indices occur in each term. At most four independent copies
of the coefficient variables represent the corresponding squared
polynomials. The two residue factors have the finite polynomial
coefficient formula used in the proof of
Theorem~\ref{push7:membership}.

For a prime larger than every resulting individual degree,
finite-field orthogonality expands all these extractions into one
sum of polynomial-time computable residue weights. Its indices
consist of the two family indices, the basis index, and the
coefficient-extraction points; together they have polynomial binary
length. The signed sum has one raw $\GapP$ representative, or one
can sum nonnegative residues in $\SharpP$. A $\PP$ oracle retrieves
its residue. No canonical counting output is substituted into a
later counting summand.

The integer $U(F)$ has a singly exponential bit-height bound.
The prime sampler and divisibility argument of
Lemma~\ref{push7:coefficient-zero} therefore give a one-sided zero
test. A true identity always passes; a false identity fails except
with bounded probability.
\end{proof}

Separate scalar denominators are also possible without a common
denominator circuit. Suppose the right side of
\eqref{push9:certificate} is instead
$\sum_i P_i^2/c_i^2$, where each $c_i$ is an \emph{ordinary scalar}
integer SLP generated from $i$. Include in the certificate a short
integer $w\ge2$ such that $c_i\bmod w\ne0$ for all $i$.
An $\NP$ query detects any violating index, so this condition is
checked exactly using $\PP$. Such a common $w$ exists whenever all
$c_i$ are nonzero, by applying the height bound to their product.
For a sampled prime, another $\NP$ query detects any vanishing
denominator; that trial is treated as inconclusive. At every other
prime, the inverse scalar weights are polynomial-time computable
inside the same raw sum. The product of all $c_i^2$ is used only for
the integer numerator height bound, not constructed as an SLP.
This gives the same $\mathsf{coRP}^{\PP}$ verifier for the identity
\emph{together with its supplied nonvanishing witness}. Arbitrary
coefficient-defined scalar denominators would not justify this
argument: retrieving their residues before inversion could add a
counting level.

\begin{proposition}[Conditional numerator bound]\label{push9:conditional}
Suppose every infeasible ETR instance has a coherent numerator family
of polynomial description size satisfying
\eqref{push9:certificate} in its original critical algebra, with the
prescribed scalar $c=2^L$ and fixed polynomial bounds on the
generator. Then
\[
 \ETR\subseteq\mathsf{coMA}^{\PP}
       \subseteq\Clevel2\cap\PH^{\PP}.
\]
The same conclusion holds with $c=3^L$.
\end{proposition}
\begin{proof}
On an infeasible instance, Merlin supplies the hypothesized
description. Proposition~\ref{push9:verifier} has perfect completeness.
On a feasible instance, a real critical minimum has $h<0$.
The scalar $c$ is nonzero, so no sum of real squares can equal $c^2h$
there. Every supplied description therefore fails the identity test
with bounded soundness error. This puts infeasibility in
$\MA^{\PP}$. The inclusions $\MA\subseteq\PP$ and
$\MA\subseteq\Sigma_2^{\Pclass}$ relativize; taking complements
gives the displayed bounds.
\end{proof}

The hypothesis in Proposition~\ref{push9:conditional} remains open.
Theorem~\ref{push9:dyadic} removes the unknown scalar denominator and
supplies uniform dense bounds, but does not prove this hypothesis.
Allowing coefficient-defined numerators makes the verifier more
permissive; it does not turn a large table of coefficients into a
short description.

\subsection{Why a positive moment functional is a different target}

A rational positive functional cannot replace the Gram certificate
in general. If a reduced rational algebra is a product of number
fields, every rational linear functional has the form
$\ell(a)=\Tr(ca)$. At a complex embedding with $c=u+iv$, its
quadratic form on squares contributes
\[
 2\operatorname{Re}\bigl((u+iv)(x+iy)^2\bigr)
   =2u(x^2-y^2)-4vxy.
\]
The eigenvalues of this block are $\pm2|u+iv|$. Nonnegativity on
rational vectors implies nonnegativity on the full real vector space
by density, so a complex embedding forces its value of $c$ to vanish.
Injectivity then forces $c$ to vanish on that whole field factor.
Thus a normalized rational positive functional exists only if some
field factor is totally real. For example,
$\Q[x]/(x^3-x-1)$ has a real point but admits no such functional.
This obstruction to a rational functional is compatible with the
positive definite Gram matrices constructed above; the two objects
solve different equations.

Even the trace of a square need not be positive on an infeasible
critical system. In one variable take $F=x^2+1$, so
$G=x^5+(2Q+1)x+1$ and $T=276$. Reducing $h=J-T$ modulo $G$ and
using Newton sums gives
\[
 \Tr_A(h^2)=5(6Q+3-T)^2-16(2Q+1)^3<0\qquad(Q\ge1024),
\]
although $J\ge6Q>T$ on the real line. The nonreal contributions
prevent this trace from being a positive sum over just the real
points. These examples do not rule out approximate moment methods
or algorithms using several traces. They explain why the exact SOS
identity and its still-unproved succinct numerator representation
cannot simply be replaced by a single positive trace assertion.

\subsection{A degree obstruction to positive-coefficient annihilators}

Another possible route is to use the absence of negative real roots
of the characteristic polynomial of $h$. A polynomial positive on the
nonnegative real axis has a positive-coefficient P\'olya multiple after
multiplication by a sufficiently large power of $1+s$. The required
degree, however, can be too large even for a coefficient-defined SLP.
This occurs within our small-coefficient weighted critical family.

\begin{proposition}\label{push9:polya}
For every $m\ge2$, take
\[
 F_m=(X_1^2+1)^2+(X_1+2X_2)^2+\sum_{j=3}^mX_j^2
\]
and the prescribed weighted objective $J$, threshold $T$, and
$Q=A_0 2^{E_0}$. Let $q_h(z)=\Norm_A(z-h)$, where $h=J-T$.
Then $f(s)=(-1)^Dq_h(-s)$ is positive for $s\ge0$, but every nonzero
polynomial with nonnegative coefficients that is a multiple of $f$
has degree greater than $(\pi/3)Q2^m$.
\end{proposition}
\begin{proof}
Write $x=X_1,y=X_2$, $F_{12}=(x^2+1)^2+(x+2y)^2$, and
$V_{12}=b(x)+2b(y)$, where $b(u)=u^6+3(u+1)^2$.
At $z=(i,-i/2)$, both $F_{12}$ and its gradient vanish, and
\[
 H=\nabla^2F_{12}(z)=
 \begin{pmatrix}-6&4\\4&8\end{pmatrix},\qquad
 H^{-1}=\begin{pmatrix}-1/8&1/16\\1/16&3/32\end{pmatrix}.
\]
Its induced infinity norm is $3/16$. Put $t=1/Q$,
$p(u)=u^5+u+1$, and $P(x,y)=(p(x),2p(y))$. The critical equations
near $z$ are equivalent to
\[
 u=-H^{-1}\bigl((12iu_1^2+4u_1^3,0)+tP(z+u)\bigr)=:\Phi(u).
\]
On $\|u\|_\infty\le r=2t$, for $Q\ge64$, the elementary bounds
\[
 \|P\|_\infty\le6,\quad \|P'\|_\infty\le14,\quad
 |12iu_1^2+4u_1^3|\le16r^2
\]
give
\[
 \|\Phi(u)\|_\infty\le3r^2+9t/8<2t,
 \qquad \|\Phi'\|_\infty\le(51/4)t<1.
\]
There is a unique critical point $z_Q$ in this disc, and it is nonreal.
Also $|F_{12}(z+u)|\le15r^2$ and
$|V_{12}(z+u)-V_{12}(z)|\le60r$. Since $V_{12}(z)=111/32$,
\[
 \left|6QF_{12}(z_Q)+V_{12}(z_Q)-\frac{111}{32}\right|
 \le\frac{480}{Q}.
\]

For every $j\ge3$, choose the unique real minimum $a_j$ of
$J_j(s)=6Qs^2+w_jb(s)$. Its critical equation is
$s^5+(1+2Q/w_j)s+1=0$, with strictly positive derivative, and
$0\le J_j(a_j)\le3w_j$. These choices give a nonreal point of the
full critical algebra. At this point,
\[
 |\operatorname{Im}h|\le480/Q,\qquad
 \operatorname{Re}(-h)
 \ge270\,2^m-8367/32-480/Q\ge200\,2^m.
\]
The fixed objective separates this point from its conjugate, so
$\operatorname{Im}h\ne0$. Hence $\alpha=-h$ satisfies
\[
                 0<|\arg\alpha|<3/(Q2^m).
\]
It is a root of $f$. If a nonzero degree-$d$ polynomial
$\sum_{k=0}^d a_ks^k$, with all $a_k\ge0$, vanishes at $\alpha$,
then $d|\arg\alpha|\ge\pi$. Otherwise rotating its value by
$\exp(-id\arg(\alpha)/2)$ puts every nonzero term in the open
right half-plane. This proves the degree bound.

Finally, $F_m\ge1$ on $\R^m$, so $h>0$ at every real point.
The separating objective is nonreal at every nonreal point.
Real factors and conjugate pairs in
$f(s)=\prod_x(s+h(x))$ are therefore strictly positive for $s\ge0$.
\end{proof}

Both the constant and leading coefficients of $f$ are positive.
Homogeneous P\'olya positivity therefore guarantees an $N$ for which
$(1+s)^Nf(s)$ has positive coefficients, but necessarily
$N+D>(\pi/3)Q2^m$. Since $E_0$ is singly exponential in $m$, this
degree is doubly exponential. An SLP of polynomial size has at most
singly exponential degree, and extracting a coefficient in auxiliary
variables cannot increase its degree in $s$. Thus this particular
certificate route cannot supply short circuits for all our instances.

This is not an SOS lower bound. The same example has the identity
\[
\begin{split}
 h={}&6Q(X_1^2)^2+12QX_1^2+6Q(X_1+2X_2)^2
       +6Q\sum_{j\ge3}X_j^2\\
    &+\sum_iw_i\bigl((X_i^3)^2+3(X_i+1)^2\bigr)+(6Q-T)1^2,
\end{split}
\]
with short positive scalar weights. It even has a short coherent
unweighted family: write
$6Q-T=3A_0\sum_{j=0}^{E_0}2^j$ and express $2^j$ as one or two
copies of $(2^{\lfloor j/2\rfloor})^2$. Repeating the polynomial-bit
integer factors uses polynomial-length indices. The total number of
terms is singly exponential. Thus the example separates this failed
positive-coefficient route from the remaining SOS compression target.

\subsection{Selective valuation promises}
The integer problem admits a more selective arithmetic promise than a
short squarefree kernel. Write \(\oplus\Pclass\) for parity counting,
and \(\mathsf{MP}\) for the class defined by a selected bit of a
\(\SharpP\) function. We use the established inclusion
\(\BPP^{\oplus\Pclass}\subseteq\mathsf{MP}\)
\cite[Proposition 3.3(a)]{GKRST95}; we do not identify parity counting
with \(\PP\).

\begin{proposition}\label{push9:parity}
Let an integer SLP compute an odd integer \(A\). Suppose \(B\) and
an even positive integer \(K\) have polynomial value and are supplied
with the input. Promise only that
\[
 p\equiv3\pmod4,\quad v_p(|A|)\text{ odd}
 \quad\Longrightarrow\quad p\le2^B\ \text{ and }\ v_p(|A|)<K.
\]
Then positivity of \(A\) is in \(\oplus\Pclass\) on this promise.
If the corresponding condition holds for every odd prime, arbitrary
two-adic valuation and \(A=0\) are allowed at the cost of a randomized
many-one reduction to \(\oplus\Pclass\), hence an \(\mathsf{MP}\)
promise algorithm.
\end{proposition}
\begin{proof}
Count pairs \((p,j)\) satisfying
\[
 p\le2^B\text{ prime},\quad p\equiv3\pmod4,\quad
 1\le j\le K,\quad p^j\mid A.
\]
Each pair is checked in polynomial time, since \(p^j\) has at most
\(BK+1\) bits. A prime contributes \(\min(v_p(|A|),K)\). Its parity
is the full valuation parity: valuations at least \(K\) must be
even, and \(K\) is even. Omitted primes contribute one modulo four.
If \(a\) is the count, then \(|A|\equiv(-1)^a\pmod4\).
Combining this bit with \(A\bmod4\) gives the sign.

For the second assertion, choose a short \(q\equiv7\pmod8\)
coprime to \(A\), and count instead
\[
 a_q=\#\{(p,j):3\le p\le2^B\text{ prime},\ 1\le j\le K,
       \ p^j\mid A,\ (p/q)=-1\}.
\]
Here \((\cdot/q)\) is the Jacobi symbol, so primality of \(q\) is
unnecessary. Since \((2/q)=1\) and \((-1/q)=-1\),
\[
                  (A/q)=\sgn(A)(-1)^{a_q}.
\]
The coprime-modulus sieve used earlier applies to \(q=8t-1\).
For completeness, if \(0<|A|<2^H\), put
\[
 h_0=\lceil\log_2(H+2)\rceil,\quad b_0=2h_0+64,\quad z=2^{b_0},
 \quad s=40(1+\lceil\log_2 b_0\rceil)+1,
 \quad T_0=z^{2s+10}.
\]
Uniform \(t\in\{1,\ldots,T_0\}\) gives
\(\Pr[\gcd(A,8t-1)=1]\ge1/(128b_0^3)\).
The proof is the same odd Bonferroni truncation for small prime
divisors, followed by the large-prime union bound; the progression
has one forbidden residue modulo each odd prime. After
\(256b_0^3\) trials, use the first coprime modulus and one parity
query. Reject if there is none. Positive inputs accept with
probability at least \(1-e^{-2}\); negative inputs always reject.
Zero has no coprime sampled modulus and also rejects.
\end{proof}

The promise allows unrestricted even valuations, including at primes
of arbitrary size. In the first assertion it also leaves every prime
congruent to one modulo four unrestricted. Thus no short factorization
or short squarefree kernel is needed. The existing ETR norm is odd,
but oddness does not establish this valuation promise. Moreover, its
residues already use a \(\Clevel2\) oracle, so a parity decoder
there would be relativized rather than an ordinary parity computation.

\subsection{A supplied magnitude window gives one middle bit}
\begin{proposition}\label{push9:window}
Let \(Z\) have the raw modular counting representatives required by
the beta reconstruction, and \(|Z|<2^H\). Suppose binary \(e\) and
polynomial-valued \(w\ge1\) are supplied with
\[
          2^e\le |Z|<2^{e+w},\qquad e\ge0,\quad e+w\le H.
\]
Then positivity of \(Z\) is in \(\mathsf{MP}\) on this promise.
This applies to both ordinary SLP outputs and specified coefficients.
\end{proposition}
\begin{proof}
Put \(N=H+16\) and
\[
 C_N=\frac1{(2N+1)\binom{2N}{N}},\qquad m_i=N+1+i.
\]
The earlier beta identity supplies raw \(\GapP\) integers \(A_i\)
of polynomial bit length satisfying
\[
 \sum_{i=0}^N\frac{A_i}{m_i}=\text{integer}+C_NZ,
 \qquad \frac{4^{-N}}{2N+1}\le C_N\le4^{-N}.
\]
Choose \(j=2N-e-w-5\). The supplied window gives
\[
 \frac{2^{-w-5}}{2N+1}\le |2^jC_NZ|<\frac1{32}.
\]
If \(|A_i|\le2^a\), let
\[
 b=a+\lceil\log_2(N+1)\rceil+w+
             \lceil\log_2(2N+1)\rceil+10
\]
and form the single raw count
\[
 U=\sum_{i=0}^N A_i
       \left\lfloor\frac{2^b(2^j\bmod m_i)}{m_i}\right\rfloor.
\]
Every weight is polynomial-time computable; the sum has a
polynomial-length index. Its phase differs from \(2^jC_NZ\), modulo
one, by at most
\[
 (N+1)2^{a-b}\le\frac{2^{-w-10}}{2N+1}.
\]
Thus \(Z<0\) exactly when bit \(b-1\) of \(U\) is one.
Write \(U=g-h\), pad the machine for \(h\) to \(2^q\) paths with
\(q\ge b\), and add its rejecting paths to those counted by \(g\).
This gives a \(\SharpP\) function equal to \(U+2^q\), with the
same selected bit. Complementing a selected-bit predicate stays in
\(\mathsf{MP}\), proving positivity as well.
\end{proof}

A positive subtraction-free reference circuit \(R\), together with
\(2^{-w}R\le|Z|<2^wR\), gives an equivalent promise. Floating-point
evaluation of a size-\(s\) subtraction-free circuit needs only
\(s+O(1)\) mantissa bits for a constant relative approximation;
its binary exponent has polynomial length. This supplies the required
logarithmic window. It does not find such a reference for a general SLP.

Semantic positivity of the reference is insufficient. In fact, even
comparing two semantically positive SLP outputs under a constant
relative-gap promise remains \(\PosSLP\)-hard. To see this, start
with \(z=2A-1\ne0\) and apply the real rational normalizer using
integer numerator--denominator pairs:
\[
                  (N,D)\longmapsto(2sND,D^2+s^2N^2).
\]
One initial step with \(s=1\) sends \(1\le|z|<2^{H_0}\) into
\([2^{-H_0},1]\) in absolute value. If the current interval is
\([2^{-K},1]\), take \(s=2^{\lceil K/2\rceil}\). The new absolute
value lies in \([2^{-\lceil K/2\rceil},1]\), since
\(2u/(1+u^2)\ge\min(u,u^{-1})\). After \(O(\log H_0)\)
stages, \(D>0\), \(1/2\le|N/D|\le1\), and the sign is that of
\(A-1/2\). All scales have short repeated-squaring circuits.
The positive outputs
\[
                     X=2D+N,\qquad Y=2D
\]
satisfy \(|X-Y|/(X+Y)\ge1/10\), and \(X>Y\) exactly when \(A>0\).
Their internal circuits may contain cancellation. By contrast, a
polynomial relative gap between two genuinely subtraction-free
circuits is decided in polynomial time by relative approximation.
None of these observations supplies an unrestricted \(\mathsf{MP}\)
algorithm for \(\PosSLP\).

\section{Short positive certificates and arithmetic normal forms}
\label{push10:sec}

The preceding section fixes the scalar denominator in a dense
sum-of-squares representation. It leaves the description of the
numerators unresolved. We now give two more permissive certificate
formats, followed by consequences of the rational sign filter for
ordinary arithmetic circuits. None of these results lowers the
unrestricted containment \(\ETR\subseteq\Clevel4\).

\subsection{Conic identities in the critical algebra}

Call an integer polynomial circuit \emph{conic} if its leaves are
nonnegative integer constants and squares of arbitrary polynomial
SLPs, and its remaining gates are addition and multiplication.
Every value of such a circuit is nonnegative on \(\R^m\).
The circuit may share subexpressions. Expanding it into a list of
squares is unnecessary and can be very expensive.

\begin{proposition}\label{push10:conic}
Let \(I=(G_1,\ldots,G_m)\) be the fixed-degree pure-monic critical
ideal and let \(h=J-T\). A supplied pair of conic circuits \(A,B\)
satisfying
\[
                         hB=1+A\pmod I
\]
certifies that \(h>0\) at every real point of \(I\).
The identity has a \(\mathsf{coRP}^{\PP}\) verifier.
If polynomial-size certificates of this form exist for every NO
instance, then
\[
          \ETR\subseteq\mathsf{coMA}^{\PP}
          \subseteq\Clevel2\cap\PH^{\PP}.
\]
\end{proposition}

\begin{proof}
At a real point, \(A,B\ge0\), and the identity gives \(hB\ge1\).
Thus both \(B\) and \(h\) are positive there. In particular, no
separate unit or inverse certificate for \(B\) is needed.
The residual \(hB-1-A\) is an ordinary short polynomial SLP.
The pure-monic membership test of Theorem~\ref{push7:membership}
therefore verifies its vanishing in the quotient with one-sided error.
Guessing the circuits gives an \(\MA^{\PP}\) procedure for NO
instances, under the stated size hypothesis.
\end{proof}

Dense existence follows from the previous section. On a NO instance
\(h-1\) is still strictly positive at all real critical points.
Write \(c^2(h-1)=\sum_jP_j^2\), with the fixed dyadic scalar \(c\).
Then one can take
\[
       B=c^2,\qquad A=(c^2-1)+\sum_jP_j^2.
\]
The geometric-sum circuit for \(c^2-1\) is short and monotone.
The numerator description is still the missing part. Allowing a
nonconstant \(B\) and a conic circuit gives more freedom than this
particular dense construction.

There is a direct connection to the real cone proof system of
Alekseev, Grigoriev, Hirsch, and Tzameret \cite{AGHT19}.
A real-CPS refutation of \(G_i=0\) and \(-h\ge0\) gives a conic
circuit \(C(X,Y_+,Y_-,z)\) that becomes \(-1\) after substituting
the equations and \(z=-h\). Set the equation placeholders to zero;
this preserves the identity modulo \(I\). At every square gate,
substitute \(z=-h\) inside the square and retain the result as a
square leaf. The remaining circuit has conic coefficients in \(X\)
and uses the formal variable \(z\) only through addition and
multiplication.

Propagate an even/odd pair through this circuit, interpreting
\((U_0,U_1)\) as \(U_0-hU_1\). Its multiplication rule is
\[
 (U_0,U_1)(V_0,V_1)=
 (U_0V_0+h^2U_1V_1,\ U_0V_1+U_1V_0).
\]
Both components remain conic. At the output this yields
\(A-hB=-1\pmod I\), with polynomial circuit-size overhead.
This is a one-way simulation. Our verifier is allowed to check
ideal membership without receiving explicit ideal multipliers.
A universal polynomial-size real-CPS refutation theorem would be
stronger: ordinary polynomial identity testing would already give
\(\ETR\subseteq\mathsf{coMA}\).

The conic format also permits a restricted coefficient constructor.
Let \(A_0(X,Y)\) be built from \(Y\)-monomials, nonnegative
constants, and squares of \(Y\)-free polynomial SLPs, using only
addition and multiplication. Every \(Y\)-coefficient is a sum of
squares in \(X\). Thus \(A=[Y^\alpha]A_0\), and an independently
specified \(B\) of the same kind, are nonnegative by syntax.
Their membership residual and its residue-pairing square use only
a constant number of coefficient copies, so the same
\(\mathsf{coRP}^{\PP}\) verifier applies. Arbitrary
\(Y\)-dependent square leaves would be unsound here:
the coefficient of \(Y\) in \((Y-X)^2\) is \(-2X\).

\subsection{Congruence of four trace forms}

There is another certificate that does not require the witness to
respect multiplication in the algebra. Let
\[
 H_1=(\Tr_A(e_ie_j))_{ij},\qquad
 H_h=(\Tr_A(he_ie_j))_{ij}
\]
in the standard basis. Both are nonsingular because the algebra is
reduced and \(h\) is a unit.

\begin{proposition}\label{push10:fourforms}
The element \(h\) is positive at every real point if and only if
the rational quadratic forms \(H_1\otimes I_4\) and
\(H_h\otimes I_4\) are isometric.
\end{proposition}

\begin{proof}
If there are \(r\) real points, of which \(s\) have \(h<0\),
the two signatures are \(r\) and \(r-2s\). Complex conjugate
pairs contribute one positive and one negative direction to each
form. An isometry therefore implies \(s=0\).

Conversely, rationally diagonalize both forms. Equal signatures
allow their diagonal entries \(a_i,b_i\) to be paired with
\(b_i/a_i>0\). Write this ratio as \(u_i/v_i\) with positive
integers \(u_i,v_i\). Choose four integers whose squares sum to
\(u_iv_i\). Quaternion multiplication by their vector divided
by \(v_i\) gives a rational matrix \(Q_i\) satisfying
\(Q_i^tQ_i=(b_i/a_i)I_4\). The direct sum of these matrices,
followed by the inverse diagonalizations, gives the isometry.
\end{proof}

If the form entries have bit height \(H\), the dense isometry
can be chosen with bit heights polynomial in \(D,H\). To see
this, first use a small integer congruence making all leading
principal minors nonzero. Their product is a nonzero polynomial
of degree \(D(D+1)\) in the congruence entries, so a grid of
that size suffices. Rational LDL decomposition and Hadamard's
bound control all subsequent entries. The four-square vectors
above and a common denominator obey the same kind of bound.
This argument establishes height bounds, not short descriptions.

A supplied integer matrix \(B\), described by a polynomial-time
generator of coefficient-SLP entries, and a short nonzero scalar
\(c\), can certify
\[
 B^t(H_1\otimes I_4)B=c^2(H_h\otimes I_4).
\]
Square and sum the integer entry residuals. Each summand contains
only a constant number of matrix-entry and trace-entry coefficient
copies. All indices have polynomial binary length, so one raw
modular counting sum verifies the resulting integer is zero.
A supplied short modulus with \(c\not\equiv0\) certifies the
scalar's nonvanishing. Nonsingularity of the trace forms then
forces \(B\) to be invertible. A universal short-description
theorem for these witnesses would again give
\(\ETR\subseteq\mathsf{coMA}^{\PP}\). It is not proved here.

The number four has a local obstruction in the canonical family.
In fact \(\Norm_A(h)\equiv3\pmod4\). Modulo four the critical
equations are \(p(X_i)=X_i^5+X_i+1\), and the perturbation reduces
to \(b(X)=(X+1)(2X+3)\). Its rank-five norm is
\(\Norm_p(b)=259\equiv3\pmod4\), while \(\Tr_p(b)=15\).
For \(m\ge2\), write \(a=b(X_1)-T\), with \(T\equiv2\pmod4\).
Taking the second-coordinate norm of \(a+2b(X_2)\) gives
\(a^4(a+2)\pmod4\). Its first-coordinate norm is again three
modulo four; all remaining norms are odd fifth powers.
Odd numbers of trace-form copies therefore fail the determinant
square-class condition. Numbers of copies congruent to two modulo
four fail the two-adic Hasse condition
\((\Norm_A(h),-1)_2=-1\). The supplementary note gives the
complete calculation and a further three-square obstruction in
the algebra itself for \(m\ge3\).

\subsection{Four integral squares in the bounded critical algebra}

For the actual bounded ETR encoding, the scalar denominator can be
removed entirely from a four-square existence statement. This uses
more structure than the general quantitative theorem of the previous
section. The new statement is qualitative: it does not supply the
numerator-height or circuit-size bounds needed by the verifier.

Let \(s\) be the number of squared residuals in \(F\). Give each
original variable weight two and each slack variable weight one,
and write \(D_\theta=\sum_i\theta_iX_i\partial_{X_i}\).
Every constraint residual has the form \(f_j=P_j-c_j\), where
\(P_j\) is weighted homogeneous of degree \(d_j\in\{2,4\}\)
and \(|d_jc_j|\le4\). Thus
\[
 D_\theta F=2\sum_jd_j f_j(f_j+c_j)
       \ge4F-8\sqrt{sF}\ge2F-8s.
\]
At a real critical point, the equations give
\[
 D_\theta F=-Q_0^{-1}\sum_i\theta_iw_i
                         (X_i^6+X_i^2+X_i)
       \le\frac{2^m-1}{2Q_0}<\frac12.
\]
Here \(t^6+t^2+t\ge-1/4\). Consequently \(F<4s+1\).
The slack inequalities then show that every coordinate of every
real critical point has magnitude at most \(B=2s+3\).
This bound is independent of the enormous parameter \(Q_0\).

For this subsection, enlarge the exponent in \(Q_0\), if needed,
so that
\[
                 Q_0/2^{m-1}>B^5+B+1.
\]
An explicitly computable additional polynomial-value power of two
suffices. The objective gap, the separation argument, and the
counting-depth bounds are all preserved.

\begin{proposition}\label{push10:no-total-real}
Every number-field factor of the resulting critical algebra has
a complex place.
\end{proposition}
\begin{proof}
Suppose a factor \(K\) were totally real, of degree \(d\).
The coordinate images \(a_i\) are algebraic integers, and every
embedding of \(K\) gives a real critical point. Put
\(p(t)=t^5+t+1\) and \(q_i=Q_0/w_i\). The critical equations
give \(p(a_i)=-q_iF_{X_i}(a)\). Its integer norm is divisible
by \(q_i^d\), but the coordinate bound gives
\[
 |\Norm_{K/\Q}p(a_i)|\le(B^5+B+1)^d<q_i^d.
\]
The norm must vanish, so \(p(a_i)=0\). However,
\[
 p(t)=(t^2+t+1)(t^3-t^2+1).
\]
The quadratic is irreducible with no real root. The cubic is
irreducible by the rational-root test and has discriminant \(-23\).
Neither irreducible factor is totally real. This is impossible
for an element of a totally real field.
\end{proof}

Write \(R\) for the integral critical order and
\(A=R\otimes\Q\). The element \(h=J-T\) is a unit in
\(A\) and modulo two; it is not asserted to be an integral unit
at every prime.

\begin{proposition}\label{push10:integral-four}
For every NO instance of the enlarged construction, the integral
order \(R\) contains \(P_1,P_2,P_3,P_4\) satisfying
\[
                         h=P_1^2+P_2^2+P_3^2+P_4^2.
\]
\end{proposition}
\begin{proof}
Strict positivity at every real place gives a rational
four-square representation in each number-field factor, by the
local-global theorem for quaternion norms. Fix such a vector
\(x_0\), viewed as a quaternion of reduced norm \(h\).

Integral local representations exist at every finite prime.
At an odd prime choose \(a,b\in\Z_p\) with \(a^2+b^2=-1\).
For any \(h\in R\otimes\Z_p\), the identity
\[
 h=\left(\frac{h+1}{2}\right)^2
    +\left(a\frac{h-1}{2}\right)^2
    +\left(b\frac{h-1}{2}\right)^2
\]
gives one. At two the order is a product of unramified extensions.
Choose a unit \(a\) with \(a^2\equiv h\pmod2\), and solve
\[
                    2c^2+ac=(h-a^2)/2.
\]
There is a unique lift from the residue-field solution, since the
derivative \(4c+a\) is a unit. Then
\(h=(a+c)^2+c^2+c^2+c^2\).

The norm-one quaternion group in each field factor is noncompact
at a complex archimedean place, by
Proposition~\ref{push10:no-total-real}. Strong approximation away
from the archimedean places therefore applies
\cite[Main Theorem 28.5.3]{Voight21}.
For every finite prime impose the open condition
\(x_0g\in(R\otimes\Z_p)^4\) on a norm-one element \(g\).
These sets are nonempty by the local representations, and form
an adelic open set: outside finitely many primes the integral
norm-one subgroup already suffices. Strong approximation gives
a rational \(g\) satisfying all the conditions. The coordinates
of \(x_0g\) belong to \(R\), and its norm is \(h\).
The argument uses the given order locally and does not require
it to be maximal.
\end{proof}

The distinction between existence and size remains decisive.
The earlier many-square construction gives explicit uniform
height bounds with a fixed scalar. This four-square construction
has scalar one but, as used here, no comparable uniform height
bound. Neither construction gives short numerator circuits.

\subsection{Eliminating adaptive PosSLP queries}

The real sign filter also gives a reduction normal form that is
independent of our counting-hierarchy arguments.

\begin{lemma}\label{push10:filter}
Given positive binary integers \(L,B\), one can construct in time
polynomial in their bit lengths a rational SLP \(R_{L,B}\) such that
all denominators are positive on the real line and
\[
 2^{-L}\le |z|\le2^L
 \quad\Longrightarrow\quad
 |R_{L,B}(z)-\sgn z|\le2^{-B}.
\]
\end{lemma}

\begin{proof}
Use the folds \(\phi_s(z)=2sz/(1+s^2z^2)\), with positive
powers of two \(s\). The first scale-one fold puts the absolute
value in \([2^{-L},1]\). From \([2^{-K},1]\), the scale
\(s=2^{\lceil K/2\rceil}\) puts it in
\([2^{-\lceil K/2\rceil},1]\). After \(O(\log L)\) stages
the interval is \([1/2,1]\). A scale-one fold now changes the
error \(e=1-z\) to \(e^2/(1+z^2)\le e^2\), so
\(O(\log B)\) further folds suffice. Every scale constant
has a short repeated-squaring circuit.
\end{proof}

For \(z>0\), these folds have the useful stability property
\[
 \frac{d\log\phi_s(z)}{d\log z}
       =\frac{1-s^2z^2}{1+s^2z^2}\in[-1,1].
\]
Hence any composition satisfies \(|R'(z)|\le1/|z|\).
The next proof uses the stronger fact that every oracle gate can
reset its error using the filter's uniform guarantee.

\begin{theorem}\label{push10:adaptive}
Every language in \(\Pclass^{\PosSLP}\) is polynomial-time
many-one reducible to \(\PosSLP\).
\end{theorem}

\begin{proof}
Unroll a polynomial-time oracle computation as a polynomial-size
Boolean circuit with oracle gates. Pad every query to one fixed
polynomial length \(q\), and let \(G\) bound the number of ordinary
Boolean gates. There is a polynomial-size polynomial circuit
\(W_q(b)\) that evaluates an encoded SLP on Boolean descriptions:
it returns \(2Z_b-1\) on a valid description of value \(Z_b\),
and \(-1\) on an invalid description. Decode each instruction
by Boolean polynomials and use polynomial-size multiplexers over
the preceding registers. On every Boolean input the output is
a nonzero odd integer with the desired oracle sign.

Before choosing any approximation precision, bound this fixed
evaluator and its derivatives. The usual arithmetic-circuit
recurrences give a computable \(L=2^{\poly(n)}\) such that on
\([-1,2]^q\),
\[
 |W_q(b)|\le2^L,\qquad
 \sum_i|\partial_iW_q(b)|\le2^L.
\]
The integer \(L\) has polynomial binary length. Put
\(B=L+2G+20\) and \(\eta=2^{-B}\). Replace Boolean NOT and
AND by \(1-a\) and \(ab\), and replace each oracle gate by
\[
                 \frac{1+R_{L+2,B+2}(W_q(b))}{2}.
\]
Inductively, every wire stays within
\(\rho=4^G\eta<1/4\) of its ideal Boolean value. An ordinary
gate increases the current error by at most four. At an oracle
gate the evaluator error is at most
\(2^L\rho=2^{-20}\). Its ideal value is odd, so the perturbed
value has the same sign and lies in the filter's permitted range.
The filter resets the output error to at most \(\eta\).
This induction allows arbitrary adaptive dependencies.

The final rational value exceeds \(1/2\) precisely on accepting
inputs. Positive-denominator clearing produces an integer SLP
for this comparison. Its size is polynomial because every
filter parameter has polynomial binary length.
\end{proof}

The order of the bounds in this proof matters. Bounding derivatives
of the entire precision-dependent simulation would create a
circular accuracy requirement. Only the fixed universal evaluator
is bounded before precision is selected. The theorem gives one
query in place of adaptive queries; it does not remove an oracle
level from the existing upper bound for \(\PosSLP\). In particular,
ordinary SLP values cannot be substituted for arbitrary GapP values.

\subsection{Nested principal square roots}

There is a scalar square-root approximation inside the same folds.
Starting with \(c=1\), use the scales from \(R_{L,B}\) and set
\[
 c\longleftarrow\frac{2sc}{1+s^2ac^2},\qquad q(a)=ac_{\rm final}.
\]
For \(a>0\), one has
\(\sqrt a\,c_{\rm final}=R_{L,B}(\sqrt a)\). Thus
\[
 2^{-2L}\le a\le2^{2L}
 \quad\Longrightarrow\quad |q(a)-\sqrt a|\le2^{L-B}.
\]
All denominators are positive. This is also the sign iteration
applied to \(\left(\begin{smallmatrix}0&a\\1&0\end{smallmatrix}\right)\),
whose square is \(aI\).

\begin{proposition}\label{push10:radicals}
The sign problem for polynomial-size circuits with rational
leaves, arithmetic operations, exact division, and nested principal
square roots is polynomial-time many-one equivalent to \(\PosSLP\).
One may either promise valid divisions and radicands, or reject
invalid circuits as part of the language.
\end{proposition}

\begin{proof}
Let \(s\) be the number of gates and \(n\) the input length.
Every defined gate has algebraic degree at most \(2^s\) and
absolute logarithmic height at most \((n+1)2^s\). The standard
height inequalities for addition, multiplication, division, and
square roots give these bounds directly. Hence
\[
 H=(n+2)4^{s+1}
\]
bounds the logarithmic magnitude of every nonzero gate in both
directions: \(2^{-H}\le|v|\le2^H\).

Put \(c_0=3H+4\), and prescribe the error schedule
\[
 K_j=2^{s-j}(H+2c_0+10)-2c_0,
 \qquad E_j=2^{-K_j},\qquad 0\le j\le s.
\]
Then \(K_j=K_{j-1}/2-c_0\) and \(K_s=H+10\).
Approximate each gate by a rational SLP. Addition and multiplication
amplify the previous error by at most \(2^{H+3}\), and a valid
division by at most \(2^{3H+2}\). For a radicand \(a\ge0\)
with approximation \(\widetilde a\), use
\(a'=\widetilde a+2E_{j-1}>0\). The square-root circuit above
can approximate \(\sqrt{a'}\) within \(E_{j-1}/4\), while
\[
 |\sqrt{a'}-\sqrt a|\le\sqrt{3E_{j-1}}.
\]
Every new error is therefore at most
\(2^{c_0}\sqrt{E_{j-1}}=E_j\). This includes zero radicands.
All exponents have polynomial binary length, so the complete
rational approximation has polynomial circuit size.

The output error is at most \(2^{-H-10}\). Its comparison with
\(2^{-H-1}\) has the sign of the true positivity predicate and
gives one PosSLP instance after denominator clearing. To reject
invalid circuits, test the approximants of each divisor and
radicand against the same separated thresholds before using
them. These are polynomially many PosSLP queries, which
Theorem~\ref{push10:adaptive} combines into one. The reverse
reduction uses no radical gates.
\end{proof}

This statement concerns exact rational division. Integer division
and floor are different operations and are not included.

\subsection{Polynomial-degree real-root gates}

The square-root result extends to more general specified real
algebraic roots. We give the arithmetic-operation argument because
a bit-complexity bound polynomial in coefficient height would be
insufficient: ordinary SLP coefficients can have exponentially
many bits.

\begin{theorem}\label{push10:root-gates}
Consider a polynomial-size arithmetic DAG with rational constants,
arithmetic operations, exact division, and real-root gates. Each
root gate lists polynomially many earlier values
\(a_0,\ldots,a_d\) and an index \(j\), and returns the \(j\)-th
distinct real root of \(\sum_{k=0}^d a_kX^k\), in increasing order.
The output-sign problem is in \(\Pclass^{\PosSLP}\), and hence
is polynomial-time many-one reducible to \(\PosSLP\).
Invalid operations or root specifications can be rejected.
\end{theorem}

\paragraph{Floor-free rational isolation.}
First let the coefficients be ordinary rational SLP values.
Polynomial-degree Euclidean and Sturm computations use polynomially
many field operations and PosSLP tests. Compute the square-free
part before isolating its roots. Let its degree be \(d\), let all
roots have magnitude at most \(2^\Gamma\), and let their separation
be at least \(\sigma\).

We adapt the accelerated subdivision analysis of
Sagraloff and Mehlhorn \cite{SM15}. An active interval \(I=(a,b)\)
has width \(w\), nonroot endpoints, and a parameter
\(N_I=2^{2^{\ell_I}}\), with \(\ell_I\ge1\).
Exact Sturm counts discard empty intervals and return single-root
intervals. A successful accelerated step retains all real roots in
an interval of width in \([w/(8N_I),w/N_I]\), and increments
\(\ell_I\). Otherwise split in the middle half and decrement
\(\ell_I\), with minimum one.

The published Newton test uses a floor to locate a small interval
on a very fine grid. We avoid that integer. Try each multiplicity
\(k=1,\ldots,d\) and each quarter point \(\xi\) with
\(P'(\xi)\ne0\), and form
\[
             \lambda=\xi-kP(\xi)/P'(\xi),
             \qquad r=w/(4N_I).
\]
Clip \([\lambda-r,\lambda+r]\) to the parent, rejecting width
below \(3r/4\). Snap any endpoint within \(r/2\) of the
corresponding parent boundary to that boundary. Move each remaining
endpoint by at most \(r/16\) to avoid roots; a grid of \(d+1\)
points suffices. Accept if the exact real-root count is unchanged.
The final width lies in \([5r/8,25r/8]\), and every endpoint
that moves does so by at least \(7r/16\), which is at least
one eighth of the final width. Also test intervals of width between
\(11w/(16N_I)\) and \(13w/(16N_I)\) adjacent to either parent
boundary, using nonroot endpoints.

These modifications preserve the geometric invariants of
\cite[Lemmas 24--26]{SM15}. In particular, if a cluster lies
in a subinterval of width at most \(w/(8192N_I)\) and all other
complex roots are at distance at least \(1024dN_Iw\), one trial
with its true multiplicity puts \(\lambda\) within
\(w/(128N_I)\) of the cluster center. The clipped interval retains
it, so the exact count accepts. The endpoint-movement bound is the
other essential invariant in the nonsplitting-chain argument.
Count splitting nodes by losses or divisions of the complex-root
sets in the one-circle disks, not just by changes in real-root
counts. There are \(O(d)\) such nodes. The resulting forest has
\[
 O\!\left(d\left[\log(d+1)+
       \log(\Gamma+\log(1/\sigma)+2)\right]\right)
\]
nodes. A single strict Cauchy interval suffices for initialization.
All endpoints are short rational SLPs; no long floor is written.
The supplementary root-isolation note details the invariant
transfer and provides an independent geometric audit.

To refine to error \(2^{-B}\), choose
\(\varepsilon<\min(2^{-B-2},\sigma/4)\) and isolate the roots of
\[
                 P(X-\varepsilon)P(X)P(X+\varepsilon).
\]
Every original real root lies between its two shifted copies.
Its isolating interval therefore has width below \(2\varepsilon\).
The threefold degree increase and the short power-of-two shift
preserve the arithmetic bound, polynomial in \(d\) and the binary
lengths of the magnitude, separation, and precision exponents.

\paragraph{Nested coefficients and stable root choices.}
We explain why nesting does not cause repeated circuit-size or
precision blowup. Let \(N\) bound the input DAG size and all
listed degrees. Its selected roots generate a field of degree at
most \(N^N\). Standard denominator and conjugate bounds give
a computable \(H_0=2^{\poly(N)}\) which bounds all semantic
gate magnitudes, inverse nonzero magnitudes, root separations,
and the conditioning of polynomial-degree gcd computations.
These bounds include the auxiliary rational expressions in every
possible Euclidean branch, whose sizes have one polynomial bound.
For example, integer norms clear inverse denominators and bound
nonzero values away from zero; a root gate increases field degree
by at most its listed degree. Two roots require an extension of
degree at most \(d^2\) over the coefficient field to bound their
difference. No full splitting field is needed.

Use approximations to the earlier selected roots to determine
the \emph{true} degree and square-free Euclidean branch, comparing
tested quantities with \(\pm2^{-H_0-1}\). Values determined to
be zero are set to literal zero in that branch. Then evaluate the
selected square-free coefficient formulas at the earlier rational
approximants. This produces a rational perturbation \(Q\) of the
true square-free polynomial \(P\); it does not square-free an
arbitrary perturbation of a multiple-root polynomial.

Suppose the leading coefficient of \(P\) has magnitude at least
\(2^{-H_0}\), its roots have magnitude at most \(2^{H_0}\), and
their separation is at least \(2^{-H_0}\). On the circle of radius
\(\eta=2^{-B-2}<2^{-H_0-2}\) about a root,
\[
 |P(z)|\ge2^{-H_0}\eta(2^{-H_0-1})^{d-1}.
\]
A coefficient error at most
\(2^{-B-8(d+1)(H_0+1)}\) is smaller than this margin after
evaluation. Rouch\'e's theorem gives exactly one \(Q\)-root in
each circle. Conjugation shows that real roots stay real; nonreal
roots have imaginary part at least half the separation and stay
nonreal. Thus real-root counts and their ordering are preserved.

If there are \(r\) root gates, put
\[
 C=100(N+1)(H_0+1),\qquad B_i=(r-i+1)C.
\]
Cache one approximation to gate \(i\), with error below
\(2^{-B_i}\). Earlier errors are smaller by a factor \(2^{-C}\),
which pays for coefficient conditioning and the Rouch\'e margin.
Rational isolation of \(Q_i\) then produces the required cache.
Apply its geometric bound using the magnitude and separation
just proved, rather than the artificial expanded denominators of
the cached approximants. Each root gate consequently adds only
polynomially many arithmetic operations. Earlier DAG nodes are
shared, not copied or recursively recomputed. The final gap and
conditioning bounds determine the output sign by PosSLP queries.
Theorem~\ref{push10:adaptive} combines them into one.

Polynomial-value degree is essential to this proof: Sturm
sequences, the multiplicity trials, and the subdivision forest
all pay for the degree. The exponentially large characteristic
polynomial in the general ETR construction is outside its scope.

There is also an exact-witness obstruction. The bounded quadratic
system
\[
 1\le x_i\le2,\qquad x_{i+1}=x_i^2\ (0\le i<m),
 \qquad x_m^2=x_0+1
\]
has a unique solution in this box. Its first coordinate has
minimal polynomial \(X^D-X-1\), \(D=2^{m+1}\), whose
Galois group is \(S_D\) by Selmer and Osada
\cite{Selmer56,Osada87}. In a tower of root adjunctions of
degree at most \(d\), the normal-closure Galois group has no
prime divisor greater than \(d\): each step embeds its new
kernel in a product of subgroups of \(S_d\). A field containing
this first coordinate has normal closure containing its splitting
field. Bertrand's postulate therefore forces \(d>D/2\).
Thus low-degree root gates do not describe every exact bounded
ETR witness.

This obstruction does not apply to approximation. For the same
example put \(u=D(x_0-1)\). Its root lies in \((3/5,1)\)
for \(D\ge4\), and Newton iteration from one on
\[
                  (1+u/D)^D-2-u/D
\]
has error \(e_{t+1}\le2e_t^2\), since its derivative is at
least \(3/4\) and its second derivative is below three there.
Each iteration has a short powering circuit. It yields rational
SLP approximations to all the coordinates to error \(2^{-H}\)
in \(O(\log H)\) iterations, even though the exact-root
description requires a large-degree gate. The example is also
trivially feasible; no decision lower bound is implied.

\subsection{A stochastic-system application}

Etessami, Stewart, and Yannakakis already prove that the strict
rational threshold problem for the least fixed point of a
probabilistic polynomial system is many-one equivalent to
\(\PosSLP\) \cite[Appendix C]{ESY13}.
Their exact Newton iteration has a polynomial warmup followed by
quadratic convergence, and hence reaches error \(2^{-H}\) in
\(\poly(n,\log H)\) arithmetic operations. This is a published
direct route to PosSLP, independent of our norm construction.

We use their proof also for signs of coordinate differences.
After removing coordinates equal to zero or one, the least
fixed point is nonsingular. Adjoining \(y=x_i-x_j\) preserves
isolation, so the same algebraic separation and Newton bounds
give exact difference signs in \(\Pclass^{\PosSLP}\).

\begin{proposition}\label{push10:games}
Rational least-fixed-point thresholds for maxPPSs, minPPSs,
and two-player branching simple stochastic games belong to
\[
             \NP^{\PosSLP}\cap\mathsf{coNP}^{\PosSLP}.
\]
Using our PosSLP bounds, they consequently belong to both
\(\MA^{\PP}\cap\mathsf{coMA}^{\PP}\) and
\(\Sigma_2^{\PP}\cap\Pi_2^{\PP}\).
\end{proposition}

\begin{proof}
For a one-player system, use the polynomial-time qualitative
algorithm to remove zero and one coordinates, leaving
\(0<q^*<1\). Guess a pure policy \(\sigma\), let \(z\) be
the PPS least fixed point after fixing it, and verify
\(0<z<1\) and \(P(z)=z\). In simple normal form, a choice
row is \(\max(x_i,x_j)\) or \(\min(x_i,x_j)\), so policy
consistency uses only the coordinate comparisons just described.

For maxPPS, \(z\le q^*\) by the policy comparison, and
\(q^*\le z\) by leastness. For minPPS, take an optimal
policy \(\tau\), for which \(q_\tau=q^*\). Then
\(P_\tau(z)\ge z\), and the exact quadratic identity gives
\[
 (I-P_\tau'((q^*+z)/2))(z-q^*)=z-P_\tau(z)\le0.
\]
The inverse in \cite[Theorem 4.6(i)]{ESY12} is nonnegative
for every \(0\le z<1\). Thus \(z\le q^*\), and equality
again follows. An optimal policy supplies a valid certificate
for either threshold side.

For two players, optimal pure memoryless policies exist
\cite[Section 5]{ESY12}. To certify a lower threshold, guess
max's policy and a one-player certificate for the resulting
minPPS. For the complementary upper threshold, guess min's
policy and a certificate for the resulting maxPPS. These
existential witnesses concatenate. They do not require an
identity such as \(\NP^{\NP^A}=\NP^A\).

Finally, a fixed complete witness can be checked with bounded
error using our \(\BPP^{\PP}\) PosSLP algorithm, giving the
two Merlin--Arthur inclusions. Substituting
\(\PosSLP\subseteq\Pclass^{\NP^{\PP}}\) gives the displayed
second-level polynomial-hierarchy inclusions.
\end{proof}

The least-fixed-point condition in this certificate is essential.
For example, \(x=\max(x,1/2)\) has every point of \([1/2,1]\)
as a fixed point. Its least fixed point is only \(1/2\).
These stochastic systems provide an application below the general
ETR bound; no reduction of arbitrary ETR to them is established.

\subsection{Boolean oracle computations above ETR}

The randomized ETR bound also has consequences for stronger
Boolean oracle classes, without improving its numerical level.

\begin{corollary}\label{push10:oracle-extensions}
The proposed ETR containment implies
\[
 \BPP^{\ETR}\subseteq\BPP^{\Clevel3},\qquad
 \PP^{\ETR}\subseteq\Clevel4,\qquad
 \PH^{\ETR}\subseteq\Pclass^{\Clevel4}\subseteq\Clevel5.
\]
\end{corollary}
\begin{proof}
Amplify every ETR oracle simulation and take a union bound over
the polynomially many adaptive queries. This proves the first
inclusion. For the second, a PP machine with \(p(n)\) random bits
has acceptance probabilities on the grid \(2^{-p(n)}\).
Shift its decision threshold from \(1/2\) to
\(1/2+2^{-p(n)-1}\), which preserves its language and gives
a margin \(2^{-p(n)-1}\) on both sides, including ties.
Amplify the ETR simulations until their total error is below
\(2^{-p(n)-3}\). The resulting \(\PP^{\Clevel3}\) machine
has the same threshold decision. This is the usual relativizing
lowness of BPP for PP, applied with the existing counting oracle.

Finally, relativized Toda gives
\(\PH^{\ETR}\subseteq\Pclass^{\PP^{\ETR}}\), and the
second inclusion finishes the proof.
\end{proof}

These are Boolean oracle computations. They do not establish a
containment for alternating blocks of real quantifiers, and they
do not assert lowness for every level of a fully relativized
counting hierarchy. In particular, the concrete equality supplied
by the proof is \(\PP^{\Clevel3\oplus\ETR}=\Clevel4\).

\subsection{A two-adic product-flattening lemma}

There is an arithmetic improvement at a fixed small characteristic.
It does not cover all the moduli needed for sign reconstruction.

\begin{proposition}\label{push10:two-adic}
Let \(\mathcal R=\operatorname{GR}(2^a,e)\), with \(a,e\) of
polynomial value and a supplied short ring presentation. Suppose
an exponentially indexed family of units \(u_i\in\mathcal R\)
has raw GapP representatives for its coordinates. If
\(P=\prod_i u_i\) belongs to the base ring \(\Z/2^a\Z\), then
\(P\) also has a raw GapP representative.
\end{proposition}
\begin{proof}
Let \(\sigma\) be the Frobenius automorphism. It is a linear map
over the base ring with polynomial-time computable coefficients.
For a unit \(u\), its norm
\(n(u)=\prod_{j<e}\sigma^j(u)\) is an odd scalar. Hence
\[
 u^{-1}=\left(\prod_{1\le j<e}\sigma^j(u)\right)
           \sum_{k<a}(1-n(u))^k\pmod{2^a}.
\]
This uses polynomially many products and sums of the raw coordinate
counts and therefore retains raw GapP representatives. The degrees
and the number of counting-path copies are polynomial.

Put \(v_i=u_i^2/\sigma(u_i)\). Its coordinates have the same
representation, and \(v_i\equiv1\pmod2\). Thus
\[
 \prod_i v_i=
 \sum_{k<a}\ \sum_{i_1<\cdots<i_k}
            \prod_{j=1}^k(v_{i_j}-1)\pmod{2^a}.
\]
Each surviving subset has polynomial description length, and only
polynomially many counting computations are multiplied in a term.
The entire sum is raw GapP. Finally,
\(\prod_i v_i=P^2/\sigma(P)=P\), because \(P\) is in the
base ring. This proves the claim without extracting a root or
retrieving a canonical counting value.
\end{proof}

The base-ring condition is substantive. A homotopy norm evaluated
at an extension-valued Fourier node usually does not satisfy it.
Also, sign reconstruction uses arbitrary polynomial-bit moduli,
whose prime factors need not have polynomial numerical value.
At an odd prime the analogous Frobenius twist produces
\(P^{p-1}\) and loses the base residue-field component. The lemma
therefore does not lower the general ETR or PosSLP bound.

An independent tool in the same direction is an exact short
Mahler expansion for division by a small power of two. Let
\(m=2^b\) have polynomial value and write
\(\lfloor n/m\rfloor=\sum_k c_k\binom nk\) for \(n\ge0\).
The generating function for the finite-difference coefficients is
\[
 \sum_kc_kz^k=
 \frac{z^m}{(1+z)^m-z^m}=\frac{z^m}{1+2Q(z)},
 \qquad \deg Q=m-1.
\]
Modulo \(2^a\), truncate the geometric series after \(a\) terms.
Only degrees at most \(m+(a-1)(m-1)\) remain. Since
\(\binom{F}{k}\) is GapP for a GapP function \(F\) and
polynomial-value \(k\), this gives raw modular representatives
for the corresponding quotient. Signed inputs can first be
offset by a sufficiently large known power of two. The size
condition is on \(2^b\), not merely on \(b\); it does not
allow arbitrary bit extraction from a count for free.

\subsection{Fixed blocks of real quantifiers}
\label{push10:real-blocks}

There is a separate extension on the real side.  A fixed number of
quantifier blocks is different from a Boolean oracle computation over
$\ETR$: the parameters of a later block may themselves be algebraic
sample points.  Two additional interfaces make this distinction
manageable.  We describe them here; the supplementary block-elimination
notes give the full construction and its specialization audit.

\begin{theorem}\label{push10:rph}
For each fixed $w$, truth of integer polynomial sentences with $w$
blocks of real quantifiers belongs to a fixed level of CH.  Consequently
the discrete, constant-free real polynomial hierarchy belongs to CH.
\end{theorem}

Here the input sentence has polynomial description length; the number
of variables in a block is unrestricted.  The terminology for the
discrete real polynomial hierarchy is that of \cite{HMMP25}.  The
following argument establishes the structural containment;
Section~\ref{sec:real-hierarchy} gives its explicit numerical refinement.

\paragraph{Short addresses for real roots.}
Let $q$ be square-free of degree at most $D$, with positive leading
coefficient.  Its coefficients may be integers of singly exponential
bit length, or bounded rational expressions in a fixed number of
previously selected algebraic roots.  In the latter case assume a CH
oracle for signs of the resulting bounded polynomial expressions in
those earlier roots.  This is stronger than knowing the signs of the
individual input coefficients.
All degree and expression-height bounds are singly exponential, and
the indexed integer coefficient arrays belong uniformly to one fixed
level of CH.

Use the full derivative Thom vector of a real root $\alpha$ and the
integer hashes
\[
 T_{p,a}(\alpha)=\sum_{j=1}^{D}
       (a^j\bmod p)\,\sgn(q^{(j)}(\alpha)),
 \qquad D^3<p<2(D^3+1),\quad 0\le a<p.
\]
Enumerating all integer moduli in this interval is harmless.  At least
one is prime, and the usual polynomial hash argument gives an
injective member: each pair of Thom vectors excludes at most $D$
parameters modulo that prime.

For the constant-depth implementation, use the product filter
\[
 P(y)=\prod_{k=-L}^{L}
 (2^{\max(0,-k)}y+2^{\max(0,k)})^m,
 \qquad R(y)=\frac{P(y)-P(-y)}{P(y)+P(-y)}.
\]
Here $2^{-L}\le |v|\le2^L$ bounds every nonzero derivative or test
value, $P_{\max}=2(D^3+1)$, and
$m=\lceil\log_2(64DP_{\max}+64)\rceil$.  The denominator is positive
on the real line, $R(0)=0$, and the error from the sign of each
nonzero value is below $1/(16DP_{\max})$.  One dyadic scale within a
factor of two supplies a factor at most $1/3$ in $|P(-y)/P(y)|$.
The degree, height, and number of factors are singly exponential,
while the product has constant unbounded-fan-in arithmetic depth.
Thus membership in a hash bucket becomes a strict polynomial sign
at the roots of $q$.  If $h$ is negative exactly on
that bucket and positive at the other real roots, the sign of the
trailing nonzero coefficient of
\[
 \Res_X(q(X),h(X)+u)
\]
is $(-1)$ to its cardinality.  Zeros at nonreal roots occur in
conjugate pairs and do not change this sign.  If $n_t$ are the bucket
sizes and $C$ is the number of odd buckets, then
\[
 r-C=\sum_t\bigl(n_t-(n_t\bmod2)\bigr),
\]
where $r$ is the number of real roots.  Consequently the maximum of
$C$ over the hash family is $r$, and \emph{every} maximizing hash is
injective.  Choose the first such hash.  An occupied bucket is now a
short, coherent name for one real root.

To test $g(\alpha)<0$, keep this same hash and bucket and add a
rational sign filter for $g$ to the bucket predicate.  The corresponding
parity is the desired truth value.  An exponentially indexed family of
test polynomials therefore has a coherent CH-accessible sign profile;
one never guesses the whole profile.  Resultants, polynomial products,
and square-free preprocessing have a constant-depth implementation.
For the last operation, the specialized trailing- and leading-coefficient
selections in \cite[Lemma 6.13]{AGS26} are made using the earlier
sample's sign oracle.  A fixed number of algebraic extensions adds
only a fixed number of CH levels.

The domain conditions matter.  A zero specialized defining polynomial
is discarded; a nonzero constant gives no root addresses.  For rational coordinates $g_i(\alpha)/g_0(\alpha)$,
test $g_0(\alpha)\ne0$ at the addressed root; global coprimality of
$q$ and $g_0$ is unnecessary and can discard valid samples.  Clear
denominators by even powers on this domain.

\paragraph{Uniform block sampling.}
The remaining task is to supply the sample polynomials without a
sequential elimination computation of exponential length.  The block
elimination construction of Basu, Pollack, and Roy
\cite[Algorithms 14.1 and 14.2]{BPR06} provides suitable geometry.
Its algebraic operations require the following replacements.

First, let $Q$ be a nonnegative sampling polynomial with zero set
in the ball of radius $1/c$, where $0<c\le1$.  Choose an even
$d>\deg_X Q+2$ and set
\[
 F=\zeta c^d\left(X_1^d+\sum_{i=2}^k(X_i^d+X_i^2)\right)
     -\zeta(2k-1)+(1-\zeta)Q.
\]
Use
\[
 G_1=dF-\sum_{i=2}^{k}X_i\partial_iF,
 \qquad G_i=\partial_iF\quad(2\le i\le k).
\]
The leading monomials are $X_1^d,X_2^{d-1},\ldots,X_k^{d-1}$,
with a common nonzero scalar coefficient, and every remainder has
degree below $d-1$.  This simultaneous identity gives the required
special basis.  Its homogenizations have no common point at infinity.
The AGS resultant therefore gives its norm polynomial and, by
derivatives with respect to a linear projection, its coordinate
numerators.  The scalar denominator depends only on the deformation
parameters, not on the free real parameters.

Enumerate the projection choices, infinitesimal coefficient indices,
and derivative orders used for limits.  For each fixed parameter
value, a separating choice supplies the required bounded limits.
Repeated limiting roots are handled by the appropriate derivatives
and the rootwise denominator guard.  Other choices can supply extra
ambient real points; these cannot create a sign condition that is
not realized in the fiber.

Second, restricted subresultant elimination can be replaced by
coefficient lists of ordinary resultants.  For nonzero $f,g$ over a
characteristic-zero field and an independent indeterminate $\lambda$,
\[
 \deg\gcd(f,g)=
 \operatorname{ord}_t\Res_X(f(X)+t,g(X)+\lambda t).
\]
At a common root of multiplicities $m,n$, the local order is
$\min(m,n)$; the indeterminate $\lambda$ prevents cancellation when
the two orders coincide.  Include every degree truncation to handle
specialization drops.  Keeping the degrees and these gcd degrees
constant preserves the real root branches and all required signs on
each connected parameter cell.  All of these coefficient lists use
a fixed number of resultant and coefficient-selection layers.
Each admissible root branch and its rational image are continuous on
such a cell.  Their fixed input signs therefore keep the image in one
connected sign component, on which the continuation truth is constant.
This supplies the inductive invariant for subsequent blocks; equality
of disconnected fibers' sign sets alone would not suffice.

Finally, fix the number $w$ of real quantifier blocks.  All the
degree, height, and family-size bounds after $w$ block eliminations
are singly exponential in the input length.  A full sample tuple has
$w$ root variables.  Represent a node by the appropriate coordinate
prefix of a full valid tuple, and match parent and child by coordinate
equality.  Expanding the finite choices of representation indices in
the resulting truth formula still gives singly exponential size.
The formula uses at most $w^2$ root variables and $2w$ infinitesimals.
Effective complete CAD in this fixed dimension bounds all projection
degrees and coefficient heights by singly exponential quantities
$D_*,H_*$.  We use this bound without computing the CAD.

Put $r=2w$, $M=rD_*+2$, and
$B_*=H_*+\lceil r\log_2(D_*+1)\rceil+5$.  Set $\eta=2^{-B_*}$ and
\[
 \epsilon_i=\eta^{(M+1)^{i-1}}\qquad(1\le i\le r).
\]
These constants have short rational SLPs.  On the connected region
$0<\epsilon_1<2\eta$ and
$0<\epsilon_{i+1}<\epsilon_i^M$, each nonleading infinitesimal monomial
is at most $\epsilon_1$ times the leading one.  The degree and height
bounds make the leading term dominate in every nonzero cumulative
projection polynomial in the infinitesimal variables.  The region consequently lies in one open
infinitesimal-base CAD cell.  The formal infinitesimal tuple and the
displayed finite tuple give the same sampled-tree truth there, using
real-closed-field transfer for the formal tuple.
This step precedes hashing and preserves truth and validity; it does
not assert a finite real gap for genuinely infinitesimal values.

At the finite specialization, apply the root-address construction
successively along each full sample tuple.  Each tuple has polynomial
description length, and all its atom signs and rootwise validity
conditions are available in CH.  Coordinate comparison between two
tuples uses at most $2w$ successive root extensions and remains in CH.
Replace each real block by the corresponding Boolean choice of a
full valid tuple with the required coordinate prefix.  Invalid choices
are excluded by conjunction for an existential block and implication
for a universal block.  Duplicate representations do not affect
either quantifier.  The fixed number of Boolean quantifiers still
gives a language in CH, proving Theorem~\ref{push10:rph}.

This construction is intended for each \emph{fixed} number of real
blocks.  It gives no bound independent of $w$, and no improvement
to the numerical $C_4P$ bound for the existential case.  The
supplementary notes separate the root-address lemma, the block
sampling construction, and the finite-specialization argument so
that these additional interfaces can be reviewed independently.

\section{Structural affine depth and restricted integer algorithms}\label{push11:sec}

The next refinements concern the dependence on the number of real
blocks and several restricted arithmetic constructions. They do not
lower the general bounds \(\ETR\subseteq\Clevel4\) and
\(\PosSLP\subseteq\BPP^{\PP}\cap\Pclass^{\NP^{\PP}}\).

\subsection{An affine level for real quantifier blocks}

\begin{theorem}\label{push11:affine}
There are absolute constants \(A,B\) such that, for every fixed
\(w\), truth of integer polynomial sentences with \(w\) real
quantifier blocks belongs to \(\Clevel{Aw+B}\).
\end{theorem}

This structural proof first gives unspecified absolute constants.
The evaluator refinement and numerical ledger in
Section~\ref{sec:real-hierarchy} give the explicit choice
\(A=9,B=17\). The degree of the input-size polynomial can
still depend on \(w\).

\begin{proof}
We strengthen the representation invariant used for root addresses.
Keep every expression as a division-free integer polynomial circuit
in the earlier root variables. Its size, degrees, and coefficient
bit heights are singly exponential, and its unbounded-fan-in
arithmetic depth is \(d\). Suppose its input-bit and connection
families are available in \(\Clevel h\). Let \(E(j,d)\) denote the
level needed to validate a tower of \(j\) addressed real roots and
evaluate the sign of such an expression there. We will show
\[
 E(j,d)\le E(j-1,d+\kappa)+K
\]
for absolute constants \(\kappa,K\). The integer arithmetic
simulation gives \(E(0,d)\le h+ad+c\), with absolute \(a,c\).
The recurrence therefore gives
\begin{equation}\label{push11:eval-depth}
 E(j,d)\le h+ad+(a\kappa+K)j+c.
\end{equation}

First, coefficient extraction adds constant depth above its input
polynomial circuit. Evaluate at the integer interpolation nodes
in parallel, then apply the uniform inverse Vandermonde matrix.
The same applies to the polynomial primitives used in the root
construction, including the resultant coefficient circuits of
\cite[Corollary 5.13]{AGS26}. Their divisions involve only constants.
Convert each \emph{local} constant-depth primitive into integer
pairs, treating its input polynomial expressions as leaves with
denominator one. Do not reconvert the whole circuit at each stage.
For a coefficient array \(n_i/d_i\), use the common positive scalar
\(D=\prod_i d_i^2\), with integer outputs
\(n_i d_i\prod_{k\ne i}d_k^2\). Every coefficient in an array
receives the same scalar. The depth becomes \(d+O(1)\).

Second, the choices in specialized square-free preprocessing can
be supplied as constantly many short indices. For a defining
polynomial \(f(T,Z)\), supply its actual degree at the preceding
root prefix and validate the leading nonzero coefficient and all
higher zeros. Use the two GCD resultants
\[
 R_1=\Res_Z(z(y-Z)+f,f+uf'),\qquad
 R_2=\Res_Z(z+f,f+uf').
\]
Let their trailing \(z\)-coefficients be \(A_1(y,u)\) and
\(A_2(u)\). The quotient is a nonzero scalar multiple of the
monic specialized GCD, by \cite[Lemmas 6.12--6.13]{AGS26}.
If \(e\) is the leading \(u\)-index of \(A_2\), then
\[
 H(Z)=\sum_i [u^e y^i]A_1\,Z^i
\]
is a scalar multiple of that GCD. The same index \(e\) works for
all coefficients. Supply the two trailing orders, \(e\), the
degree of \(H\), and nonzero coefficient witnesses. Their
validity uses only a fixed number of universal short-index tests
over signs at the earlier prefix.

There is no need for a division gate in forming the quotient.
Write \(n=\deg f\), \(e_0=\deg H\), \(b=\operatorname{lc}(H)\),
and \(M=n-e_0\). Set
\[
 F^*(Z)=Z^n f(1/Z),\qquad
 H^*(Z)=Z^{e_0}H(1/Z)=b-L(Z),\quad L(0)=0.
\]
Then
\begin{equation}\label{push11:pseudoquotient}
 Q^*(Z)=[Z^{\le M}]F^*(Z)
           \sum_{r=0}^M b^{M-r}L(Z)^r
       =b^{M+1}\operatorname{rev}(f/H).
\end{equation}
This follows by multiplying the finite geometric sum by \(b-L\)
and truncating. Reverse \(Q^*\) and multiply by its own leading
coefficient. At a valid prefix the result is square-free, has
positive leading coefficient, and has exactly the roots of \(f\).
At other parameter values the displayed expression is still a
well-defined polynomial; the division identity is only needed
at valid prefixes.

Third, a parity test can use one earlier-prefix sign query. Form
the bucket polynomial and its resultant \(F(U)\), as in
Section~\ref{push10:real-blocks}. Before choosing \(U\), obtain
uniform bounds
\[
 \deg F\le N,\qquad
 2^{-L}\le |F_i|\le2^L\quad(F_i\ne0).
\]
Put \(v=2L+\lceil\log_2(N+1)\rceil+3\). The sign of
\(F(2^{-v})\) equals its trailing-coefficient sign: the first
nonzero term dominates all later terms. Equivalently, query the
sign of the integer polynomial expression
\[
 \sum_{i=0}^N F_i2^{v(N-i)}.
\]
This adds constant depth and uses constants of singly exponential
bit height. There is no circular precision choice, since \(L\)
bounds the coefficients before this substitution. Nonreal common
zeros contribute an even initial power of \(U\), exactly as in
the original parity argument.

Finally, supply a maximizing hash in each root address instead of
computing the first one. If \(C_s\) is the number of odd buckets
for hash \(s\), the assertion
\[
 \forall t\quad C_s\ge C_t
\]
adds two counting levels above the bucket-parity predicate.
Every maximizing hash is injective, so an odd bucket is one root.
Use the same hash and bucket for subsequent sign tests. Multiple
valid addresses for a root cause no problem in the guarded sample
tree. Together with the global GCD metadata checks, this proves
the recurrence and \eqref{push11:eval-depth}.

The block-elimination coefficient circuits have depth \(O(w)\).
For rational coordinate substitutions, expand the coefficient
array and homogenize once: \(g_0^E P(g_1/g_0,\ldots,g_k/g_0)\)
is a sum of monomial products. Its depth is the maximum of the
coefficient-array depth and the \(g_i\) depths, plus a constant.
Use an even clearing exponent for signs and retain the rootwise
guard \(g_0\ne0\). Finite specialization of the infinitesimals
adds explicit constants and constant depth. A comparison of full
sample tuples uses at most \(2w\) roots. Equation
\eqref{push11:eval-depth} puts all these predicates at level
\(O(w)\), and the outer \(w\) Boolean quantifiers add \(O(w)\)
levels. This proves the theorem.
\end{proof}

\subsection{Recovering sparse signed digits}

The next promise removes the unknown scale by recovering all occupied
digit positions. A nonadjacent signed binary expansion, or NAF, has
digits in \(\{-1,0,1\}\), with no two adjacent nonzero digits.
Every integer has a unique such finite expansion. Its weight is the
number of nonzero digits; the positions are encoded in binary.

\begin{theorem}\label{push11:sparse}
Given an integer SLP and a polynomial-valued bound \(t\), its complete
NAF can be recovered in deterministic polynomial time under the
promise that its weight is at most \(t\). The digit positions may be
exponentially large and need not be supplied. In particular, promised
positivity belongs to \(\Pclass\).
\end{theorem}

\begin{proof}
Fix a uniform upper bound \(H\ge1\) on the NAF positions, with
\(\log H\) polynomial in the input length. For odd \(\ell\ge3\),
consider signed digit vectors on a cycle of length \(\ell\), again
with no adjacent nonzeros, now including the first and last positions.
They represent the residues modulo \(2^\ell-1\) bijectively.

For injectivity, any two supports have a common zero, since each has
at most \((\ell-1)/2\) positions. Rotate that zero to the highest
position. Each represented integer then has absolute value less
than \(2^\ell/3\), so their difference has absolute value less
than \(2^\ell-1\). Congruence forces equality as integers, and
ordinary NAF uniqueness forces equality of the vectors. The number
of cyclic vectors is \(2^\ell+(-1)^\ell=2^\ell-1\): it is the
trace of the \(\ell\)-th power of the transfer matrix
\(\left(\begin{smallmatrix}1&2\\1&0\end{smallmatrix}\right)\).
Thus injectivity gives bijectivity. To decode a canonical residue
\(r\), compute the ordinary NAFs of \(r\) and
\(r-(2^\ell-1)\). Exactly one has length at most \(\ell\) and
satisfies cyclic nonadjacency. This takes polynomial time in \(\ell\).

Write the unknown expansion as
\(Z=\sum_{i=1}^{s}\epsilon_i2^{e_i}\), with \(s\le t\).
Call an odd prime \(p\) good if no difference \(e_i-e_j\) is
\(0,1\), or \(-1\) modulo \(p\). There are at most
\[
 B=3\binom t2\lceil\log_2(H+1)\rceil
\]
bad primes. Indeed every bad prime divides one of the nonzero
integers \(e_i-e_j\) or \(e_i-e_j\pm1\); nonadjacency makes
these integers nonzero.

Choose \(2B+1\) small odd base primes. More than half are good.
For a good \(p\), the cyclic decoder modulo \(2^p-1\) returns
the original signs and the distinct positions \(e_i\bmod p\).
Choose distinct small odd primes \(q\ne p\) with product greater
than \(H\). The same no-collision and no-adjacency conditions hold
modulo \(pq\). Decode modulo \(2^{pq}-1\), match each term to
its position modulo \(p\), and obtain \(e_i\bmod q\). CRT
recovers every \(e_i\) exactly.

For bad \(p\), either discard an invalid reconstruction or keep
its proposed expansion after checking the weight, position bounds,
signs, and nonadjacency. Every good base prime returns the same
correct sorted list. Its strict majority therefore determines the
output, without an equality oracle to certify the bad-prime lists.
All primes and the products \(pq\) have polynomial numerical
value. The queried Mersenne moduli consequently have polynomial
bit length, and an SLP is evaluated modulo each in polynomial time.
The most significant signed digit determines the sign.
\end{proof}

The total language asserting that an SLP output is positive and
has NAF weight at most \(t\) is in \(\mathsf{coRP}\). Run the
decoder with fixed polynomial bounds, check that its proposed
expansion has the required form and sign, and compare the resulting
short SLP with the input using the usual one-sided modular equality
test. Every true instance gives the correct candidate. On a false
instance a candidate of the prescribed form cannot equal the output.
The same proof recognizes the sparsity condition without the sign
restriction.

For specified SLP coefficients, the modular values have raw
\(\GapP\) representatives instead of polynomial-time residues.
The reconstruction uses a fixed polynomial list of moduli. Pack all
positive and negative branch counts for these representatives into
one \(\SharpP\) value, using polynomially wide fields. One exact
retrieval gives the entire list, after which reconstruction is
ordinary polynomial-time computation. Thus the sparse-coefficient
promise is in \(\Pclass^{\SharpP[1]}\subseteq\Pclass^{\PP}\).
One retrieval of a count is not being identified with one PP
threshold question.

\paragraph{An unknown sparse multiple.}
The same method works when a supplied polynomial-size list of
positive sparse multipliers contains some \(D_i\) for which
\(D_iZ\) has weight at most \(t\). Decode every product to get
a list of candidate fractions \(N_i/D_i\), one of which is right.
Every pairwise cross-difference \(N_iD_j-N_jD_i\) is a sum of
polynomially many signed powers of two; normalization to NAF does
not increase their number. The small Mersenne moduli from the preceding
proof give a deterministic hitting set for every nonzero such
difference. Test \(ZD_i\equiv N_i\) at this entire set. A true
candidate survives, and two unequal candidate fractions cannot
both survive. This determines the value and its sign.

In particular, the multiplier list can be
\(D_p=2^p-1\) for \(1\le p\le P\), with \(P\) polynomially
bounded. This handles an unknown short period even when the
output itself has exponentially many nonzero signed digits.
The promise remains essential; no such sparsity theorem is asserted
for general SLP outputs or for the selector norms in the ETR proof.

\subsection{Recovering an exponentially long geometric block}

The preceding sparse-multiple method allows an unknown short period.
A different invariant handles one geometric block even when its
period is exponentially large.

\begin{theorem}\label{push11:geometric}
Suppose the output of an integer SLP is zero or has the form
\[
 Z=\sigma2^a\sum_{j=0}^{N-1}2^{jd},\qquad
 \sigma\in\{-1,1\},\quad a\ge0,\quad d,N\ge1.
\]
The sign and all these parameters can be recovered in deterministic
polynomial time, without being supplied. For \(N=1\), use the
canonical convention \(d=1\).
\end{theorem}

\begin{proof}
Take a valid height bound \(|Z|<2^H\), with \(H\ge1\) and
\(\log H\) polynomial in the input length. The nonzero canonical
parameters satisfy \(a,d,N\le H\). We first describe the small
robust CRT calculation used below.

For distinct primes \(p_i\), let \(P=\max p_i\) and
\(M=\prod p_i\). Suppose reported residues agree with an integer
\(U\in[-W,W]\) at all but at most \(B\) primes, and
\begin{equation}\label{push11:robust-crt}
 M>8WP^{2B}.
\end{equation}
CRT the reported residues to \(R\in[0,M)\). If \(E\) is the
product of the bad primes, then \(E\le P^B\) and
\(E(R-U)=kM\). After reducing \(k/E=k'/e'\),
\[
 \left|\frac RM-\frac{k'}{e'}\right|
      =\frac{|U|}{M}<\frac1{2e'^2}.
\]
Thus \(k'/e'\) occurs among the continued-fraction convergents of
\(R/M\), including the exact final convergent if \(U=0\).
Include endpoints \(0/1,1/1\). The fraction is in \([0,1]\),
since leaving that interval would require \(|U|\ge M/E>W\).
For each convergent, test whether \(R-k'M/e'\) is an integer in
\([-W,W]\) agreeing with all but \(B\) reported residues.
The true integer passes. Two survivors agree modulo a product at
least \(M/P^{2B}>2W\), and hence coincide. This recovers \(U\)
using only ordinary arithmetic on the explicit CRT integers.

Set \(b=\lceil\log_2(H+1)\rceil\), \(B=4b+1\), and
\(W=H^2\). Enumerate primes at least five until
\eqref{push11:robust-crt} holds. Their number and numerical values
are polynomial in \(b\): the \(r\)-th prime is
\(O(r\log(r+2))\), whereas \(\log M\ge r\log5\).
All the CRT integers consequently have polynomial bit length.

Query the canonical residues \(R_p=Z\bmod(2^p-1)\). If
\(p\nmid d\), the exponent positions of a full block of \(p\)
terms traverse all residues modulo \(p\), contributing
\(2^p-1\). If \(b_p=N\bmod p\), a positive block therefore
reduces to the binary indicator of
\[
 \{a,a+d,\ldots,a+(b_p-1)d\}\pmod p.
\]
These positions are distinct. Negation complements this word when
\(b_p>0\); if \(b_p=0\), both signs give residue zero. Hence
\begin{equation}\label{push11:popcount}
 \operatorname{popcount}(R_p)\equiv\sigma N\pmod p.
\end{equation}
Only primes dividing \(d\) can be wrong, and there are at most
\(b\). Robust CRT recovers \(\sigma N\), giving the sign and
number of terms. Zero gives the empty block.

For a nonzero output, change the cached residues to those of
\(\sigma Z\), so their words represent the positive block.
Let \(S_1,S_2\) be the sums of the occupied positions and their
squares. For \(N\ge2\), at primes not dividing \(d\),
\[
 b_pS_2-S_1^2
       \equiv \frac{d^2b_p^2(b_p^2-1)}{12}\pmod p.
\]
Thus these moments give \(d^2\bmod p\) whenever
\(p\nmid dN(N-1)(N+1)\). Assign arbitrary residues when the
displayed denominator vanishes. At most \(4b\) primes are bad,
so robust CRT recovers \(d^2\le H^2\); an ordinary integer square
root gives \(d\). Then
\[
 a\equiv S_1/b_p-d(b_p-1)/2\pmod p
\]
outside primes dividing \(dN\), and another robust CRT calculation
gives \(a\). For \(N=1\), every positive residue is a single
rotated bit, whose position is \(a\bmod p\); recover \(a\)
directly and set \(d=1\).

Each queried Mersenne modulus has polynomial bit length. Bit scans,
moments, square roots, and continued fractions here also concern
polynomial-bit integers, never the full binary output \(Z\).
This proves the running-time claim.
\end{proof}

This supplies a further promise family with inputs outside the earlier
polynomial-NAF-weight and polynomial-value period-multiplier promises. Take
\(N=d=2^n\), \(a=0\). Binary doubling constructs the geometric
sum with \(O(n)\) SLP gates, but its NAF has \(N\) nonzero
digits. For every \(2\le p\le d-2\),
\[
 (2^p-1)Z=\sum_{j=0}^{N-1}
              \bigl(2^{jd+p}-2^{jd}\bigr)
\]
is already nonadjacent and has \(2N\) terms; for \(p=1\) the
weight is \(N\). No polynomial-valued \(p\) makes this family
sparse. Both signs are handled, and the input circuit need not
display the geometric pattern syntactically.

As with sparse digits, total recognition of the geometric-block
language is in \(\mathsf{coRP}\): construct the recovered candidate
by binary doubling and use the one-sided SLP equality test. The
deterministic theorem itself depends on the promise. It does not
extend to arbitrary sums of geometric blocks, where carries can
destroy \eqref{push11:popcount}.

\subsection{Approximate squares with exact integral identities}

The integral four-square theorem still leaves the size of its numerator
circuits open. There is a short approximation construction that controls
both real and dyadic errors. Its exact identities also clarify the
remaining obstacle.

Retain the bounded critical objective and put \(h=J_w-T\). Define
\[
 a_0=X_1^3+X_1+1,\quad a=a_0+2M,\quad
 \Delta=4h-3a^2.
\]
The identity
\[
 a_0^2-\bigl(X_1^6+3(X_1+1)^2\bigr)
   =2\bigl(X_1^4+X_1^3-(X_1+1)^2\bigr)
\]
shows that \(K_0=(a^2-h)/2\) is an ordinary integer polynomial
with a short circuit. Thus \(\Delta=a^2-8K_0\). Modulo the
critical ideal and two, \(a^2=h=X_1+1\), a unit.
Choose \(M=B^3+B+1\), where \(B=2s+3\) is the real critical-box
bound. Then \(1\le a\le3M\) there. Increasing \(Q\) by a known
power of two of polynomial numerical value makes \(h\ge9M^2\)
on infeasible instances, so \(h\le\Delta\le4h\) there. The
earlier separation and real-box arguments remain valid.

\begin{proposition}\label{push11:approx-squares}
For binary precision parameters \(b,n\), polynomial-time construction
gives ordinary integer-polynomial SLPs \(N,D,K,V\) satisfying
\begin{align*}
 D^2-\Delta N^2&=8K,& \Delta N-aD&=4V,\\
 hD^2&=(aD+V)^2+3V^2+2\Delta K.
\end{align*}
The denominator \(D\) is nonzero at every complex critical point.
Writing \(P_0=(aD+V)/D\) and \(P_1=P_2=P_3=V/D\), the residual
\(E=h-\sum_iP_i^2\) belongs to \(2^n(R\otimes\Z_2)\).
On infeasible instances it also satisfies \(0\le E\le2^{-b}h\)
at every real critical point.
\end{proposition}

\begin{proof}
Start with \((N,D,K,V)=(1,a,K_0,-2K_0)\). For \(s=1\) or
\(s=2^k+1\), \(k\ge2\), set \(u=(s^2-1)/8\), \(v=(s-1)/4\).
These are explicit integral constants with short powering circuits.
Perform the simultaneous updates
\begin{align*}
 L&=s^2K-uD^2,&N'&=sND,\\
 D'&=D^2-4L,&K'&=2L^2,\\
 V'&=vaD^2+sDV+aL.
\end{align*}
Expansion preserves the two invariants and the four-square identity.
Modulo two, \(D'=D^2\); its initial value \(a\) is a unit.
Its multiplication determinant is therefore odd, so \(D\) is
invertible in the rational critical algebra and has no complex zero.

At a real point on an infeasible instance set
\(y=\sqrt\Delta N/D\), only for the analysis. The update is
\[
 D'=(D^2+s^2\Delta N^2)/2>0,\qquad
 y'=\frac{2sy}{1+s^2y^2}.
\]
A first unit-scale fold puts \(y\) in \([2^{-H},1]\), for a
uniform \(H\) of polynomial binary length. From \([2^{-j},1]\),
choose \(k=\lceil j/2\rceil\) and \(s=2^k+1\); the image lies
in \([2^{-(k+1)},1]\). After \(O(\log H)\) folds it lies in
\([1/16,1]\). At unit scale, \(e=1-y\) satisfies
\(e'=e^2/(1+y^2)\le e^2\). A further \(r\) folds with
\(2^r\ge\max\{32(b+2),n\}\) give the real error bound.
For the same folds, \(K'=2K^2\), so the cleared residual
\(2\Delta K\) is divisible coefficientwise by \(2^{2^r}\).
Since \(D\) is a dyadic unit, this proves the second precision
claim. The number of updates is polynomial in the input length,
\(\log(b+2)\), and \(\log(n+2)\).
\end{proof}

The exactness gap is structural. Three coordinates coincide; an
exact representation in this form would require
\((\Delta N/D)^2=\Delta\). For \(h=5,a=1\), the discriminant
is \(17\), a nonsquare. At unit scales the nonzero residual is
\[
 E_t=17\,2^{2^{t+1}}/D_t^2\qquad(t\ge1).
\]
Its real value tends to zero and its dyadic valuation tends to
infinity, while its cleared numerator grows. Neither convergence
statement implies an exact certificate.

\paragraph{Free quaternion coordinates.}
Let \(q=(P_0,P_1,P_2,P_3)\) and \(n_q=\sum_iP_i^2\). At dyadic
precision four, \(n_q/h\equiv1\pmod{16}\). Thus \(n_q\) is a
unit in the rational algebra, and is strictly positive at every
real point on every instance. A purely orthogonal correction
\(q\mapsto q(1+u_1i+u_2j+u_3k)\) would require
\((h-n_q)/n_q\) to be three squares. This can fail locally. For
\(h=7,a=1\), two folds give
\[
 q=(154,57,57,57)/97,\qquad
 (h-n_q)/n_q=180^2/33463.
\]
Since \(33463\equiv7\pmod8\), this is not three rational squares.

A fixed radial change removes that obstruction. Set
\(\rho=h/n_q\) and \(r=\rho-25/81\). If the real relative
error is at most \(\epsilon\le1/2\), then on infeasible instances
\[
 \frac{56}{81}\le r\le\frac{56}{81}
        +\frac{\epsilon}{1-\epsilon},\qquad
 r\equiv\frac{56}{81}\pmod{16(R\otimes\Z_2)}.
\]
Thus \(r\) has dyadic valuation three at every dyadic place.
It is three rational squares in each number-field factor:
positivity handles real places; at odd finite places the
three-square form is isotropic; and at dyadic places the form
\(\langle1,1,1,-r\rangle\) has nonsquare determinant and is
therefore isotropic. Hasse--Minkowski gives
\(r=u_1^2+u_2^2+u_3^2\). Multiplication of \(q\) by the
quaternion \((5/9,u_1,u_2,u_3)\) changes its norm exactly to \(h\).

With \(S=(aD+V)^2+3V^2\), the corresponding certificate verifier
checks
\[
 c^2(81hD^2-25S)=81S\sum_{i=1}^3 U_i^2\pmod I.
\]
The known \(D,S\) are units, and \(S>0\) at every real point.
Short descriptions of \(c,U_1,U_2,U_3\) for this normalized
ternary problem would give \(\ETR\subseteq\mathsf{coMA}^{\PP}\),
using the established identity verifier. Their existence with
polynomial description length remains unproved.

\subsection{A limit of a single modular counting gap}

There is an exact obstruction to one direct flattening proposal.
The model hides an integer \(Z\) and supplies only its canonical
residues modulo polynomial-bit integers. A counting path makes
polynomially many such queries, possibly adaptively, and contributes
\(+1\) or \(-1\). Fix a nondeterministic string and its full
residue-answer transcript. The compatible integers form either
the empty set or one residue class, by generalized CRT. Thus the
gap has the form
\[
 g(Z)=\sum_{\tau\in\mathcal T}\epsilon_\tau
                 1_{Z\equiv a_\tau\pmod{d_\tau}},
 \qquad |\mathcal T|\le2^{\poly(n)}.
\]
Any one residue class has counts differing by at most one in
the two intervals \([1,M]\) and \([-M,-1]\). Consequently
\[
 \left|\sum_{z=1}^M g(z)-\sum_{z=-M}^{-1}g(z)\right|
       \le|\mathcal T|.
\]
If \(g(z)>0\) exactly when \(z>0\), integrality makes the
left side at least \(M\). One such gap therefore cannot decide
all signs over the doubly exponential range of possible hidden
integers: that would require \(M\le2^{\poly(n)}\).

This statement concerns the specified residue-only model. It does
not give a lower bound against an algorithm inspecting the SLP,
against multiple PP queries, or against nonlinear processing of
a retrieved count. It explains why simply summing all modular
branches and taking the resulting gap's sign cannot provide a
universal sign decoder. The randomized phase algorithm and the
restricted reconstruction theorems use operations outside that
proposal.

\section{Further root operations and certificate targets}
\label{push12:sec}

The unrestricted containments remain those of
Theorems~\ref{thm:posslp} and~\ref{thm:depth}. We can nevertheless
enlarge the class of algebraic computations handled by the integer
argument. The degrees of the new root operations have binary
encodings, and may be exponentially large. A separate certificate
refinement removes an unnecessary constant from the exact norm target.

\subsection{Binary-degree radicals}
\label{push12:binary-roots}

Consider arithmetic circuits with rational constants and operations
\(+, -,\times,\div\), together with \(N\)-th root gates whose positive
integer degrees \(N\) are written in binary. An even-degree gate takes
the nonnegative real root; an odd-degree gate takes the unique real
root. Negative even radicands and zero divisors are invalid.

\begin{theorem}\label{push12:binary-root-theorem}
Sign and validity for these circuits belong to \(\Pclass^{\PosSLP}\).
By the adaptive-query compilation theorem, they reduce many-one to
\(\PosSLP\), and hence belong to
\(\BPP^{\PP}\cap\Pclass^{\NP^{\PP}}\).
\end{theorem}

The point is to avoid a Newton transient proportional to \(N\).
Suppose first that \(a>0\) is rational, given by an SLP, and
\(2^{-L}\le a\le2^L\). Binary search with \(\PosSLP\) comparisons
finds the integer \(k\) such that
\[
 2^{Nk}\le a<2^{N(k+1)}.
\]
This takes \(O(\log L)\) comparisons. The powers have short rational
SLPs because \(Nk\) has \(O(\log N+\log L)\) bits. Writing
\(r=a^{1/N}\), bisect \([2^k,2^{k+1}]\) another
\(\lceil\log_2N\rceil+4\) times. Its upper endpoint \(u\) then has
relative error
\[
 0\le\varepsilon=u/r-1\le\frac1{16N}.
\]
Every comparison uses \(O(\log N)\) multiplication gates.

For the exact Newton step
\[
 u'=\frac{(N-1)u+a/u^{N-1}}N,
\]
weighted AM--GM and Taylor's theorem give
\begin{equation}\label{push12:nth-newton}
 0\le\frac{u'}r-1
 =\frac{(N-1)(1+\varepsilon)+(1+\varepsilon)^{-(N-1)}}N-1
 \le\frac{N-1}{2}\varepsilon^2.
\end{equation}
Indeed the second derivative of \((1+t)^{-(N-1)}\) is at most
\(N(N-1)\) for \(t\ge0\). Thus \(\eta=N\varepsilon/2\) satisfies
\(\eta'\le\eta^2\), starting below \(1/32\). After \(t\) steps,
\[
 0\le u-r\le2^L2^{-5\cdot2^t}.
\]
Only \(O(\log(L+H+1))\) steps are needed for error below \(2^{-H}\).
Neither \(N\) nor \(H\) needs to have polynomial value. The cases
\(N=1\), \(a=0\), and negative odd radicands are immediate.

For nesting, fix bounds from the original circuit before choosing any
approximation precision. If \(s\) is its number of gates and
\(b=\sum_i\lceil\log_2N_i\rceil\), every gate belongs to a field of
degree at most \(2^b\). Absolute logarithmic height satisfies
\[
 h(x+y)\le h(x)+h(y)+1,\quad
 h(xy),h(x/y)\le h(x)+h(y),\quad
 h(\sqrt[N]{x})=h(x)/N.
\]
Consequently a computable \(L_0=2^{\poly(n)}\), with polynomial
binary length, bounds every nonzero defined gate by
\(2^{-L_0}\le |z|\le2^{L_0}\).
Let \(M=2^{L_0}\). Once cached approximations have error below
\(1/(4M)\), comparison with \(\pm1/(2M)\) identifies exact zero
and strict sign. Zero radicands are therefore returned exactly.
On a nonzero radicand's neighborhood, the derivative of its chosen
real root has magnitude at most \(2M\), independently of \(N\).

Ordinary arithmetic amplifies an error \(e\) by at most \(4M^3\).
Using local root error \(2^{-B}\), with
\[
 B=s(4L_0+5)+L_0+20,
\]
keeps the total error below \(1/(8M)\) at every gate. The preceding
domain tests and the final positivity comparison are therefore valid.
All approximation SLPs share the cached earlier gates. Their sizes and
the number of oracle calls are polynomial, proving the first assertion.
The many-one conclusion uses the earlier compilation theorem separately.

\subsection{Inverting a polynomial given by a monotone circuit}
\label{push12:monotone}

The same estimate applies to a larger family. An inverse gate supplies
a univariate polynomial \(F\) by a circuit using addition and
multiplication, \(X\), and nonnegative coefficient leaves. The
coefficient leaves and a target \(b\) may be values of earlier gates.
The gate returns the unique \(r\ge0\) satisfying \(F(r)=b\).
It is valid when \(F\) is nonconstant and \(b\ge F(0)\).
The entire internal circuit is counted in the input length; its
expanded coefficient list is not part of the input.

\begin{theorem}\label{push12:monotone-theorem}
Sign and validity for arithmetic circuits with polynomially many nested
inverse gates of this form reduce many-one to \(\PosSLP\).
\end{theorem}

For \(F(X)=\sum_{j=0}^D a_jX^j\), \(a_j\ge0\), one has
\[
 F'(x)>0,\qquad F''(x)\ge0,\qquad
 xF''(x)\le(D-1)F'(x)\quad(x>0).
\]
If \(F(r)=b\), \(r>0\), and \(u\ge r\), then
\[
 u'=u-\frac{F(u)-b}{F'(u)}
\]
lies in \([r,u]\). Moreover,
\begin{align*}
 u'-r
 &=\frac1{F'(u)}\int_r^u\bigl(F'(u)-F'(v)\bigr)\,dv\\
 &\le\frac{D-1}{2r}(u-r)^2.
\end{align*}
Thus the relative error obeys the same bound as
\eqref{push12:nth-newton}, with \(D\) in place of \(N\).
After an exponent search and \(O(\log D)\) bisections, Newton's
method converges quadratically. Automatic differentiation supplies
\(F'\) with a polynomial-size circuit, even when \(F\) has
exponentially many terms. Nonconstancy is tested by \(F'(1)>0\).

We give the stability argument needed for nested gates. It is not enough
to bound the expanded heights of the eventual rational approximants.
An internal circuit of size \(S\) has degree at most \(2^S\).
If its coefficient leaves have height at most \(T\), ordinary
coefficientwise sum and product bounds give
\[
 \max_j h(a_j)\le(T+1)2^{3(S+1)^2}.
\]
For multiplication, each coefficient has at most \(2^S+1\)
summands, which proves this deliberately coarse recurrence.
The algebraic root-height inequality then gives
\[
 h(r)\le h([a_0:\cdots:a_D])+\log_2(D+1)
       \le(T+1)2^{5(S+2)^2},
\]
after including the target. The products of field degrees and all
height bounds remain singly exponential through the full circuit.

Augment the semantic description by the arithmetic computations of
\(F(0)\), \(F'(1)\), \(b-F(0)\), and \(F'(r)\). Its size is
still polynomial. Fix \(L=2^{\poly(n)}\) such that every nonzero
augmented gate lies in magnitude between \(2^{-L}\) and \(2^L\).
In particular this bounds \(F'(r)\) below at each positive inverse
root. This derivative value is used to prove a bound; it is not an
oracle required before computing the root.

Put \(M=2^L\). The circuit magnitude and derivative recurrence supplies
a computable \(G=2^{\poly(n)}\) such that
\[
 \sum_j|\partial_{c_j}F(x)|+1\le2^G
\]
when \(|x|\le2M\) and each coefficient leaf \(c_j\) has magnitude
at most \(2M\). These are derivatives with respect to the supplied
leaves, so no coefficient expansion is needed. Exact zero leaves are
identified from the semantic gap and replaced by zero; positive leaves
retain their signs. The perturbed addition-and-multiplication circuit
then has the same coefficient support as the true polynomial.

Suppose its leaves and target have error at most \(e\). For
\(0<t\le r/(4D)\), the nonnegative derivative coefficients and
Bernoulli's inequality imply
\[
 F'(x)\ge\tfrac12F'(r)\quad(r-t\le x\le r).
\]
Hence the two values \(F(r\pm t)-b\) have opposite signs and
absolute values at least \(t2^{-L}/2\). Their perturbations have
absolute value at most \(e2^G\). Taking
\[
 t=4e2^{G+L},\qquad
 e\le2^{-G-2L-\lceil\log_2D\rceil-5}
\]
preserves both endpoint signs. The perturbed inverse root therefore
exists, lies in \([1/(2M),2M]\), and differs from \(r\) by at most
\(4e2^{G+L}\). A true zero inverse is recognized through
\(b-F(0)=0\) and returned exactly.

Choose a common error amplification bound
\(K=2^{G+3L+10}\), and local root precision \(2^{-B}\), where
\[
 B=s(G+3L+11)+G+2L+n+30.
\]
The accumulated error after \(j\le s\) gates is at most
\((1+K+\cdots+K^{j-1})2^{-B}\). It is below both the displayed
stability requirement and \(1/(8M)\). This justifies the domain
tests, all perturbed inverse calls, and the final sign test by induction.
The root algorithms depend polynomially on the binary length of \(B\).
Sharing earlier approximations gives a polynomial-size rational DAG,
computed in \(\Pclass^{\PosSLP}\). Compilation proves the theorem.

These gates do not solve an arbitrary coupled system of monotone
equations. They invert one variable at a time in a supplied circuit.
No reduction of unrestricted ETR to this representation is known here.

\subsection{Sparse real-root gates of binary degree}
\label{sec:sparse-root-gates}

The root-operation model can be enlarged beyond polynomial-degree
polynomials.  A \emph{sparse real-root gate} specifies
\[
 p(X)=\sum_{j=1}^{k}a_jX^{e_j}
\]
and the index of a distinct real root, in increasing order.  Exponents
are nonnegative integers written in binary, and coefficients are values
of earlier arithmetic or root gates.  Thus the degree can be exponential
in the circuit description, while the displayed number of terms is
polynomial.  A zero polynomial, a nonexistent root index, or an ordinary
division by zero makes the circuit invalid.

\begin{theorem}\label{thm:sparse-root-gates}
Validity and output sign for polynomial-size circuits with nested sparse
real-root gates are in $\Pclass^{\PosSLP}$.  Consequently, using the
adaptive $\PosSLP$ compilation established above, these problems reduce
in polynomial time to $\PosSLP$.
\end{theorem}

Multiple roots are permitted.  The theorem does not require an algorithm
that expands a polynomial of exponential degree, computes its dense
square-free part, or writes a high-precision root in binary.  It also
does not apply to arbitrary short circuits for univariate polynomials:
such a circuit need not describe a polynomial with few terms.

We use the sparse isolation algorithms
of~\cite{Sagraloff14,JindalSagraloff17}.  There are two adjustments
to make before using their arithmetic bounds with succinct numbers.
First, rounding a Newton center to a grid with exponentially many bits
cannot be treated as an elementary operation.  Second, the precision
bound for nested gates must depend on the magnitudes of their
coefficients, rather than the bit heights of previously constructed
rational approximants.  We give these adjustments explicitly.

\paragraph{Sparse weak covers.}
After collecting equal exponents and deleting zero coefficients, divide
out the lowest monomial.  Differentiating and again dividing out the
lowest monomial removes one term.  The resulting fractional-derivative
chain has length at most $k$.  On the positive axis, the union of its
root sets has fewer than $k^2$ elements, by Descartes' rule.  Accordingly,
among $k^2+1$ distinct positive grid points, one is a nonroot of every
member of this chain.  It can be selected by exact sign comparisons;
one may maximize the minimum absolute value of the chain polynomials.
All newly selected endpoints below use this entire original chain.

The weak-cover algorithm of~\cite[Algorithm~1]{JindalSagraloff17}
recursively retains small intervals covering derivative roots and refines
the intervening sign-change intervals.  It makes polynomially many
refinement calls and returns polynomially many disjoint intervals of
width at most $2^{-L}$, covering all real roots in its domain.  Some
intervals may contain no root.  For a polynomial with simple real roots
separated by more than $2^{-L+2}$, retaining exactly those intervals
with opposite endpoint signs gives genuine isolating intervals.
Root-free boundary sentinels are discarded.

If the nonzero coefficient magnitudes lie in $[2^{-\tau},2^\tau]$,
the refinement bound is polynomial in
$k,\log(n+1),\log(\tau+L+1)$, where $n$ is the degree
\cite[Theorem~11]{JindalSagraloff17}.  Here $\tau$ is an archimedean
magnitude bound, not the bit height of a rational representation.
The fractional-derivative chain satisfies the same type of bound with
$T=\tau+k\lceil\log_2(n+1)\rceil$.
Short powers of two give positive endpoints below and above all its
positive roots by the Cauchy bounds.  Alternatively, dilation $X=SY$
by a power of two above the upper root bound places the search in the
source's unit interval.  Dilation preserves sparsity and increases the
magnitude exponent by at most $n\log_2S$, which has short binary length.
The negative roots are treated by replacing $X$ with $-X$.

\paragraph{Removing the large floor operation.}
Let $I=(a,b)$ contain exactly one simple real root, put $w=b-a$, and
maintain $N=2^{2^\ell}$ with $\ell\geq1$.  At the three quarter points
$x_i=a+iw/4$, evaluate $v_i=p(x_i)/p'(x_i)$ whenever the denominator
is nonzero.  For each of the at most three pairs with $v_i\ne v_j$, use
the exact rational center
\[
 \lambda=x_i-\frac{(x_j-x_i)v_i}{v_j-v_i}.
\]
Set $r=w/(4N)$ and $\delta=r/16$.  Clip
$[\lambda-r,\lambda+r]$ to $[a,b]$ and reject it if its width is less
than $3r/4$.  An endpoint within $r/2$ of a parent endpoint is snapped
to that parent endpoint.  Otherwise perturb it by at most $\delta$,
using the nonroot grid just described.  Accept the resulting interval
$I'$ if its endpoint values have opposite signs.  Directly from the
clipping, snapping, and perturbation rules,
\begin{equation}\label{eq:sparse-width}
 \frac{5r}{8}\leq w(I')\leq\frac{25r}{8},\qquad
 \frac{w}{8N}\leq w(I')\leq\frac{w}{N}.
\end{equation}
Every changed endpoint moves at least $7r/16\geq w(I')/8$.

For completeness, the center has the cluster-success property needed
in the sparse refinement analysis.  Suppose a subinterval $J$ of width
at most $w/(8192N)$ has a one-circle disk containing $c$ roots, counted
with multiplicity, and there are no additional roots within distance
$1024nNw$ of the midpoint of $I$.  Write $m$ for the midpoint of $J$.
Two quarter points are separated by at least $w/4$ and each has distance
at least $w/8$ from $m$.  Logarithmic differentiation gives, at either,
\[
 \frac{(x-m)p'(x)}{c p(x)}=1+e_x,
 \qquad |e_x|<\frac1{256N}.
\]
Indeed, the clustered roots contribute at most $9/(16384N)$ to the
error and the remaining roots contribute at most $1/(512cN)$.
Writing $v_x=(x-m)/c+d_x$ gives $|d_x|\leq w/(255cN)$,
$|v_j-v_i|\geq w/(5c)$, and hence
\[
 |\lambda-m|\leq\frac{10w}{255N}<\frac r6.
\]
Clipping therefore leaves width at least $5r/6$, and the perturbed
interval retains $J$.  Its sign test succeeds.  In particular, no
enumeration of the possibly exponential cluster multiplicity is used.

Boundary tests use points at distance $3w/(4N)$ from either endpoint,
perturbed by at most $w/(16N)$.  They succeed when the root is within
$w/(4N)$ of that endpoint.  Otherwise split within $w/16$ of the
midpoint.  A successful accelerated step increments $\ell$; a linear
step replaces it by $\max(1,\ell-1)$.  The invariant
\[
 \log_2N\leq2\log_2(w_0/w)+2
\]
holds throughout, where $w_0$ is the initial width.  Before termination,
$w\geq2^{-L}$; if $w_0\leq2^\Gamma$, then
\[
 \ell\leq\log_2\bigl(2(\Gamma+L)+2\bigr).
\]
Thus $N$ always has a short repeated-squaring description.  If a quarter
point is exactly the root, return a small interval around it with
nonroot endpoints, rather than a singleton later lost by sign filtering.

These are the geometric invariants used in the sparse refinement
analysis~\cite[Theorem~7]{Sagraloff14}: accelerated width reduction,
fixed-fraction endpoint movement, boundary success, and cluster success.
The initial endpoint/Obreshkoff-lens argument takes
$O(\log n+\log(T+L))$ steps.  Thereafter the one-circle regions lose
roots at most $k-1$ times, and a chain with unchanged root set has that
same length bound.  Changing the endpoint-movement fraction from $1/4$
to $1/8$ only changes absolute constants.  Each sparse evaluation uses
$O(k\log(n+1))$ arithmetic operations.  The resulting weak cover
therefore uses polynomially many rational-SLP operations and
$\PosSLP$ comparisons, without a large floor or endpoint bit expansion.

\paragraph{Splitting multiple roots without a dense GCD.}
Suppose $p$ has $k\geq2$ retained terms and a nonzero constant term.
Choose $H\geq1$ such that its nonzero coefficients and root moduli lie
in $[2^{-H},2^H]$, and distinct complex roots are separated by
$\sigma=2^{-H}$.  Put $s=k-1$ and
\[
 q_\epsilon=(1+\epsilon D_X)^s p.
\]
This polynomial has at most $k^2$ terms.  Every root of $p$ has
multiplicity at most $k-1$: otherwise the first $k$ derivative equations
at a nonzero root give an invertible falling-factorial Vandermonde
system in the distinct exponents.

At a root of multiplicity $m$, the local model is
\[
 L_{m,s}(z)=(1+D_z)^s z^m
 =\sum_{j=0}^m\binom{s}{j}(m)_jz^{m-j}.
\]
For $s\geq m$ it has $m$ distinct negative real roots in $[-sm,0]$.
To see this, for a real-rooted polynomial $P$ the function
$(P+P')/P=1+\sum_i m_i/(z-r_i)$ decreases strictly in every gap.
It adds one simple root per gap and one to the left, and lowers each
old multiplicity by one.  Its leftmost new root moves by at most
$\deg P$.  Iteration starting from $z^m$ proves both assertions.

The following constants make the perturbation quantitative.  For desired
output precision $B\geq1$, set
\begin{align*}
 R&=2(k+1)^2,\\
 A&=4(k+1)^2\lceil\log_2(2k+2)\rceil,\qquad \eta=2^{-A},\\
 C&=k(A+3)+(k+1)\lceil\log_2(n+1)\rceil\\
  &\hspace{2em}+(k+1)\lceil\log_2R\rceil+20,\\
 E&=B+4H+C,\qquad \epsilon=2^{-E}.
\end{align*}
The nonzero integer discriminant and the root bound for $L_{m,s}$
give model-root separation at least $\eta$.  On a circle of radius
$\eta/8$ about a model root, $|L_{m,s}|\geq(\eta/8)^k$.
Writing
\[
 p(\alpha+\epsilon z)=c\epsilon^mz^m
       \prod_{\beta\ne\alpha}
       \left(1+\frac{\epsilon z}{\alpha-\beta}\right)
\]
with repeated roots included, the normalized local-model error on
$|z|\leq R$ is at most
\[
 R^m(n+1)^s\bigl((1+\epsilon R/\sigma)^n-1\bigr)
 \leq2n(n+1)^sR^{m+1}\epsilon/\sigma
 <\tfrac12(\eta/8)^k.
\]
Rouch\'e's theorem gives one $q_\epsilon$-root near each local-model
root.  The resulting $n$ disjoint disks account for the whole degree:
$q_\epsilon$ is square-free, its separation is at least
$\epsilon\eta/2$, and each root moves less than $\epsilon R$.
For real $\alpha$, conjugation fixes each small disk and thus its unique
root.  For nonreal $\alpha$, $|\operatorname{Im}\alpha|\geq\sigma/2$,
so no real root is created.  Each real root of multiplicity $m$ has
therefore become exactly $m$ simple real roots.

There is also a necessary coefficient bound.  At exponent $e$, the
coefficient of $q_\epsilon$ is
$\sum_j\binom{s}{j}(e+j)_j a_{e+j}\epsilon^j$.
The first nonzero summand is unique and dominates the remaining terms
by a factor greater than two.  Every retained coefficient consequently
has magnitude at least $2^{-H-1-sE}$ and at most $k2^H(n+1)^s$.
Thus the weak-cover magnitude parameter is linear in $E$, even when
the rational representation of a coefficient is much larger.

\paragraph{Rational approximation and root indices.}
Approximate the retained coefficients of $p$ to error $2^{-P}$ and
apply the same differential operator to obtain a rational-SLP polynomial
$\widetilde q$.  It suffices to take
\[
 P=H+n(H+E+A+6)+(k-1)\lceil\log_2(n+1)\rceil
                          +3\lceil\log_2(k+1)\rceil+10.
\]
Each coefficient changes by at most $k(n+1)^s2^{-P}$.
On circles of radius $\nu=\epsilon\eta/16$ about the roots of
$q_\epsilon$, its magnitude is at least $2^{-H}\nu^n$.
The displayed precision makes the coefficient perturbation smaller
than this bound throughout those circles, where $|X|\leq2^{H+2}$.
Rouch\'e's theorem again preserves all roots, their reality, and their
order.  The separation is at least $\epsilon\eta/4$.

Run the weak cover on $\widetilde q$ with $L=E+A+12$ and retain
sign-change intervals.  Midpoints belonging to one original root have
diameter less than $\sigma/8$, whereas distinct original-root clusters
are separated by more than $\sigma/2$.  Grouping consecutive midpoints
at gap $\sigma/4$ recovers the distinct-root order and multiplicities.
Any midpoint in the selected group approximates that root within
$2^{-B}$.  The removed monomial contributes zero once to the distinct
root list, with its binary multiplicity recorded separately.

\paragraph{A shared precision schedule for nested gates.}
Let $M$ be the input description length and $r$ its number of root
gates.  The product of their degrees is at most $2^{\operatorname{poly}(M)}$.
Standard algebraic-height and norm bounds give a fixed
$H_0=2^{\operatorname{poly}(M)}$ controlling all valid-prefix values,
their nonzero conjugates, coefficients, roots, and nonzero root
differences.  Include the collected coefficient expressions and the
final ordinary-arithmetic output in this bound.  Enlarge it to control
the conditioning of their rational arithmetic on neighborhoods of the
true earlier roots.  This bound concerns the input's algebraic values,
not the heights of later rational approximants.

Exact zero coefficients and zero denominators can be recognized by
approximating past this fixed nonzero gap.  The preceding construction
needs coefficient precision at most $K(B+H_0)$ to deliver root error
$2^{-B}$, for one common $K=2^{\operatorname{poly}(M)}$.
Crucially, this dependence is linear in $B$.  Choose
\[
 B_i=K^{r-i+1}(B_{\rm out}+H_0+1),\qquad 1\leq i\leq r,
\]
enlarging $K$ once for coefficient conditioning.  All these exponents
have polynomial binary length.  Compute each root approximant once
and keep it in a shared arithmetic DAG.  Each gate uses polynomially
many operations in $M$ and $\log(B_i+H_0)$, because the weak-cover
bound uses the controlled coefficient magnitudes and root separation.
It does not use the bit heights of the cached numerators.

There are consequently only polynomially many new arithmetic operations
and $\PosSLP$ queries.  Approximating the final output past its semantic
gap determines its sign, including zero.  Invalid gates are detected
when first encountered.  This proves the theorem.

The restriction remains substantive for real feasibility.  In particular,
the characteristic polynomial of the critical algebra need not have
polynomially many terms.  The theorem enlarges the root-certificate
language without supplying such certificates for unrestricted ETR.

\subsection{Inverses of series/parallel rational circuits}
\label{sec:series-parallel-inverses}

A second extension permits rational inverse operations.  Consider a
shared DAG with one variable $X$, nonnegative scalar leaves, positive
scalar multiplication, and the operations
\[
 P(F,G)=F+G,\qquad S(F,G)=\frac{FG}{F+G}.
\]
We call these parallel and series composition, respectively.  Identically
zero operands are eliminated using $S(0,G)=0$ and $P(0,G)=G$, including
$S(0,0)=0$.  All retained nodes are positive for $X>0$.  Values at zero
are defined by continuity.  An inverse gate specifies such a function
$F$ and a target $b$, and returns the unique finite $r\geq0$ with
$F(r)=b$.  Negative scalar leaves, a constant function, or an absent
finite inverse make the gate invalid.

\begin{theorem}\label{thm:series-parallel-inverses}
Sign and validity for polynomial-size arithmetic circuits with nested
series/parallel inverse gates are in $\Pclass^{\PosSLP}$, and therefore
reduce to $\PosSLP$ using the adaptive-query compiler.
\end{theorem}

The degree of the reduced rational function may be exponential in the
shared-DAG size.  The convergence bound below has no degree parameter.
The theorem concerns the stated positive grammar, rather than arbitrary
rational circuits accompanied by an assertion of monotonicity.

\paragraph{Concavity and duality.}
On the positive quadrant, $H(a,b)=ab/(a+b)$ has positive first partials
and Hessian
\[
 \operatorname{Hess}H(a,b)
 =-\frac{2}{(a+b)^3}
   \begin{pmatrix}b^2&-ab\\-ab&a^2\end{pmatrix}\preceq0.
\]
Composition therefore shows that every retained function is increasing
and concave.  Its derivative is strictly positive unless it is constant:
an active variable occurrence contributes positively at every operation.
There is also a useful second concavity.  Define
\[
 F^*(t)=\frac{1}{F(1/t)}.
\]
Duality exchanges series and parallel composition, replaces a positive
constant by its reciprocal, and replaces a positive scalar multiplier
by its reciprocal.  Thus $F^*$ belongs to the same grammar and is
concave.  Differentiating twice, with $x=1/t$, gives
\[
 F(x)F''(x)-2F'(x)^2+\frac{2F(x)F'(x)}x\geq0.
\]
Together with $F''\leq0$, this proves
\begin{equation}\label{eq:series-curvature}
 0\leq-xF''(x)
 \leq2F'(x)-\frac{2xF'(x)^2}{F(x)}\leq2F'(x).
\end{equation}
Only elementary concavity and the circuit grammar are used here.

\paragraph{Newton convergence from below.}
Let $F(r)=b$ with $r>0$, and take $0<y\leq r$.  Concavity puts the
Newton update
\[
 \mathcal N(y)=y+\frac{b-F(y)}{F'(y)}
\]
in $[y,r]$.  Integrating~\eqref{eq:series-curvature} gives
$F'(t)/F'(y)\geq(y/t)^2$ for $t\geq y$.  Consequently
\[
 \mathcal N(y)-y
 =\int_y^r\frac{F'(t)}{F'(y)}\,dt
 \geq\int_y^r\frac{y^2}{t^2}\,dt
 =\frac{y(r-y)}r,
\]
and therefore
\begin{equation}\label{eq:series-newton}
 0\leq\frac{r-\mathcal N(y)}r
 \leq\left(\frac{r-y}r\right)^2.
\end{equation}
A factor-two lower bracket gives relative error at most $1/2$; after
$j$ steps it is at most $2^{-2^j}$.  If $2^{-L}\leq r\leq2^L$,
binary search on the exponent finds such a bracket using $O(\log L)$
comparisons of $F(2^k)$ with $b$.  Each $2^k$ is a short rational SLP.
Automatic differentiation and $O(\log(L+B))$ Newton steps then give
absolute error below $2^{-B}$.  No high-precision mantissa is written.

\paragraph{Endpoint validity.}
The finite value $F(0)$ is computed recursively using the continuous
series rule.  For $F(\infty)$, use an infinity tag: parallel composition
with infinity is infinite, while
\[
 S(\infty,a)=a,\qquad S(\infty,\infty)=\infty.
\]
Zero nodes have already been eliminated.  For a nonconstant function,
the exact validity condition is
\[
 F(0)\leq b<F(\infty).
\]
Equality at the lower endpoint returns zero.  Equality at a finite upper
asymptote is invalid.  Nonconstancy is detected by the active-variable
tag, or equivalently by $F(1)>F(0)$.  These rules handle removable zero
denominators such as the apparent $0/0$ in $S(X,X)$ at zero.

\paragraph{Coefficient stability and nested gates.}
We spell out the precision bound to avoid a dependence on the heights
of generated Newton endpoints.  For total description length $n$, the
cleared root equations have degree $2^{\operatorname{poly}(n)}$.
Algebraic height and norm bounds supply a fixed
$L=2^{\operatorname{poly}(n)}$ controlling exact gate values, finite
endpoints, their nonzero differences, and $F'(r)$.  Every listed nonzero
value has magnitude between $2^{-L}$ and $2^L$.

Fix the exact zero pattern.  On the positive parameter box in which
nonzero leaves and $X$ lie in $[2^{-L-1},2^{L+1}]$, positive
numerator/denominator clearing gives a uniform derivative bound $2^G$,
with $G=2^{\operatorname{poly}(n)}$, for $F(X)-b$ and finite endpoint
tests.  One positive denominator monomial supplies its lower bound.
These estimates depend on the input grammar, not the rational cache.

Suppose incoming scalar errors are at most $e$.  For $0<t\leq r/2$,
\eqref{eq:series-curvature} gives $F'(x)\geq2^{-L}/4$ on
$[r-t,r+t]$.  Choose
\[
 t=8e\,2^{G+L},\qquad e\leq2^{-G-2L-4}.
\]
Then $t\leq r/2$, and the two true endpoint residuals dominate their
perturbations.  The nearby positive grammar has a unique inverse
$\widetilde r$ in this interval, with
$|\widetilde r-r|\leq8e\,2^{G+L}$.
Its Newton computation still obeys~\eqref{eq:series-newton} exactly.
The zero-root and endpoint-equality cases are handled by separated
algebraic sign tests before this argument.

Thus a local gate has error at most $Ke+\delta$, where one may take
$K=2^{G+3L+10}$ and enlarge it for ordinary arithmetic gates.
For at most $s$ gates, a common local precision $\delta=2^{-B}$ with
\[
 B=s(\lceil\log_2K\rceil+1)+G+2L+B_{\rm out}+30
\]
makes the accumulated error smaller than both the stability threshold
and $2^{-B_{\rm out}}$.  All precision exponents have polynomial
binary length.  Each approximant is computed once in a shared cache,
using polynomially many rational operations and $\PosSLP$ comparisons.
Taking $B_{\rm out}\geq L+5$ completes the sign and validity tests.

\paragraph{An exponential reduced-degree example.}
Set $F_0=X+1$ and
\[
 F_{j+1}=S(F_j,1)+S(F_j,X).
\]
Only constantly many nodes are added per step.  If $F_j=P/Q$ is
reduced, the next pair is
\[
 P_{\rm new}=P((1+X)P+2XQ),\qquad
 Q_{\rm new}=(P+Q)(P+XQ).
\]
These polynomials are coprime.  Indeed, put $A=P+Q$, $B=P+XQ$,
and $C=B+XA$.  A common divisor of $A,B$ would divide $X-1$,
but $A(1)>0$.  Also $P(0)>0$, so $P$ is coprime to $A,B$, and
$C$ is coprime to both $A$ and $B$.  Positive coefficients persist.
The reduced degrees consequently satisfy
\[
 \deg P_j=2^{j+1}-1,\qquad \deg Q_j=2^{j+1}-2.
\]
This verifies exponential degree without relying on unreduced formal
expressions.  No representation of the general ETR critical eliminant
by this grammar is claimed.

\subsection{Monomial-weighted positive rational inverse gates}
\label{push12:rational-inverses}

Let $P,Q$ be nonzero univariate polynomial circuits using only $+$ and
$\times$, the variable $X$, and nonnegative scalar leaves from earlier
gates.  Their support and degree may be exponential.  For a binary
integer $D>\deg Q$, put
\[
 F(x)=x^D P(x)/Q(x)\quad(x>0),\qquad F(0)=0.
\]
The value at zero denotes the continuous extension, so it does not
evaluate $Q(0)^{-1}$ when $Q(0)=0$.  An inverse gate returns the unique
nonnegative solution of $F(x)=b$ for $b\ge0$, returning zero directly
when $b=0$.  Negative scalar leaves or $b$, zero $P$ or $Q$, and a failed
degree condition are invalid.

\begin{proposition}
Sign and validity for polynomial-size rational arithmetic circuits with
polynomially many nested inverse gates of this form belong to
$\mathsf P^{\mathrm{PosSLP}}$, and hence admit the previously established
many-one compilation to $\mathrm{PosSLP}$.  The gates fit the same shared
scalar pool as the preceding root and inverse constructions.
\end{proposition}

\begin{proof}
Write $p=\deg P$ and $q=\deg Q$.  The nonnegative coefficient weights
at any $x>0$ define probability distributions on their exponents.  Thus,
with differentiation in $y=\log x$,
\[
 \mu:=\frac{xF'}F=D+\mathbb E_P k-\mathbb E_Q l,
 \qquad 1\le\mu\le D+p,
 \qquad \frac{d\mu}{dy}=\operatorname{Var}_P(k)-\operatorname{Var}_Q(l).
\]
Consequently
\[
 \left|\frac{xF''}{F'}\right|
 =\left|\mu-1+\frac{d\mu/dy}{\mu}\right|
 \le K:=D+p+(p+q)^2.
\]
The orders at zero and infinity are positive because $D>q$, proving
existence and uniqueness.  Convexity is not required.  Algebraically,
the selected positive root satisfies $X^D P-bQ=0$, whose negative
coefficients occur below degree $D$ and whose positive coefficients
occur at or above it.

If $r=F^{-1}(b)>0$ and $\varepsilon=|x-r|/r\le1/(16K)$, integrating
the displayed curvature bound along the interval between $x$ and $r$
and applying the Newton remainder formula gives
\[
 \frac{|x_{
m new}-r|}{r}
 \le\frac{K\exp(K\varepsilon/(1-\varepsilon))}
               {2(1-\varepsilon)}\varepsilon^2
 \le K\varepsilon^2.
\]
This preserves positivity and the basin even if Newton crosses the
root.  A binary search on integer exponents first gives a factor-two
bracket; $O(\log K)$ ordinary bisections enter the basin.  Thereafter
$K\varepsilon$ squares at every step, so $O(\log(B+H+2))$ steps give
absolute error $2^{-B}$ when $r\in[2^{-H},2^H]$.  Every evaluation and
automatic derivative uses a polynomial-size shared rational SLP;
binary powering handles $x^D$.  Neither the degree nor the precision
exponent is expanded into a unary computation.

For nesting, let $H_0=2^{\operatorname{poly}(n)}$ be the fixed semantic
height and nonzero-separation bound for the original algebraic circuit.
After exact zero pruning, let $\lambda_P,\lambda_Q$ bound the total
degrees in the positive scalar leaves, and put
$\Lambda=1+\lambda_P+\lambda_Q$, including $b$.  Relative leaf errors
at most $\delta\le1/4$ give, at every positive argument,
\[
 (1-\delta)^{\lambda_P}
 \le\frac{\widetilde P}{P}
 \le(1+\delta)^{\lambda_P},
\]
and similarly for $Q$.  Since the perturbed logarithmic derivative
remains at least one, the corresponding roots obey
\[
 \left|\log\frac{\widetilde r}{r}\right|\le2\Lambda\delta.
\]
When the true positive leaves and root lie in
$[2^{-H_0},2^{H_0}]$, absolute leaf precision
\[
 T=B+2H_0+\lceil\log_2\Lambda\rceil+5
\]
makes the root perturbation at most $2^{-B-3}$.  Newton accuracy is
budgeted separately.  This requirement is linear in $B$ and does not
depend on the bit height of old rational approximants.  Ordinary
arithmetic conditioning and the existing scalar interfaces therefore
fit a common requirement $K_*(B+H_0+1)$ with
$K_*=2^{\operatorname{poly}(n)}$.  A reverse precision schedule with
base $2K_*$ supplies that requirement because $B\ge H_0+1$.  This and
one stored rational approximation per scalar gate keep both the total
DAG and the number of $\mathrm{PosSLP}$ comparisons polynomial.
Semantic zero gaps handle invalid inputs and the final sign.  The
adaptive compiler yields the stated many-one reduction.
\end{proof}

This is a compilation theorem for the weighted function $x^D P/Q$.
It supplies no improved upper bound for unrestricted $\mathrm{PosSLP}$ or
$\exists\mathbb R$.

There is also an exact positive-circuit certificate for this
monotonicity.  Set $H_j(x,a)=(x^j-a^j)/(x-a)$, interpreted as a
polynomial, and let $\operatorname{DD}_P=(P(x)-P(a))/(x-a)$.  Binary
recursion computes $H_j$ using $O(\log(j+2))$ positive gates, while the
usual sum and product rules compute $\operatorname{DD}_P$ with sharing.
For $q=\deg Q$, put
\[
 G_Q=\frac{x^qQ(a)-a^qQ(x)}{x-a}.
\]
If $Q=AB$ and $u=\deg A$, $v=\deg B$, then
$G_Q=x^vB(a)G_A+a^uA(x)G_B$.  At a sum gate, the contribution of
$A$ to $G_{A+B}$ is
$x^{q-u}G_A+a^uA(x)H_{q-u}$, and similarly for $B$.
The leaves have $G=0$.  These positive recurrences give a polynomial
size circuit for
\[
 T=x^D\operatorname{DD}_P Q(a)
   +P(a)\bigl[x^{D-q}G_Q+a^qQ(x)H_{D-q}\bigr].
\]
Direct expansion gives the exact polynomial identity
\[
 (x-a)T=x^DP(x)Q(a)-a^DP(a)Q(x),
\]
and hence
\[
 Q(x)Q(a)(x-a)(F(x)-F(a))=(x-a)^2T.
\]
Here $T>0$ for $x,a>0$, since $D-q\ge1$.  At an inverse root
$F(a)=b$, the cleared identity is
\[
 Q(a)(x-a)\bigl[x^DP(x)-bQ(x)\bigr]=(x-a)^2T.
\]
To interpret the rational identity literally on the whole nonnegative
axis, first strip $Q=x^vQ_0$ and replace $D$ by $D-v$; then
$Q_0(0)>0$ and $D-v>\deg Q_0$.  The order at zero is computed by
minimum at sums and addition at products, and normalization uses only
positive gates and binary powers.  This is a succinct positive-circuit
certificate, not a claim of a polynomial-size expanded sum of squares.

\subsection{Root gates with ordered coefficient blocks}
\label{sec:ordered-block-roots}

The sparse-root construction has an extension in which the number of
nonzero coefficients may be exponential.  What is supplied instead is
an ordered decomposition into positive coefficient blocks.
Write
\[
 \begin{gathered}
 f_j(X)=X^{a_j}P_j(X),\qquad \deg P_j\leq b_j-a_j,\\
 0\leq a_j\leq b_j,\qquad b_j<a_{j+1}\quad(j<s),
 \end{gathered}
\]
where each nonzero $P_j$ is given by a circuit with nonnegative scalar
leaves and the operations addition and multiplication.  The exponents
are written in binary.  Consider
\[
 p(X)=\sum_{j=0}^{s}c_jf_j(X).
\]
Zero $c_j$ are removed, and zero leaves are identified before the degree
recurrences are used.  The signs of the remaining $c_j$ are arbitrary.
If all $c_j$ are zero, the selected-positive-root gate is invalid.
After the sign tests, write $c_j=\operatorname{sgn}(c_j)|c_j|$ and
include $|c_j|$ among the shared positive scalar leaves for rational
perturbation; its sign is held fixed.
The number $s+1$ of blocks may be polynomial in the input length.

\begin{theorem}[Ordered-block root gates]
\label{thm:ordered-block-roots}
Subject to the semantic-height and adaptive-compilation lemmas used
above, sign and validity for polynomial-size arithmetic circuits with
nested selected positive roots of the supplied ordered-block
polynomials are in $\Pclass^{\PosSLP}$, and reduce to $\PosSLP$.
Either distinct-root indexing or indexing with multiplicities is
allowed.  Scalar coefficients are shared with the previously supported
root and inverse gates.
\end{theorem}

This is a theorem about the supplied representation.  It does not
provide a short ordered-block decomposition of an arbitrary arithmetic
circuit.  For example, the family
\[
 \sum_{j=0}^{s} c_j X^{j(N+1)}(1+(j+1)X)^N
\]
fits the representation when $N$ is binary and $s$ is polynomial,
although each block can have exponentially many coefficients.

We first construct the derivative chain, then give the complex-root
bound needed to refine its roots quickly.

\begin{lemma}[Positive Wronskians]
\label{lem:positive-block-wronskians}
Every Wronskian of an ordered subfamily of the $f_j$ has nonnegative
coefficients and is strictly positive on $(0,\infty)$.  It has an
ordinary polynomial-size arithmetic circuit.  For a fixed number of
columns, it also has a polynomial-size circuit using only nonnegative
constants, addition, and multiplication.
\end{lemma}

\begin{proof}
Take $k$ columns, abbreviate their lower and upper support bounds by
$A_i,B_i$, and introduce separate variables $Y_i$.
Starting with $C=\prod_iP_i(Y_i)$, apply
\[
 \prod_{i<j}\bigl((A_j-A_i)+E_j-E_i\bigr),
 \qquad E_i=Y_i\frac{\partial}{\partial Y_i}.
\]
After setting every $Y_i=X$, multiply by
$X^{\sum_iA_i-k(k-1)/2}$.  The result is the Wronskian:
on a choice of exponents $e_i$, both constructions multiply the
coefficient product by the Vandermonde factor
$\prod_{i<j}(e_j-e_i)$.
All these factors are positive.
The outside exponent is nonnegative because the ordered integer
supports imply $A_i\geq i-1$.

For the positive circuit claim, implement an Euler difference as
\[
 (A_j-B_i)C+E_jC+
       \bigl((B_i-A_i)-E_i\bigr)C.
\]
All three summands have nonnegative coefficients.
The degree-complement operation $uC-E_iC$ has a shared linear-size
positive compiler whenever $u\geq\deg_{Y_i}C$.
For a product, use the product rule with the degrees added.
For a sum, add the two degree complements and the nonnegative
degree-padding terms.  Euler operators do not enlarge support, so
the same bounds hold at every step.
This uses a fixed number of circuit transformations when $k$ is fixed.

For polynomial $k$, compute derivatives through order $k-1$ by shared
truncated automatic differentiation and use a division-free determinant
circuit.  This gives an ordinary polynomial-size circuit.
The Vandermonde argument still proves coefficient positivity, including
after any nonnegative perturbation of the original leaves.
\end{proof}

Set $W_k=W(f_0,\ldots,f_k)$ and $W_{-1}=1$, and define
\[
 N_k=W(f_0,\ldots,f_{k-1},p),\qquad
 L_k=\frac{N_k}{W_k},\qquad 0\leq k\leq s.
\]
The empty-prefix convention gives $L_0=p/f_0$.
Multilinearity expresses $N_k$ as a signed sum of the positive
Wronskians with final column $f_j$, $j\geq k$.
The Jacobi identity gives
\begin{equation}
 L_k'=\frac{W_{k-1}W_{k+1}}{W_k^2}L_{k+1},
 \qquad L_s=c_s.
 \label{eq:block-rolle-chain}
\end{equation}
Thus all prefactors are positive on the positive half-line.
Rolle's theorem bounds the number of zeros of $L_k$, with
multiplicities, by $s-k$.

The coefficient-sign variation of $N_k$ need not satisfy that bound.
For example, with
$f_0=(1+X^2)^N$, $f_1=X^{2N+1}$, and $f_2=X^{2N+2}$,
\[
 W(f_0,f_1-f_2)
 =X^{2N}(1+X^2)^{N-1}
   \bigl((2N+1)-(2N+2)X+X^2-2X^3\bigr)
\]
has $2N+1$ coefficient-sign changes.
We therefore use a local complex-root estimate rather than applying
Descartes' bound to every transformed numerator.

\begin{lemma}[Local complex disconjugacy]
\label{lem:block-local-disconjugacy}
Suppose the initial Wronskians of $g_0,\ldots,g_{r-1}$ are nonzero
polynomials with nonnegative coefficients and degree at most $D_W$.
For $x_0>0$, put
\[
 R_0=\frac{x_0}{100(D_W+1)},\qquad
 C_r=r!\,2^{r(r-1)/2},\qquad
 \rho=\frac1{8rC_r}.
\]
Every nonzero complex linear combination of the $g_j$ has at most
$r-1$ zeros, counted with multiplicities, in
$|z-x_0|\leq\rho R_0$.
\end{lemma}

\begin{proof}
Write $V_j=W(g_0,\ldots,g_j)$.
For any nonzero nonnegative polynomial $V$ of degree at most $D_W$,
its coefficient weights at $x_0$ show that
\[
 \left|\frac{V(z)}{V(x_0)}-1\right|
 \leq (1+R_0/x_0)^{D_W}-1<\frac1{50}
 \quad (|z-x_0|\leq R_0).
\]
The canonical weights
\[
 w_0=g_0,\qquad
 w_j=\frac{V_jV_{j-2}}{V_{j-1}^2}\quad(j\geq1),
 \qquad V_{-1}=1,
\]
are consequently analytic and nonzero on this disk.
Their normalized versions $v_j=w_j/w_j(x_0)$ have modulus below $2$.
Successive division and differentiation, using the same Jacobi identity
as in \eqref{eq:block-rolle-chain}, shows that the span divided by
$g_0$ has the canonical basis $\psi_0=1$ and
\[
 \psi_j(z)=\int_{x_0}^{z}v_1(t_1)
       \int_{x_0}^{t_1}v_2(t_2)\cdots
       \int_{x_0}^{t_{j-1}}v_j(t_j)\,dt_j\cdots dt_1.
\]
These analytic integrals are used in the proof only.
No integration gate is required in the computed representation.

Writing $\psi_j(z)=\sum_e A_{e,j}(z-x_0)^e$, the iterated-integral
and Cauchy bounds give
\[
 A_{e,j}=0\ (e<j),\qquad A_{j,j}=\frac1{j!},\qquad
 |A_{e,j}|\leq \frac{2^jR_0^{j-e}}{j!}.
\]
Expand the evaluation determinant at $r$ nodes by Cauchy--Binet and
divide by their Vandermonde determinant.
Its leading term is $a_0=1/\prod_{j=0}^{r-1}j!$.
A coefficient minor with excess degree $L$ has modulus at most
$a_0C_rR_0^{-L}$.
The corresponding Schur polynomial, on the smaller disk, has modulus
at most $r^L(\rho R_0)^L$: there are at most $r^L$ fillings of its
$L$ cells.
There are at most $2^{L-1}$ partitions of $L\geq1$.
The relative tail is therefore bounded by
\[
 C_r\sum_{L\geq1}(2r\rho)^L\leq\frac13.
\]
The normalized determinant is nonzero.
The same convergent expansion at coalescing nodes gives the confluent
determinants, proving the assertion including multiplicities.
\end{proof}

For a chain stage, take
$g_j=W(f_0,\ldots,f_{k-1},f_j)$ for $j=k,\ldots,s$.
The generalized Sylvester identity expresses each of their initial
Wronskians as
\[
 W(f_0,\ldots,f_{k-1})^\ell
 W(f_0,\ldots,f_{k-1},f_{j_0},\ldots,f_{j_\ell}).
\]
Lemma~\ref{lem:positive-block-wronskians} applies.
If the original degrees are at most $n$, one may take
$D_W=2(s+1)^2n$.
The logarithm of the inverse relative radius in
Lemma~\ref{lem:block-local-disconjugacy} is polynomial in
$s+\log(n+1)$.

\begin{proof}[Proof of Theorem~\ref{thm:ordered-block-roots}]
Construct the true Rolle chain and work downward from $L_s$.
There are $O(s^2)$ internal roots.
Fix a semantic exponent $H_0=2^{\mathrm{poly}(N)}$ covering their
magnitudes, nonzero critical-value gaps, and separations between
distinct roots of all chain numerators; here $N$ is the full input
length.
This uses the same algebraic degree and norm bounds as the earlier
root-gate compiler.
Uniform one-root and two-root estimates supply the complex separation
and imaginary-part bounds.  We do not adjoin exponentially many
complex roots to the shared semantic field.

Suppose the roots $\alpha_i$ of $N_{k+1}$ have been constructed.
Test $N_k(\alpha_i)=0$ using the fixed semantic gap.
If it is zero, carry the same root into the current list.
Equation~\eqref{eq:block-rolle-chain} increases its multiplicity by one.
Every carried root is thus an alias of a simple root at a deeper
stage; it is not obtained by splitting an approximate multiple root.

Choose buffers of radius $\delta=2^{-H_0-10}$ about sufficiently
accurate approximants to the $\alpha_i$.
They contain no other current-stage root.
On each remaining interval, $L_k$ is strictly monotone, so endpoint
signs identify at most one new simple root.
Strict dyadic lower and upper bounds close the two unbounded ends.
One may take $2^{-H_0-2}=256\delta$ and $2^{H_0+1}$, and use
fixed cluster disks of radius $\delta/4$; a zero-root cluster then
stays below the searched positive interval.
A second exponent $H_1=\mathrm{poly}(n,H_0,s)$ bounds the nonzero
endpoint values and derivatives.
This distinction matters: an endpoint evaluation can cost $nH_0$
bits in its product bound.
The buffers remain fixed by $H_0$; they are not recursively shrunk
to demand a self-referential gap estimate.

For a new simple target, use one fixed rational approximation of the
original positive leaves.
In $N_k$, the largest active suffix uniquely supplies the highest-degree
term and the smallest active suffix the lowest nonzero order.
Their coefficients remain nonzero under these perturbations, so degree
and multiplicity at zero are unchanged.
Rouch\'e bounds preserve all root clusters in fixed small disks.
Any split multiple-root clusters remain in the excluded buffers,
and nonreal and negative clusters remain outside the positive
isolators.
The selected simple isolator contains exactly one perturbed root.
The rational numerator still has the Wronskian structure.

Binary search on dyadic exponents first gives a factor-two positive
bracket using $O(\log(H_0+2))$ sign tests.
A further polynomial number of midpoint bisections puts its entire
one-circle disk inside the disk of
Lemma~\ref{lem:block-local-disconjugacy}.
There are now at most $s-k$ complex roots in that initial disk.
One-circle disks of nested subintervals are nested, so their root set
can change at most $s-k$ times.
Sagraloff's refinement lemma~\cite[Lemma~6]{Sagraloff14} bounds a
consecutive portion with unchanged root set by
$O(\log n+\log(\Gamma+B+2))$ iterations.
Here $\Gamma$ controls the initial magnitude and $2^{-B}$ is the
requested width.
The floor-free secant, clipping, and boundary construction already
proved above retains the required width and endpoint-movement bounds.
It uses only ordinary circuit evaluations of the numerator and its
derivative, and PosSLP sign tests.
No coefficient variation is computed or bounded globally.

If the rational refiner encounters an exact rational root, that is
an approximation to its matched true internal root.
The semantic label is retained; no later zero test assumes that the
true algebraic-coefficient polynomial vanishes at that rational point.

For a simple target, on a circle of radius $2^{-B}$ the local
polynomial product has only one factor costing $B$ bits.
All other factors use fixed separation bounds.
The required earlier-scalar precision therefore has the form
$K(B+H_1+1)$ with $K=2^{\mathrm{poly}(N)}$.
Carried roots forward that request to their deeper simple source.
Critical-value tests use their fixed semantic gaps.
A reverse cache schedule with base $2K$ absorbs the additive term.
Each of the polynomially many internal roots is approximated once,
and all later arithmetic shares its stored rational circuit.
Degree products, height exponents, and precision encodings remain
within the previously proved semantic bounds.
Thus the number of operations and PosSLP comparisons is polynomial.
The adaptive compiler gives the stated final reduction.
\end{proof}

\subsection{A pure quaternion target}
\label{push12:pure-target}

Let \(A\) be the rational critical algebra and \(h\) its observable.
The polynomial fold compiler gives
\[
 q=(aD+V,V,V,V)/D,\qquad
 S=(aD+V)^2+3V^2,\qquad n=\operatorname{Nrd}(q)=S/D^2.
\]
Run it so that \(hD^2-S\in8R\) and, on a NO instance,
\(0\le(h-n)/h\le\varepsilon\le1/2\).
The established dyadic unit properties of \(h,D\) imply that \(S\)
is also a dyadic unit. All three are units in \(A\), and \(S>0\)
at every real critical point. Put
\begin{equation}\label{push12:rho}
 \rho=hD^2/S.
\end{equation}
On a NO instance,
\[
 1\le\rho\le1+\frac{\varepsilon}{1-\varepsilon},
 \qquad \rho\in1+8(R\otimes\mathbb Z_2).
\]

\begin{proposition}\label{push12:rho-three}
On every NO instance, \(\rho\) is a sum of three rational squares
in \(A\).
\end{proposition}
\begin{proof}
Work in each number-field factor. At every dyadic place \(\rho\)
is itself a square, by Hensel's lemma applied to an element of
\(1+8\mathcal O\). This observation holds in any finite extension
of \(\mathbb Q_2\). At an odd finite place the ternary form
\(\langle1,1,1\rangle\) is isotropic and represents every scalar.
At real places \(\rho>0\); complex places impose no restriction.
Hasse--Minkowski applied to
\(\langle1,1,1,-\rho\rangle\) gives the assertion. If the ternary
subform is isotropic globally it already represents every scalar;
otherwise the last coordinate of the isotropic vector is nonzero.
\end{proof}

If \(\rho=\sum_{j=1}^3u_j^2\), take the pure quaternion
\(w=u_1i+u_2j+u_3k\). Then \(\operatorname{Nrd}(qw)=h\).
This removes the former offset \(25/81\). The corresponding
short-certificate hypothesis becomes the existence of three
short numerators \(U_j\) and a nonzero scalar denominator \(c\)
satisfying
\begin{equation}\label{push12:pure-certificate}
 c^2hD^2=S\sum_{j=1}^3U_j^2\pmod I.
\end{equation}
At a real critical point this implies \(h\ge0\), and the independently
known unit property of \(h\) excludes equality. The ordinary
pure-monic verifier still checks the rational-coefficient format in
\(\mathrm{coRP}^{\PP}\). Existence in Proposition~\ref{push12:rho-three}
does not prove polynomial description length.

\subsection{An exact family and an obstruction to fixed formulas}
\label{push12:exact-family}

For the scalar fold family \(h=7,a=1,\Delta=25\), use only
scale-one folds. After \(t\ge1\) folds, with \(k=2^t\),
\[
 D=3^k+2^k,\quad N=(3^k-2^k)/5,\quad
 V=3^k-3\cdot2^{k-1},\quad 7D^2-S=25\cdot6^k.
\]
Put \(z=(2/3)^{k/2}\) and \(f(z)=7z^4-11z^2+7\).
Then \(\rho=7(1+z^2)^2/f(z)\), and the exact identity
\[
 7f(z)=\left(7-\frac{41}{7}z^2\right)^2
       +\left(z+\frac{24}{7}z^2\right)^2
       +\left(2z-\frac{12}{7}z^2\right)^2
\]
gives three squares for \(\rho\) by multiplying each displayed
coordinate by \((1+z^2)/f(z)\). Their rational SLP sizes are
\(O(t)\); quaternion multiplication then gives an exact
four-square representation of seven. No search or quotient
normal-form computation is used.

Three coordinates are necessary in this family. Since \(k\) is
two or four modulo six, \(D\equiv6\pmod7\) and \(S\equiv3\pmod7\).
Thus \(v_7(\rho)=1\), which is impossible for two rational squares.
Nevertheless \(\rho\to1\) over the reals and
\(\rho\equiv1\pmod{2^k}\) for \(t\ge2\).

There is a different obstruction even within positive rational curves.
For \(h=13,a=1,\Delta=49\), the fold targets are
\[
 \rho(z)=\frac{13(1+z^2)^2}{13z^4-23z^2+13},
 \qquad z=z_t=(3/4)^{2^{t-1}}.
\]
The denominator is positive for every real \(z\).
There are no three rational functions over \(\mathbb Q(z)\)
whose squares sum to this \(\rho(z)\): none could have a pole at
the real point \(z=1\), since the highest even-order poles in a
sum of real squares cannot cancel. Specialization would represent
\(\rho(1)=52/3\) by three rational squares, and multiplication
by \((3/2)^2\) would represent \(39\equiv7\pmod8\), a contradiction.
Nor can a finite collection of fixed rational triples cover every
fold point \(z_t\). One triple would work at infinitely many
distinct points, forcing the forbidden rational-function identity.
This does not exclude a growing, \(t\)-dependent recursive circuit.

The additional fold conic does not remove this obstruction. In fact
\[
 n=\frac{3a^2+(\Delta N/D)^2}{4},
\]
and, with \(W=4L_{\rm old}\), the next fold has
\[
 D_{\rm new}+W=D_{\rm old}^2,\qquad
 D_{\rm new}-W=s^2\Delta N_{\rm old}^2.
\]
These are the ordinary split-conic parametrization, valid for every
\(\Delta\), including negative real values. They impose no new
algebraic condition that alone supplies a rational norm point.

\subsection{A shared pool of real root constants}
\label{push12:root-coefficient-verifier}

The certificate numerators can be enlarged to polynomial circuits
with a polynomial-size shared pool of real scalar constants.
The constants are defined by the supported binary-root and monotone
inverse circuits of Theorems~\ref{push12:binary-root-theorem}
and~\ref{push12:monotone-theorem}. Their entire defining DAG counts
toward witness length. Root gates do not depend on the polynomial
variables. We also permit the original disjoint-variable flattenable
format: each numerator is one selected coefficient of a polynomial
SLP, and only the fixed-degree norm or conic residual operations are
applied before the perfect-pairing square. Fresh coefficient variables
give at most constantly many copies of each generating circuit.
Arbitrary-depth arithmetic on coefficient-extraction leaves is excluded.

\begin{proposition}\label{push12:root-coefficient-verifier-prop}
Ideal membership for this format in the fixed-degree pure-monic
critical algebra is in \(\BPP^{\PP}\). Polynomial-size norm or
conic certificates in this format would therefore still imply
\(\exists\mathbb R\subseteq\mathrm{coMA}^{\PP}\).
\end{proposition}
\begin{proof}
Let \(F\) be the residual and let \(L\subset\mathbb R\) be the field
generated by the shared scalar pool. The perfect residue pairing gives
\[
 U=\sum_\alpha\lambda(e_\alpha F)^2,\qquad
 F=0\text{ in }A\otimes L\ \Longleftrightarrow\ U=0.
\]
Perfection survives scalar extension. At the selected real embedding,
all summands are nonnegative. The finite Laurent formula and the
bounded coefficient extractions represent \(U\) as one selected
coefficient of a short polynomial circuit using the same scalar pool.
For a norm numerator, its two appearances in a square and the final
residual square use at most four disjoint-variable copies.

If the total description length is \(s\), the exact pool has field
degree and logarithmic heights at most \(2^{\operatorname{poly}(s)}\).
The coefficient expression has singly exponential degree,
coefficient bit height, and monomial count. The ordinary height
inequalities therefore give a computable
\(L_U=2^{\operatorname{poly}(s)}\) such that
\[
 U\ne0\quad\Longrightarrow\quad U\ge2^{-L_U}.
\]
This bound concerns the exact semantic expression, before any
approximation circuit is built. Formal differentiation and
coefficient \(\ell_1\)-norm bounds likewise give
\(G=2^{\operatorname{poly}(s)}\) such that the sum of absolute
partial derivatives of this coefficient, on a box containing the
scalar pool and its perturbations, is at most \(2^G\).

Approximate every shared scalar to error at most
\(2^{-L_U-G-3}\) using rational SLPs. The root algorithms allow this
binary precision with polynomially many operations and
\(\PosSLP\) queries; their uniform schedule is increased by the
requested final accuracy. Sharing preserves polynomial circuit size.
The resulting coefficient \(\widetilde U\) satisfies
\[
 |\widetilde U-U|\le2^{-L_U-3}.
\]
Clear scalar denominators by numerator/denominator pairs and compare
\(\widetilde U\) with \(2^{-L_U-1}\) using
\(\mathsf{CoeffPosSLP}\). This distinguishes zero from nonzero.
Validity of the scalar pool and nonzero scalar denominators is
handled by its sign procedure. The \(\BPP^\PP\) simulations of
\(\PosSLP\) and coefficient sign compose after amplification.
Two-sided bounded error is sufficient for the stated MA verifier.
\end{proof}

This is an archimedean verification argument, not modular evaluation
of algebraic root constants. It does not give short ordinary integer
SLPs for their exact values or residues. The shared-pool restriction
also matters: exponentially many independently generated radicals
can have a double-exponential compositum degree, invalidating the
polynomial-bit separation exponent used above.

\subsection{Exact conic certificates for dense convex compositions}
\label{push12:convex-compositions}

Let \(P(T)\) be given by a \(+,\times\) circuit with nonnegative
scalar leaves. For each gate polynomial \(f\), define
\[
 \operatorname{DD}_f(U,V)=\frac{f(U)-f(V)}{U-V},\qquad
 \mathcal B_f(U,V)=f(U)-f(V)-f'(V)(U-V).
\]
These are polynomial identities, including at \(U=V\).
At a constant leaf both expressions are zero; at \(T\) they are
one and zero. Addition is componentwise. For multiplication,
\begin{align*}
 \operatorname{DD}_{fg}
   &=f(U)\operatorname{DD}_g+g(V)\operatorname{DD}_f,\\
 \mathcal B_{fg}
   &=g(V)\mathcal B_f+f(V)\mathcal B_g
     +(U-V)^2\operatorname{DD}_f\operatorname{DD}_g.
\end{align*}
The second identity follows by adding
\((f(U)-f(V))(g(U)-g(V))\) to its first two terms.
Memoizing \(f(U),f(V),f'(V),\operatorname{DD}_f,\mathcal B_f\)
at every original gate adds only constantly many gates per node.
If \(U,V\) are SOS-conic, all the displayed factors other than the
explicit square are SOS-conic by nonnegative-coefficient syntax.
Thus \(\mathcal B_P(U,V)\) has a short exact conic description,
even for exponential degree and exponentially many expanded terms.

\begin{theorem}\label{push12:convex-composition-theorem}
Let
\[
 q(X)=(X-u)^TA(X-u)+d,\qquad
 h(X)=P(q(X))+b^TX+c,
\]
where \(A\) is a rational positive semidefinite matrix, \(d\ge0\),
and \(P\) has the preceding nonnegative-leaf circuit.
Global strict positivity of \(h\), and existence of \(X\) with
\(h(X)\le0\), belong to \(\Pclass^{\PosSLP}\).
Every globally strictly positive instance has a polynomial-size
exact conic certificate using a shared pool of supported real root
constants. The result also holds for sums on disjoint blocks and
invertible rational affine changes of coordinates.
\end{theorem}
\begin{proof}
For any real \(\alpha\), put \(a=q(\alpha)\). The exact identity is
\begin{align*}
 &P(q(X))-P(a)-P'(a)\nabla q(\alpha)^T(X-\alpha)\\
 &\hspace{8mm}
 =\mathcal B_P(q(X),a)+P'(a)(X-\alpha)^TA(X-\alpha).
\end{align*}
Both terms on the right are SOS-conic. Rational LDL decomposition
supplies short quadratic conic descriptions. The preceding compiler
supplies the first term with shared DAG nodes.

Suppose \(P\) is nonconstant. The polynomial
\[
 F(r)=2rP'(d+r^2)
\]
has a nonnegative-leaf circuit, obtained by automatic
differentiation, and is nonconstant and strictly increasing on
\(r>0\), with \(F(0)=0\).
If \(b\notin\operatorname{im}A\), a kernel direction of \(A\)
with nonzero pairing against \(b\) makes \(h\) unbounded below.
If \(b=0\), choose \(\alpha=u\), giving minimum \(P(d)+c\).
Otherwise solve \(Av=b\) and set
\[
 \beta=b^Tv>0,\quad \gamma=\sqrt{\beta},\quad
 r=F^{-1}(\gamma),\quad \alpha=u-(r/\gamma)v.
\]
The value \(\beta=v^TAv\) is positive and independent of the
chosen solution. These expressions use one square root and one
supported inverse. They satisfy \(q(\alpha)=d+r^2\) and
\(\nabla h(\alpha)=0\). Consequently
\[
 h(X)=\mu+C(X),\qquad
 \mu=P(d+r^2)+b^Tu-r\gamma+c,
\]
where
\[
 C=\mathcal B_P(q(X),q(\alpha))
       +P'(q(\alpha))(X-\alpha)^TA(X-\alpha)
\]
is SOS-conic and vanishes at \(\alpha\). This proves the minimum
claim directly, also for singular \(A\).

Nonconstancy is tested by \(P'(1)>0\), after validating the scalar
leaves. If \(P\) is constant, a nonzero \(b\) makes \(h\)
unbounded below, and \(b=0\) leaves a constant sign test.
If \(A=0\), the image test likewise handles all nonzero \(b\).
No formula divides by \(\gamma=0\). Symmetry, PSD, rank, image
membership, and the rational solves use polynomial-size linear
algebra with \(\PosSLP\) comparisons. PSD LDL elimination pivots
on a positive diagonal; negative diagonals reject, and a residual
block with zero diagonal must be identically zero.

The sign of \(\mu\) is decided by the root-circuit theorem.
When \(\mu>0\), one further constant \(\tau=\sqrt\mu\) gives
\(h=\tau^2+C\), with \(\tau\ne0\). Equivalently,
\(h/\tau^2=1+C/\tau^2\) is the usual conic certificate.
There are only polynomially many shared scalar gates, so
Proposition~\ref{push12:root-coefficient-verifier-prop} applies.
Block sums add the minima and conic remainders; an invertible
affine substitution preserves these identities.
\end{proof}

The theorem includes
\(P(T)=\prod_j(1+a_jT^{2^j})\), whose expanded support may be
exponential. It does not cover a general sum of overlapping
quadratic compositions: their minimizer may require a coupled
system. In particular the squared quadratic residuals of the
unrestricted ETR critical objective do not satisfy the hypotheses.
No rational-coefficient descent of the new root-constant
certificates, and no improved unrestricted ETR containment,
is asserted.

\subsection{A coupled convex family with an exact rational-center certificate}
\label{push12:coupled-convex}

The preceding constructor separates a univariate polynomial from its
quadratic argument. A different argument permits coupled quadratic terms
and constructs a certificate from a rational approximation to the
minimizer. The minimizer box in the next statement is an explicit
hypothesis; the bounded-ETR box proof uses quadratic slack residuals and
cannot be imported after deleting them.

\begin{proposition}\label{push12:coupled-convex-theorem}
Let
\[
 h(x)=\sum_iw_ix_i^6+\tfrac12x^{\mathsf T}Ax+b^{\mathsf T}x+c,
 \qquad W=\operatorname{diag}(w_i),\quad A\succeq6W,
\]
where the coefficients are rational SLP values, \(A\) is symmetric,
and \(1\le w_i\le2^{\poly(n)}\). Suppose a supplied integer
\(B_0\ge1\), of polynomial numerical value, bounds every coordinate
of the unique minimizer. Global strict positivity is in
\(\Pclass^{\PosSLP}\). Every positive instance has a short exact conic
certificate using only square roots of positive rational-SLP scalars
as additional real constants.
\end{proposition}
\begin{proof}
Put \(B=B_0+1\), \(\lambda=30B^4\), and iterate from zero by
\[
 T(x)=\operatorname{clip}_{[-B,B]^m}
       \left(x-(A+\lambda W)^{-1}\nabla h(x)\right).
\]
The fixed inverse uses polynomial-dimensional rational arithmetic;
clipping uses \(\PosSLP\) comparisons. Analyze the iteration in
\(\|v\|_W=(v^{\mathsf T}Wv)^{1/2}\), without constructing whitened
coordinates. If \(\widetilde A=W^{-1/2}AW^{-1/2}\), the derivative
before clipping is conjugate to
\[
 (\widetilde A+\lambda I)^{-1}
       (\lambda I-30\operatorname{diag}(x_i^4)).
\]
On the box its operator norm is at most
\(q=\lambda/(6+\lambda)<1\). This uses the product of two norm
bounds and assumes no commutation. Coordinate clipping is nonexpansive
in the diagonal weighted norm. The minimizer \(x_*\) is a fixed point,
so the error contracts by \(q\). Since
\(q^{5B^4+1}<1/2\), and the logarithm of
\(B_0\sqrt{\sum_iw_i}\) is polynomial, polynomially many iterations
give
\[
 \|x-x_*\|_W\le(80B^3)^{-1}.
\]

On the larger box the whitened Hessian has inverse norm at most
\(1/6\) and Lipschitz constant at most \(120B^3\). Exact Newton
therefore satisfies
\[
 e_{j+1}\le10B^3e_j^2,\qquad e_j=\|x_j-x_*\|_W.
\]
All subsequent points remain within distance one of the minimizer
and hence inside the box. An absolute error \(2^{-H}\), for binary
\(H\), takes only \(O(\log H)\) additional matrix operations.
Sharing the earlier iterates gives rational SLPs of polynomial size.

The gradient equations, divided by \(6w_i\), have leading terms
\(x_i^5\) and lower terms of degree one. Their quotient has dimension
\(5^m\). The usual multiplication-matrix height bound, used only
semantically, gives a computable \(L=2^{\poly(n)}\) such that
the nonzero minimum \(\mu=h(x_*)\) has \(|\mu|\ge2^{-L}\).
Choose another computable \(G=2^{\poly(n)}\) bounding the whitened
Hessian norm on the box by \(2^G\), and construct a rational center
\(z\) with error at most \(2^{-L-G-4}\).

Put \(g=\nabla h(z)\) and
\[
 \ell=h(z)-\tfrac12g^{\mathsf T}A^{-1}g.
\]
The exact sextic identity
\[
 B_6(x,z)=x^6-z^6-6z^5(x-z)
 =(x^2-z^2)^2(x^2+2z^2)+3z^4(x-z)^2
\]
and quadratic completion give
\[
 h(x)=\ell+
 \tfrac12(x-z+A^{-1}g)^{\mathsf T}A(x-z+A^{-1}g)
       +\sum_iw_iB_6(x_i,z_i).
\]
Thus \(\ell\le\mu\). Conversely, the gradient bound and
\(A^{-1}\preceq W^{-1}/6\) imply
\[
 \ell\ge\mu-\tfrac1{12}2^{2G}\|z-x_*\|_W^2.
\]
If \(\mu>0\), this is positive by the gap bound. If \(\mu\le0\),
then \(\ell\le0\). One exact sign test of \(\ell\) therefore
decides strict positivity, including the zero-minimum case.
For a positive instance, add the scalar root \(\sqrt\ell\).
A rational LDL decomposition of \(A\) and the displayed Bregman
identity give the claimed short exact conic certificate. No exact
coupled-minimizer coordinate is needed as a scalar gate.
\end{proof}

The box can be independently certified, for example, by outward
gradient inequalities on each face. A useful family with arbitrarily
large couplings is \(A=6W+QL_{\rm graph}\), where \(Q>0\) and
\(L_{\rm graph}\) is a nonnegative weighted graph Laplacian.
If \(b=-Av\) for a supplied vector \(v\), the maximum principle
at a largest or smallest coordinate of \(x_*-v\) gives
\(\|x_*\|_\infty\le2\|v\|_\infty\). A polynomial-valued bound
on \(v\) supplies the hypothesis independently of \(Q\).
The coupled nonconvex residuals of unrestricted ETR are not covered.

\subsection{A residue congruence and approximate Gram certificates}
\label{push12:approx-gram}

Consider a pure-monic quotient
\[
 R=\mathbb Z[X_1,\ldots,X_m]/(X_i^d+r_i:1\le i\le m),
 \qquad \deg r_i<d,
\]
where \(d\ge2\) is fixed and \(m\ge1\).
Set \(D=d^m\), \(M=m(d-1)\), and
\(t=(d-1,\ldots,d-1)\). The standard basis is
\(e_\alpha=X^\alpha\), \(0\le\alpha_i<d\).
Let \(\lambda\) extract the coefficient of \(X^t\) in normal form,
and define
\[
 H_{\alpha\beta}=\lambda(e_\alpha e_\beta),\qquad
 J_{\alpha\beta}=\mathbf1_{\alpha+\beta=t},\qquad
 N=HJ-I.
\]

\begin{proposition}\label{push12:residue-congruence}
The matrix \(N\) satisfies \(N^{M+1}=0\). The finite expression
\[
 S=\sum_{k=0}^{M}\binom{-1/2}{k}N^k
\]
satisfies \(SHS^\mathsf T=J\). Its entries, and those of
\(H^{-1}\), have uniform coherent coefficient-SLP descriptions
of polynomial size. The scalar \(4^M\) clears the denominators of
\(S\).
\end{proposition}
\begin{proof}
Every monic reduction strictly lowers total degree. Therefore
\((HJ)_{\alpha\beta}=0\) if \(|\alpha|<|\beta|\), and equals
\(\mathbf1_{\alpha=\beta}\) if the degrees agree. Thus \(N\)
strictly raises degree, proving nilpotence. Put \(A=I+N=HJ\).
Symmetry of \(H=AJ\) gives \(AJ=JA^\mathsf T\), and hence
\(SJ=JS^\mathsf T\). Finite binomial calculus gives \(S^2A=I\), so
\[
 SHS^\mathsf T=SAJS^\mathsf T=SASJ=J.
\]
Also \(H^{-1}=J\sum_{k=0}^{M}(-N)^k\).
Each matrix entry is a sum over paths of length at most \(M\).
Each path uses polynomially many residue coefficients. The existing
finite Laurent formula and Boolean index-summation compiler flatten
these products with disjoint auxiliary variables. All path indices
have polynomial total bit length. This is a polynomial-length path
expansion, not a claim about exponentially long matrix powers.
Finally,
\(\binom{-1/2}{k}=(-1)^k\binom{2k}{k}/4^k\).
\end{proof}

The triangular inverse was already available from the residue
dualization argument; the additional conclusion is the explicit
congruence. For odd \(d\), the reversal \(J\) has signature one.
This does not determine the signatures of the trace forms.

\begin{lemma}\label{push12:residue-error-bound}
Suppose every real root \(z\) satisfies \(|z_i|\le B\), where
\(B\ge1\). Let \(C\ge1\) bound the coefficient \(\ell_1\) norm
of every \(r_i\), and put
\[
 A_0=2C^{2M},\qquad K=DB^M(1+DA_0)^M,\qquad
 U(F)=\sum_\alpha\lambda(e_\alpha F)^2.
\]
For every polynomial \(F\) with real coefficients and every real
root \(z\), one has \(|F(z)|\le K\sqrt{U(F)}\).
\end{lemma}
\begin{proof}
Let \(v\) be the normal-form coefficient vector of \(F\). Its
residue-value vector is \(Hv\). A reduction tree starting in degree
at most \(2M\) has depth at most \(2M\), and increases coefficient
\(\ell_1\) norm by a factor at most \(C\) at each level. Hence
\(|H_{\alpha\beta}|\le C^{2M}\), and \(\|N\|_2\le DA_0\).
Proposition~\ref{push12:residue-congruence} yields
\[
 \|H^{-1}\|_2\le\sum_{k=0}^{M}(DA_0)^k
                   \le(1+DA_0)^M.
\]
Since \(\|e(z)\|_2\le\sqrt D B^M\le DB^M\), the claimed bound
follows from \(F(z)=e(z)^\mathsf T H^{-1}(Hv)\).
\end{proof}

The bound \(K\) has a short positive scalar SLP whenever the input
coefficient and real-root bounds do. Its possibly large value is
retained in the exact inequality below.

\begin{proposition}\label{push12:approx-gram-certificate}
Let
\[
 E(X)=\sum_\alpha e_\alpha(X)^2
     =\prod_{i=1}^{m}\sum_{j=0}^{d-1}X_i^{2j},\qquad
 F=h-\sum_j P_j^2-\eta E.
\]
If \(\eta>0\) and \(4K^2U(F)\le\eta^2\), then
\(h(z)\ge\eta/2>0\) at every real root.
These assertions have a \(\BPP^{\PP}\) verifier in the coherent
certificate format specified below. Polynomial-size certificates
for all NO critical instances would imply
\(\mathrm{ETR}\in\mathrm{coMA}^{\PP}\subseteq
 C_2\Pclass\cap\mathrm{PH}^{\PP}\).
\end{proposition}
\begin{proof}
The lemma gives \(|F(z)|\le\eta/2\). Since \(E(z)\ge1\),
\[
 h(z)=\sum_jP_j(z)^2+\eta E(z)+F(z)\ge\eta/2.
\]
In particular, an approximate identity suffices even when the norm
inequality is an equality.

The index \(j\) may have polynomial bit length. A fixed polynomial-time
generator supplies each \(P_j\) as one coefficient of a polynomial
SLP, divided by one common scalar denominator whose nonvanishing
is checked. Its
size is polynomial, and its syntactic degree and coefficient
bit-height bounds are singly exponential. Ordinary polynomial SLPs
are included. Polynomially many supported real scalar root/inverse
constants may be shared by all circuits. They are independent of
\(X\) and \(j\). The scalar \(\eta\) belongs to this same pool.
Arbitrary arithmetic of polynomial depth on coefficient-extraction
leaves is not included.

Expanding \(U(F)\) uses at most four copies of a \(P_j\) coefficient
description per summand. All SOS and basis indices can be summed
with polynomially many index bits and disjoint coefficient variables.
Clearing the known scalar denominators by a positive square leaves
one coefficient-sign comparison. The shared real-pool verifier first
fixes semantic degree/height and absolute gradient bounds, then
approximates the pool to the corresponding precision. It checks
this comparison and \(\eta>0\) in \(\BPP^{\PP}\). Its error is
two-sided. Guessing a short certificate gives the stated conditional
\(\mathrm{coMA}^{\PP}\) consequence.
\end{proof}

Dense completeness does not supply the missing size bound. If
\(h\ge\delta>0\) at all real roots, choose a short \(\eta>0\) with
\(2\eta DB^{2M}\le\delta\). In the etale real algebra
\(R\otimes\mathbb R\simeq\mathbb R^r\times\mathbb C^q\), the
element \(h-\eta E\) has a square root: it is positive in the real
components, and every complex component has a square root.
Its coefficient vector can be approximated by a rational vector so
closely that even one dense rational polynomial \(P\) satisfies the
norm inequality. No short description of that vector follows.
For completeness, these dense coefficients have singly exponential
bit heights. Every complex coordinate has modulus at most \(C\),
so every basis value is bounded by \(U_b=C^M\). If \(T\) is the
real evaluation matrix, the nonzero integral discriminant gives
\(|\det T|\ge2^{-D/2}\). Cofactors imply
\(\|T^{-1}\|_2\le I_0:=D2^D(DU_b)^D\).
For a short bound \(A_{\max}\ge|h(z)-\eta E(z)|\) at every complex
root, the square-root coefficient vector is bounded entrywise by
\(V:=I_0D(A_{\max}+1)\). Put \(B_0=C^{2M}\), bounding all
multiplication structure constants. Rounding these coefficients
with error \(\epsilon\le1\) gives
\[
 \sqrt{U(h-P^2-\eta E)}
 \le D^4B_0^2(2V+1)\epsilon.
\]
Thus dyadic precision with singly exponential value suffices.
Its exponent has polynomial bit length; the common denominator
has a short scalar SLP, and the dense numerators have singly
exponential bit heights. This supplies no coherent short generator.
The positive slack plays the same role as the residual lower bound
in the convex rational-center certificate. In that setting a short
Newton computation and an explicit Bregman identity construct the
certificate; such a constructor is still missing here.

There is a useful distinction between numerical convergence and
matrix representation. For a real symmetric nonsingular matrix with
eigenvalue magnitudes in \([l,u]\), the scaled step
\[
 A\longmapsto\frac{A+luA^{-1}}{2\sqrt{lu}}
\]
preserves signs and puts the magnitudes in
\([1,(u+l)/(2\sqrt{lu})]\). Its condition ratio is at most
\(\sqrt{u/l}\). Thus polynomially many scaled steps suffice even
when the logarithm of the initial condition ratio is exponential.
Ordinary Newton steps subsequently converge quadratically.
The supported scalar root pool supplies the scales; it does not
supply coherent descriptions of the matrix inverses or products.

Indeed, if \(\delta\) is the Jacobian, the trace matrices are
\(H_1=HM_\delta\) and \(H_h=HM_{\delta h}\), with \(M_f\)
denoting quotient multiplication. The short inverse of \(H\) does
not provide a short inverse of \(M_\delta\) or its scalar denominator.
After the congruence, a weighted form is \(JM'\), while its inverse
is \(M'^{-1}J\). Newton steps therefore involve conjugation by
\(J\), which need not preserve the commutative multiplication
algebra. Ordinary matrix powers cannot be replaced by quotient powers.
Nor does real-point convergence imply residue-norm convergence:
in \(\mathbb Q[x]/(x^3+x)\), the error \(1-(1+Tx)^2\) vanishes
at the only real point for every \(T\), while its residue-square
norm grows as a nonzero degree-four polynomial in \(T\).
The outstanding step is a short coherent Gram approximation with
the global error and known scalar scale required above.

\subsection{Formal square roots and controlled continuation}
\label{push12:homotopy}

Let \(a=h-\eta E\), with \(\eta>0\) chosen sufficiently small
that \(a\) is a unit and is positive at all real points on a NO
instance. The known complex value bounds permit a short dyadic
choice of \(\eta\).

Formal series themselves have short coefficient descriptions.
For a short polynomial \(g(T,X)\) with \(g(0,X)=1\), put \(w=g-1\).
Modulo \(T^N\), the formal square root is
\[
 \sum_{k<N}\binom{1/2}{k}w^k,\qquad
 \binom{1/2}{k}=\frac{2(-1)^{k-1}}{4^k}
 \left(\binom{2k-2}{k-1}-\binom{2k-2}{k-2}\right)
 \quad(k\ge1).
\]
The binomial coefficients are coefficients of \((1+Y)^{2k-2}\).
Represent \(k\) by Boolean index bits, compile its powers by binary
selectors, pad the coefficient indices to one fixed exponent, and
use the common denominator \(4^N\). For a bit \(v\), a degree bound
\(B\) gives the explicit summation identity
\[
 P(0)+P(1)=[v^B]P(v)
       \left(\sum_{j=0}^{B-1}v^j+2v^B\right).
\]
Invalid indices are masked by the Boolean comparison circuit.
Finally,
\[
 \operatorname{trunc}_{<N}B(T)
 =[Z^{N-1}]B(TZ)\sum_{j=0}^{N-1}Z^j.
\]
Disjoint-variable flattening therefore gives one coefficient-SLP
of size polynomial in \(\log N\) and the input size. This is not
an unrestricted closure assertion for indexed circuit families.

Analytic evaluation is a separate issue. In
\[
 A=\mathbb Q[x]/(x^3+x),\qquad h=1+2x^2,
\]
the only real point has \(h=1\), and
\(h=(1+x+x^2)^2\) in \(A\). Yet \(h=-1\) at both nonreal points
\(\pm i\), so every defined real rational function of \(h\) has
the same real value at those points and cannot square to \(h\).
The straight homotopy \(1+T(h-1)\) has a branch point at \(T=1/2\).
Neither a short formal truncation nor a common complex continuation
path supplies the required conjugate imaginary endpoint values.
Since \(E=1\) in this example, even the approximate Gram certificate
is impossible when all its square coordinates are restricted to
\(\mathbb R[h]\). The displayed short exact square demonstrates
the scope of this obstruction.

\begin{proposition}\label{push12:real-homotopy}
There is a polynomial-size real polynomial homotopy from \(1\) to
\(a\) that is a unit at every complex point throughout \(0\le T\le1\)
and is positive at every real point throughout that interval.
\end{proposition}
\begin{proof}
Choose \(\ell=\sum c_iX_i\) with integer \(0\le c_i\le D\)
such that \(\operatorname{Im}\ell(z)\ne0\) at each nonreal point.
Each bad conjugate pair excludes one proper hyperplane from the
grid, containing at most \((D+1)^{m-1}\) grid points. There are at
most \(D/2\) pairs, so a random choice succeeds with probability
greater than \(1/2\).
Let \(M_\ell\ge1\) bound all \(|\ell(z)|\), and let \(\nu\le1\)
be a positive lower bound for every nonzero
\(|\operatorname{Im}a(z)|\). The conjugate-pair compositum has
degree at most \(D^2\), so the existing height bounds give a
short dyadic \(\nu\). Choose a short \(\epsilon>0\) satisfying
\(\epsilon M_\ell\le\nu/2\), and set
\[
 g(T)=1+T(a-1)+\epsilon T(1-T)\ell.
\]
If \(\operatorname{Im}a(z)\ne0\), the imaginary part is
\(T[\operatorname{Im}a(z)+\epsilon(1-T)\operatorname{Im}\ell(z)]\)
and retains its sign for \(T>0\). If \(a(z)\) is real but \(z\)
is nonreal, the perturbation has nonzero imaginary part in the open
interval. At a real point,
\[
 g(T,z)\ge(1-T)/2+Ta(z)>0.
\]
Both endpoints are units. Continuation from the root \(1\) is thus
well-defined and conjugacy compatible. The minimum modulus also
has a singly exponential logarithmic bound: its square is a real
quartic in \(T\), with interior extrema at roots of its cubic
derivative and the usual algebraic degree/height bounds.
\end{proof}

\begin{proposition}\label{push12:continuation-compiler}
Suppose a short real polynomial homotopy is accompanied by
polynomially many short rational scalar knots
\(0=t_0<\cdots<t_s=1\), all unit values, satisfying
\[
 \left|\frac{g(t_{j+1},z)}{g(t_j,z)}-1\right|\le1/16
\]
at every complex point. Then a polynomial-size ordinary rational
polynomial SLP constructs an approximate Gram certificate.
\end{proposition}
\begin{proof}
Start \(r=1\). At each new knot iterate
\[
 r\longleftarrow r(3-g(t_j)r^2)/2.
\]
For \(e=1-g(t_j)r^2\), the exact recurrence is
\(e_{\rm new}=e^2(3+e)/4\). Error at most \(1/64\) at the old
knot becomes at most \(81/1024<1/8\) at the new knot. Quadratic
convergence then achieves error \(2^{-L}\) at all complex points
in \(O(\log L)\) steps, where \(L\) has polynomial bit length.
At the endpoint set \(P=ar\). The error is
\(a-P^2=a(1-ar^2)\). The interpolation and residue bounds in
Section~\ref{push12:approx-gram} convert a sufficiently small
pointwise error into \(\sqrt{U(a-P^2)}\le\eta/(2K)\).
Every update is ordinary polynomial arithmetic. There are no
quotient inverses or hidden normal-form operations.
\end{proof}

Short zero-free path syntax alone does not supply such a schedule.
For \(g(T)=1+2Tx^2+\epsilon T(1-T)x\) in the preceding cubic
algebra, its value at \(i\) is \(1-2T+i\epsilon T(1-T)\).
Any increasing chain with relative changes bounded by \(c<1\)
has \(\Omega(\log(1/\epsilon))\) steps. At its first crossing of
\(T=1/2\), the preceding value \(v\) satisfies
\(\operatorname{Re}v\le c|v|\), hence
\(|v|\le\epsilon/(4\sqrt{1-c^2})\). Each earlier step decreases
modulus by at most \(1-c\). Taking \(\epsilon=2^{-2^n}\) gives
an exponential step bound despite an \(O(n)\)-size path. The
endpoint has a constant-size square, so this is only a limitation
of the chosen continuation schedule.

One explicit sufficient class bypasses continuation. Suppose
\(\operatorname{Re}a(z)\ge\delta>0\), \(|a(z)|\le M\) at all
complex points, with short bounds and \(\kappa=M/\delta\le2^{\mathrm{poly}}\).
Take \(C=M^2/\delta\) exactly. Then
\(|1-a(z)/C|^2\le1-\kappa^{-2}\). The coherent binomial polynomial
\[
 P_N=\sqrt C\sum_{k<N}\binom{1/2}{k}(a/C-1)^k
\]
has pointwise root error at most
\(\sqrt C\rho^N/(1-\rho)\), \(\rho=\sqrt{1-\kappa^{-2}}\).
A polynomial-bit \(N\) reaches the required residue accuracy,
using only the shared scalar root \(\sqrt C\). This is a
constructive sufficient condition, not an unrestricted ETR bound.
If \(M=\delta\), all point values equal \(M\), and \(N=1\)
already gives an exact square root.

More general rational reparametrizations remain open. For
\(T=P(s)/Q(s)\), clearing a quadratic homotopy gives
\(Q(s)^2g(P(s)/Q(s))\), whose endpoint is \(aQ(1)^2\).
The denominator scale must be retained in the certificate.
A global short map must also handle exponentially many conjugate
branch points while preserving real coefficients. The step-count
example does not rule out such compact reparametrizations, and no
general construction is established here.

\subsection{A generic recursive norm lift is also obstructed}
\label{push12:dynamic-norm-obstruction}

The preceding fixed-formula argument does not allow the new square
coordinates to depend on an old ternary point. A Brauer-group argument
rules out that stronger generic construction for the \(h=13\) family.
This is not a lower bound for step-specific arithmetic circuits.

We use the following standard fact, with its short proof included.
For a smooth projective quadric surface \(X\) over a characteristic-zero
field \(K\), a nonsquare discriminant implies that
\begin{equation}\label{push12:quadric-brauer-injection}
 \operatorname{Br}(K)\longrightarrow\operatorname{Br}(K(X))
 \quad\text{is injective}.
\end{equation}
Over a separable closure, the Picard group is generated by the two
rulings. The nonsquare discriminant exchanges these rulings, so the
invariant Picard group is generated by the hyperplane class \((1,1)\),
which already descends to \(X\). The Hochschild--Serre exact sequence
\[
 \operatorname{Pic}(X)\longrightarrow
 \operatorname{Pic}(\overline X)^{\operatorname{Gal}}
 \longrightarrow\operatorname{Br}(K)\longrightarrow\operatorname{Br}(X)
\]
therefore makes the last map injective. Smoothness also gives
\(\operatorname{Br}(X)\hookrightarrow\operatorname{Br}(K(X))\).
These two exact-sequence statements are recorded, for example, in
Colliot-Th\'el\`ene's Arizona Winter School 2015 lectures \cite{CT15},
Sections~4.2 and~4.4, printed pages~19--20.

\begin{proposition}\label{push12:no-generic-ternary-map}
For \(r,s\in F^\times\), put
\(Q_r=\{\langle1,1,1,-r\rangle=0\}\).
If \(-r,-s\) are distinct nontrivial square classes and \(Q_s\) is
anisotropic over \(F\), then \(Q_s\) remains anisotropic over \(F(Q_r)\).
In particular there is no rational map \(Q_r\dashrightarrow Q_s\).
\end{proposition}
\begin{proof}
Let \(H=(-1,-1)_F\) and \(K=F(\sqrt{-s})\).
The quaternion class \(H_K\) is nonzero: otherwise the quadratic
splitting/embedding criterion would embed \(K\) into \(H\), giving
a pure quaternion of norm \(s\) and an \(F\)-point of \(Q_s\).
If \(H\) were already split, the ternary norm would represent every
scalar, with the same contradiction.

Suppose \(Q_s\) becomes isotropic over \(E=F(Q_r)\). Its ternary form
then represents \(s\), so \(H\) splits over \(E(\sqrt{-s})=K(Q_r)\).
The quadric function field is regular over \(F\), justifying this
base change. The square class of \(-r\) remains nontrivial in \(K\):
the kernel of \(F^\times/F^{\times2}\to K^\times/K^{\times2}\) is
\(\{1,-s\}\). Thus the base-changed surface \(Q_r/K\) has nonsquare
discriminant. Equation~\eqref{push12:quadric-brauer-injection}
contradicts the vanishing of its nonzero quaternion class \(H_K\).
\end{proof}

Apply the proposition over \(F=\mathbb Q(z)\) with
\[
 f(z)=13z^4-23z^2+13,\quad
 r(z)=\frac{13(1+z^2)^2}{f(z)},\quad s(z)=r(z^2).
\]
Both quadrics are anisotropic by specialization at \(z=1\), as above.
Their discriminants are nontrivial, and their ratio has square class
\(f(z^2)/f(z)\). The polynomial \(f\) is squarefree and
\[
 f(z^2)\equiv-\frac{108}{13}z^4\pmod{f(z)}.
\]
Since \(f(0)\ne0\), this proves \(\gcd(f(z),f(z^2))=1\);
the ratio has odd valuation at a factor of \(f\) and is nonsquare.
There are therefore no rational functions of
\((z,u_1,u_2,u_3)\), on the generic quadric
\(\sum u_i^2=r(z)\), whose squares sum to \(r(z^2)\).

This rules out a rational update valid on the whole generic old
ternary quadric, even with all three old coordinates available.
Each individual arithmetic fold fiber can still have rational points.
An update confined to a special orbit with extra invariants,
step-dependent arithmetic choices, or growing circuits is not excluded.

\paragraph{A curve cannot carry a fixed rational squaring update.}
The obstruction can be strengthened for a state space of dimension one.
Write
\[
 \rho(z)=\frac{13(1+z^2)^2}{13z^4-23z^2+13}.
\]
\begin{proposition}\label{prop:curve-state-obstruction}
For any even $d\geq2$, there is no smooth geometrically integral
projective curve $C/\mathbb Q$ with a nonconstant $z\in\mathbb Q(C)$,
a rational self-map $\phi$ over $\mathbb Q$, and
$u_1,u_2,u_3\in\mathbb Q(C)$ such that
\[
 z\circ\phi=z^d,\qquad \sum_i u_i^2=\rho(z).
\]
\end{proposition}
\begin{proof}
The self-map has degree $d$.  Riemann--Hurwitz forces the genus to be
at most one.  Let $S$ be the geometric zero and pole set of $z$.
It is finite, has at least two points, and satisfies
$\phi^{-1}(S)=S$.  In genus one the self-map is unramified, giving
$|\phi^{-1}(S)|=d|S|$, a contradiction.  In genus zero,
\[
 (d-1)|S|=\sum_{P\in S}(e_\phi(P)-1)\leq2(d-1).
\]
Thus $S$ has two points, one in each fiber.  Both are rational because
the fibers are Galois invariant.  Hence $C\cong\mathbb P^1_{\mathbb Q}$.
With a rational coordinate $w$ vanishing at the zero and having its
pole at the pole of $z$, we obtain
\[
 z=cw^m,\qquad \phi(w)=aw^d,\qquad a^m=c^{d-1}.
\]
For $d=2$, this gives $z=(aw)^m$, so the rational point $w=1/a$
has $z=1$.

More generally write $m=2^v m_0$ with $m_0$ odd.  Since $d-1$ is
odd, prime valuations in the last identity show that $c=b^{2^v}$
for some $b\in\mathbb Q^*$; when $v>0$ the sign is positive as well.
The odd-degree polynomial $w^{m_0}-1/b$ has an irreducible factor of
odd degree.  Its root gives an odd-degree point of $C$ with $z=1$.
The functions $u_i$ have no poles there, since their leading squared
pole coefficients cannot cancel in a number field with a real embedding.
They therefore represent $\rho(1)=52/3$ by three squares over an
odd-degree extension of $\mathbb Q$.

For completeness, anisotropic quadratic forms remain anisotropic over
odd-degree extensions.  In a least odd-degree counterexample
$F[X]/(f)$, choose polynomial coordinates $a_i$ of degree less than
$n=\deg f$ and divide out their common polynomial factor.
If $m=\max\deg a_i$, anisotropy over $F$ makes $q(a_i)$ have degree
exactly $2m$.  Since $f$ divides it, the quotient has smaller odd
degree.  An odd-degree irreducible factor of the quotient gives a
smaller odd-degree isotropy extension; the vector remains nonzero
there because its coordinates have gcd one.  This is a contradiction.

Finally, $\langle1,1,1,-52/3\rangle$ is anisotropic over $\mathbb Q$:
a three-square representation, scaled by $(3/2)^2$, would represent
$39$, excluded over $\mathbb Q_2$ by the integral square classes
modulo eight.  The preceding odd-degree fact supplies the contradiction.
\end{proof}

This also applies to return maps with $d=2^k$.  In particular a fixed
rational map over $\mathbb Q$, defined along an infinite rational orbit
with distinct $z$-values, cannot have a one-dimensional orbit closure
while maintaining the displayed square identity.  An eventual cycle
of curve components would return to a geometrically integral component
and give the forbidden map on its smooth projective normalization;
infinitely many rational points force that geometric integrality.
The conclusion does not exclude invariant surfaces, arithmetic branching,
nonrational constant fields, or a different update formula at each step.

\subsection{A proposed modular closure would collapse the counting hierarchy}
\label{push12:modular-closure}

Consider the following unproved hypothesis.  For every uniform
polynomial-size integer polynomial SLP $C_x$ on polynomially many raw
$\mathsf{GapP}$ inputs $g_i(x)$, there is a uniform, polynomial-output-bit
$\mathsf{GapP}$ family $R(x,m)$ satisfying
\[
 R(x,m)\equiv C_x(g_1(x),\ldots,g_\ell(x))\pmod m
\]
for every positive polynomial-bit modulus $m$.  The degree of $C_x$ and
the bit length of its integer output may be exponential.  This hypothesis
requires an ordinary raw counting representative; retrieving the $g_i$
with counting queries and then evaluating the residue is insufficient.

\begin{proposition}[Conditional consequence of modular composition]
Under this hypothesis, the magnitude-window lemma implies
\[
 \mathsf{PP}^{\mathsf{PP}}=\mathsf{MP}=\mathsf{P}^{\mathsf{PP}},
 \qquad \mathsf{CH}=\mathsf{PP}^{\mathsf{PP}}.
\]
These equalities are conditional conclusions, not assumptions about the
known relationships among these classes.
\end{proposition}

\begin{proof}
For an odd integer $g$ with $1\le |g|\le M=2^u$, set
\[
 \beta_0=\frac{M+g}{2M},\qquad
 \beta_{j+1}=3\beta_j^2-2\beta_j^3.
\]
All iterates lie in $[0,1]$.  In the coordinate $z=2\beta-1$, the map is
$z\mapsto(3z-z^3)/2$.  Below $|z|=1/2$ its magnitude grows by at least
$4/3$; afterwards the error $e=1-|z|$ satisfies
$e'\le 3e^2/2$.  Therefore
\[
 j=3u+3+\left\lceil\log_2(3(v+2))\right\rceil
\]
iterations give error at most $2^{-v}$ from $\mathbf1_{g>0}$.
Writing $\beta_j=A_j/D_j$, the exact recurrences
\[
 A_0=M+g,\quad D_0=2M,\qquad
 A_{j+1}=3A_j^2D_j-2A_j^3,\quad D_{j+1}=D_j^3
\]
give $D_j=2^d$ with $d=(u+1)3^j$.  Thus the numerator is an integer
polynomial represented by a short SLP, and the denominator is fixed,
although its bit length can be exponential.

Let a $\mathsf{PP}^{\mathsf{PP}}$ machine use $p\ge1$ random bits and at most
$q\ge1$ oracle queries, padded with prescribed dummy answers.  Given a
random tape $r$ and a guessed answer transcript $t$, all queries and the
final acceptance bit $a(x,r,t)$ are computable in polynomial time.  If
$h$ is an ordinary $\mathsf{GapP}$ gap for the oracle, use the nonzero
odd gaps
\[
 g_i(x,r,t)=2h(\operatorname{query}_i(x,r,t))-1.
\]
They have one common bound $2^u$.  Apply the preceding filter with
$v=p+q+10$ and write its error as $\rho=2^{-v}$.  For each transcript put
\[
 W_{r,t}=\prod_{i=1}^q
 \begin{cases}
  \beta_j(g_i),&t_i=1,\\
  1-\beta_j(g_i),&t_i=0.
 \end{cases}
\]
For each $r$ there is exactly one consistent transcript, whose weight
is at least $1-q\rho$.  Every other transcript has weight at most $\rho$.
Consequently $F_r=\sum_t a(x,r,t)W_{r,t}$ approximates the actual
acceptance bit on $r$ within $(q+2^q)\rho\le2^{q+1}\rho$.

All weights have denominator $D=2^E$, where $E=qd$ has polynomial binary
length.  Write $aW_{r,t}=A_{r,t}/D$, and define
\[
 N=2\sum_{r,t} A_{r,t}-(2^p+1)D.
\]
If $A$ is the actual number of accepting outer paths, then $N/D$ differs
by at most $2^{-8}$ from the odd integer $2A-2^p-1$.  Hence its sign is
the machine's decision, including rejection on a tie, and
\[
 2^{E-1}\le |N|<2^{E+p+2}.
\]
The quantity $A$ is used only for this correctness argument; it is not
asserted to be an ordinary $\#\mathsf P$ value.

For each $(r,t)$ the numerator $A_{r,t}$ is one polynomial-size integer
polynomial SLP on ordinary raw $\mathsf{GapP}$ inputs.  One application
of the hypothesis supplies its representative modulo $m$.  Summation
over the polynomial-length indices $(r,t)$ preserves $\mathsf{GapP}$,
and $(2^p+1)2^E\bmod m$ is computable in polynomial time.  Thus $N$ has
uniform polynomial-output-bit raw representatives modulo all required
short moduli.  The magnitude-window lemma puts its sign in
$\mathsf{MP}$.  This proves
\[
 \mathsf{PP}^{\mathsf{PP}}\subseteq\mathsf{MP}
 \subseteq\mathsf P^{\mathsf{PP}}
 \subseteq\mathsf{PP}^{\mathsf{PP}}.
\]
Finally, a deterministic $\mathsf P^{\mathsf{PP}}$ subroutine can be
inlined inside an outer $\mathsf{PP}$ machine, so
$\mathsf{PP}^{\mathsf{PP}^{\mathsf{PP}}}
=\mathsf{PP}^{\mathsf P^{\mathsf{PP}}}
=\mathsf{PP}^{\mathsf{PP}}$.  Induction gives the asserted collapse.
\end{proof}

The available closure modulo $2^a$ for polynomially bounded $a$ does not
prove the hypothesis used here.  The relevant selected-bit position is
near $E$, whose value can be exponential.  Direct extraction would
require a modulus $2^{E+O(p)}$ with exponentially many bits.  Knowing
the magnitude window does not move this bit down: for $E>a$, the
opposite-sign integers $2^E+1$ and $-2^E+1$ have identical residues modulo
$2^a$ and both lie in a constant-width magnitude window.

Even a cosine-specific version of the proposed closure would be too
broad if it allowed arbitrary raw $\mathsf{GapP}$ arguments.  Small-angle
Taylor approximation with dyadic coefficients, followed by polynomially
many doublings $c\mapsto2c^2-1$, gives a short integer-polynomial numerator
for $\cos(2\pi T/Q)$ divided by a fixed known power of two.  For a nonzero
integer $|g|\le M$ and $Q=4M$, the square-wave expansion
\[
 \operatorname{sgn}\sin\theta
 =\frac4\pi\sum_{j\ge0}\frac{\sin((2j+1)\theta)}{2j+1},
 \qquad \theta=2\pi g/Q,
\]
has truncation error $O(1/(K|\sin\theta|))=O(Q/K)$ after $K$ terms.
Taking $K=2^{\operatorname{poly}}$, using integer quarter-period shifts
to replace sine by cosine, and rounding the short weights to a common
dyadic denominator gives a sign filter by one exponentially indexed raw
sum.  Polynomially many ordinary products then form the transcript
weights above.  Thus closure for those cosine numerators on arbitrary
raw counting inputs already implies the same collapse.  Any successful
phase-specific closure must use a narrower input family or a different
mechanism.

The same conditional consequence holds for a restricted Gaussian
product principle.  Suppose that whenever $\epsilon(x,w)\in\{-1,1\}$
is polynomial-time computable and $w$ has polynomial length, the real
coordinate of
\[
 \prod_w(a+i\epsilon(x,w)b),\qquad
 a^2+b^2=D^2,
\]
has uniform polynomial-output-bit raw $\mathsf{GapP}$ representatives
modulo every short positive modulus.  Here $a,b,D$ are positive short
integers, $a,D$ are odd, and $b$ is even.  This remains an unproved
closure hypothesis, even though only two conjugate factors occur.

To see its strength, write an arbitrary $\mathsf{GapP}$ value as
$T=f-h$, pad both counting machines to $2^t$ potential witnesses, and
let $A(w),B(w)$ be their acceptance predicates.  The $L=2^{t+1}$ signs
$2A(w)-1$ and $1-2B(w)$ have sum $2T$.  For $Q\ge4$ a short power of
two, choose a dyadic denominator $q=2^s$ and a positive odd $p$ with
\[
 \left|p/q-\tan\frac{\pi}{2Q}\right|\le2/q,
 \qquad q\ge16Q,
\]
and put $a=q^2-p^2$, $b=2pq$, $D=q^2+p^2$.
Then $\zeta=(a+ib)/D$ is exactly on the unit circle and its argument
differs from $\pi/Q$ by at most $4/q$.  If $|T|\le2^u$, taking
$s\ge u+h+10$ gives
\[
 \left|D^{-L}\operatorname{Re}\prod_w(a+i\epsilon(x,w)b)
       -\cos(2\pi T/Q)\right|<2^{-h}.
\]
All parameters are computable to the required polynomial precision,
and the common denominator $D^L$ is independent of the value of $T$.
Uniform padding makes it common across all harmonic and transcript
indices.  The preceding Fourier and transcript construction therefore
applies.  Its final denominator has the known form
$(2BD^L)^{3q_0}$, where $B$ is the dyadic Fourier denominator and $q_0$
is the oracle-query bound.  Approximating $\log_2D$ to
$t+\lceil\log_2q_0\rceil+O(1)$ fractional bits supplies a constant-width
window for its logarithm.  Thus it supplies the same polynomial-width
magnitude window for the final threshold numerator, and again implies
$\mathsf{CH}=\mathsf{PP}^{\mathsf{PP}}
=\mathsf{MP}=\mathsf P^{\mathsf{PP}}$.

This product principle should not be inferred from the established
two-adic closure.  Since $q$ is divisible by $4$, the negatives of the
two factors lie in $1+4\mathbb Z[i]$; the even number of factors removes
their common sign.  Finite principal-unit logarithm and exponential
polynomials work modulo $2^k$ for polynomial $k$, but their low
two-adic residue does not reveal the archimedean quotient by $D^L$.

In fact, ordinary integer products suffice for this conditional collapse.
The magnitude-window lemma can use only moduli $4r^2+1$.  Put
$k_j=2(N+1+j)$ and $m_j=k_j^2+1$.  The exact partial-fraction identity
\[
 c_N:=\sum_{j=0}^N(-1)^j\binom Nj\frac{k_j}{m_j}
 =\frac{C_N}{2}\operatorname{Re}
   \prod_{k=N+1}^{2N+1}\left(1+\frac{i}{2k}\right)^{-1}
\]
satisfies $C_N/4\le c_N\le C_N/2$.  Indeed, each factor has modulus
at most one and differs from one by at most $1/(2k)$, so the product
differs from one by at most $1/2$.  If $R_j$ represents $Z$ modulo $m_j$
and $W_j$ represents $\binom Nj$ there, then
\[
 \sum_{j=0}^N\frac{(-1)^j k_jW_jR_j}{m_j}
       \in\mathbb Z+c_N Z.
\]
For a supplied window $2^e\le|Z|<2^{e+w}$, choose $N=H+16$ and
$J=2N-e-w-5$ as before.  Then
\[
 \frac{2^{-w-7}}{2N+1}\le|2^Jc_N Z|<\frac1{64}.
\]
The same fixed-point raw sum, with two additional guard bits, has error
at most $2^{-w-12}/(2N+1)$ and therefore gives one middle-bit test.
Thus representatives at these split moduli suffice for the
magnitude-window conclusion.

At $m=k^2+1$ with even $k$, evaluation at $i=\pm k$ is valid even
when $m$ is composite.  Consequently
\[
 \operatorname{Re}\prod_w(a+i\epsilon(w)b)
 \equiv\frac12\left[
     \prod_w(a+\epsilon(w)bk)+
     \prod_w(a-\epsilon(w)bk)\right]\pmod m.
\]
The inverse of $2$ is polynomial-time computable.  Hence a uniform
raw $\mathsf{GapP}$ product rule for exponentially many short
polynomial-time computable integer factors would supply precisely the
representatives used above.  Only two factor values, and only moduli
$4r^2+1$, are required.  The factors can even be made strictly positive
by replacing each by its residue modulo $m$ plus $m$, with $m$ included
in the factor-family input.  This remains a conditional consequence:
the ordinary product rule is not established, and its known two-adic
special case does not cover these odd split moduli.

\subsection{A limit on deterministic phase samples}
\label{push12:phase-samples}

There is a simple obstacle to replacing the random magnitude test by
one fixed short list of universal multipliers. Fix integers \(T>B+2\).
Let \(K_1,\ldots,K_q\) be arbitrary integers and let nonnegative
weights \(w_i\) sum to one. Suppose
\[
 \sum_iw_i\cos(2\pi K_i\delta)\le\tfrac12
 \quad\text{for every }\delta\in[4/2^T,2^{-B}).
\]
Then
\begin{equation}\label{push12:sampler-bound}
 2^{T-B-2}\le6^q,
 \qquad q\ge\frac{T-B-2}{\log_2 6}.
\end{equation}
To see this, consider the phase vectors
\((4jK_i/2^T\bmod1)_{i=1}^q\), for
\(0\le j<2^{T-B-2}\). Partition each coordinate circle into six
half-open intervals of length \(1/6\). If two vectors occupy the
same box, their difference supplies an allowed nonzero \(\delta\)
whose every cosine is strictly greater than \(1/2\), a contradiction.

In our magnitude test \(B=32\), while \(T\) can have exponential
value. Thus a universal deterministic sampler requires exponentially
many support points. This does not rule out a sampler depending on the
input SLP: the collision index in the proof need not have a short
arithmetic description. Nor does it concern a random list that succeeds
with high probability for each fixed displacement. It excludes only
the proposed replacement by a fixed polynomial number of multipliers
valid on the entire phase interval.

\section{Remaining obstacles and unsuccessful routes}
The real and integer arguments share a way of recovering signs, but they
leave different arithmetic objects to be evaluated. For an SLP, a short
modular value is already computable in polynomial time. For the real
construction, even after the ordered-rank refinement, we have two
nested products over formal branches. A product algorithm retrieves a
count and then exponentiates by its value. We do not know how to absorb
this nonlinear operation into the same counting function as the later
sum.

There is a restricted way to remove the coefficient layer. If the factors
\(c_i\in\operatorname{GR}(p^a,e)\) of \(\prod_{i<D}(u+c_i)\)
occupy only polynomially many known residue classes modulo \(p\),
every specified coefficient is in \(\FP^{\SharpP}\).
Write \(c_i=\lambda_{v(i)}+p\beta_i\), retrieve the class
multiplicities \(N_v\), put \(n_v(S)=|\{i\in S:v(i)=v\}|\), and expand
\[
 \prod_i(u+c_i)=\sum_{j<a}p^j
  \sum_{|S|=j}\left(\prod_{i\in S}\beta_i\right)
  \prod_v(u+\lambda_v)^{N_v-n_v(S)}.
\]
Each surviving subset is polynomially describable, and its remaining
factor is a short SLP. Apply Lemma~\ref{lem:ring-coefficients} inside
the raw subset count, after retrieving the multiplicities.
If this residue-class condition held at every node, it would give
another way to remove the high-derivative layer. General Fourier
nodes do not preserve the condition. The ordered-rank proof in
Section~\ref{push6:depth} reaches \(\Clevel4\) without it.
More generally, the same expansion works when the reference product
\(\prod_i(u+\lambda_i)\) has a short SLP and the \(\lambda_i\)
are individually computable: divide out the polynomially many deleted
factors by a truncated reciprocal at infinity. This includes some
families with exponentially many residue classes, but no such
description is established for the general homotopy factors.

This identifies a concrete route to a stronger result. Suppose the
integer family \(Z_h\) in Section~\ref{push6:depth}, uniformly over
its random seeds, had \(\GapP^{\PH}\) representatives modulo every
short integer. Lemma~\ref{push4:sign} would decide each sign in
\(\BPP^{\PP}\); random isolation and repetition would then give
\(\ETR\subseteq\BPP^{\PP}\subseteq\Clevel2\).
The first-event reconstruction would also give a containment in
\(\PH^{\PP}\) after simulating the outer random choices by
relativized polynomial-hierarchy quantifiers. A reduction of those integers
to coefficients of polynomial-size SLPs would also suffice. The present
norm construction gives neither representation. This is a separate
problem from reducing a constant in a counting-level ledger.

For integer SLPs, the remaining distinction is between sampling a phase
and averaging it exactly. Each sampled phase involves nonlinear
postprocessing of a retrieved counting value. Putting that
postprocessing inside another exact count would raise the oracle level.
The present argument therefore gives \(\BPP^{\PP}\), and does not
derandomize it to \(\Pclass^{\PP}\) or replace it by a randomized
reduction to a single PP predicate.

There is a related source-side question. The original short-conjunction argument
only proves that some root can be isolated by some list. It does not
count roots, and different lists may select the same root or overlapping
subsets. The small decision bound therefore does not automatically give
the same bound for counting real solutions or connected components.
The ordered-rank refinement labels every real root, and
Theorem~\ref{app:gcr} now supplies the additional signed-Morse
experiment that gives total point counting in \(\FP^{\Clevel4}\).
This does not by itself count connected components at the same level.

The accompanying supplementary material includes the earlier general
resultant construction, unsuccessful attempts to remove these remaining
operations, and finite exact checks. The checks cover the rational
filter, nonreal common zeros, small-characteristic homotopies, Galois-ring
coefficient identities, and signed noncanonical beta representatives.
Additional checks cover the small-bias generator, its random-access
algorithm, the signed phase sums, and the scale search under arbitrary
labels in its transition interval.
They test the individual identities and boundary cases; the uniform
complexity conclusions depend on the arguments in the paper.


\begin{thebibliography}{99}
\small\raggedright

\bibitem{HAB02}
William Hesse, Eric Allender, and David A. Mix Barrington.
Uniform Constant-Depth Threshold Circuits for Division and Iterated
Multiplication.
\emph{Journal of Computer and System Sciences}, 65(4):695--716, 2002.
\href{https://doi.org/10.1016/S0022-0000(02)00025-9}{doi:10.1016/S0022-0000(02)00025-9}.

\bibitem{BC06}
Peter B\"urgisser and Felipe Cucker.
Counting Complexity Classes for Numeric Computations II:
Algebraic and Semialgebraic Sets.
\emph{Journal of Complexity}, 22(2):147--191, 2006.
\href{https://doi.org/10.1016/j.jco.2005.11.001}{doi:10.1016/j.jco.2005.11.001}.
Preprint: \href{https://arxiv.org/abs/cs/0312007}{arXiv:cs/0312007}.

\bibitem{BLS25}
Manuel Bodirsky, Georg Loho, and Mateusz Skomra.
Reducing Stochastic Games to Semidefinite Programming.
In \emph{52nd International Colloquium on Automata, Languages, and
Programming}, LIPIcs 334, pages 145:1--145:15, 2025.
\href{https://doi.org/10.4230/LIPIcs.ICALP.2025.145}{doi:10.4230/LIPIcs.ICALP.2025.145}.

\bibitem{CSW19}
Cassio P. de Campos, Georgios Stamoulis, and Dennis Weyland.
\newblock A structured view on weighted counting with relations to counting,
quantum computation and applications.
\newblock arXiv:1701.06386v2, 2019.

\bibitem{AGHP92}
Noga Alon, Oded Goldreich, Johan H\aa stad, and Ren\'e Peralta.
Simple constructions of almost \(k\)-wise independent random variables.
\emph{Random Structures \& Algorithms}, 3(3):289--304, 1992.

\bibitem{NN93}
Joseph Naor and Moni Naor.
Small-bias probability spaces: Efficient constructions and applications.
\emph{SIAM Journal on Computing}, 22(4):838--856, 1993.

\bibitem{AAM18}
Mikkel Abrahamsen, Anna Adamaszek, and Tillmann Miltzow.
The Art Gallery Problem is \(\exists\R\)-complete.
\emph{Journal of the ACM}, 69(1), 2022.
\href{https://doi.org/10.1145/3486220}{doi:10.1145/3486220}.
Preprint: \href{https://arxiv.org/abs/1704.06969}{arXiv:1704.06969}.

\bibitem{Canny88}
John Canny.
Some algebraic and geometric computations in PSPACE.
In \emph{Proceedings of the Twentieth Annual ACM Symposium on Theory
of Computing}, pages 460--469, 1988.
\href{https://doi.org/10.1145/62212.62257}{doi:10.1145/62212.62257}.

\bibitem{JP10}
Gabriela Jeronimo and Daniel Perrucci.
On the minimum of a positive polynomial over the standard simplex.
\emph{Journal of Symbolic Computation}, 45(4):434--442, 2010.
\href{https://doi.org/10.1016/j.jsc.2010.01.001}{doi:10.1016/j.jsc.2010.01.001}.

\bibitem{AGS26}
Robert Andrews, Abhibhav Garg, and \'Eric Schost.
Hilbert's Nullstellensatz is in the Counting Hierarchy.
Preprint, 2026.
\href{https://arxiv.org/abs/2602.17904}{arXiv:2602.17904}.

\bibitem{ABKPM09}
Eric Allender, Peter B\"urgisser, Johan Kjeldgaard-Pedersen, and
Peter Bro Miltersen.
On the Complexity of Numerical Analysis.
\emph{SIAM Journal on Computing}, 38(5):1987--2006, 2009.
\href{https://doi.org/10.1137/070697926}{doi:10.1137/070697926}.

\bibitem{KP07}
Pascal Koiran and Sylvain Perifel.
The Complexity of Two Problems on Arithmetic Circuits.
\emph{Theoretical Computer Science}, 389(1--2):172--181, 2007.
\href{https://doi.org/10.1016/j.tcs.2007.08.008}{doi:10.1016/j.tcs.2007.08.008}.

\bibitem{Toda91}
Seinosuke Toda.
PP is as Hard as the Polynomial-Time Hierarchy.
\emph{SIAM Journal on Computing}, 20(5):865--877, 1991.
\href{https://doi.org/10.1137/0220053}{doi:10.1137/0220053}.

\bibitem{TW92}
Seinosuke Toda and Osamu Watanabe.
Polynomial-time 1-Turing reductions from \(\SharpP\mathsf H\) to
\(\SharpP\).
\emph{Theoretical Computer Science}, 100(1):205--221, 1992.
\href{https://doi.org/10.1016/0304-3975(92)90369-Q}{doi:10.1016/0304-3975(92)90369-Q}.

\bibitem{GKRST95}
Frederic Green, Johannes K\"obler, Kenneth W.~Regan, Thomas Schwentick,
and Jacobo Tor\'an.
The Power of the Middle Bit of a \(\SharpP\) Function.
\emph{Journal of Computer and System Sciences}, 50(3):456--467, 1995.
\href{https://doi.org/10.1006/jcss.1995.1036}{doi:10.1006/jcss.1995.1036}.

\bibitem{BJ24}
Peter B\"urgisser and Gorav Jindal.
On the Hardness of PosSLP.
In \emph{Proceedings of the 2024 Annual ACM-SIAM Symposium on Discrete
Algorithms}, pages 1872--1886, 2024.
\href{https://doi.org/10.1137/1.9781611977912.75}{doi:10.1137/1.9781611977912.75}.

\bibitem{BDJ24}
Markus Bl\"aser, Julian D\"orfler, and Gorav Jindal.
PosSLP and Sum of Squares.
In \emph{44th IARCS Annual Conference on Foundations of Software
Technology and Theoretical Computer Science}, LIPIcs 323,
pages 13:1--13:18, 2024.
\href{https://doi.org/10.4230/LIPIcs.FSTTCS.2024.13}{doi:10.4230/LIPIcs.FSTTCS.2024.13}.

\bibitem{Lam05}
Tsit-Yuen Lam.
\emph{Introduction to Quadratic Forms over Fields}.
Graduate Studies in Mathematics, volume 67.
American Mathematical Society, 2005.

\bibitem{SJ12}
Thatchaphol Saranurak and Gorav Jindal.
Subtraction makes computing integers faster.
\href{https://arxiv.org/abs/1212.2549}{arXiv:1212.2549}, 2012.

\bibitem{AGHT19}
Yaroslav Alekseev, Dima Grigoriev, Edward A.~Hirsch, and Iddo Tzameret.
Semi-Algebraic Proofs, IPS Lower Bounds and the \(\tau\)-Conjecture: Can a Natural Number be Negative?
\href{https://arxiv.org/abs/1911.06738}{arXiv:1911.06738}, 2019.

\bibitem{BKM25}
Lorenzo Baldi, Teresa Krick, and Bernard Mourrain.
An Effective Positivstellensatz over the Rational Numbers for Finite
Semialgebraic Sets.
\href{https://arxiv.org/abs/2410.04845}{arXiv:2410.04845v2}, 2025.


\bibitem{SM15}
Michael Sagraloff and Kurt Mehlhorn.
Computing Real Roots of Real Polynomials.
\href{https://arxiv.org/abs/1308.4088}{arXiv:1308.4088v2}.

\bibitem{Voight21}
John Voight.
\emph{Quaternion Algebras}.
Graduate Texts in Mathematics, volume 288. Springer, 2021.
\href{https://doi.org/10.1007/978-3-030-56694-4}{doi:10.1007/978-3-030-56694-4}.

\bibitem{ESY13}
Kousha Etessami, Alistair Stewart, and Mihalis Yannakakis.
Polynomial Time Algorithms for Multi-Type Branching Processes and
Stochastic Context-Free Grammars.
\href{https://arxiv.org/abs/1201.2374}{arXiv:1201.2374v2}, 2013.

\bibitem{ESY12}
Kousha Etessami, Alistair Stewart, and Mihalis Yannakakis.
Polynomial Time Algorithms for Branching Markov Decision Processes and
Probabilistic Min(Max) Polynomial Bellman Equations.
\href{https://arxiv.org/abs/1202.4798}{arXiv:1202.4798v2}, 2012.


\bibitem{Selmer56}
Ernst S. Selmer.
On the Irreducibility of Certain Trinomials.
\emph{Mathematica Scandinavica}, 4:287--302, 1956.

\bibitem{Osada87}
Hiroyuki Osada.
The Galois Groups of the Polynomials \(X^n+aX^s+b\), II.
\emph{Tohoku Mathematical Journal}, 39(3):437--445, 1987.
\href{https://doi.org/10.2748/tmj/1178228289}{doi:10.2748/tmj/1178228289}.


\bibitem{BPR06}
Saugata Basu, Richard Pollack, and Marie-Fran\c{c}oise Roy.
\emph{Algorithms in Real Algebraic Geometry}.
Second edition. Springer, 2006.

\bibitem{HMMP25}
Thekla Hamm, Lucas Meijer, Tillmann Miltzow, and Subhasree Patro.
Oracle Separations for the Real Polynomial Hierarchy.
\href{https://arxiv.org/abs/2502.09279}{arXiv:2502.09279v3}, 2025.


\bibitem{Sagraloff14}
Michael Sagraloff.
A Near-Optimal Algorithm for Computing Real Roots of Sparse Polynomials.
\href{https://arxiv.org/abs/1401.6011}{arXiv:1401.6011}, 2014.

\bibitem{JindalSagraloff17}
Gorav Jindal and Michael Sagraloff.
Efficiently Computing Real Roots of Sparse Polynomials.
\href{https://arxiv.org/abs/1704.06979}{arXiv:1704.06979}, 2017.

\bibitem{CT15}
Jean-Louis Colliot-Th\'el\`ene.
Arizona Winter School lectures, 2015.
\href{https://www.imo.universite-paris-saclay.fr/~colliot/AWS30MAI2015.pdf}
{Author's lecture notes}, Sections 4.2 and 4.4.

\end{thebibliography}
\end{document}